\documentclass[aps,twocolumn,superscriptaddress]{revtex4-2}

\usepackage{amsmath}
\usepackage{siunitx}
\usepackage[inline]{enumitem}

\usepackage{xcolor}  
\usepackage{graphicx}
\usepackage{amsfonts}
\usepackage{mathtools}
\usepackage{booktabs}
\usepackage{xcolor}
\usepackage{braket}
\usepackage{MnSymbol}
\usepackage{tikz}
\usepackage{cancel}

\usepackage{amsthm}
\makeatletter
\newtheorem*{rep@theorem}{\rep@title}
\newcommand{\newreptheorem}[2]{%
\newenvironment{rep#1}[1]{%
 \def\rep@title{#2 \ref*{##1}, repeated}%
 \begin{rep@theorem}}%
 {\end{rep@theorem}}}
\makeatother

\newtheorem{lemma}{Lemma}
\newreptheorem{lemma}{Lemma}

\newtheorem{problem}{Problem}

\newtheorem{proposition}{Proposition}
\newreptheorem{proposition}{Proposition}
\newtheorem{corollary}{Corollary}[proposition]
\newreptheorem{corollary}{Corollary}

\usetikzlibrary{decorations}

\pgfdeclaredecoration{ignore}{final}
{
\state{final}{}
}

\pgfdeclaremetadecoration{middle}{initial}{
    \state{initial}[
        width={(\pgfmetadecoratedpathlength - \the\pgfdecorationsegmentlength)/2},
        next state=middle
    ]
    {\decoration{moveto}}

    \state{middle}[
        width={\the\pgfdecorationsegmentlength},
        next state=final
    ]
    {\decoration{curveto}}

    \state{final}
    {\decoration{ignore}}
}

\tikzset{middle segment/.style={decoration={middle},decorate, segment length=#1}}

\def\norm#1{ {|\hspace{-.015in}|#1|\hspace{-.015in}|} }
\def\bignorm#1{ {\left \Vert #1 \right \Vert} }
\def\abs#1{ {|#1|}}
\def\bigabs#1{ {\left|#1 \right|}}
\newcommand{\tr}[2]{\textnormal{Tr}_{#1}(#2)}
\newcommand{\bigtr}[2]{\textnormal{Tr}_{#1}\bigg(#2\bigg)}
\newcommand{\expect}[1]{\langle #1 \rangle}
\newcommand{\vecket}[1]{|{#1} \rrangle}
\newcommand{\bigket}[1]{\left|{#1} \right\rangle}
\newcommand{\bigbra}[1]{\left\langle{#1} \right|}
\newcommand{\vecbra}[1]{\llangle {#1}|}
\newcommand\numeq[1]%
  {\stackrel{\scriptscriptstyle(\mkern-1.5mu#1\mkern-1.5mu)}{=}}
\newcommand\numleq[1]%
  {\stackrel{\scriptscriptstyle(\mkern-1.5mu#1\mkern-1.5mu)}{\leq}}
\newcommand\numgeq[1]%
  {\stackrel{\scriptscriptstyle(\mkern-1.5mu#1\mkern-1.5mu)}{\geq}}

\graphicspath{{Images/}}

\begin{document}


\title{Accuracy guarantees and quantum advantage in analogue open quantum simulation with and without noise}
\author{Vikram Kashyap}
\affiliation{Department of Electrical and Computer Engineering, University of Washington - 98195, USA}
\author{Georgios Styliaris}
\affiliation{Max-Planck-Institut f\"ur Quantenoptik, Garching bei M\"unchen - 85748, Germany}
\affiliation{Munich Center for Quantum Science and Technology (MCQST), Schellingstraße 4, D-80799 München, Germany}
\author{Sara Mouradian}
\affiliation{Department of Electrical and Computer Engineering, University of Washington - 98195, USA}
\author{J.~Ignacio Cirac}
\affiliation{Max-Planck-Institut f\"ur Quantenoptik, Garching bei M\"unchen - 85748, Germany}
\affiliation{Munich Center for Quantum Science and Technology (MCQST), Schellingstraße 4, D-80799 München, Germany}
\author{Rahul Trivedi}
\affiliation{Max-Planck-Institut f\"ur Quantenoptik, Garching bei M\"unchen - 85748, Germany}
\affiliation{Munich Center for Quantum Science and Technology (MCQST), Schellingstraße 4, D-80799 München, Germany}
\affiliation{Department of Electrical and Computer Engineering, University of Washington - 98195, USA}

\newcolumntype{L}[1]{>{\raggedright\let\newline\\\arraybackslash}p{#1}}





\date{\today}

\begin{abstract}
    Many-body open quantum systems, described by Lindbladian master equations, are a rich class of physical models that display complex equilibrium and out-of-equilibrium phenomena which remain to be understood. In this paper, we theoretically analyze noisy analogue quantum simulation of geometrically local open quantum systems and provide evidence that this problem is both hard to simulate on classical computers and could be approximately solved on near-term quantum devices. First, given a noiseless quantum simulator, we show that the dynamics of local observables and the fixed-point expectation values of rapidly-mixing local observables in geometrically local Lindbladians can be obtained to a precision of $\varepsilon$ in time that is $\text{poly}(\varepsilon^{-1})$ and uniform in system size. Furthermore,  we establish that the quantum simulator would provide a superpolynomial advantage, in run-time scaling with respect to the target precision and either the evolution time (when simulating dynamics) or the Lindbladian's decay rate (when simulating fixed-points), over any classical algorithm for these problems, assuming BQP $\neq$ BPP. We then consider the presence of noise in the quantum simulator in the form of additional geometrically-local Linbdladian terms. We show that the simulation tasks considered in this paper are stable to errors, i.e.~they can be solved to a noise-limited, but system-size independent, precision. Finally, we establish that, assuming BQP $\neq$ BPP, there are stable geometrically local Lindbladian simulation problems such that as the noise rate on the simulator is reduced, classical algorithms must take time superpolynomially longer in the inverse noise rate to attain the same precision as the analog quantum simulator. 
\end{abstract}
\maketitle


\section{Introduction}
An extensive body of results suggest that, with a fault-tolerant quantum computer, several many-body problems relating to both dynamics and equlibrium properties can be efficiently simulated~\cite{bauer2023quantum, bauer2023quantum2}. However, despite the recent experimental demonstrations of quantum error correction to build logical qubits~\cite{sivak2023real,livingston2022experimental, yao2012experimental, bluvstein2024logical,Ryan-Anderson2021Dec}, building a large-scale fault-tolerant quantum computer still remains a massive technological challenge. At the same time, there has been increasing interest in using available quantum devices without error correction to obtain approximate solutions to quantum many-body problems. A major effort in this direction is to use quantum devices as ``analogue quantum simulators", where the target many-body Hamiltonian is configured on a (semi-) programmable quantum device, and due to its own physics the quantum device simulates the target many-body problem \cite{daley2022practical}.

Most of the focus of analogue quantum simulation has been confined to simulating closed many-body systems i.e.~many-body problems that are specified by a Hamiltonian. The goal is typically to measure an intensive observable in either a non-equilibrium state associated with the Hamiltonian (such as the state generated by time-evolving an initial product state under the Hamiltonian), or an associated equilibrium state (such as its Gibb's state or ground state). The mappings of a variety of physically interesting Hamiltonians arising in solid-state physics, condensed matter physics and high energy physics onto existing quantum computing platforms such as trapped ions, superconducting qubits, or atomic systems have been extensively developed~\cite{daley2022practical, buluta_nori_2009, bauer2023quantum, Houck2012Apr, Schafer2020Aug, georgescu_ashhab_nori_2014}.

However, closed many-body systems are just a special case of open many-body systems, which can be modelled by Lindblad master equations. Dynamics and equilibrium properties of open many-body systems have a number of physical effects -- such as driven-dissipative phase transitions \cite{kessler2012dissipative, horstmann2013noise}, super and sub-radiance \cite{benedict2018super, sierra2022dicke, masson2022universality, asenjo2017exponential, rubies2022superradiance}, and exceptional spectral points \cite{ding2022non, kawabata2019classification, graefe2008non, alvarez2018non, minganti2019quantum} -- that are qualitatively different from closed systems and make them independently interesting target problems for quantum simulation. Furthermore, there is evidence to suggest that open quantum system problems could inherently be more robust to noise than closed quantum system problems \cite{trivedi2024quantum}. In particular, if a quantum simulator is implementing only a Hamiltonian, then well-known no-go results indicate that due to entropy accumulation in the presence of a constant rate of depolarizing noise, the simulator state converges exponentially with simulation time to the maximally mixed state \cite{aharonov1996limitations, stilck2021limitations, de2023limitations, gonzalez2022error, mishra2023classically}. However, these no-go theorems are not applicable to a quantum simulator implementing a Lindbladian, since the simulator dynamics inherently could counter the accumulation of entropy due to external noise.

\begin{figure}
\includegraphics[width=.5\textwidth]{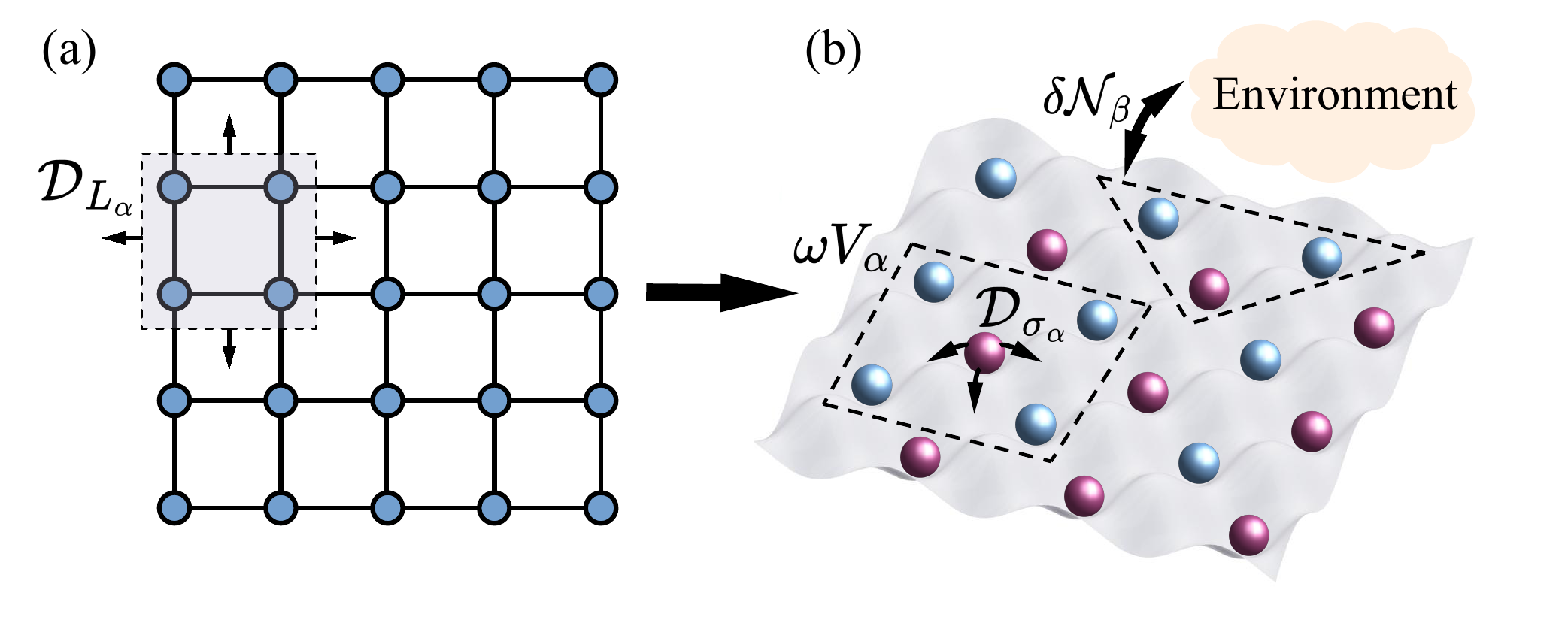}
\caption{\label{fig:setup}Diagrammatic depiction of the analogue quantum simulator. (a) The target Lindbladian acting on a lattice of sites is the sum of dissipator Linbdladians $\mathcal{D}_{L_\alpha}$, each corresponding to a geometrically local jump operator $L_\alpha$. (b) The analog quantum simulator approximates each Lindbladian $\mathcal{D}_{L_\alpha}$ by implementing a Hamiltonian $\omega V_\alpha$ that couples the sites (blue) on the lattice region to an ancilla qubit (purple) that undergoes dissipation resetting it to $\ket{0}$. Noise due to unwanted coupling to an environment is modeled as additional local Lindbladians $\mathcal{N}_\beta$ of strength $\delta$ each acting on regions that may include both simulated lattice sites and ancillae. This figure is inspired by the schematic in  Ref.~\cite{trivedi2024quantum}}.
\end{figure}

Several proposals for implementing the digital quantum simulation of Markovian \cite{ pocrnic2023quantum, cattaneo2023quantum, suri2023two, cleve2016efficient, childs2016efficient, li2022simulating} and non-Markovian open quantum systems \cite{trivedi2022description, li2023succinct} on fault-tolerant quantum computers have been recently developed. Furthermore, Ref.~\cite{guimaraes2023noise} has recently proposed and experimentally demonstrated using noise inherent in near-term devices to design digital circuits that can simulate several resctricted, but physically motivated, classes of Lindbladians and without using any ancilla qubits. 

Analogue quantum simulation of general Markovian open quantum systems has also been proposed  \cite{pastawski_clemente_cirac_2011, zanardi_marshall_venuti_2016, ding2024simulating}, although it has not received as much attention as its closed system counterpart. A relatively simple way of simulating a given Lindbladian is to use a set of ancillae, each continuously reset to its ground state (i.e.~experiencing an amplitude damping channel), and couple the ancillae to the system via a time-independent Hamiltonian that depends on the jump operators in the target Lindbladian. By tuning the strength of the Hamiltonian in comparison to the rate of rate of amplitude damping, an approximation of the target Lindbladian dynamics can be implemented on this system [Fig.~\ref{fig:setup}].  Since implementing both amplitude-damping dissipation and a Hamiltonian interaction is possible on many experimental platforms, this method is particularly suitable for analogue quantum simulation of open quantum systems. The method we consider here has been used to implement dissipative quantum memories \cite{pastawski_clemente_cirac_2011} and been identified as a strategy for Lindbladian simulation \cite{zanardi_marshall_venuti_2016}. However, it remains unclear if, and under what circumstances, this experimentally simple analogue quantum simulation protocol can provide a good approximation to many-body Lindbladians, especially in the presence of errors in the quantum simulator.
 
In this paper, we analyze the analogue quantum simulation of open quantum systems, both in the noiseless and noisy settings, and rigorously prove accuracy guarantees on the quantum simulation protocol. We first consider the noiseless setting, where a user-specified Hamiltonian and an amplitude damping channel can be perfectly implemented on the analog quantum simulator. We establish that a generic $k-$local Lindbladian implemented on $n$ qudits (of any dimension) for time $t$  can be simulated with an overhead in simulation time which is at most polynomial in the number of qudits $n$ and the evolution time $t$. We then focus on the case of a spatially local Lindbladian on a lattice and show that, as a consequence of Lieb-Robinson bounds \cite{poulin2010lieb, barthel_kliesch_2012}, local observables on an analogue quantum simulator can be measured in a simulation time that scales polynomially with $t$ but is uniform in the system size $n$. For geometrically local Lindbladians, we also consider the problem of measuring either long-time dynamics of the local observable, or its fixed point expectation value. We show that, under the assumption of rapid mixing \cite{cubitt2015, lucia2015rapid}, this can again be done on an analogue quantum simulator in a simulation time that is uniform in the system size $n$ (and additionally independent of time $t$ in the case of long time dynamics). Our main theoretical technical contribution that enables us to establish these results is to develop a rigorous analysis of adiabatic elimination of the damped ancillae which explicitly leverages the finite Lieb-Robinson velocity in the target system. 

Simultaneously, we build upon the quantum-circuit to Lindbladian mapping developed in Ref.~\cite{verstraete2009quantum} and show that unless Bounded Quantum Polynomial Time (BQP) = Bounded Probabilistic Polynomial Time (BPP), we do not expect there to be a classical algorithm that has a run-time which is polynomial in both the inverse precision $\varepsilon^{-1}$ and the evolution time $t$ (when simulating dynamics) or the inverse of the Lindbladian's decay rate $\gamma^{-1}$ (when simulating fixed-points), and consequently we expect there to be a superpolynomial separation between the run-time of the quantum simulator and the best possible classical algorithm. This holds true even when restricting ourselves to the physically relevant setting of time-independent and spatially local 2D Lindbladians. To establish this result for 2D local Lindbladians, we adapt the circuit-to-2D-Hamiltonian-ground-state mapping developed in Ref.~\cite{aharonov2008adiabatic} to the dissipative setting and theoretically establish that the resulting master equation mixes in a time that scales polynomially in the system size.

Finally, we consider the presence of errors or noise in the quantum simulator, which can either be coherent configuration errors or incoherent errors due to interaction with an external environment. Since typically every qudit on the quantum simulator would be noisy, there are extensively many errors on the simulator. In the worst case, these errors could accumulate and result in the observable being measured on an $n$-qudit quantum simulator incurring an error $\sim \delta \times n$. Consequently, as the system size is increased (e.g.~to simulate the thermodynamic limit of a many-body observable), any constant error rate $\delta$ would cause the simulated observable to diverge from its true value. However, we establish that when computing dynamics of local observables  or the fixed point expected value of rapidly mixing local observables, a noisy quantum simulator incurs a \emph{system-size independent} error---the precision in the simulated observable is only limited by noise rate and scales subpolynomially with the noise rate, reaching perfect precision for vanishing noise rate. Thus, these simulation task are \emph{stable} in the sense defined in Ref.~\cite{trivedi2024quantum} and can be solved using near-term experimental platforms to a hardware-error limited precision. Finally, by combining these stability results with the circuit-to-2D-geometrically-local-Lindbladian encoding, we also establish that to solve these problems to the same precision as achieved by the noisy quantum device, any classical algorithm would require time that scales superpolynomially with the hardware error rate.

\section{Summary of Results}
\subsection{Theoretical results}\label{sec:theory_result}
Throughout this paper, we will concern ourselves with simulating a Lindladian. Any Lindbladian over a Hilbert space $\mathcal{H}$ ($\text{dim}(\mathcal{H}) < \infty$) can be specified by a set of jump operators $L_1, L_2 \dots L_M$, and a Hamiltonian $H_\text{sys}(t)$. The corresponding Lindbladian master equation is 
\begin{subequations}\label{eq:lindbladian}
\begin{align}
    \frac{d}{dt}\rho(t) = -i[H_\text{sys}, \rho(t)] + \sum_{\alpha  = 1}^M \mathcal{D}_{L_\alpha}(\rho(t)),
\end{align}
where $\rho(t)$ is the time-dependent density matrix of the physical system under consideration and, for any operator $A$, the dissipator corresponding to $A$ is the superoperator
\begin{align}
\mathcal{D}_A(X) = AXA^\dagger - \frac{1}{2}\{X, A^\dagger A\}.
\end{align}
\end{subequations}
We assume that the time-unit in the physical problem is normalized such that $\norm{L_\alpha} \leq 1$ for all $\alpha$. For typical problems in physics the qudits are laid out in a lattice, the jump operators are local operators that are supported on a subset of neighbouring qudits, and the Hamiltonian is also geometrically local i.e.~a sum of such local operators. 

If a quantum simulator could configure any specified dissipator, there would be a direct mapping between the target master equation and the quantum simulator. However, in most available experimental systems the quantum simulator can controllably implement a (family of) Hamiltonians and simple single qubit dissipators, necessitating an experimentally simple way of mapping a Lindbladian to a quantum simulator.

This problem has been addressed previously --- in particular, Ref.~\cite{pastawski_clemente_cirac_2011} identified that a dissipator $\mathcal{D}_{L_\alpha}$ can be effectively implemented by using an ancilla qubit with a large amplitude damping dissipation (i.e.~a dissipator that maps $\ket{1}$ of the ancilla to $\ket{0}$) and coupling it weakly to the system with a Hamiltonian that depends on the jump operator $L_\alpha$. More specifically, to simulate the Lindbladian in Eq.~\ref{eq:lindbladian} over Hilbert space $\mathcal{H}_\mathcal{S}$, we define another Hilbert space $\mathcal{H}_\mathcal{A}=(\mathbb{C}^2)^{\otimes M}$ of $M$ ancilla qubits and  consider the following Lindbladian over $\mathcal{H}_\mathcal{S}\otimes \mathcal{H}_\mathcal{A}$
\begin{subequations}\label{eq:qsim_basic}
    \begin{align}
        \frac{d}{dt}\rho_\omega(t) = \mathcal{L}_\omega (\rho_\omega(t))
    \end{align}
    with
    \begin{align}
        \mathcal{L}_\omega(X) = -i\omega^2 [H_\text{sys}, X] -i\omega\bigg[\sum_{\alpha = 1}^M V_\alpha, X\bigg]+ 4\sum_{\alpha = 1}^M \mathcal{D}_{\sigma_\alpha}(X),
    \end{align}
where $\sigma_\alpha = \ket{0_\alpha}\!\bra{1_\alpha}$ is the lowering operator on the $\alpha^\text{th}$ ancilla qubit, 
\begin{align}
V_\alpha = L_\alpha \sigma_\alpha^\dagger + L_\alpha^\dagger \sigma_\alpha
\end{align}
\end{subequations}
is the interaction Hamiltonian between the $\alpha^\text{th}$ ancilla and the system, and the dimentionless parameter $\omega$ controls the strength of the interaction between the ancillae and the system. 

The parameter $\omega$ also controls the extent to which the reduced state of the system, $\text{Tr}_{\mathcal{A}}(\rho_\omega(t))$, is a faithful approximation of the target state $\rho(t)$ satisfying Eq.~\ref{eq:lindbladian}a. To see this physically, we treat the ancillae as an environment for the system qudits. The limit of small $\omega$ then corresponds to the limit of weak system-environment interaction in which we expect the ancillae to behave like a bath describable by the Born-Markov approximation \cite{carmichael2009open}. Equivalently, we expect $\text{Tr}_\mathcal{A}(\rho_\omega(t))$ to satisfy a Markovian master equation. Consequently, as we will rigorously establish in this paper with a concrete bound for finite $\omega$, it is true that a smaller $\omega$ guarantees a more accurate simulation of the target quantum dynamics. In particular,
\begin{align}\label{eq:quantum_simulator_qualitative_correctness}
\lim_{\omega \to 0}\text{Tr}_\mathcal{A}\bigg(\rho_\omega\bigg(\frac{t}{\omega^2}\bigg)\bigg) = \rho(t).
\end{align}
It is important to note that Eq.~\ref{eq:quantum_simulator_qualitative_correctness} suggests that even for a fixed time $t$, the simulation time needed on the quantum simulator increases on decreasing $\omega$, i.e.~the greater the accuracy required in simulating the target master equation, the slower the quantum simulator will be.

While this provides a qualitative reason why a quantum simulator can reproduce the dynamics of the target system, it leaves open several theoretical questions that are important from the point of view of current experiments. \emph{First}, can we provide concrete run-time bounds on the quantum simulator and gauge its performance for various physically relevant models? And is it true that an analogue open quantum system simulator can potentially provide a quantum advantage over classical algorithms for physically relevant  models? \emph{Second}, does the presence of errors (coherent or incoherent) catastrophically impact the performance of the analogue open quantum simulator, and if not, is the analogue open quantum simulator stable to errors for any problems which are both physically interesting \emph{and} classically hard? This question is even more important in view of Eq.~\ref{eq:quantum_simulator_qualitative_correctness} which suggests that the quantum simulation time must be increased to increase the accuracy of the simulation, and hence it is possible that even a small but constant error rate could accumulate over this long simulation time and result in the simulator's output being completely incorrect. It is therefore crucial to understand if and under what circumstances this quantum simulation protocol could be stable to errors in the quantum simulator and trusted to produce a faithful approximation of the target problem. In this paper, we provide answers to all of these questions.

\subsubsection{Noiseless setting}
\emph{Dynamics}. We first analyze the noiseless quantum simulator and provide concrete run-time bounds for the simulation tasks considered in this paper. We will also provide complexity-theoretic evidence of quantum advantage, even for the physically motivated restricted settings that we consider in this paper. We begin by considering the general Lindblad master equation (Eq.~\ref{eq:lindbladian}) as the simulation target. Our goal is to obtain the full density matrix at time $t$ to a given precision $\varepsilon$ in trace-norm. Assuming that $\norm{L_\alpha} \leq 1$ for all $\alpha$, we establish the following simulation time bound on the analogue open quantum simulator.
\begin{proposition}[Noiseless simulator runtime]
\label{prop:general_lindbladian}
    With all the ancillae initialized to the state $\ket{0}$ and for any time $t > 0$, $\norm{\rho(t) - \textnormal{Tr}_\mathcal{A}(\rho_\omega(t/\omega^2))}_1 \leq \varepsilon$ can be obtained with
\[
\omega = \Theta\left(\frac{\varepsilon^{1/2}}{\big(M + 4M \norm{H_\textnormal{sys}} t + 4M^2 t\big)^{1/2}} \right),
\]
which implies that the required run-time on the simulator $t_\textnormal{sim} = t/\omega^2$ scales as
\[
t_\textnormal{sim} =  \Theta\left(\frac{Mt + 4M \norm{H_\textnormal{sys}}t^2 + 4M^2 t^2}{\varepsilon} \right).
\]
\end{proposition}
\noindent Importantly, in this general setting and if the required target is the full quantum state to a desired precision, then the analogue quantum simulator must run for a time that increases with the number of jump operators $M$, which in a typical $n$-qudit problem scales as $\textnormal{poly}(n)$. Furthermore, we expect this scaling to be tight as far as a target precision in the full many-body state is concerned since we incur an extensive error in using damped ancillae to mimic a Markovian environment.  Finally, we remark that our bound on the quantum simulation time exhibits a quadratic slowdown for the analogue quantum simulator with respect to the target dynamics i.e.~$t_\text{sim} \sim t^2$ --- this is expected and consistent with lower bounds on digital quantum simulation time for Lindbladians previously described in Ref.~\cite{cleve2016efficient}.

The key technical tools that we develop in this paper that enabled us to prove this proposition are (1) a procedure to rigorously adiabatically eliminate the ancillae while accounting for the adiabatic elimination errors in terms of the excitation number in the ancillae together with (2) an excitation number bound on the ancillae which decreases sufficiently fast with $\omega$. This analysis approach also forms the basis of the proof of the better run-time bounds for spatially local Lindbladians provided in the propositions that follow.

We remark that, to the best of our knowledge, such an analogue simulation protocol has been analyzed previously only in Refs.~\cite{zanardi_marshall_venuti_2016, zanardi2014coherent} and \cite{pastawski_clemente_cirac_2011}. In Refs.~\cite{zanardi_marshall_venuti_2016, zanardi2014coherent}, the bounds provided effectively do not explicitly evaluate the dependence of the simulator time $t_\text{sim}$ on the target time $t$ compared to the result in Proposition~\ref{prop:general_lindbladian}. Furthermore, our proof of this proposition also significantly differs from Refs.~\cite{zanardi_marshall_venuti_2016, zanardi2014coherent} and is more amenable to the analysis of geometrically local models considered in the remainder of this paper. Additionally, in Ref.~\cite{pastawski_clemente_cirac_2011}, the authors also considered an encoding similar to the one considered in Proposition~\ref{prop:general_lindbladian}, but their error analysis followed a significantly different approach and was restricted to a single jump operator.
\begin{figure*}
\includegraphics[width=1\textwidth]{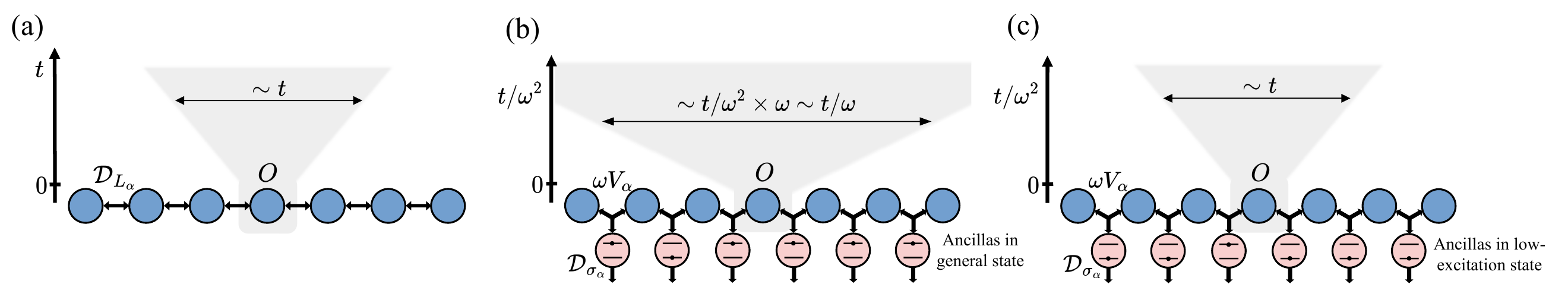}
\caption{\label{fig:lightcone}(a) In a system evolving under the target Lindbladian $\mathcal{L}$, the support of a local observable $O$ spreads in the Heisenberg picture in a ``light-cone" with a Lieb-Robinson velocity scaling as $\sim\norm{\mathcal{L}_\alpha}$. (b) In the quantum simulator we analyze, the Lieb-Robinson velocity scales as $\sim \omega\norm{V_\alpha} \sim \omega$; however the simulator also runs slower by a factor of $1/\omega^2$, leading to a  diverging Lieb-Robinson light-cone as we take $\omega \to 0$. (c) Our analysis shows that the effective light-cone is no longer diverging if we account for the ancillae (red) being heavily damped and thus in a low-excitation state.}
\end{figure*}

Next, we consider geometrically local models, which frequently appear in the study of many-body open quantum systems in physics. We make the assumption that the system $\mathcal{S}$ consists of a set of qudits arranged on a lattice $\mathbb{Z}^d$ and the target Lindbladian $\mathcal{L}$ is assumed to be of the form
\[
\mathcal{L} = \sum_{\alpha} \mathcal{L}_\alpha,
\]
where $\mathcal{L}_\alpha$ is a Lindbladian corresponding to a jump operator $L_\alpha$, with $\norm{L_\alpha} \leq 1$, and Hamiltonian $h_\alpha$ with $\norm{h_\alpha} \leq 1$. Both $L_\alpha$ and $h_\alpha$ are assumed to be supported on qudits in $S_\alpha$ which has a diameter $a$, i.e. $\text{diam}(S_\alpha) = \max_{x, y \in S_\alpha} d({x, y}) \leq a$, where $d(x, y)$ is the Manhattan distance between $x, y$.
Such models have a finite velocity at which correlations can propagate across the lattice, formalized by the well-known Lieb-Robinson bounds \cite{poulin2010lieb}.

In the following we analyze the restricted problem of measuring local observables in the dynamics and fixed points of such Lindbladians. The physical motivation behind this choice stems from the fact that in many-body physics settings it is typically of interest to only measure intensive order parameters that are expressible as single (or weighted sums of) local observables. Furthermore, in the settings that we consider, the local observables have a well-defined thermodynamic limit --- the system size $n$ can be taken to infinity and the local observable will converge. The computational problem of interest is to compute this thermodynamic limit of the local observable. Importantly, we will rigorously establish that the run-time of the analog simulator needed to simulate these observables to a given precision will be independent of the system size --- this independence from the system size $n$ is also consistent with the fact that the local observables converge to a thermodynamic limit as $n \to \infty$. 

Consider first the problem of using the analogue quantum simulator to compute a local observable, at time $t$, for a spatially local Lindbladian. Given that a spatially local Lindbladian has a finite Lieb-Robinson velocity [Fig.~\ref{fig:lightcone}(a)]\cite{poulin2010lieb}, and as long as $t$ is independent of the system size $n$, local observables have a well defined thermodynamic limit. Furthermore, also due to the Lieb-Robinson bounds, it could be expected that the choice of $\omega$ needed to obtain a good approximation for a local observable, instead of depending on the total number of jump operators in the Lindbladians as suggested by Proposition~\ref{prop:general_lindbladian}, should depend on the number jump operators within the light-cone until time $t$. However, a closer look at the quantum simulator Lindbladian (Eq.~\ref{eq:qsim_basic}) reveals that its Lieb-Robinson velocity, being dependent on the norm of the interaction terms between different qudits, scales as $\omega$, while the total simulation time scales as $t / \omega^2$ [Fig.~\ref{fig:lightcone}(b)]. Consequently, it would appear that the effective light cone of a local observable on the analogue quantum simulator would scale as $\omega \times t / \omega^2 = t/\omega$, which would diverge as we take $\omega \to 0$ to increase the accuracy of the quantum simulation.

Despite this issue, in the next proposition, we show that the parameter $\omega$ needed to estimate a local observable at $t$ evolving under a spatially local Lindbladian can be chosen to be uniform in the system size (i.e.~number of jump operators). More specifically, we consider local observables (i.e.~observables that are supported on a geometrically local subset of system qudits) and establish that

\begin{proposition}[Geometrically local dynamics --- noiseless simulator runtime]\label{prop:dynamics_noiseless}
Suppose $\mathcal{L}$ is a $d-$dimensional geometrically local Lindbladian and $O$ with $\norm{O}\leq 1$ is a local observable. To achieve an additive error $\varepsilon$ in the expected local observable with the analogue quantum simulator, we need to choose
\[
\omega = \Theta\left(t^{-(d + 1/2)}\sqrt{\varepsilon}\right)
\]
and this corresponds to a simulator evolution time 
\[
t_\textnormal{sim} = t/\omega^2 = O\left(t^{2d + 2} \varepsilon^{-1}\right),
\]
which is uniform in the system size.
\end{proposition}
\noindent The key physical insight that allows us to side-step the issue of the diverging light-cone as $\omega \to 0$ is that the light-cone predicted by the Lieb-Robinson bound upper bounds the spread of a local observable for \emph{any} initial state of the ancillae. However, in the analogue quantum simulation experiment, the ancillae are initialized to $\ket{0}$ --- moreover, they are strongly damped by the amplitude damping channel applied on them and thus remain approximately in $\ket{0}$ throughout the evolution of the simulator. At a qualitative level, restricting the quantum states of interest to only slightly excited ancillae results in a much slower growth of the light cone on the quantum simulator than predicted by a direct application of the Lieb-Robinson bounds [Fig.~\ref{fig:lightcone}(c)]. Our key technical contribution in the proof of Proposition~\ref{prop:dynamics_noiseless} is to make this expectation precise by combining the rigorous adiabatic elimination of ancillae developed in the proof of Proposition~\ref{prop:general_lindbladian} together with the Lieb-Robinson bounds for dissipative systems \cite{poulin2010lieb}.

\emph{Steady state under rapid mixing.} Similarly, we can consider the quantum simulation of long-time dynamics of local observables that are rapidly mixing in a geometrically local Lindbladian. Rapidly mixing Lindbladians were introduced in Ref.~\cite{cubitt2015} as the dissipative counterparts to gapped Hamiltonians. It has been established that the fixed points of rapidly mixing Lindbladians have properties similar to those of ground states of gapped Hamiltonians, such as the stability of local observables to local perturbations in the Lindbladian \cite{cubitt2015, lucia2015rapid} and a mutual information area law for the fixed point \cite{brandao2015area}, and have also been proposed as a key tool to rigorously define phases of mixed many-body states \cite{coser2019classification}. 

More specifically, a local observable $O$ in a spatially local Lindbladian $\mathcal{L}$, with a unique fixed point $\sigma$, is considered to be rapidly mixing if it converges exponentially fast to $\text{Tr}(O\sigma)$; i.e.
\begin{align}\label{eq:rapid_mixing_observable}
    \abs{\text{Tr}(Oe^{\mathcal{L}t}(\rho(0)) - \text{Tr}(O\sigma)} \leq k\big(\abs{S_O}, \gamma \big) e^{-\gamma t},
\end{align}
where $S_O$ is the support of observable $O$, $\abs{S_O}$ is the number of lattice sites in $S_O$ (i.e.~its volume), and $k(l, \gamma)$ is $O(\text{poly}(l))$ for a fixed $\gamma$, and, for some $\kappa > 0$, $O(\exp(\gamma^{-\kappa}))$ as $\gamma\to0$ for a fixed $l$. The parameter $\gamma$ controls the rate of convergence of the local observable to the fixed point. Such observables have a well defined thermodynamic limit both for dynamics and in the fixed point \cite{cubitt2015}. For observables and spatially local Lindbladians that satisfy Eq.~\ref{eq:rapid_mixing_observable}, we establish that the $\omega$ needed to estimate the observable, at any time $t$, depends entirely on the precision required in the observable and is uniform in both system size $n$ as well as time $t$.
\begin{proposition}[Rapid mixing dynamics --- noiseless simulator runtime]\label{prop:fp_noiseless}
Suppose $\mathcal{L}$ is a $d-$dimensional geometrically local Lindbladian and $O$ with $\norm{O}\leq 1$ is a local observable supported on $O(1)$ lattice sites satisfying rapid mixing (Eq.~\ref{eq:rapid_mixing_observable}). To achieve an additive error $\varepsilon$ in the expectation value of $O$ at time $t$ with the analogue quantum simulator, we need to choose
\[
\omega = \Theta\left(\gamma^{(d+1/2)(\kappa + 1)}\sqrt{\varepsilon}\right)
\]
which corresponds to a simulator evolution time 
\[
t_\textnormal{sim} = t / \omega^2 = O\left(t\gamma^{-(2d + 1)(\kappa + 1)} \varepsilon^{-1}\right).
\]
\end{proposition}
Since rapidly mixing local observables (Eq.~\ref{eq:rapid_mixing_observable}) are within $\varepsilon$ of their fixed point expectation value $\text{Tr}(O\sigma)$ in time $t = \Theta\left(\gamma^{-(\kappa + 1)}) + \Theta(\gamma^{-1}\log(\varepsilon^{-1})\right)$, an immediate consequence of the run-time bound in Proposition~\ref{prop:fp_noiseless} is that analogue quantum simulators can be used to efficiently simulate such observables. More specifically,
\begin{corollary}[Rapid mixing fixed points --- noiseless simulator runtime]
    For a $d-$dimensional geometrically local Lindbladian, an analogue quantum simulator can compute the fixed point expected value of a rapidly mixing local observable $O$ (Eq.~\ref{eq:rapid_mixing_observable}), with $\norm{O} \leq 1$, to precision $\varepsilon$ in simulator evolution time
    \[
    t_\textnormal{sim} = O\left(\gamma^{-(2d + 2)(\kappa + 1)} \varepsilon^{-1}\log(\varepsilon^{-1})\right).
    \]
\end{corollary}
\noindent The proof of this Proposition~\ref{prop:fp_noiseless} builds on the proof of Proposition~\ref{prop:dynamics_noiseless} together with an application of the rapid mixing assumption in Eq.~\ref{eq:rapid_mixing_observable}. To show that $\omega$ can be chosen to be uniform in time, we separately analyze the error incurred by the quantum simulator in the short-time and the long-time regimes. In the short-time regime, the observable error is estimated using Lieb-Robinson bounds with the same approach used to prove Proposition~\ref{prop:dynamics_noiseless}, and in the long-time regime the error is bounded using the rapid mixing property (Eq.~\ref{eq:rapid_mixing_observable}).

\emph{Quantum Advantage}. Propositions \ref{prop:dynamics_noiseless} and \ref{prop:fp_noiseless} show that, for the specific problem of spatially local Lindbladians, significantly better run-time bounds that are uniform in the system size can be obtained when compared to the general setting addressed in Proposition~\ref{prop:general_lindbladian}. However, this immediately raises the question of whether we even expect a quantum advantage over classical algorithms in these restricted settings. In our next proposition, we show that there is indeed complexity-theoretic evidence for quantum advantage with respect to physically relevant problem parameters even in these restricted settings. 
\begin{proposition}[Rapid mixing fixed points --- noiseless quantum advantage]\label{prop:quantum_advantage}
There cannot exist a classical algorithm that can, for every geometrically local 2D Lindbladian and a corresponding rapidly mixing local observable, compute the fixed point expected value of the local observable to additive error $\varepsilon$ in time $\textnormal{poly}(\gamma^{-1}, \varepsilon^{-1})$, unless \textnormal{BQP = BPP}.
\end{proposition}
\noindent Again, since from Eq.~\ref{eq:rapid_mixing_observable} it follows that rapidly mixing local observables reach a precision $\varepsilon$ in time $t = \Theta(\gamma^{-(\kappa+1)}\log(\varepsilon^{-1}))$, Proposition~\ref{prop:quantum_advantage} also implies the hardness of measuring local observables in dynamics, as is made precise by the following corollary.
\begin{corollary}[Geometrically local dynamics --- noiseless quantum advantage]\label{prop:quantum_advantage_dynamics}
    There cannot exist a classical algorithm that can compute every local observable at any given time $t$ in every 2D Lindbladian to additive error $\varepsilon$ in time $\textnormal{poly}(t, \varepsilon^{-1})$ unless \textnormal{BQP = BPP}.
\end{corollary}
\noindent Our proof of Proposition~\ref{prop:quantum_advantage} follows the  strategy of encoding any given quantum circuit on $N$ qubits and depth $\text{poly}(N)$ into the unique fixed point of a 2D geometrically local Lindbladian on $n=\text{poly}(N)$ 6-level qudits. Furthermore, we also exhibit a local observable whose expected value in the fixed point determines the probability of a pauli-Z measurement on the first qubit of the encoded circuit resulting in a 1. Since every decision problem in the BQP class can be solved by measuring only one qubit at the circuit output, this encoding establishes an equivalence between the classical hardness of simulating a $\text{poly}(N)$ depth quantum circuit and the fixed point of a geometrically local 2D Lindbladian. With this encoding, we can see that if there indeed existed a classical algorithm to obtain a rapidly mixing observable in the fixed point of a geometrically local Lindbladian to a precision $\varepsilon$ in time polynomial in both $\gamma^{-1}$ and $\varepsilon^{-1}$, then these parameters could be chosen as polynomials of $N$ such that the encoding Lindbladian and local observable would also satisfy Eq.~\ref{eq:rapid_mixing_observable} and thus we would obtain a classical algorithm to simulate any $\text{poly}(N)$ depth quantum circuit, implying BQP = BPP.

Our key technical contribution to the proof of Proposition~\ref{prop:quantum_advantage}, therefore, is an encoding of a quantum circuit on $N$ qubits and of depth $\text{poly}(N)$ into the fixed point of a geometrically local 2D Lindbladian, as well as a rigorous analysis of the convergence of the resulting Lindbladian to its fixed point. This builds on Refs.~\cite{verstraete2009quantum} and \cite{aharonov2008adiabatic}. In particular, in Ref.~\cite{verstraete2009quantum}, the authors demonstrated a strategy to encode a given quantum circuit into the unique fixed point of a 5-local (but not geometrically local) Lindbladian and analyzed the spectrum of the Lindbladian to assess its convergence to the fixed point. We extend this construction to geometrically local 2D Lindbladian by adapting the construction provided in Ref.~\cite{aharonov2008adiabatic}, where the authors encoded a quantum circuit into the unique ground state of a 2D geometrically local Hamiltonian. Furthermore, we provide a detailed convergence analysis of the 2D geometrically local Lindbladian, which is significantly different from the analysis of the gap of the geometrically local Hamiltonian constructed in Ref.~\cite{aharonov2008adiabatic}.

The simulation times of the noiseless quantum simulator for all of these problems and complexity-theoretic limitations of classical algorithms are summarized in Table~\ref{tab:noiseless}.

\subsubsection{Noisy setting}
Next, we address the stability of the quantum simulation protocol to noise in the quantum simulator. The notion of stability in quantum simulation tasks has been considered previously in Ref.~\cite{trivedi2024quantum} --- a physically meaningful and stable quantum simulation task is one in which in the presence of a constant rate of error $\delta$ on the simulator, the error in the observable being measured is only perturbed by an amount $\varepsilon(\delta) \leq O(\delta^c)$ for some $c>0$ dependent on the error rate $\delta$, and does not grow with the number of qudits in the simulator. The observable can thus be computed to a hardware-limited and system-size independent precision, $\varepsilon(\delta)$, on a noisy quantum simulator. A stable quantum simulation task is a special computational task where the quantum simulator avoids an accumulation of errors incurred on all the qudits --- such an accumulation would typically yield an observable error that grows not only with $\delta$ but also with the system size $n$, and would be a worst case scenario for the performance of the quantum simulator. Importantly, for stable many-body simulation tasks where the observables additionally have a well-defined thermodynamic limit, a noisy quantum simulator can simulate the thermodynamic limit to a hardware-limited precision $f(\delta)$.
\begin{table*}
    \begin{tabular}{|L{5.0cm}| L{5.0cm} | L{5.0cm}|}
    \hline 
       \textbf{Problem} & 
       \textbf{Simulator run-time} &
       \textbf{Quantum advantage} \\ \hline 
       
        General Lindbladian and observable (Dynamics)  &
        $O\left(\varepsilon^{-1}{Mt}\right) + O\left(\varepsilon^{-1}M^2 t^2\right)+ O\left(\varepsilon^{-1}M \norm{H_\textnormal{sys}}t^2\right)$  &
       
        \ \ \ \ \ \ \ \ \ \ \ \ \ \ \ \ \ \ ---\\
 \hline 
            
Geometrically Local Lindbladian on $\mathbb{Z}^d$, local observable \ \ \ \ \ \  (Dynamics) &
        \ \ \ \ \ \ \ \ \ \ \ \ \ \ \ \ \ \ \ \ \ \  \ \ \ \ \ \ \ \ \  \ \ \ \ 
        $O\left(t^{2d + 2}\varepsilon^{-1}\right)$ &
        No $\text{poly}(t, \varepsilon^{-1})$ classical algorithm for $d\geq 2$. \\  \hline  

        Geometrically Local Lindbladian on $\mathbb{Z}^d$, rapidly-mixing local observable (Dynamics)  & 
          \ \ \ \ \ \ \ \ \ \ \ \ \ \ \ \ \ \ \ \ \ \  \ \ \ \ \ \ \ \ \  \ \ \ \ 
        $O\left(t\gamma^{-(\kappa + 1)(2d + 1)}\varepsilon^{-1}\right)$ &
        No $\text{poly}(t, \gamma^{-1}, \varepsilon^{-1})$ classical algorithm for $d \geq 2$.  \\  \hline

        Geometrically Local Lindbladian on $\mathbb{Z}^d$, rapidly-mixing local observable (Fixed points) &
          \ \ \ \ \ \ \ \ \ \ \ \ \ \ \ \ \ \ \ \ \ \  \ \ \ \ \ \ \ \ \  \ \ \ \ 
        $O\left(\gamma^{-(\kappa + 1)(2d + 2)}\varepsilon^{-1} \log(\varepsilon^{-1})\right)$  &
        No $\text{poly}(\gamma^{-1}, \varepsilon^{-1})$ classical algorithm for $d\geq 2$. \\
    \hline
    \end{tabular}
    \caption{Summary of the upper bounds on the simulator run-time for the problems considered in Propositions~\ref{prop:general_lindbladian}-\ref{prop:fp_noiseless}, as well as the classical hardness results for these problems following from Proposition~\ref{prop:quantum_advantage}.}
    \label{tab:noiseless}
\end{table*}

We start by considering a concrete but general and realistic error model for the analogue quantum simulator. There can be two sources of errors --- coherent errors on the quantum simulator which lead to the incorrect configuration of Hamiltonian interactions between different qudits, or incoherent errors that arise from the interaction of the simulator qudits with an external environment. In order to model both of these sources of errors, we will assume that the quantum simulator implements the Lindbladian $\mathcal{L}_{\omega, \delta}$ instead of the Lindbladian $\mathcal{L}_{\omega}$ with
\begin{subequations}    \label{eq:lindbladian_omega_delta}
\begin{align}
    \mathcal{L}_{\omega, \delta} = \mathcal{L}_\omega +\delta \mathcal{N},
\end{align}
where $\delta$ is the error rate and $\mathcal{N}$ is itself a spatially local Lindbladian given by
\begin{align}
\mathcal{N} = \sum_{\beta = 1}^{M'} \mathcal{N}_\beta,
\end{align}
\end{subequations}
where each term $\mathcal{N}_\beta$ is a local Lindbladian acting on the system and ancillae and the number of these term $M' \leq O(n)$ with $n$ being the number of qudits. We may consider $\delta$ to be the error rate since it mediates the strength of $\mathcal{N}$, which describes decohering interactions with an environment. We assume the normalization $\norm{\mathcal{N}_\beta}_\diamond \leq 1$. In general, $\mathcal{N}_\beta(X) = -i[v_\beta, X] + \sum_{j}\mathcal{D}_{Q_{j,\beta}}$ --- the Hamiltonian term $v_\beta$ in $\mathcal{N}_\beta$ corresponds to coherent errors in the quantum simulator that perturb the Hamiltonians, and the jump operators $Q_{j, \beta}$ can model different incoherent errors in the Lindbladian.
\begin{table*}[htpb]
    \begin{tabular}{|L{4.0cm}| L{4.0cm} | L{4.0cm}| L{4.5cm} |}
    \hline 
       \textbf{Problem} & 
       \textbf{Noise-limited precision} &
       \textbf{Simulator run-time} &
       \textbf{No $\text{poly}(\delta^{-1})$ algorithm for (in 2D)} \\ \hline             
Geometrically Local Lindbladian on $\mathbb{Z}^d$, local observable (Dynamics) &
        \ \ \ \ \ \ \ \ \ \ \ \ \ \ \ \ \ \ \ \ \ \  \ \ \ \ \ \ \ \ \ 
        $O(\delta^{1/2}t^{2d + 1})$ &
         \ \ \ \ \ \ \ \ \ \ \ \ \ \ \ \ \ \ \ \ \ \  \ \ \ \ \ \ \ \ \ 
        $O(t\delta^{-1/2})$ &
         \ \ \ \ \ \ \ \ \ \ \ \ \ \ \ \ \ \ \ \ \ \  \ \ \ \ \ \ \ \ \ \ \ \ \ \ \ \ \ \ \ \ \ \ \ \ \ \ \ \ \ \ \ \ \ \ \ \ \ \ \
        $t = O(\delta^{\alpha-1/10 })${ for any }$\alpha > 0$. \\ 
        \hline  
        Geometrically Local Lindbladian on $\mathbb{Z}^d$, rapidly-mixing local observable (Dynamics)  &
        \ \ \ \ \ \ \ \ \ \ \ \ \ \ \ \ \ \ \ \ \ \  \ \ \ \ \ \ \ \ \ 
        $O(\delta^{1/2}\gamma^{-(\kappa + 1)(2d + 1)})$ &
        \ \ \ \ \ \ \ \ \ \ \ \ \ \ \ \ \ \ \ \ \ \  \ \ \ \ \ \ \ \ \ 
        $O(t\delta^{-1/2})$ &
                 \ \ \ \ \ \ \ \ \ \ \ \ \ \ \ \ \ \ \ \ \ \  \ \ \ \ \ \ \ \ \ \ \ \ \ \ \ \ \ \ \ \ \ \ \ \ \ \ \ \ \ \ \ \ \ \ \ \ \ \ \                 
         \ \ \ \ \ \ \ \ \ \ \ \ \ \ \ \ \ \ \ 
         ---
          \\  \hline
        Geometrically Local Lindbladian on $\mathbb{Z}^d$, rapidly-mixing local observable (Fixed points) &
         \ \ \ \ \ \ \ \ \ \ \ \ \ \ \ \ \ \ \ \ \ \  \ \ \ \ \ \ \ \ \ 
        $O(\delta^{1/2}\gamma^{-(\kappa + 1)(2d + 1)})$  &
                 \ \ \ \ \ \ \ \ \ \ \ \ \ \ \ \ \ \ \ \ \ \  \ \ \ \ \ \ \ \ \ 
        $\tilde{O}(\gamma^{-(\kappa + 1)} \delta^{-1/2})$  &
               \ \ \ \ \ \ \ \ \ \ \ \ \ \ \ \ \ \ \ \ \ \  \ \ \ \ \ \ \ \ \ \ \ \ \ \ \ \ \ \ \ \ \ \ \ \ \ \ \ \ \ \ \ \ \ \ \ \ \ \ \
       $\gamma^{-1} = O(\delta^{\alpha-1/(10(\kappa + 1))})$ for any $\kappa, \alpha > 0$. \\
    \hline
    \end{tabular}
    \caption{Summary of the results from the analysis of noisy open quantum simulator for the dynamics and fixed points of geometrically local Lindbladians. The table lists the estimates of noise-limited precision and simulator run-time as a function of the noise-rate $\delta$ that follow from Propositions~\ref{prop:dynamics_noisy}-\ref{prop:fp_noisy}, as well as the family of problems as a function of $\delta$ that where there is expected to be no $\text{poly}(\delta^{-1})$ classical algorithm. The $\tilde{O}$ in the last row suppresses polynomial factors of $\log(\gamma^{-1})$ and $\log(\delta^{-1})$.}
    \label{tab:noisy}
\end{table*}

We now revisit the problem of computing local observables in dynamics and fixed points, and analyze how the presence of errors and noise effects the results of the computation. We will show that the analogue quantum simulator is stable for both the problem of local observables in dynamics of geometrically local Lindbladians as well as that of rapidly mixing local observables in fixed points. In the next two propositions, we establish that the quantum simulator is robust to both coherent and incoherent errors --- in particular, we show that in the presence of hardware error at rate $\delta$, we can pick the parameter $\omega$ dependent on $\delta$, and uniform in the system size, to obtain a precision $\varepsilon$ in the observable that also scales as $O(\delta^c)$, for some constant $c>0$, and is uniform in system size. More specifically, for the problem of dynamics of geometrically local Lindbladians, we establish that
\begin{proposition}[Geometrically local dynamics --- noisy simulator precision and runtime]\label{prop:dynamics_noisy}
Suppose $\mathcal{L}$ is a $d-$dimensional geometrically local Lindbladian and $O$ with $\norm{O}\leq 1$ is a local observable. Then, in the presence of noise with noise rate $\delta$,  the expected local observable at time $t$ can be obtained to a precision $\varepsilon = O\left(\delta^{1/2}t^{2d + 1}\right)$, independent of the system size, with the analogue quantum simulator. Furthermore, to obtain this precision, we need to choose $\omega = \Theta(\delta^{1/4})$ which results in a simulator run-time $t_\textnormal{sim} = t/\omega^2 = O(t \delta^{-1/2})$.
\end{proposition}
\noindent We establish a similar stability result for the problem of long-time dynamics of rapidly mixing observables (defined in Eq.~\ref{eq:rapid_mixing_observable}), i.e.~these observables can be obtained on quantum simulators to a precision $\varepsilon$ that scales as $\text{poly}(\delta)$ and is independent of the system size $n$ as well as the time $t$.
\begin{proposition}[Rapid mixing dynamics --- noisy simulator precision and runtime]\label{prop:fp_noisy}
Suppose $\mathcal{L}$ is a $d-$dimensional geometrically local Lindbladian and $O$ with $\norm{O}\leq 1$ is a local observable supported on $O(1)$ lattice sites satisfying rapid mixing (Eq.~\ref{eq:rapid_mixing_observable}). Then, in the presence of noise with noise rate $\delta$,  the expected local observable at any time $t$ can be obtained to a precision $\varepsilon = O(\delta^{1/2}\gamma^{-(\kappa + 1)(2d + 1)})$, independent of $n$ and $t$, with the analogue quantum simulator. Furthermore, to obtain this precision, we need to choose $\omega = \Theta(\delta^{1/4})$ which results in a simulator run-time $t_\textnormal{sim} = t/\omega^2 = O(t \delta^{-1/2})$.
\end{proposition}
\noindent For rapidly mixing observables, Eq.~\ref{eq:rapid_mixing_observable} implies that after $t \geq \Theta\left(\gamma^{-(\kappa + 1)}\log(\varepsilon^{-1}\right)$, the expected value of the observable at time $t$ is $\varepsilon$-close to its fixed point expected value. Consequently, from Proposition~\ref{prop:fp_noisy}, it immediately follows that the problem of measuring the fixed point expectation value of such observables is also stable to errors. More specifically, we obtain that
\begin{corollary}[Rapid mixing fixed points --- noisy simulator precision and runtime]
\label{prop:fp_noisy_fixed_point}
Suppose $\mathcal{L}$ is a $d-$dimensional geometrically local Lindbladian and $O$ with $\norm{O}\leq 1$ is a local observable supported on $O(1)$ lattice sites satisfying rapid mixing (Eq.~\ref{eq:rapid_mixing_observable}). Then, in the presence of noise with noise rate $\delta$,  the fixed point expected value of the local observable can be obtained to a precision $\varepsilon = O\left(\delta^{1/2}\gamma^{-(\kappa + 1)(2d + 1)}\right)$, independent of $n$, with the analogue quantum simulator with a simulator run-time $t_\textnormal{sim}  = O\left(\gamma^{-(\kappa + 1)} \delta^{-1/2} \log(\delta^{-1})\right) + O\left(\gamma^{-(\kappa + 1)} \delta^{-1/2} \log(\gamma^{-1})\right)$.
\end{corollary}
\emph{Quantum advantage with errors}. Even for a stable quantum simulation task, the observable of interest cannot be determined to an arbitrary precision on a noisy quantum simulator, but only to a precision $\varepsilon(\delta) \leq O(\delta^c)$ for some $c> 0$. Furthermore, the settings where we have been able to establish stability of the analogue quantum simulator are also settings where the observables being computed have a well-defined thermodynamic limit and thus becomes asymptotically independent of the system size. To compare the performance of the noisy quantum simulator with classical algorithms for such problems, we take the perspective laid out in Ref.~\cite{trivedi2024quantum} ---  Since we are considering problems where the observables have a well defined thermodynamic limit, and this limit is the quantity of interest to be computed, we cannot study the classical and quantum computational time of such observables as a function of system size (which has been implicitly taken to $\infty$). The relevant parameter with respect to which we can measure the computational time is the obtained precision --- since the precision on a noisy quantum simulator, with noise rate $\delta$, is limited to $\varepsilon(\delta)$ we ask if there are families of problems 
\begin{enumerate}
    \item[(i)] that are stable and as $\delta \to 0$, the noise-limited precision $\varepsilon(\delta)$ in the observables also $\to 0$ and
    \item[(ii)] the scaling of the (best) classical algorithm run-time with respect to $\delta$ is at least superpolynomially worse than the scaling of the analogue simulator's run-time with respect to $\delta$.
\end{enumerate}  
Explicitly, if the run-time of the analogue simulator scales as $O(\text{poly}(\delta^{-1}))$ to achieve the hardware-limited precision $\varepsilon(\delta)$ and the run-time of the (best) classical algorithm scales as $\textnormal{superpoly}(\delta^{-1})$ to reach the same precision, then the noisy quantum simulator provides a superpolynomial advantage over the classical algorithms. Stated differently, this notion of quantum advantage implies that as the noise rate $\delta$ is reduced the classical algorithm would find it superpolynomially harder to compute the observable to the noise-limited precision.

Our main result is to show that this notion of quantum advantage holds for certain families of problems of computing local observables in the dynamics or fixed points of geometrically local 2D Lindbladians that are also stable to noise on the quantum simulator. In particular, we consider the following problem for fixed points.
\begin{problem}
    \label{prob:fixed_points}
    For $\alpha, \kappa > 0$, given a family of geometrically local Lindbladians $\mathcal{L}_{(\delta)}$ and rapid mixing local observables $O_{(\delta)}$, indexed by $\delta \to 0$, such that the observable satisfies Eq.~\ref{eq:rapid_mixing_observable} with $\gamma^{-1} \leq O(\delta^{\alpha - 1/(10(\kappa + 1))})$, compute the expected value of $O_{(\delta)}$ in the fixed point of $\mathcal{L}_{(\delta)}$.
\end{problem}
\begin{proposition}[Rapid mixing fixed points --- noisy quantum advantage]\label{prop:quantum_advantage_noisy}
A noisy analogue quantum simulator with noise rate $\delta$ can solve Problem~\ref{prob:fixed_points}, with parameters $\alpha, \kappa$, to a noise-limited precision $O(\delta^{5\alpha (\kappa + 1)})$ (which $\to 0$ as $\delta \to 0$). Furthermore, there cannot exist a randomized classical algorithm with run-time $\text{poly}(\delta^{-1})$ that can solve Problem~\ref{prob:fixed_points} to the same precisions for every given $\alpha, \kappa$ unless \textnormal{BQP = BPP}.
\end{proposition}

\noindent From a physical standpoint, the constraint  $\gamma^{-1} \leq O(\delta^{\alpha - 1/(10(\kappa + 1))})$, for $\alpha <1/(10(\kappa + 1))$, simply states that a quantum simulator with lower noise can be used to simulate the fixed point expected value of a rapidly mixing local observable that takes longer to reach its fixed point value while still obtaining a hardware-limited precision which decreases polynomially with $\delta$. This proposition, which rigorously proves the notion of noisy quantum advantage laid out in Ref.~\cite{trivedi2024quantum}, follows almost directly from the quantum-circuit-to-2D-Lindbladian encoding developed for the proof of Proposition~\ref{prop:quantum_advantage}. The only additional detail in proving this proposition is to show that, given any quantum circuit on $N$ qubits with depth $T = \text{poly}(N)$, the encoding Lindbladian can be embedded into the family of problems considered in Proposition~\ref{prop:quantum_advantage_noisy} while accounting for the additional constraint on $\gamma$. We show that this is possible simply by choosing $\alpha, \kappa$ depending on the degree of the polynomial of $N$ describing the depth $T$, and then choosing $N$ depending on $\delta$ as $N = \text{poly}(\delta^{-1})$. This implies that if there did exist a $\text{poly}(\delta^{-1})$ classical algorithm to simulate this family of problems for any $\alpha, \kappa > 0$, then it would necessarily also be able to simulate an arbitrary $\text{poly}(N)$ depth quantum circuit, thus implying BQP = BPP. Furthermore, we also establish an implication of Proposition \ref{prop:quantum_advantage_noisy} i.e.~a similar notion of quantum advantage holds for the problem of dynamics of local observables. More specifically,
\begin{problem}
    \label{prob:time_dynamics}
    For $\alpha > 0$, given a family of geometrically local Lindbladians $\mathcal{L}_{(\delta)}$, local observables $O_{(\delta)}$, and and evolution times $t_{(\delta)}$, indexed by $\delta \to 0$, such that $t_{(\delta)} \leq O(\delta^{\alpha - 1/10})$, compute the expected value of $O_{(\delta)}$ after evolving under $\mathcal{L}_{(\delta)}$ for time  $t_{(\delta)}$ when all the qudits are initially in a product state.
\end{problem}
\begin{corollary}[Geometrically local dynamics --- noisy quantum advantage] \label{prop:quantum_advantage_dynamics_noisy}
    A noisy analogue quantum simulator with noise rate $\delta$ can solve Problem~\ref{prob:time_dynamics}, with parameter $\alpha$, to a noise-limited precision $O(\delta^{5\alpha})$ (which $\to 0$ as $\delta \to 0$) in simulator run-time $O(\textnormal{poly}(\delta^{-1}))$ and there cannot exist a $\textnormal{poly}(\delta^{-1})$ randomized classical algorithm to estimate this local observable to the same precision for every given $\alpha > 0$ unless \textnormal{BQP = BPP}.
\end{corollary}
\noindent We summarize the results pertaining to noisy quantum simulation of geometrically local Lindbladians in Table \ref{tab:noisy}.


\subsection{Numerical example}
As an illustrative example of the analogue quantum simulation and the impact of noise on the simulator, we study the analogue quantum simulation of a gaussian fermion model. We choose a gaussian fermion model since it can be numerically simulated efficiently for large system sizes \cite{bravyi2011classical, horstmann2013noise}, allowing us to verify the scalings predicted by Propositions~\ref{prop:general_lindbladian}, \ref{prop:dynamics_noiseless}, \ref{prop:quantum_advantage}, and \ref{prop:dynamics_noisy}. We consider a family of target Lindbladians on $n = 2L + 1$ fermions arranged on a 1D lattice and described by the Hamiltonian
\[
H_\text{sys} = \sum_{x = -L}^{L} K \big(a_x^\dagger a_{x + 1} + \text{h.c.}\big) + \sum_{x = -L}^L J\big(a_x a_{x + 1} + \text{h.c.}\big),
\]
where $a_x$ is the annihilation operator on the fermionic mode at $x$, and we assume periodic boundary conditions and therefore set $a_{x = L + 1} \cong a_{x = -L}$. We associate one 2-fermion jump operator per site $x$, $L_x$, given by
\[
L_x = \lambda_0 a_x + \lambda_1 a_{x + 1}.
\]
The parameters $J, K, \lambda_0$ and $\lambda_1$ specify the model. We will consider the problem of measuring the observable $O$ given by
\[
O = \frac{1}{n} \sum_{x} a_x^\dagger a_x,
\]
which measures the particle density (i.e. the particle number per unit lattice size), in the fixed point of this dissipative system. Figure \ref{fig:just_model} shows this observable as a function of the parameter $J$ --- at $J = 0$, since the Hamiltonian $H_\text{sys}$ is particle number conserving and the jump operators annihilate any fermions on the lattice, the particle density in the fixed point is $0$ and becomes non-zero when $J \neq 0$. Furthermore, as shown in Fig.~\ref{fig:just_model}, this observable also has a well defined thermodynamic limit and converges exponentially to this limit.
\begin{figure}
    \centering
\includegraphics[scale=0.45]{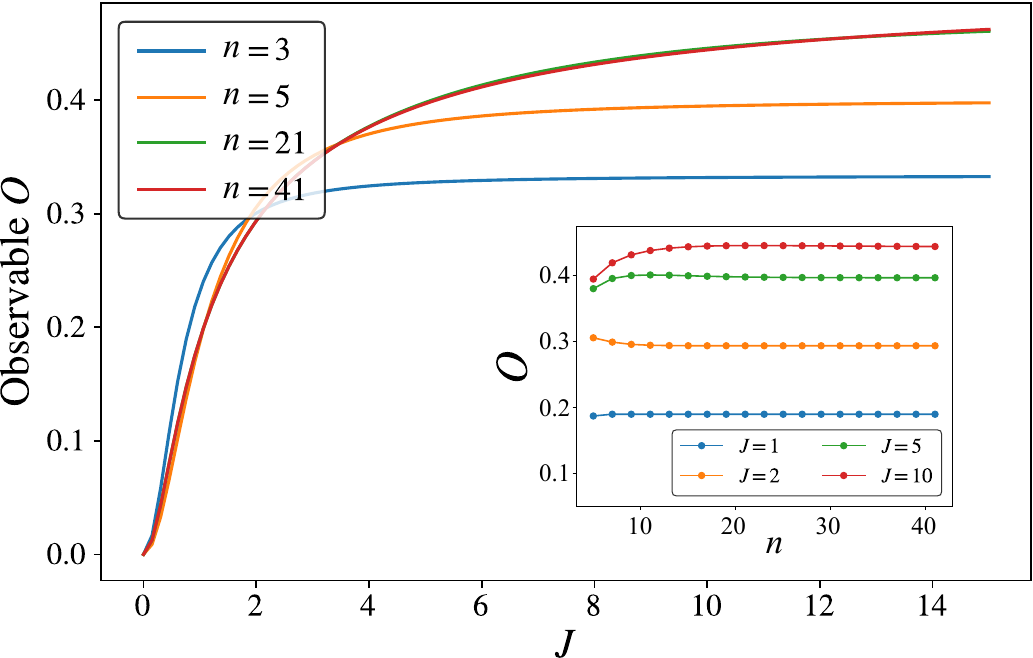}
    \caption{The particle density observable $O$ in the fixed point of the dissipative fermionic system as a function of the coupling parameter $J$, while setting $K = 1, \lambda_0 = 1.1$ and $\lambda = 1$. The inset shows the convergence of the particle-density observable to its thermodynamic limit. Note that, due to the fast convergence of $O$ to the thermodynamic limit, the curves for $n = 21$ and $41$ overlap with eachother.}
    \label{fig:just_model}
\end{figure}

To perform a quantum simulation of this model, as described in the previous subsection, we introduce ancillary fermions with annihilation operators $b_x$ and couple them to the system fermions. The quantum simulator dynamics is given by the Lindbladian
\begin{align}\label{eq:ff_simulator_lindbladian}
\mathcal{L}_\omega(X) = -i\omega^2 [H_\text{sys}, X] + \sum_x \bigg(4\mathcal{D}_{b_x}(X) - i\omega [V_x, X]\bigg),
\end{align}
where $V_x = b_x^\dagger L_x + L_x^\dagger b_x$. The quantum simulator becomes increasingly accurate as $\omega \to 0$ --- this is numerically demonstrated in Fig.~\ref{fig:obs_error_noiseless}a for different system sizes, where we compare the expected value of $O$ in the fixed point of the simulator Lindbladian $\mathcal{L}_\omega$ to its expected value in the fixed point of the target Lindbladian $\mathcal{L}$. Furthermore, the scalings of the observable error incurred on the quantum simulator with the system size $n$ and the parameter $\omega$ are studied in Fig.~\ref{fig:obs_error_noiseless}b and c. We note from Fig.~\ref{fig:obs_error_noiseless}b that the for a fixed $\omega$, the observable error saturates on increasing the system size. This uniformity with system size is consistent with our expectation from Proposition~\ref{prop:fp_noiseless}. Furthermore, consistent with the Proposition~\ref{prop:fp_noiseless}, Fig.~\ref{fig:obs_error_noiseless}c shows that the observable error, in the limit of large system-size, scales polynomially with $\omega$.
\begin{figure*}[htpb]
    \centering
    \includegraphics[scale=0.475]{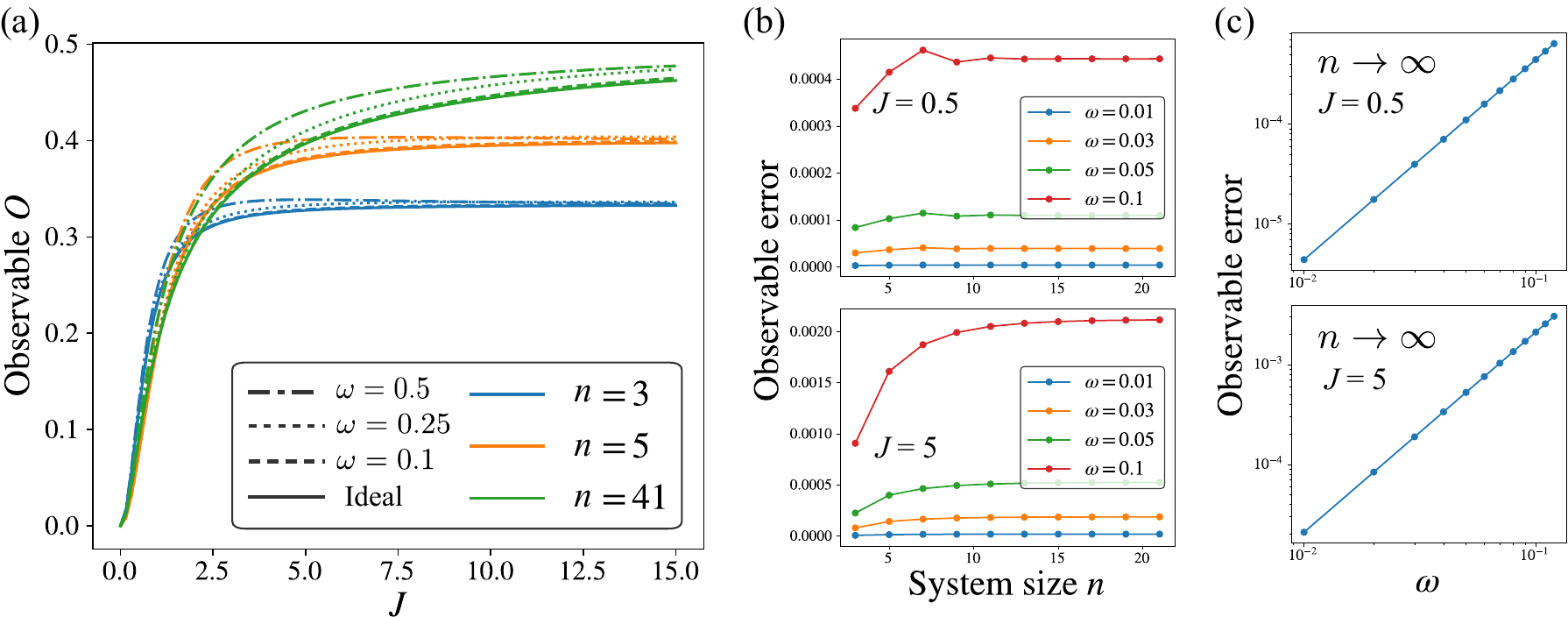}
    \caption{Numerical studies of the performance of the noiseless analogue quantum simulator described by Eq.~\ref{eq:ff_simulator_lindbladian}. (a) The particle density observable as a function of $J$ --- the solid line shows the target (ideal)  observable and dotted-dashed lines show the observable measured on the quantum simulator for different $\omega$. (b) The observable error incurred by the quantum simulator as a function of the system size, indicating that the error becomes independent of $n$ as $n \to \infty$. (c) The observable error for large $n$ as a function of $\omega$, showing that the error scales polynomially with $\omega$ (note that both the axes are shown on log-scale).}
    \label{fig:obs_error_noiseless}
\end{figure*}

Next, we add noise to the quantum simulator and analyze its impact on the performance of the quantum simulator. We consider single-site depolarizing noise acting on the fermions at a rate $\delta$ --- this can be theoretically modelled by assuming that the quantum simulator has a Lindbladian $\mathcal{L}_{\omega, \delta}$
\begin{align}\label{eq:ff_simulator_noise}
\mathcal{L}_{\omega, \delta} = \mathcal{L}_\omega + \delta \sum_{x}\bigg(\text{Tr}_x(\cdot) \otimes \frac{I_x}{2} - \textnormal{id}\bigg).
\end{align}
Such a Lindbladian has hermitian jump operators that are quadratic in the fermionic annihilation/creation operator and can be simulated following the formalism in Ref.~\cite{horstmann2013noise}. In Fig.~\ref{fig:noisy_qsim}, we numerically study the deviation in the observable $O$ computed in the fixed point of $\mathcal{L}_{\omega, \delta}$ from its expected value in the fixed point of $\mathcal{L}$. Fig.~\ref{fig:noisy_qsim}a shows the dependence of the measured observable on the parameter $\omega$ in the presence of noise. For a fixed $\delta$, unlike the noiseless case, reducing $\omega$ no-longer results in an increasingly accurate simulation. We instead observe that there is an optimal $\omega$, that is dependent on $\delta$, at which the noisy quantum simulator best approximates the target observable. Figure~\ref{fig:noisy_qsim}b shows the scaling of the observable error at this optimal point with the parameter $\delta$. Consistent with Proposition~\ref{prop:fp_noisy}, we see that the observable error becomes independent of the system size $n$ as $n\to \infty$, and the error in the large $n$ limit scales polynomially with $\delta$. The choice of $\omega$ that achieves this optimal error is also shown in Fig.~\ref{fig:noisy_qsim}, and we see it too becomes independent of $n$ as $n\to \infty$ and scales polynomially with $\delta$ consistent with Proposition~\ref{prop:fp_noisy}.
\begin{figure*}
    \centering
    \includegraphics[scale=0.6]{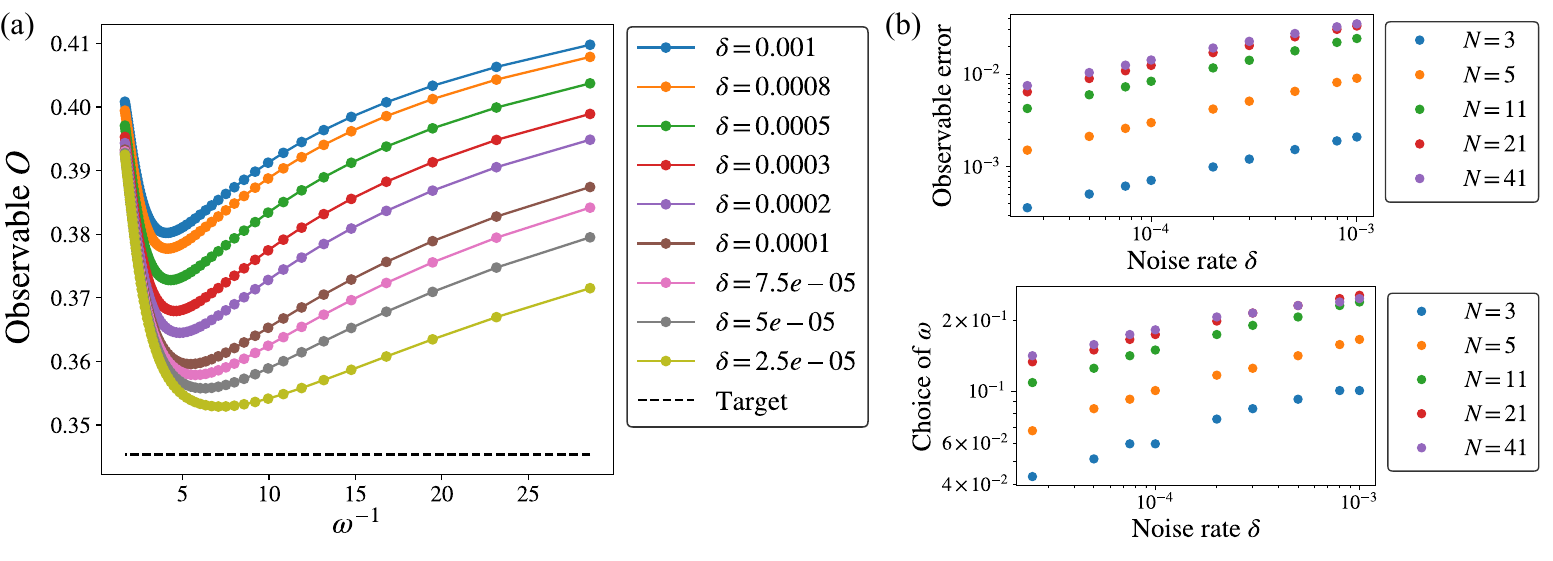}
    \caption{Numerical studies of the performance of a noisy quantum simulator described by Eq.~\ref{eq:ff_simulator_noise}. (a) The dependence of the particle density observable on $\omega$ with the dashed line indicating the target observable. In the presence of noise there is an optimal $\omega$ depending on $\delta$ at which this error is the lowest. (b) The (smallest) observable error as a function of the noise rate $\delta$, as well as the choice of $\omega$ that yields this error.}
    \label{fig:noisy_qsim}
\end{figure*}

\section{Notation and Preliminaries}\label{sec:notation}
Given a Hilbert space $\mathcal{H}$, we will denote by $\text{L}(\mathcal{H})$ the set of bounded linear operators from $\mathcal{H}\to \mathcal{H}$, define $\text{M}(\mathcal{H})$ to be the set of the bounded Hermitian operators from $\mathcal{H}\to \mathcal{H}$ and define $\text{D}_1(\mathcal{H})$ as the set of valid density matrices on $\mathcal{H}$. We will typically use the $\dagger$ superscript to indicate the adjoint, or Hermitian conjugate, of an operator or superoperator. However in some cases, more compact expression can obtained by using the following notation --- for some operator or superoperator $X$, we define $X^{(-)}:=X$ and $X^{(+)}:=X^\dagger$. Furthermore we will use $\bar{+}=-$ and $\bar{-}=+$. For example, $X^{(\bar{u})}:=X^\dagger$ for $u=-$.

While dealing with mixed states and their dynamics, it will often be convenient to adopt the vectorized notation, where we map operators on a (finite-dimensional) Hilbert space to state vectors via $\rho = \sum_{i_l, i_r}\rho_{i_l, i_r}\ket{i_l}\!\bra{i_r} \to \vecket{\rho} = \sum_{i_l, i_r}\rho_{i_l, i_r} \ket{i_l, i_r}$. Superoperators, such as Lindladians or channels, will map to ordinary operators in this picture. Given an operator $X \in \text{M}(\mathcal{H})$, we will define $X_l, X_r \in \text{M}(\mathcal{H}\otimes \mathcal{H})$ by $X_l\vecket{\rho} = (X\otimes I)\vecket{\rho} = \vecket{X\rho}$ and $X_r\vecket{\rho} = (I\otimes X^\text{T})\vecket{\rho} =  \vecket{\rho X}$. $X_l (X_r)$ can also be interpreted as a superoperator which left (right) multiplies its argument with $X$ i.e. $X_l(Y) = XY$ and $X_r(Y) = YX$. A Lindbladian superoperator $\mathcal{L}$ specified by a Hamiltonian $H$ and jump operators $\{L_\alpha\}_{\alpha \in \{1, 2 \dots M\}}$, 
\[
\mathcal{L}(X) = -i[H, X] + \sum_{\alpha = 1}^M \bigg(L_\alpha X L_\alpha^\dagger - \frac{1}{2}\{X, L_\alpha^\dagger L_\alpha\}\bigg),
\]
can be vectorized 
\[\mathcal{L} = -i(H_l - H_r) + \sum_{\alpha = 1}^M \bigg(L_{\alpha, l} L_{\alpha, r}^\dagger -\frac{1}{2} \left(L_{\alpha, l}^\dagger L_{\alpha, l} + L_{\alpha, r} L_{\alpha, r}^\dagger\right)\bigg).
\]
The adjoint of the Lindbladian (with respect to the Hilbert-Schmidt inner product), $\mathcal{L}^\dagger$ will be given by
\[
\mathcal{L}^\dagger(X) = i[H, X] + \sum_{\alpha = 1}^M \bigg(L_\alpha^\dagger X L_\alpha - \frac{1}{2}\{L_\alpha^\dagger L_\alpha, X\}\bigg).
\]
The adjoint of the Lindbladian, by definition, will satisfy $\textnormal{Tr}(A \mathcal{L}(B)) = \textnormal{Tr}(\mathcal{L}^\dagger(A) B)$ and $\mathcal{L}^\dagger(I) = 0$. In the vectorized notation, $\mathcal{L}^\dagger$ will be expressed as
\[
\mathcal{L}^\dagger = i(H_l - H_r) + \sum_{\alpha = 1}^M \bigg(L_{\alpha, l}^\dagger L_{\alpha, r} - \frac{1}{2}\left(L_{\alpha, l}^\dagger L_{\alpha, l} + L_{\alpha, r} L_{\alpha, r}^\dagger\right)\bigg).
\]

\emph{Norms}. $\norm{A}_p$ denotes the Schatten $p$-norm of an operator $A$. We will denote the operator norm, which is also the Schatten-$\infty$ norm, by $\norm{A}$ without an subscript. $\norm{\mathcal{A}}_{p \to q} := \max_{O,\norm{O}_p=1} \norm{\mathcal{A}(O)}_q$ indicates the norm of a superoperator $\mathcal{A}$. We define the completely-bounded norm of a superoperator $\mathcal{A}$ as $\norm{\mathcal{A}}_{cb, p \to q} := \sup_{n\geq 2} \norm{\mathcal{A} \otimes \textnormal{id}_n}_{p \to q}$. The diamond norm is the completely bounded $1\to1$ norm --- for the diamond norm, we will use the standard notation $\norm{\mathcal{A}}_\diamond := \norm{\mathcal{A}}_{cb, 1\to 1}$.

Lattices in this work are lattice graphs. $d(x,y)$ for two lattice sites $x$ and $y$ denotes the Manhattan distance between $x$ and $y$, i.e. the graph path length to reach $y$ from $x$. We write the distance between a set of lattice sites $S$ and a single site $x$ as $d(S,x):=\min_{y \in S}(d(x,y))$. The distance between two sets of lattice sites $S_1$ and $S_2$ is accordingly $d(S_1, S_2) := \min_{x \in S_1, y \in S_2}(d(x,y))$. The diameter of set of lattice sites $S$ is written $\textnormal{diam}(S) := \max_{x\in S,y\in S}(d(x,y))$.

For two real-valued functions $f(x)$ and $g(x)$, we will write $f(x) \leq O(g(x))$ to indicate that there exists $C,x_0 \in \mathbb{R}$ such that $f(x) \leq C g(x)$ for all $x>x_0$. Similarly $f(x) \geq \Omega(g(x))$ indicates that there exists $C,x_0 \in \mathbb{R}$ such that $f(x) \geq C g(x)$ for all $x>x_0$. Finally $f(x) = \Theta(g(x))$ indicates $\Omega(g(x)) \leq f(x) \leq O(g(x))$.

For notational conciseness, we will often use a shorthand for list of indices --- the list $\{m, m + 1, m + 2 \dots n\}$ will be abbreviated as $[m: n]$.\\

\section{Analysis of the noiseless simulator}
\label{sec:analysis_noiseless}


\subsection{Rigorous adiabatic elimination of the ancilla}
\label{subsec:rigorous_adiabatic_elimination}
We begin with analyzing the analogue open quantum simulator in the absence of any noise. Supposing system qudits $\mathcal{S}$ starts in the state $\rho(0)$ at time $t=0$, we wish to simulate the state $\rho(t) = e^{\mathcal{L} t}\rho(0)$. To begin the analysis, we set up equations of motion for $\tr{\mathcal{A}}{{\rho}_\omega(t)}$ and $\tr{\mathcal{A}}{\sigma_\alpha {\rho}_\omega(t)}$. Following from the definition of ${\mathcal{L}}_\omega$ in Eq.~\ref{eq:qsim_basic}, we have that
\begin{widetext}
\begin{subequations}
\label{eq:dynamics}
\begin{align}
    &\frac{d}{dt}\tr{\mathcal{A}}{{\rho}_\omega(t)} = 
    -i\omega \sum_{\alpha}\sum_{\substack{u\in\{+, -\}}} [L_\alpha^{(u)}, \tr{\mathcal{A}}{\sigma^{(\bar{u})}_\alpha {\rho}_\omega(t)}]  -i\omega^2 [H_\text{sys}, \tr{\mathcal{A}}{{\rho}_\omega(t)}],
    \\
    &\frac{d}{dt}\tr{\mathcal{A}}{\sigma_\alpha \rho_\omega(t)} = 
    - 2 \tr{\mathcal{A}}{\sigma_\alpha \rho_\omega (t)} - i\omega L_\alpha \tr{\mathcal{A}}{\rho_\omega(t)}
  - i\omega^2[H_\text{sys}, \tr{\mathcal{A}}{\sigma_\alpha\rho_\omega(t)}]+ \omega\sum_{\alpha'} {E_{\alpha,\alpha'}(t)}, \\
     &\frac{d}{dt}\tr{\mathcal{A}}{\sigma_\alpha^\dagger \rho_\omega(t)} = 
    - 2 \tr{\mathcal{A}}{\sigma_\alpha^\dagger \rho_\omega (t)} +  i\omega \tr{\mathcal{A}}{\rho_\omega(t)}L_\alpha^\dagger
  - i\omega^2 [H_\text{sys}, \tr{\mathcal{A}}{\sigma_\alpha^\dagger\rho_\omega(t)}]+ \omega \sum_{\alpha'}  {E_{\alpha,\alpha'}^\dagger(t)},
\end{align}
\end{subequations}
\end{widetext}
where
\begin{equation}
    E_{\alpha, \alpha'}(t) =
    \begin{dcases}
        i \{L_\alpha, \tr{\mathcal{A}}{n_\alpha \rho_\omega (t)}\} &  \ \  \alpha=\alpha'
        \\
        - i \sum_{u \in \{+,-\}} [L_{\alpha'}^{(u)}, \tr{\mathcal{A}}{\sigma_\alpha \sigma_{\alpha'}^{(\bar{u})}\rho_\omega(t)}] & \ \  \alpha \neq \alpha'
    \end{dcases} 
\end{equation}
with $n_\alpha = \sigma_\alpha^\dagger \sigma_\alpha$ being the excitation number operator for the $\alpha^\text{th}$ ancillary qubit. 




To develop concrete error bound on the deviation of the quantum simulator and the target dynamics, we will carefully analyze the remainder $\mathcal{R}_\omega(t)$ defined by
\begin{align}\label{eq:remainder_def}
\mathcal{R}_\omega(t):=\frac{d}{dt} \tr{\mathcal{A}}{\rho_\omega(t)}-\omega^2 \mathcal{L} \tr{\mathcal{A}}{\rho_\omega(t)},
\end{align}
which can be physically interpreted as the error in the rate of change of the reduced state of the system qudits $\mathcal{S}$ on the analogue quantum simulator compared to the target Lindbladian. The next lemma provides an expression for $\mathcal{R}_\omega(t)$, which we will repeatedly use throughout this paper while analyzing the analogue quantum simulator. This expression follows directly from Eqs.~\ref{eq:dynamics} --- a detailed proof of this is provided in Appendix~\ref{appendix:analysis_noiseless_proofs}.

\begin{lemma}
    \label{lemma:remainder}
    For any $t > 0$, the remainder $\mathcal{R}_\omega(t)$ satisfies 
    \begin{subequations}
    \label{eq:remainder}
    \begin{align}
    \mathcal{R}_\omega(t) & 
        = \omega^2\sum_{\alpha}  e^{-2t} q_{\alpha} +\omega^4 \smashoperator[l]{\sum_{j \in \{1, 2\}}}\sum_{\alpha} \int_0^t e^{-2(t - s)}\mathcal{Q}^{(j)}_{\alpha, H_\textnormal{sys}}(s)ds
        \nonumber \\
        & \quad 
        + \omega^4\sum_{j \in \{3, 4\}}\sum_{\substack{\alpha, \alpha'}}\int_0^t e^{-2(t - s)}Q_{\alpha, \alpha'}^{(j)}(s)ds.
    \end{align}
where
\begin{align*}
& q_\alpha = -\mathcal{D}_{L_\alpha}(\rho(0)), \\
&\mathcal{Q}_{\alpha, h}^{(1)}(t) = - \frac{1}{\omega}[L_\alpha^\dagger, [h, \tr{\mathcal{A}}{\sigma_\alpha \rho_\omega(t)}]] + \textnormal{h.c.},
\\
&\mathcal{Q}_{\alpha, h}^{(2)}(t) = -\frac{i}{2}[L_\alpha^\dagger, L_\alpha [h,\tr{\mathcal{A}}{\rho_{\omega}(t)}]] + \textnormal{h.c.},\\
&\mathcal{Q}_{\alpha, \alpha'}^{(3)}(t) =
\frac{i}{\omega} \sum_{u \in \{+, -\}} \mathcal{D}_{L_\alpha}\big([L_{\alpha'}^{(u)}, \tr{\mathcal{A}}{\sigma_{\alpha'}^{(\bar{u})} \rho_\omega(t)}]\big),\\
&\text{If }\alpha = \alpha',\\
&\ \mathcal{Q}^{(4)}_{\alpha, \alpha'}(t) =
    \frac{2}{\omega^2}(\mathcal{D}_{L_\alpha^\dagger}-\mathcal{D}_{L_\alpha})(\tr{\mathcal{A}}{n_\alpha \rho_\omega(t)}), \\
&\text{If }\alpha \neq \alpha', \\
&\ \mathcal{Q}_{\alpha, \alpha'}^{(4)}(t) = 
    -\frac{1}{\omega^2}\smashoperator[l]{\sum_{\substack{u, u' \\ \in \{+, -\}}}} [L_\alpha^{(u)}, [L_{\alpha'}^{(u')}, \tr{\mathcal{A}}{\sigma_\alpha^{({\bar{u}})}\sigma_{\alpha'}^{(\bar{u}')}\rho_\omega(t)}]].
\end{align*}
\end{subequations}
\end{lemma}

To theoretically analyze the fidelity of the quantum simulation, we thus need to provide an upper bound on $\norm{\mathcal{R}(t)}_1$ which $\to 0$ as $\omega \to 0$. From the expression provided in Lemma~\ref{lemma:remainder}, we see that the terms $\mathcal{Q}_{\alpha, H_\text{sys}}^{(1)}, \mathcal{Q}_{\alpha, H_\text{sys}}^{(2)}, \mathcal{Q}_{\alpha, \alpha'}^{(3)}$ and $\mathcal{Q}_{\alpha, \alpha'}^{(4)}$, which contribute to the remainder, depend on the operators ${\tr{\mathcal{A}}{n_\alpha \rho_\omega(t)}}$, ${\tr{\mathcal{A}}{\sigma_\alpha \rho_\omega(t)}}$, ${\tr{\mathcal{A}}{ \sigma^\dagger_{\alpha'} \sigma_\alpha \rho_\omega(t)}}$ and ${\tr{\mathcal{A}}{\sigma_\alpha \sigma_{\alpha'} \rho_\omega(t)}}$. We expect these operators to be small --- to see this physically, we note that if all the ancillae were exactly in $\ket{0}$ at time $t$, then all of these operators would be exactly 0. In the analogue simulation, the ancillae are only weakly coupled to the system, with the coupling strength $\sim \omega$, while simultaneously being strongly damped. Thus, we could expect the ancillae to never be significantly excited during the simulation, and consequently we expect the operators $\mathcal{Q}_{\alpha, H_\text{sys}}^{(1)}, \mathcal{Q}_{\alpha, H_\text{sys}}^{(2)}, \mathcal{Q}_{\alpha, \alpha'}^{(3)}$ and $\mathcal{Q}_{\alpha, \alpha'}^{(4)}$ to be small. The next lemma translates this physical intuition into a rigorous upper bound,

\begin{lemma}
    \label{lemma:tr_sigma_bounds}
    Suppose $\rho_\omega(t)$ is the joint state of the system and ancilla qubits with the ancilla qubits initially being in state $\ket{0}$, then for all $\alpha, \alpha'$
    \begin{align*}
        &\norm{\tr{\mathcal{A}}{\sigma_{\alpha}\rho_\omega(t)}}_1 \leq \frac{\omega}{2} \text{ and },\\
        &\norm{\tr{\mathcal{A}}{\sigma_{\alpha}^\dagger \sigma_{\alpha'}\rho_\omega(t)}}_1,
        \norm{\tr{\mathcal{A}}{\sigma_{\alpha} \sigma_{\alpha'}\rho_\omega(t)}}_1\leq \frac{\omega^2}{4}.
    \end{align*}
\end{lemma}
To obtain these upper bounds, we develop a set of input-output equations for the ancilla lowering operator $\sigma_\alpha$. These are analogous to the input-output equations used in quantum optics \cite{carmichael2009open}. We provide a detailed proof of the lemma in Appendix~\ref{appendix:analysis_noiseless_proofs} --- to illustrate the basic idea behind the proof, we explicitly bound $\norm{\sigma_\alpha \rho_\omega(t)}_1$. We begin by noting that
\[
\sigma_\alpha \rho_\omega(t) = \mathcal{E}_\omega(t,0)(\sigma_{\alpha,l}(t)(\rho(0))),
\]
where $\mathcal{E}_\omega(t, s) = e^{\mathcal{L}_\omega (t - s)}$ is the channel generated by the Lindbladian of the quantum simulator and $\sigma_{\alpha, l}(t)$ is a superoperator defined by
\[
\sigma_{\alpha, l}(t) = \mathcal{E}_\omega^{-1}(t, 0)\sigma_{\alpha, l} \mathcal{E}_\omega(t, 0),
\]
where, as defined in Section~\ref{sec:notation}, $\sigma_{\alpha, l}(X) = \sigma_\alpha X$ is a superoperator that left-multiplies its argument by $\sigma_\alpha$.
We can now obtain a ``Heisenberg-like" equation of motion for $\sigma_{\alpha, l}(t)$:
\begin{align*}
    &\frac{d}{dt}\sigma_{\alpha, l}(t) =\mathcal{E}_\omega^{-1}(t, 0) \big(\sigma_{\alpha, l} \mathcal{L}_\omega - \mathcal{L}_{\omega} \sigma_{\alpha, l}\big) \mathcal{E}_{\omega}(t,0), \nonumber \\
    &=-i\omega \mathcal{E}_\omega^{-1}(t, 0)  [\sigma_{\alpha, l}, V_{\alpha, l} - V_{\alpha, r}] \mathcal{E}_\omega(t, 0) + \nonumber\\
    &\qquad \qquad 4 \mathcal{E}_\omega^{-1}(t, 0)  [\sigma_{\alpha, l}, \mathcal{D}_{\sigma_\alpha}] \mathcal{E}_\omega(t, 0), \nonumber\\
    &\numeq{1}-2 \sigma_{\alpha, l}(t) -i\omega \mathcal{E}_{\omega}^{-1}(t, 0) L_{\alpha, l} [\sigma_{\alpha, l}, \sigma_{\alpha, l}^\dagger] \mathcal{E}_{\omega}^{-1}(t, 0),
\end{align*}
where in obtaining (1) we have used the fact that $\sigma_{\alpha, l}$ commutes with $V_{\alpha, r}$ (since it left multiplies while $V_{\alpha, r}$ right multiplies) and with itself, and that $[\sigma_{\alpha, l}, \mathcal{D}_{\sigma_{\alpha}}] = \sigma_{\alpha, l}/2 $. Integrating this equation, we can obtain that
\begin{align*}
&\sigma_{\alpha, l}(t) = e^{-2t}\sigma_{\alpha, l} - \nonumber\\
&\qquad i\omega \int_0^t e^{-2(t - s)}\mathcal{E}_{\omega}^{-1}(s, 0) L_{\alpha, l} [\sigma_{\alpha, l}, \sigma_{\alpha, l}^\dagger]\mathcal{E}_{\omega}(s, 0)ds,
\end{align*}
which can be considered to be an input-output equation for $\sigma_{\alpha}$ since it relates its action on the quantum state at time $t$ to its action on the initial state at $t = 0$. Now, since $\sigma_\alpha$ annihilates the initial state, i.e.~$\sigma_{\alpha}\rho(0) = 0$, we obtain that
\begin{align*}
    &\sigma_\alpha \rho_\omega(t) \nonumber\\
    &=-i\omega\int_0^t e^{-2(t - s)}\mathcal{E}_\omega(t, s)\big(L_{\alpha}[\sigma_{\alpha}, \sigma_{\alpha}^\dagger]\mathcal{E}_\omega(s, 0)(\rho(0))\big)ds.
\end{align*}
We can now bound $\tr{\mathcal{A}}{\sigma_\alpha \rho_\omega(t)}$ --- in particular, using the fact that $\norm{[\sigma_\alpha, \sigma_\alpha^\dagger]} \leq 1$, $\norm{L_\alpha} \leq 1$  and that $\norm{\mathcal{E}_{\omega}(t, s)(X)}_1 \leq \norm{X}_1$ (i.e.~trace norm is contractive under application of  quantum channels), we obtain that
\begin{align*}
    \norm{\tr{\mathcal{A}}{\sigma_\alpha \rho_\omega(t)}}_1 \leq \omega \int_0^t e^{-2(t - s)} \norm{\rho(s)}_1 ds = \frac{\omega}{2},
\end{align*}
which is the bound provided in Lemma~\ref{lemma:tr_sigma_bounds}. A similar procedure can be followed to obtain bounds on $\norm{\tr{\mathcal{A}}{\sigma_\alpha^\dagger \sigma_{\alpha'}\rho_\omega(t)}}_1$ and $\norm{\tr{\mathcal{A}}{\sigma_\alpha \sigma_{\alpha'}\rho_\omega(t)}}_1$.

\begin{repproposition}{prop:general_lindbladian}
 With all the ancillae initialized to the state $\ket{0}$ and for any time $t > 0$, $\norm{\rho(t) - \textnormal{Tr}_\mathcal{A}(\rho_\omega(t/\omega^2))}_1 \leq \varepsilon$ can be obtained with
\[
\omega = \Theta\left(\frac{\varepsilon^{1/2}}{\big(M + 4M \norm{H_\textnormal{sys}} t + 4M^2 t\big)^{1/2}} \right),
\]
which implies that the required simulation time on the simulator $t_\textnormal{sim} = t/\omega^2$ scales as
\[
t_\textnormal{sim} =  \Theta\bigg(\frac{Mt + 4M \norm{H_\textnormal{sys}}t^2 + 4M^2 t^2}{\varepsilon} \bigg).
\]
\end{repproposition}
\begin{proof}
We note that the remainder defined in Eq.~\ref{eq:remainder_def},
\[
\frac{d}{dt}\tr{\mathcal{A}}{\rho_\omega(t)} - \omega^2 \mathcal{L} \tr{\mathcal{A}}{\rho_\omega(t)} = \mathcal{R}_\omega(t), 
\]
can be integrated to obtain
\begin{align}\label{eq:remainder_integral_form}
\text{Tr}_\mathcal{A}\bigg(\rho_\omega \bigg(\frac{t}{\omega^2}\bigg)\bigg) = \rho(t) + \int_0^{t/\omega^2} e^{\mathcal{L}(t - \omega^2 s)}\big(\mathcal{R}_\omega(s)\big) ds,
\end{align}
where $\rho(t) = e^{\mathcal{L}t} (\rho(0))$ is the target state to be simulated. Using the contractivity of the trace norm under the quantum channel $e^{\mathcal{L}t}$, we obtain that
\begin{align}\label{eq:obs_err_nonlocal}
    & \bignorm{\text{Tr}_{\mathcal{A}}\bigg({\rho_\omega\bigg(\frac{t}{\omega^2}\bigg)}\bigg) - \rho(t)}_1 \leq \int_0^{t/\omega^2} \norm{\mathcal{R}_\omega(s)}_1 ds.
\end{align}
We can now use the explicit expression for $\mathcal{R}_\omega$ in Lemma~\ref{lemma:remainder} together with Lemma~\ref{lemma:tr_sigma_bounds} to bound $\norm{\mathcal{R}_\omega(s)}_1$ term by term. Consider first bound on $q_{\alpha}$ --- note that since $\norm{L_\alpha} \leq 1$, $\norm{q_\alpha}_1 \leq 2$. Next, consider bounds on $\norm{\mathcal{Q}_{\alpha, H_\text{sys}}^{(j)}(s)}_1$ for $j \in \{1,2\}$ --- we obtain that
\begin{align*}
&\norm{\mathcal{Q}^{(1)}_{\alpha, H_\text{sys}}(s)}_1 \leq \frac{8}{\omega} \norm{H_\text{sys}} \bignorm{\tr{\mathcal{A}}{\sigma_\alpha^{(\bar{u})}\rho_\omega(s)}}_1  \leq 4\norm{H_\text{sys}},\\
&\norm{\mathcal{Q}^{(2)}_{\alpha, H_\text{sys}}(s)} \leq \norm{[L_\alpha^\dagger, L_\alpha [H_\text{sys},\tr{\mathcal{A}}{\rho_{\omega}(t)}]]} \leq 4 \norm{H_\text{sys}}.
\end{align*}
Similarly, a bound on $\norm{Q^{(3)}_{\alpha,\alpha'}(s)}_1$ can be obtained via
\[
\norm{\mathcal{Q}^{(3)}_{\alpha, \alpha'}(s)}_1 \leq \frac{1}{\omega}\sum_{u \in \{+, -\}} 4 \bignorm{\tr{\mathcal{A}}{\sigma_{\alpha'}^{(\bar{u})} \rho_\omega(s)}}_1  \leq 4.
\]
Consider next $\norm{\mathcal{Q}^{(4)}_{\alpha, \alpha'}(s)}_1$ --- for $\alpha \neq \alpha'$, we obtain that
\[\norm{\mathcal{Q}_{\alpha, \alpha'}^{(4)}(s)}_1 \leq \frac{4}{\omega^2}\sum_{u, u' \in \{+, -\}} \bignorm{\tr{\mathcal{A}}{\sigma_\alpha^{(\bar{u})} \sigma_{\alpha'}^{(\bar{u}')}\rho_\omega(s)}}_1 \leq 4.\] 
For $\alpha = \alpha'$, it similarly follows that
\begin{align*}
\norm{\mathcal{Q}^{(4)}_{\alpha, \alpha}(s)}_1 &= \frac{2}{\omega^2}\norm{(\mathcal{D}_{L_\alpha^\dagger} - \mathcal{D}_{L_\alpha})(\tr{\mathcal{A}}{n_\alpha \rho_\omega(t)}}_1, \nonumber\\
&\leq \frac{8}{\omega^2}\norm{\tr{\mathcal{A}}{n_\alpha \rho_\omega(t)}}_1  \leq 4.
\end{align*}
With these bounds, we obtain that
\begin{align*}
&\norm{\mathcal{R}_\omega(s)}_1  \nonumber\\
&\leq 2M \omega^2 e^{-2s} + \omega^4\int_0^s e^{-2(s - s')} \big( 8M \norm{H_\text{sys}}  + 8M^2 \big) ds' ,\\
& \leq 2M \omega^2 e^{-2s} + 4M\omega^4 s(\norm{H_\text{sys}} + M).
\end{align*}
Using this bound on $\norm{\mathcal{R}_\omega(s)}_1$ together with Eq.~\ref{eq:obs_err_nonlocal}, we obtain that
\begin{align*}
& \bignorm{\text{Tr}_{\mathcal{A}}\bigg({\rho_\omega\bigg(\frac{t}{\omega^2}\bigg)}\bigg) - \rho(t)}_1
\\
& \quad \leq \big(M\omega^2 + 4M\omega^2 \norm{H_\text{sys}}t + 4M^2 \omega^2 t \big),
\end{align*}
from which the theorem statement follows.
\end{proof}

\subsection{Geometrically local models}\label{sec:geometric_locality}
In this subsection, we confine ourselves to geometrically local models, which frequently appear in the study of many-body open quantum systems in physics. We make the assumption that the system $\mathcal{S}$ consists of a set of $n$ qudits arranged on a lattice $\mathbb{Z}^d$ and the target Lindbladian $\mathcal{L}$ is assumed to be of the form
\begin{subequations}\label{eq:geom_local_lind}
\begin{align}
\mathcal{L} = \sum_{\alpha} \mathcal{L}_\alpha,
\end{align}
where $\mathcal{L}_\alpha$ is a Lindbladian correspond to a jump operator $L_{\alpha}$ i.e.
\begin{align}
    \mathcal{L}_\alpha(X) = -i[h_\alpha, X] + \mathcal{D}_{L_\alpha}(X),
\end{align}
\end{subequations}
with $\norm{L_{\alpha}} \leq 1$, and Hamiltonian $h_\alpha$ with $\norm{h_\alpha} \leq 1$. The operators $L_{\alpha}$ and $h_\alpha$ are assumed to be supported on qudits in $S_\alpha$, which has a maximum  diameter of $a$, 
\[
\text{diam}(S_\alpha) = \max_{x, y \in S_\alpha} d({x, y}) \leq a.
\]
Such models are known have a finite velocity at which correlations can propagate across the lattice, formalized by the well-known Lieb-Robinson bounds. While originally derived for geometrically local Hamiltonian (i.e.~closed) systems, Lieb-Robinson bounds have been extended to open quantum systems \cite{poulin2010lieb}. Below, we quote a Lieb-Robinson bound from Ref.~\cite{barthel_kliesch_2012}, which we will use in the remainder of this section. To compactly express this bound, following Ref.~\cite{barthel_kliesch_2012}, it is convenient to define the degree of a subset $S \subseteq \mathbb{Z}^d$
\begin{align*}
\partial_S = \bigabs{\{\alpha | S_\alpha \cap S \neq \emptyset\}} \text{ and }\mathcal{Z} = \max_{\alpha} \partial(S_\alpha). 
\end{align*}
$\partial_S$ is thus the number of terms in the Lindbladian that act on the qudits in $S$, and consequently $\mathcal{Z}$ is an upper bound on the number of terms in the Lindbladian that intersect with any one term in the Lindbladian. For notational convenience, given a set $S \subset \mathbb{Z}^d$ of lattice sites, we will define
\[
\eta_S = \frac{\partial_S}{\mathcal{Z}}.
\]
The Lieb-Robinson bound from Ref.~\cite{barthel_kliesch_2012} is precisely stated below.
\begin{lemma}
\label{lemma:lieb_robinson}
Suppose $\mathcal{K}$ is a superoperator supported on region $S_\mathcal{K}$ and satisfies $\mathcal{K}(I) = 0$, and $O$ is a local observable, with $\norm{O} \leq 1$, supported on region $S_O$, then
\begin{align*}
    \norm{\mathcal{K}  e^{\mathcal{L}^\dagger t}(O)} \leq \eta_{S_O} \norm{\mathcal{K}}_{cb,\infty\to\infty}\exp\bigg(4e\mathcal{Z}t - \frac{d({S_\mathcal{K}, S_O})}{a}\bigg).
\end{align*}
\end{lemma}
\noindent The superoperator $\mathcal{K}$ that will typically arise in our analysis throughout this paper would either be commutator $\mathcal{K}(X) = [K, X]$ with some operator supported on $S_\mathcal{K}$ or the adjoint of a Lindbladian supported on $S_\mathcal{K}$ i.e. $\mathcal{K} = \mathcal{L}_\mathcal{K}^\dagger$ for a Lindbladian superoperator $\mathcal{L}_\mathcal{K}$. Both of these examples can be seen to satisfy $\mathcal{K}(I) = 0$.

As we will outline in the rest of this section, while analyzing quantum simulation of both dynamics and fixed points our strategy will be to use Lieb-Robinson bounds to obtain an upper bound on the remainder calculated in Lemma~\ref{lemma:remainder} that is uniform in the system size $n$ (the total number of qudits in the lattice). However, we note that the remainder expression, unlike the Lieb-Robinson bound given in Lemma~\ref{lemma:lieb_robinson}, has contributions from terms such as $\mathcal{Q}_{\alpha, \alpha'}^{(3)}, \mathcal{Q}_{\alpha, \alpha'}^{(4)}$, that involve the consecutive application of two superoperators supported on two different regions of the lattice. Motivated from this observation, we will need a Lieb-Robinson bound that accounts for the application of two superoperators on a Heisenberg picture observable as opposed to a single superoperator. We provide this in the lemma below, and it follows straightforwardly from the Lieb-Robinson bound in Lemma~\ref{lemma:lieb_robinson}.
\begin{lemma}
\label{lemma:2superopLR}
$\mathcal{K}_X$ and $\mathcal{K}_Y$ are superoperators supported on regions $X$ and $Y$ such that $\mathcal{K}_X(I) = \mathcal{K}_Y(I) = 0$ and $\textnormal{diam}(X),\textnormal{diam}(Y)\leq a$, then for any local observable $O$, with $\norm{O} \leq 1$, supported on region $S_O$
\begin{align*}
    &\norm{\mathcal{K}_X \mathcal{K}_Y e^{\mathcal{L}^\dagger t}(O)} \leq
    e \eta_{S_O}\norm{\mathcal{K}_X}_{cb,\infty\to\infty} \norm{\mathcal{K}_Y}_{cb,\infty\to\infty} \times \\
    &\qquad \qquad \exp\bigg(4e\mathcal{Z}t - \frac{1}{2a}(d({X, S_O}) + d({Y, S_O}))\bigg).
\end{align*}
\end{lemma}
\subsubsection{Dynamics}
With the two Lieb-Robinson bounds in Lemmas~\ref{lemma:lieb_robinson} and \ref{lemma:2superopLR}, we can now analyze the remainder expression given in Lemma~\ref{lemma:remainder} to obtain bounds that are uniform in the system size $n$ when we are interested in evaluating local observables. In particular, for a local observable $O$, let us consider the error between its ideal expected value
\[
\expect{O}_\text{ideal} := \text{Tr}[Oe^{\mathcal{L}t}(\rho(0))],
\]
and the expected value on the quantum simulator after evolving it for time $t / \omega^2$,
\begin{align}
\expect{O}_\text{sim} = \text{Tr}\bigg[O \rho_\omega\bigg(\frac{t}{\omega^2}\bigg)\bigg] = \text{Tr}\bigg[O\text{Tr}_{\mathcal{A}}\bigg({\rho_\omega\bigg(\frac{t}{\omega^2}\bigg)\bigg)}\bigg]. 
\end{align}
Using the integral form of the remainder definition, Eq.~\ref{eq:remainder_integral_form}, we can express the error between $\expect{O}_{\text{ideal}}$ and $\expect{O}_\text{sim}$ as
\begin{align}\label{eq:remainder_with_observable}
    &\bigabs{\expect{O}_\text{ideal} - \expect{O}_\text{sim}} = \bigabs{\int_0^{t/ \omega^2}\text{Tr}\bigg(O e^{\mathcal{L}(t - \omega^2 s)}\mathcal{R}_\omega(s) \bigg)ds}, \nonumber\\
    &\qquad =\frac{1}{\omega^2}\bigabs{\int_0^t \text{Tr}\bigg(O (t - s)\bigtr{\mathcal{A}}{\mathcal{R}_\omega\bigg(\frac{s}{\omega^2}\bigg)} \bigg)ds},
\end{align}
where $O(t) = e^{\mathcal{L}^\dagger t}(O)$ is the Heisenberg picture evolution of the observable $O$ under the \emph{target} Lindbladian $\mathcal{L}$. Using Lemma~\ref{lemma:remainder}, we express the remainder $\mathcal{R}_\omega(s)$ in terms of the operators $q_\alpha$, $\mathcal{Q}^{(1)}_{\alpha,H_\textnormal{sys}}(s)$, $\mathcal{Q}^{(2)}_{\alpha,H_\textnormal{sys}}(s)$, $\mathcal{Q}^{(3)}_{\alpha,\alpha'}(s)$, and $\mathcal{Q}^{(4)}_{\alpha,\alpha'}(s)$ and analyze the resulting terms individually. The next lemma shows that applying the Lieb-Robinson bounds in Lemmas~\ref{lemma:lieb_robinson} and \ref{lemma:2superopLR} provides upper bounds on these remainder terms that are uniform in the system size.
\begin{lemma}\label{lemma:bounds_remainder_lr}
Suppose $O$ is a local observable with $\norm{O} \leq 1$ supported on $S_O$, and for $\tau>0$, let $O(\tau) = \exp(\mathcal{L}^\dagger \tau)(O)$ where $\mathcal{L}$ is a geometrically local Lindbladian of the form in Eq.~\ref{eq:geom_local_lind}. Then for\ $q_\alpha$ as defined in Lemma~\ref{lemma:remainder}, then there is a non-decreasing piecewise continuous function $\nu$ such that $\nu(t) \leq O\left(t^d\right)$ as $t \to \infty$ and
        \[
        \sum_{\alpha}\bigabs{\tr{}{O(\tau) q_\alpha}} \leq  \nu(\tau),
        \]
and for $j \in \{1, 2, 3, 4\}$
        \begin{align*}
        &\sum_{\alpha, \alpha'}\bigabs{\tr{}{O(\tau)\mathcal{Q}_{\alpha, \alpha'}^{(j)}(s)}} \leq \nu^2(\tau),
        \end{align*}
where, for $j \in \{3, 4\}$, $\mathcal{Q}_{\alpha, \alpha'}^{(j)}$ is defined in Lemma~\ref{lemma:remainder} and for $j\in \{1, 2\}$, we define $\mathcal{Q}_{\alpha, \alpha'}^{(j)} = \mathcal{Q}_{\alpha, h_{\alpha'}}^{(j)}$ where $\mathcal{Q}_{\alpha, h}^{(j)}$ is defined in Lemma~\ref{lemma:remainder}.
\end{lemma}
This lemma follows directly from the application of the Lieb-Robinson bounds. The proof of this lemma is detailed in Appendix~\ref{appendix:geometric_locality_lemma_proofs}, but we provide a sketch here. Consider first $\sum_{\alpha}\abs{\text{Tr}(O(\tau) q_\alpha}$ --- using the explicit expression for $q_\alpha$ from Lemma~\ref{lemma:remainder}, we obtain that
\begin{align}\label{eq:example_lr_bound_1}
    &\sum_{\alpha}\abs{\text{Tr}(O(\tau) q_\alpha} = \sum_{\alpha}\abs{\text{Tr}[O(\tau)\mathcal{L}_\alpha(\rho(0))]}, \nonumber \\
    &\qquad =\sum_{\alpha}\bigabs{\text{Tr}[\mathcal{L}_\alpha^\dagger(O(\tau))\rho(0)]}, \nonumber \\
    &\qquad = \sum_{\alpha}{\norm{\mathcal{L}_\alpha^\dagger(O(\tau))}}.
\end{align}
Now, we can use the Lieb-Robinson bound in Lemma~\ref{lemma:lieb_robinson}, together with $\norm{\mathcal{L}_\alpha^\dagger}_{cb,\infty\to\infty} \leq 4$ to obtain
\begin{align}\label{eq:example_lr_bound_2}
    &\sum_{\alpha}\abs{\text{Tr}(O(\tau) q_\alpha}  \nonumber\\
    &\qquad \leq 4 \sum_{\alpha} \min\bigg(\eta_{S_O}\exp\bigg(4e\mathcal{Z}\tau - \frac{d(S_\alpha, S_O)}{a}\bigg), 1\bigg).
\end{align}
However, the summation in the above equation is simply a summation over an exponential over the lattice $\mathbb{Z}^d$, which will be convergent and independent of the size of system. A careful analysis of this summation, performed in Appendix~\ref{appendix:geometric_locality_lemma_proofs}, yields that, as stated in the lemma, it can be upper bounded by a non-decreasing function which is $O(\tau^d)$. Similarly, as an example, consider $\sum_{\alpha, \alpha'}\text{Tr}(O(\tau) \mathcal{Q}^{(3)}_{\alpha, \alpha'}(s))$ --- by using the explicit expression for $\mathcal{Q}^{(3)}_{\alpha, \alpha'}(s)$ from Lemma~\ref{lemma:remainder}, we obtain that
\begin{align*}
&\sum_{\alpha, \alpha'}\bigabs{\text{Tr}(O(\tau) \mathcal{Q}^{(3)}_{\alpha, \alpha'}(s))} \nonumber\\
&\quad= \frac{1}{\omega}\sum_{\alpha, \alpha'} \sum_{ u \in \{-, +\}}\bigabs{\text{Tr}\big(O(\tau)\mathcal{D}_{L_\alpha}([L_{\alpha'}^{(u)}, \text{Tr}_{\mathcal{A}}(\sigma_{\alpha'}^{(\bar{u})}\rho_\omega(s))])\big)},\nonumber \\
&\quad=\frac{1}{\omega} \sum_{\alpha, \alpha'} \sum_{u \in \{-, +\}} \bigabs{\text{Tr}\big([\mathcal{D}_{L_{\alpha}}^\dagger(O(\tau)), L_{\alpha'}^{(u)}]\text{Tr}_{\mathcal{A}}(\sigma_{\alpha'}^{(\bar{u})}\rho_\omega(s))\big)}, \nonumber\\
&\quad \leq \frac{1}{\omega} \sum_{\alpha, \alpha'} \sum_{u \in \{-, +\}} \norm{[\mathcal{D}_{L_\alpha}^\dagger(O(\tau)), L_{\alpha'}^{(u)}]} \norm{\text{Tr}_{\mathcal{A}}(\sigma_{\alpha'}^{(\bar{u})}\rho_\omega(s))}_1, \nonumber\\
&\quad \leq \frac{1}{2}\sum_{\alpha, \alpha'}\sum_{u \in \{-, +\}}\norm{[\mathcal{D}_{L_\alpha}^\dagger(O(\tau)), L_{\alpha'}^{(u)}]},
\end{align*}
where in the last step we have used Lemma~\ref{lemma:tr_sigma_bounds}. We can note that both $\mathcal{D}_{L_\alpha}^\dagger$ and $[\cdot, L_{\alpha'}^{(u)}]$ are superoperators that are respectively supported on $S_\alpha$ and $S_{\alpha'}$ and  have the identity operator in their kernel. Thus from Lemma~\ref{lemma:2superopLR}, as well as using the fact that $\norm{\mathcal{D}_{L_\alpha}^\dagger}_{cb,\infty\to\infty}, \norm{[L_\alpha^{(u)}, \cdot]}_{cb,\infty\to\infty} \leq 2$ we obtain that
\begin{align*}
&\sum_{\alpha, \alpha'}\bigabs{\text{Tr}(O(\tau) \mathcal{Q}^{(3)}_{\alpha, \alpha'}(s))}\nonumber\\
&\leq 2 \sum_{\alpha, \alpha'} \min\bigg(e\eta_{S_O}\exp\bigg(4e\mathcal{Z}\tau - \nonumber\\
&\qquad \qquad \qquad \qquad \frac{1}{2a}\big(d(S_\alpha, S_O) + d(S_{\alpha'}, S_O)\big) \bigg), 1\bigg).
\end{align*}
This summation, which is approximately a double summation of a two-dimensional exponential function on a lattice, can again be analyzed and upper bounded by a non-decreasing function that grows as $O(\tau^{2d})$. The rest of the Lemma~\ref{lemma:bounds_remainder_lr} can be proved similarly --- the proofs are outlined in detail in Appendix~\ref{appendix:geometric_locality_lemma_proofs}.

\begin{repproposition}{prop:dynamics_noiseless}
Suppose $\mathcal{L}$ is a $d-$dimensional geometrically local Lindbladian and $O$ with $\norm{O}\leq 1$ is a local observable. To achieve an additive error $\varepsilon$ in the expected local observable at time $t$ with the analogue quantum simulator, we need to choose $\omega = \Theta(t^{-(2d + 1)}\sqrt{\varepsilon})$ which corresponds to a simulator evolution time $t_\textnormal{sim} = t / \omega^2 = \Theta(t^{4d + 3} \varepsilon^{-1})$.
\end{repproposition}
\begin{proof}
    The proposition follows by combining the expression of $\mathcal{R}_\omega(s)$ from Lemma~\ref{lemma:remainder} with Eq.~\ref{eq:remainder_integral_form} and then using Lemma~\ref{lemma:bounds_remainder_lr}. Using $H_\text{sys} = \sum_{\alpha} h_\alpha$, we first decompose $\mathcal{Q}_{\alpha, H_\text{sys}}^{(j)} = \sum_{\alpha'}\mathcal{Q}_{\alpha, h_{\alpha'}}^{(j)}$ for $j \in \{1, 2\}$, where $\mathcal{Q}_{\alpha, h}^{(j)}$ is defined in Lemma~\ref{lemma:remainder}. For notational convenience, for $j \in \{1, 2\}$ we will define $\mathcal{Q}_{\alpha, {\alpha'}}^{(j)} = \mathcal{Q}_{\alpha, h_{\alpha'}}^{(j)}$. From Eq.~\ref{lemma:remainder} and Lemma~\ref{lemma:remainder}, it follows that
    \begin{widetext}
    \begin{align}\label{eq:summation_local_observable_remainder}
        \abs{\expect{O}_\text{ideal} - \expect{O}_\text{sim}} &\leq \sum_{\alpha}\bigabs{\int_0^t  {\text{Tr}(O(t - s) q_\alpha)} e^{-2s/\omega^2}ds} + \omega^2\sum_{j = 1}^4 \sum_{\alpha, \alpha'} \bigabs{\int_0^t \int_0^{s/\omega^2} e^{-2(s/\omega^2 - s')} {\text{Tr}(O(t - s)\mathcal{Q}^{(j)}_{\alpha, \alpha'}(s'))}ds'}, \\
        &\numleq{1} \int_0^t \nu(t - s) e^{-2s/\omega^2}ds + 4\omega^2  \int_0^t \int_0^{s/\omega^2}e^{-2(s/\omega^2 - s')} \nu^2(t - s) ds', \nonumber\\
        &\numleq{2} \nu(t) \int_0^t e^{-2s/\omega^2}ds + 4\omega^2 \nu^2(t) \int_0^t \int_0^{s/\omega^2}e^{-2(s/\omega^2 - s')}ds'\leq \frac{\omega^2}{2}\nu(t) + 2\omega^2 t\nu^2(t),\nonumber
    \end{align}
    \end{widetext}
    where in (1) we have used Lemma~\ref{lemma:bounds_remainder_lr} --- in particular, the fact that $\nu(t)$ is a piecewise continuous function and thus can be integrated to give a legitimate upper bound. In (2), we have used the fact that the function $\nu(t)$, introduced in Lemma~\ref{lemma:bounds_remainder_lr}, is an non-decreasing function. Noting, again from Lemma~\ref{lemma:bounds_remainder_lr}, that $\nu(t) \leq O(t^d)$, we obtain that $\abs{\expect{O}_\text{ideal} - \expect{O}_\text{sim}} \leq \omega^2 O(t^{2d + 1})$ --- consequently, to attain a precision $\varepsilon$ in the simulated observable, we need to choose $\omega = \Theta(t^{-(d + 1/2)} \varepsilon^{1/2})$ which corresponds to a quantum simulation time of $t_\text{sim} = \Theta(t^{2d + 1} \varepsilon^{-1})$ 
\end{proof}
\subsubsection{Long-time dynamics and fixed point}
Next, we consider the problem of simulating long-time dynamics and fixed point values of a local observable that is rapidly mixing for given a spatially local Lindbladian. As outlined in Section~\ref{sec:theory_result}, a local observable in a spatially local Lindbladian $\mathcal{L}$, with a unique fixed point $\sigma$, is considered to be rapidly mixing if it satisfies, for any initial state $\rho(0)$,
\begin{align*}
\abs{\text{Tr}(Oe^{\mathcal{L}t}(\rho(0)) - \text{Tr}(O\sigma)} \leq k(\abs{S_O}, \gamma) e^{-\gamma t},
\end{align*}
where $k(l, \gamma)$ is a function that is $\text{poly}(l)$ for a fixed $\gamma$ and $\text{exp}(O(\gamma^{-\kappa}))$ for a fixed $l$. Since the rapid mixing condition is satisfied for any initial state $\rho(0)$, it can equivalently be reformulated in the Heisenberg picture
\begin{align}\label{eq:rapid_mixing_observable_v2}
\norm{O(t) - \text{Tr}(O\sigma) I} \leq k(\abs{S_O}, \gamma) e^{-\gamma t},
\end{align}
where $O(t) = e^{\mathcal{L}^\dagger t}(O)$ is the Heisenberg-picture observable corresponding to $O$. For the problem of simulating such rapidly mixing observables at time $t$ in geometrically local Lindbladians, we show that the $\omega$ required is uniform in both system size and time. This also allows us to use the quantum simulator to approximate the fixed point expectation value of the local observable.

To establish this result, we start from Eq.~\ref{eq:summation_local_observable_remainder} for the observable error, analyze it term by term and provide a bound uniform in $t$. The key lemma that we establish, which is the counterpart of Lemma~\ref{lemma:bounds_remainder_lr} for dynamics, is the following
\begin{lemma}\label{lemma:error_term_rapid_mixing}
    Suppose $O$ is a local observable with $\norm{O} \leq 1$ supported on $S_O$, and for $\tau>0$, let $O(\tau) = \exp(\mathcal{L}^\textnormal{\dagger} \tau)(O)$ where $\mathcal{L}$ is a geometrically local Lindbladian of the form in Eq.~\ref{eq:geom_local_lind}. Furthermore, suppose $O$ is rapidly mixing with respect to $\mathcal{L}$ and satisfies Eq.~\ref{eq:rapid_mixing_observable_v2} with $k(\abs{S_O}, \gamma) \leq O(\exp(\gamma^{-\kappa}))$ as $\gamma\to0$. Then for\ $q_\alpha$ as defined in Lemma~\ref{lemma:remainder},
        \[
        \sum_{\alpha}\bigabs{\int_0^t \tr{}{O(t - s) q_\alpha}e^{-2s/\omega^2}ds} \leq  \omega^2 \lambda^{(1)}(\gamma),
        \]
        where $\lambda^{(1)}(\gamma) \leq O(\gamma^{-d(\kappa + 1)})$ as $\gamma \to 0$; and for $j \in \{1, 2, 3, 4\}$
        \begin{align*}
        &\sum_{\alpha, \alpha'}\bigabs{\int_0^t \int_0^{s/\omega^2} \tr{}{O(t - s)\mathcal{Q}_{\alpha, \alpha'}^{(j)}(s')} e^{-2(s/\omega^2 - s')}ds' ds} \nonumber\\
        &\qquad \qquad \qquad \qquad \qquad \qquad \qquad \qquad \quad \leq  \lambda^{(2)}(\gamma),
        \end{align*}
where $\lambda^{(2)}(\gamma) \leq O(\gamma^{-(2d + 1)(\kappa + 1)})$ as $\gamma\to0$ and for $j \in \{3, 4\}$, $\mathcal{Q}_{\alpha, \alpha'}^{(j)}$ is defined in Lemma~\ref{lemma:remainder}, and for $j\in \{1, 2\}$, we define $\mathcal{Q}_{\alpha, \alpha'}^{(j)} = \mathcal{Q}_{\alpha, h_{\alpha'}}^{(j)}$ where $\mathcal{Q}_{\alpha, h}^{(j)}$ is defined in Lemma~\ref{lemma:remainder}.
\end{lemma}

Importantly, this lemma establishes bounds on the error terms appearing in Eq.~\ref{eq:summation_local_observable_remainder} that are uniform in the time $t$ --- this is a consequence of not only the spatial locality of the target Lindbladian, but also an explicit accounting of the rapid mixing of the target observable (Eq.~\ref{eq:rapid_mixing_observable_v2}). We provide a detailed proof of this lemma in Appendix~\ref{appendix:geometric_locality_lemma_proofs} --- below, we illustrate how we obtain bounds uniform in $t$. The key idea is instead of using the Lieb-Robinson bounds to obtain an upper bound on the observable error for \emph{any} time, we use it for only short times and use the rapid mixing property (Eq.~\ref{eq:rapid_mixing_observable_v2}) to upper bound the observable error for long times. For instance, consider upper bounding 
\begin{align}\label{eq:example_rapid_mixing_obs}
    &\sum_{\alpha}\bigabs{\int_0^t \text{Tr}(O(t - s) q_\alpha) e^{-2s/\omega^2}ds} \nonumber \\
    &\qquad \qquad  =\sum_{\alpha}\bigabs{\int_0^t \text{Tr}(O(s) q_\alpha) e^{-2(t - s)/\omega^2}ds}.
\end{align}
Suppose $t_\alpha$ is a time cutoff for separating long times from short times, which we will choose later and possibly allow it to vary with $\alpha$. We then split the integral with respect to time in Eq.~\ref{eq:example_rapid_mixing_obs} into a short time integral $e_{\alpha}^{\leq}(t, t_\alpha)$ and a long time integral $e_{\alpha}^{\geq}(t, t_\alpha)$ i.e.
\begin{align*}
\bigabs{\int_0^t \text{Tr}(O(s) q_\alpha) e^{-2(t - s)/\omega^2}ds} \leq  e_\alpha^{\leq}(t, t_\alpha) + e_\alpha^{\geq}(t, t_\alpha),
\end{align*}
where
\begin{align*}
    &e_\alpha^{\leq}(t, t_\alpha) = \bigabs{\int_0^{\min(t, t_\alpha)} \text{Tr}(O(s) q_\alpha)e^{-2(t - s)/\omega^2}ds} \text{ and}\\
    &e_\alpha^{\geq}(t, t_\alpha) = \bigabs{\int_{\min(t, t_\alpha)}^t \text{Tr}(O(s) q_\alpha)e^{-2(t - s)/\omega^2}ds}.
\end{align*}
The short-time integral, $e_\alpha^{\leq}(t, t_\alpha)$, can be analyzed similar to the case of dynamics --- more specifically, proceeding similar to the analyses in Eqs.~\ref{eq:example_lr_bound_1} and \ref{eq:example_lr_bound_2}, we obtain that
\begin{align}\label{eq:main_text_short_time_bound}
    &e_\alpha^{\leq}(t, t_\alpha) \nonumber\\
    &\quad \leq 2\omega^2\min\bigg(\eta_{S_O}\exp\bigg(4e\mathcal{Z}\min(t, t_\alpha) - \frac{d(S_\alpha, S_O)}{a}\bigg), 1\bigg),\nonumber\\
    &\quad \numleq{1} 2\omega^2\min\bigg(\eta_{S_O}\exp\bigg(4e\mathcal{Z} t_\alpha - \frac{d(S_\alpha, S_O)}{a}\bigg), 1\bigg).
\end{align}
where in (1) we have used the fact that $\min(\eta_{S_O}\exp(4e\mathcal{Z}s - {d(S_\alpha, S_O)}/{a}), 1)$ is a non-decreasing function of $s$.

Next, for bounding the long-time integral, $e^{\geq}_{\alpha}(t, t_\alpha)$, we explicitly use the rapid-mixing property of $O$ (Eq.~\ref{eq:rapid_mixing_observable_v2}) --- in particular, we note that for 
\begin{align}\label{eq:e_geq_inequality_1}
&\abs{\text{Tr}(O(s) q_\alpha)} \nonumber \\
&\quad =\abs{\text{Tr}((O(s) - \text{Tr}(O \sigma)I) q_\alpha) + \text{Tr}(O\sigma) \text{Tr}(q_\alpha)}, \nonumber \\
&\quad \leq \abs{\text{Tr}((O(s) - \text{Tr}(O\sigma) I) q_\alpha)} + \abs{\text{Tr}(q_\alpha)} \abs{\text{Tr}(O\sigma)}, \nonumber \\
&\quad \numleq{1} \norm{O(s) - \text{Tr}(O\sigma) I} \norm{q_\alpha}_1, \nonumber \\
&\quad \numleq{2} 4k(\abs{S_O}, \gamma) e^{-\gamma s},
\end{align}
where in (1) we have used the expression for $q_\alpha$ from Lemma~\ref{lemma:remainder} to obtain $\text{Tr}(q_\alpha) = \text{Tr}(\mathcal{L}_\alpha(\rho(0))) = 0$ and in (2) we have used the fact that, by assumption, $O$ is a rapid mixing observable and hence satisfies Eq.~\ref{eq:rapid_mixing_observable}. Now, note that if $t \leq t_\alpha$ (or $\min(t, t_\alpha) = t$), 
\[
e^{\geq}_{\alpha}(t, t_\alpha) = 0.
\]
On the other hand, if $t \geq t_\alpha$ (or $\min(t, t_\alpha) = t_\alpha$), 
\begin{align}\label{eq:e_geq_inequality_2}
e^{\geq}_{\alpha}(t, t_\alpha) &\leq \int_{t_\alpha}^t \abs{\text{Tr}(O(s)q_\alpha)}e^{-2(t - s)/\omega^2} ds, \nonumber \\
&\leq 4k(\abs{S_O}, \gamma)\int_{t_\alpha}^t e^{-\gamma s} e^{-2(t - s)/\omega^2}ds, \nonumber \\
&\leq 4k(\abs{S_O}, \gamma) e^{-\gamma t_\alpha} \int_{t_\alpha}^{t} e^{-2(t - s)/\omega^2}ds,\nonumber \\
&\leq 2\omega^2 k(\abs{S_O}, \gamma) e^{-\gamma t_\alpha}.
\end{align}
Furthermore, it is also true that for $t \geq t_\alpha$,
\begin{align}\label{eq:e_geq_inequality_3}
e^{\geq}_{\alpha}(t, t_\alpha)  &\leq \int_{t_\alpha}^t \abs{\text{Tr}(O(s)q_\alpha)}e^{-2(t - s)/\omega^2} ds, \nonumber \\
&\leq \int_{t_\alpha}^t \norm{O} \norm{q_\alpha}_1 e^{-2(t - s)/\omega^2}ds \leq 2\omega^2.
\end{align}
Combining Eqs.~\ref{eq:e_geq_inequality_1}, \ref{eq:e_geq_inequality_2} and \ref{eq:e_geq_inequality_3}, we obtain that for any $t, t_\alpha$,
\begin{align}\label{eq:main_text_long_time_bound}
e_\alpha^{\geq}(t, t_\alpha) \leq 2\omega^2 \min(k(\abs{S_O}, \gamma) e^{-\gamma t_\alpha}, 1).
\end{align}
We note already that the bounds on both the short-time integral (Eq.~\ref{eq:main_text_short_time_bound}) and the long-time integral (Eq.~\ref{eq:main_text_long_time_bound}) are independent of the time $t$, but instead depend on the time cutoff $t_\alpha$ that separate long times from short times. We can now choose $t_\alpha$ in such as a way that $\sum_{\alpha}e^{\leq}_\alpha(t, t_\alpha)$ and $\sum_{\alpha}e^{\geq}_\alpha(t, t_\alpha)$ are uniform in system size. A specific choice for $t_\alpha$ is to choose them to be
\[
t_\alpha = \frac{1}{8e\mathcal{Z}a}d(S_\alpha, S_O),
\]
i.e.~choose the time-cutoff to grow with distance of $S_\alpha$ from the observable. With this choice, we note that
\begin{align*}
&\sum_{\alpha}e^{\leq}_\alpha(t, t_\alpha) \leq 2\omega^2 \sum_\alpha \min\bigg(\eta_{S_O}\exp\bigg(-\frac{d(S_\alpha, S_O)}{2a}\bigg), 1\bigg), \\
&\sum_{\alpha}e^{\geq}_{\alpha}(t, t_\alpha) \nonumber\\
&\qquad \leq 2\omega^2 \sum_\alpha \min\bigg(k(\abs{S_O}, \gamma) \exp\bigg(-\frac{\gamma d(S_\alpha, S_O)}{8e\mathcal{Z}a}\bigg), 1\bigg),
\end{align*}
both of which, since they are sums of a decaying exponential function on a lattice, converge and are independent of the system size. As is shown in Appendix~\ref{appendix:geometric_locality_lemma_proofs}, the other error terms can be analyzed similarly. Furthermore, a more careful analysis of these summations allow us to also calculate their scaling with the parameter $\gamma$ to obtain the exact upper bounds quoted in Lemma~\ref{lemma:error_term_rapid_mixing}.

With Lemma~\ref{lemma:error_term_rapid_mixing}, we can now establish the following proposition.

\begin{repproposition}{prop:fp_noiseless}
Suppose $\mathcal{L}$ is a $d-$dimensional geometrically local Lindbladian and $O$ with $\norm{O}\leq 1$ is a local observable supported on $O(1)$ lattice sites satisfying rapid mixing (Eq.~\ref{eq:rapid_mixing_observable}). To achieve an additive error $\varepsilon$ in the expected local observable at time $t$ with the analogue quantum simulator, we need to choose
\[
\omega = \Theta\left(\gamma^{(d+1/2)(\kappa + 1)}\sqrt{\varepsilon}\right)
\]
which corresponds to a simulator evolution time 
\[
t_\textnormal{sim} = t / \omega^2 = \Theta\left(t\gamma^{-(\kappa + 1)(2d + 1)} \varepsilon^{-1}\right).
\]
\end{repproposition}
\begin{proof}
    This proposition can now be proved by directly combining Lemma~\ref{lemma:error_term_rapid_mixing} and Eq.~\ref{eq:summation_local_observable_remainder} --- in particular, we obtain that for any time $t$,
    \[
    \abs{\expect{O}_\text{ideal} - \expect{O}_\text{sim}} \leq \omega^2\left( \lambda^{(1)}(\gamma) + 4 \lambda^{(2)}(\gamma)\right).
    \]
    In the limit of $\gamma \to 0$, we obtain that $\abs{\expect{O}_\text{ideal} - \expect{O}_\text{sim}} \leq \omega^2 O\left(\gamma^{-(2d + 1)(\kappa + 1)}\right)$. Thus, to obtain a precision $\varepsilon$ in the simulated observable i.e.~to achieve $\abs{\expect{O}_\text{ideal} - \expect{O}_\text{sim}} \leq \varepsilon$, we need to choose $\omega = \Theta\left(\varepsilon^{1/2}\gamma^{(d+1/2)(\kappa + 1)}\right)$.
\end{proof}

\subsection{Quantum advantage over classical algorithms}
As discussed in the previous subsection, for the problems of simulating local observables in constant time dynamics for spatially local Lindbladians as well as in long-time dynamics and fixed points for rapidly mixing Lindbladians, an analogue quantum simulator can be very efficient, even obtaining a run-time that is independent of the system size. It is a natural question to ask if these problems are hard to solve on a classical computer, so that we can expect a quantum advantage with the analogue quantum simulation method discussed so far. In this section, we establish that, subject to the complexity assumption of BQP $\neq$ BPP, we indeed expect superpolynomial advantage with a quantum simulator applied to both of these problems.
\begin{figure*}
\centering
\includegraphics[scale=0.42]{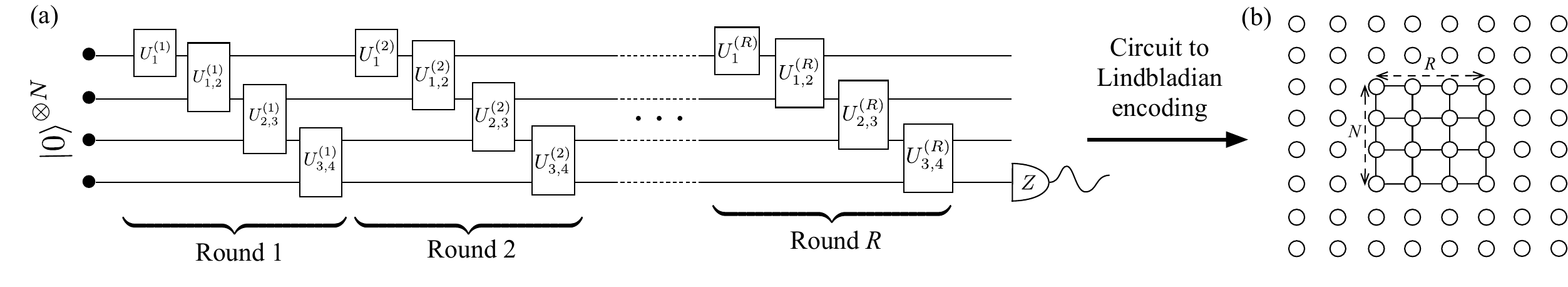}
\caption{Circuit architecture of the circuit being encoded into the geometrically local Lindbladian. In each round, the unitaries are applied in sequence, first on the first qubit, then on qubits 2 and 3, then on qubits 3 and 4 and so on.}
\label{fig:circuit_format}
\end{figure*}

The separation between the runtime of the analogue quantum simulator and classical algorithms is proven based on a technique to encode a quantum circuit on $N$ qubits and of depth $T = \text{poly}(N)$ into the fixed point of a geometrically local Lindbladian in two dimensions. This problem has been previously addressed for 5-local (but not geometrically local) Lindbladians \cite{verstraete2009quantum}. A related line of work encoded a quantum circuit on $N$ qubits and of depth $ \text{poly}(N)$ into the ground state of geometrically local Hamiltonians, starting with Hamiltonians in two dimensions \cite{aharonov2008adiabatic, oliveira2005complexity}, followed by Hamiltonians even in one dimension \cite{aharonov2009power}. We build upon the techniques developed in Ref.~\cite{verstraete2009quantum} for quantum circuit to 5-local Lindbladian encoding, and in Ref.~\cite{aharonov2008adiabatic} for quantum circuit to 2D spatially local Hamiltonian encoding and encode a quantum circuit into the fixed point of a 2D spatially local Lindbladian. Furthermore, one of our key technical contributions is to establish that the resulting Lindbladian has a mixing time of at most $\text{poly}(N)$, where $N$ is the number of qubits in the circuit, and that its depth is $\text{poly}(N)$. More precisely, we consider circuits of architecture shown in Fig.~\ref{fig:circuit_format} and establish that
\begin{lemma}\label{lemma:encoded_lindbladian}
    Suppose we are given a quantum circuit $\mathcal{C}$ on $N$ qubits with architecture shown in Fig.~\ref{fig:circuit_format}(a)  and with $R$ rounds of gates. Then, there exists a two-dimensional spatially local Lindbladian $\mathcal{L}$ on a $N\times R$ grid of 6-level qudits with a unqiue fixed point $\sigma$, as well as a local observable $O$ such that 
    \[
    \textnormal{Tr}(O\sigma) = \frac{1}{2NR }z_{\mathcal{C}},
    \]
    where $z_{\mathcal{C}}$ is the expected value of a pauli-${Z}$ operator on the last qubit at the output of $\mathcal{C}$. Furthermore, for any initial state $\rho(0)$ of the two-dimensional grid of qudits,
    \[
    \norm{e^{\mathcal{L}t}(\rho(0)) - \sigma}_1 \leq c_0(N, R) \exp(-\gamma_0(N, R) t),
    \]
    where $c_0(N, R) = O(N^6 R^2 \exp(O(NR)))$ and $\gamma_0(N, R) = \Theta(N^{-3}R^{-3})$.
\end{lemma}
\noindent The circuit is described in terms of rounds instead of layers due to the nature of the circuit-to-2D-lindbladian encoding scheme which we describe in Appendix~\ref{appendix:circuit_to_geom_local_lindbladian_encoding}. However we point out that this circuit architecture (Fig.~\ref{fig:circuit_format}) is not restrictive --- any quantum circuit on $N$ qubits and depth $\text{poly}(N)$ can be expressed in this format with at most $\text{poly}(N)$ rounds. Furthermore, we also note that, while the Lindbladian in Lemma~\ref{lemma:encoded_lindbladian} is on $NR$ qudits, it can trivially be embedded into the center of an infinite grid of qudits where the qudits other than these $NR$ qudits are uncoupled and only experience a single qudit Lindbladian with a unique fixed point [Fig.~\ref{fig:circuit_format}(b)]. The observable $O$ in Lemma~\ref{lemma:encoded_lindbladian} will then have a well-defined thermodynamic limit and $N, R$ will effectively become parameters of the Lindbladian on which the observable $O$ depends instead of being the system size.

For quantum circuits with $R = \text{poly}(N)$ rounds, Lemma~\ref{lemma:encoded_lindbladian} shows that it takes time $t = O(\text{poly}(N)\log(\varepsilon^{-1}))$ for the encoding Lindbladian $\mathcal{L}$ to converge $\varepsilon-$close to its fixed point $\sigma$. Then, measuring the local observable $O$ to a precision of $1 / (2NR) = 1 /\text{poly}(N)$ effectively measures the $Z$ operator on the last qubit at the output of the quantum circuit. Since for any decision problem in the BQP complexity class, measuring the $Z$ operator in the last qubit is the only required measurement for solving the problem, Lemma~\ref{lemma:encoded_lindbladian} indicates that being able to take a local observable in the fixed point of any spatially local two-dimensional Lindbladian, even restricting to the class of Lindbladians that converge to the fixed point in time that scales at most polynomially in the system size, is a sufficient resource to perform an arbitrary quantum computation. The construction of this Lindbladian closely follows the circuit-to-Hamiltonian ground state mapping presented in Ref.~\cite{aharonov2008adiabatic}. We detail this construction, the calculation of its fixed point, as well as analysis of the convergence rate of the Lindbladian to its fixed point in Appendix~\ref{appendix:circuit_to_geom_local_lindbladian_encoding}. Based on this lemma, we obtain the following proposition.
\begin{repproposition}{prop:quantum_advantage}
There cannot exist a classical algorithm that can, for every geometrically local 2D Lindbladian and a corresponding rapidly mixing local observable, compute the fixed point expected value of the local observable to additive error $\varepsilon$ in time $\textnormal{poly}(\gamma^{-1}, \varepsilon^{-1})$, unless \textnormal{BQP = BPP}.
\end{repproposition}
\begin{proof}
    Assume the contrary i.e.~ that there is a randomized classical algorithm that can obtain the fixed point expectation of a rapidly mixing local observable to an additive error $\varepsilon$ in time $\textnormal{poly}(\gamma^{-1}, \varepsilon^{-1})$. In such case, any decision problem $\in$ BQP could be solved with this algorithm. To see this, note that the decision problem parameterized by the problem size $m$ could be solved by measuring one qubit at the output of a  quantum circuit on $N = \textnormal{poly}(m)$ qubits with $R = \textnormal{poly}(m)$ rounds. By Lemma~\ref{lemma:encoded_lindbladian}, the expectation value of the pauli-$\text{Z}$ operator on the output qubit in such a quantum circuit can then be embedded into the fixed point expectation value of a rapidly mixing local observable with $\gamma^{-1} = \text{poly}(m)$. Furthermore, as given by Lemma~\ref{lemma:encoded_lindbladian}, the expected of value of this observable is $1 / (2NR) \leq O(1/\text{poly}(m))$ times smaller than the expected value of the pauli-$\text{Z}$ operator, and thus needs to be computed to an $O(1/\text{poly}(m))$ precision to effectively simulate the encoded circuit. Since this can be done with a randomized $\text{poly}(\gamma^{-1}, \varepsilon^{-1})$ classical algorithm in $\text{poly}(m)$ time, we contradict the complexity assumption of $\text{BQP} \neq \text{BPP}$.  
\end{proof}
\section{Stability to errors}

\subsection{Stability analysis}
Next, we consider the impact of noise on the quantum simulator. We will consider the model of the noisy quantum simulator given in Eq.~\ref{eq:lindbladian_omega_delta} --- the quantum simulator Lindbladian, in the presence of errors, will be given by
\[
\mathcal{L}_{\omega, \delta}  = \mathcal{L}_\omega + \delta \sum_\beta \mathcal{N}_\beta,
\]
where $\mathcal{N}_\beta$ is a Lindbladian acting on the system qudits and/or the ancillae in a geometrically local region $S_\beta'$. We will assume that the subset $S_\beta'$ at most intersects with $\mathcal{Z}'$ other subsets $S'_\beta$ or $S_\alpha$, where $S_\alpha$ are the subsets of system qudits on which the target Lindbladian is defined (i.e.~Eq.~\ref{eq:geom_local_lind}). More specifically, for all $\beta$,
\[
\abs{\{\beta' : S'_\beta \cap S_{\beta'}' \neq \emptyset \}} + \abs{\{\alpha : S'_\beta \cap S_{\alpha} \neq \emptyset \}} \leq \mathcal{Z}'.
\]
In the remainder of this subsection, we will carefully analyze the error incurred in using the noisy quantum simulator. Our goal would be to show that, despite the hardware error rate $\delta \neq 0$, the quantum simulator can still obtain local observables to a precision that depends only on the hardware error rate and not on the system size.

Our analysis of the stability relies strongly on the extension of Lemmas~\ref{lemma:remainder} and \ref{lemma:tr_sigma_bounds} to the noisy setting i.e.~where the quantum simulator dynamics are described by the Lindbladian in Eq.~\ref{eq:lindbladian_omega_delta}. In particular, as is shown in Appendix~\ref{appendix:analysis_noisy}, Lemma~\ref{lemma:remainder} can be modified to obtain an upper bound on the remainder, $\mathcal{R}_{\omega, \delta}(t)$, corresponding to the time evolution of the noisy quantum simulator i.e.
\[
\mathcal{R}_{\omega, \delta}(t) = \frac{d}{dt} \tr{\mathcal{A}}{\rho_{\omega, \delta}(t)} - \omega^2 \mathcal{L} \tr{\mathcal{A}}{\rho_{\omega, \delta}(t)}.
\]
More specifically, we obtain that
\begin{lemma}\label{lemma:remainder_noisy}
If the ancillae are initially in $\ket{0}$ then 
\begin{align*}
\mathcal{R}_{\omega, \delta}(t) = & \mathcal{R}_{\omega}(t) +\delta \sum_{\beta} \mathcal{K}^{(0)}_{\beta}(t) \nonumber \\
&  + \delta \sum_{\alpha, \beta} \sum_{j = 1}^2 \omega^j \int_0^t e^{-2(t-s)}  \mathcal{K}^{(j)}_{\alpha, \beta}(s)  ds,
\end{align*}
where $\mathcal{R}_{\omega}(t)$ is as defined in Lemma~\ref{lemma:remainder} but with $\rho_{\omega} \to \rho_{\omega, \delta}$ and
\begin{align*}
    \mathcal{K}_{\beta}^{(0)}(t) &= \tr{\mathcal{A}}{\mathcal{N}_\beta(\rho_{\omega,\delta}(t)},
    \nonumber \\
    \mathcal{K}_{\alpha,\beta}^{(1)}(t) &= -i[L_\alpha^\dagger, \tr{\mathcal{A}}{\sigma_\alpha\mathcal{N}_\beta(\rho_{\omega,\delta})}] + \textnormal{h.c.},
    \nonumber \\
    \mathcal{K}_{\alpha,\beta}^{(2)}(t) &= \frac{1}{2}[L_\alpha^\dagger,L_\alpha\tr{\mathcal{A}}{\mathcal{N}_\beta(\rho_{\omega,\delta}(t)}] + \textnormal{h.c.}.
\end{align*}
\end{lemma}
\noindent Another key ingredient in the stability analysis is the following modification of Lemma~\ref{lemma:tr_sigma_bounds}, whose proof is presented in Appendix~\ref{appendix:analysis_noisy}. 
\begin{lemma}\label{lemma:tr_sigma_bounds_noisy}
Suppose $\rho_{\omega, \delta}(t)$ is the joint state of the system and the ancilla qubits with the ancilla qubits initially being in the state $\ket{0}$, then for all $\alpha, \alpha'$,
\begin{align*}
    &\norm{\tr{\mathcal{A}}{\sigma_\alpha \rho_{\omega, \delta}(t)}}_1 \leq \frac{\omega}{2} + \mathcal{Z}'\delta  \text{ and },\\
    &\norm{\tr{\mathcal{A}}{\sigma_\alpha \sigma_{\alpha'} \rho_{\omega, \delta}(t)}}_1, \norm{\tr{\mathcal{A}}{\sigma_\alpha^\dagger \sigma_{\alpha'} \rho_{\omega, \delta}(t)}}_1 \leq \nonumber\\
    &\qquad \qquad \qquad \qquad \qquad \qquad \frac{\omega^2}{4}  + { \frac{\omega\mathcal{Z}'\delta}{2}} +    \mathcal{Z}'\delta.
\end{align*}
\end{lemma}
\noindent It should be noted that on setting the hardware noise $\delta$ to 0, the bounds in Lemma~\ref{lemma:tr_sigma_bounds_noisy} reproduce the bounds in the noiseless case obtained in Lemma~\ref{lemma:tr_sigma_bounds}.

\emph{Dynamics}. Next, we can analyze the modified remainder provided in Lemma~\ref{lemma:remainder_noisy} term by term and explicitly use the Lieb-Robinson bounds (Lemmas~\ref{lemma:lieb_robinson} and \ref{lemma:2superopLR}) to obtain bounds that are uniform in the system size. We first establish an extension of Lemma~\ref{lemma:bounds_remainder_lr} to the noisy setting.
\begin{lemma}\label{lemma:bounds_remainder_noisy_lr}
Suppose $O$ is a local observable on the system qudits with $\norm{O} \leq 1$ supported on $S_O$, and for $\tau>0$, let $O(\tau) = \exp(\mathcal{L}^\dagger \tau)(O)$ where $\mathcal{L}$ is a geometrically local target Lindbladian of the form in Eq.~\ref{eq:geom_local_lind}. Then for\ $q_\alpha$ as defined in Lemma~\ref{lemma:remainder}, then there are non-decreasing piecewise continuous function $\nu, \nu'$ such that $\nu(t), \nu'(t) \leq O(t^d)$ as $t \to \infty$ and for $\omega \leq 2$
        \begin{align*}
         &\sum_{\alpha}\bigabs{\tr{}{O(\tau) q_\alpha}} \leq  \nu(\tau) \text{ and}, \nonumber\\
         &\sum_{\alpha, \alpha'}\bigabs{\tr{}{O(\tau)\mathcal{Q}_{\alpha, \alpha'}^{(j)}(s)}} \leq  \bigg(1 +\frac{2\delta \mathcal{Z}'}{\omega} + \frac{4\delta \mathcal{Z}'}{\omega^2}\bigg)\nu^2(\tau).
        \end{align*}
where, for $j \in \{3, 4\}$, $\mathcal{Q}_{\alpha, \alpha'}^{(j)}$ is defined in Lemma~\ref{lemma:remainder} with $\rho_{\omega} \to \rho_{\omega, \delta}$ and for $j\in \{1, 2\}$, we define $\mathcal{Q}_{\alpha, \alpha'}^{(j)} = \mathcal{Q}_{\alpha, h_{\alpha'}}^{(j)}$ where $\mathcal{Q}_{\alpha, h}^{(j)}$ is defined in Lemma~\ref{lemma:remainder} with $\rho_{\omega} \to \rho_{\omega, \delta}$. Furthermore,
\begin{align*}
    &\sum_{\beta}\bigabs{\textnormal{Tr}(O(\tau)\mathcal{K}_{\beta}^{(0)}(s)} \leq \nu'(\tau)  \text{ and }, \\
    &\sum_{\alpha, \beta}\bigabs{\textnormal{Tr}(O(\tau)\mathcal{K}_{\alpha, \beta}^{(j)}(s)} \leq (\nu'(\tau))^2,
\end{align*}
where $\mathcal{K}^{(0)}_\beta$ and $\mathcal{K}^{(j)}_{\alpha, \beta}$ for $j \in \{1,2\}$ are defined in Lemma~\ref{lemma:remainder_noisy}.
\end{lemma}
\noindent In this lemma, the bounds on $\sum_{\alpha}\abs{\tr{}{O(\tau) q_\alpha}}$ and $\sum_{\alpha, \alpha'}\abs{\tr{}{O(\tau)\mathcal{Q}_{\alpha, \alpha'}^{(j)}(s)}} $  can be obtained by closely following the proof of Lemma~\ref{lemma:bounds_remainder_lr}, but with an application of Lemma~\ref{lemma:tr_sigma_bounds_noisy} instead of \ref{lemma:tr_sigma_bounds}. To obtain $\sum_{\beta} \abs{\text{Tr}(O(\tau) \mathcal{K}_\beta^{(0)}(s)}$, we note that since $O(\tau) = \exp(\mathcal{L}^\dagger \tau)(O)$ is an operator that acts entirely on the system qudits,
\begin{align*}
    &\abs{\text{Tr}(O(\tau) \mathcal{K}_\beta^{(0)}(s)} \nonumber\\
    &\qquad= \abs{\text{Tr}((O(\tau) \otimes I_\mathcal{A})\mathcal{N}_\beta(\rho_{\omega, \delta}(s))}, \nonumber \\
    &\qquad= \abs{\text{Tr}(\mathcal{N}_\beta^\dagger(O(\tau) \otimes I_\mathcal{A}) \rho_{\omega, \delta}(s))},\nonumber\\
    &\qquad\leq \norm{\mathcal{N}_\beta^\dagger(O(\tau) \otimes I_\mathcal{A})} = \norm{\tilde{\mathcal{N}}_\beta^\dagger(O(\tau))},
\end{align*}
where $\tilde{\mathcal{N}}_\beta^\dagger$, defined by $\tilde{\mathcal{N}}_\beta^\dagger(X) = \mathcal{N}_\beta(X\otimes I_\mathcal{A})$, is a superoperator acting entirely on the system qudits. Although $\tilde{\mathcal{N}}_\beta^\dagger$ isn't necessarily the adjoint of a Lindbladian superoperator (like $\mathcal{N}_\beta$), it has the identity (on the system qudits) in its kernel i.e. $\tilde{\mathcal{N}}_\beta^\dagger(I) = 0$. Furthermore, since $\mathcal{N}_\beta^\dagger$ is supported on $S'_\beta$, $\tilde{\mathcal{N}}_\beta^\dagger$ is supported only on the system qudits contained in $S'_\beta$, which we denote by $\tilde{S}'_\beta$. Now, with an application of the Lieb-Robinson bounds (Lemma~\ref{lemma:lieb_robinson}), we obtain that
\begin{align*}
     \abs{\text{Tr}(O(\tau) \mathcal{K}_\beta^{(0)}(s)}\leq \min\bigg(\eta_{S_O} \exp\bigg(4e\mathcal{Z}\tau - \frac{d(S_O, \tilde{S}'_\beta)}{a}\bigg), 1\bigg),
\end{align*}
where we have used that $\norm{\tilde{\mathcal{N}}_\beta^\dagger}_{cb,\infty\to\infty} \leq \norm{\mathcal{N}_\beta^\dagger}_{cb,\infty\to\infty} \leq 1$. This bound immediately yields an upper bound on $\sum_\beta \abs{\text{Tr}(O(\tau)\mathcal{K}^{(0)}_\beta(s)}$ which is uniform in the system size, since it reduces it to the summation of a decaying exponential on a lattice. An analysis of the scaling of this upper bound with $\tau$, together with similar bounds on $\sum_{\alpha, \beta}\abs{\text{Tr}(O(\tau) \mathcal{K}_{\alpha, \beta}^{(j)}(s)}$ is provided in Appendix~\ref{appendix:analysis_noisy}.
\begin{repproposition}{prop:dynamics_noisy}
Suppose $\mathcal{L}$ is a $d-$dimensional geometrically local Lindbladian and $O$ with $\norm{O}\leq 1$ is a local observable. Then, in the presence of noise with noise rate $\delta$,  the expected local observable with the analogue quantum simulator can be obtained to a precision $\varepsilon = O(\delta^{1/2}t^{2d + 1})$. Furthermore, to obtain this precision, we need to choose $\omega = \Theta(\delta^{1/4})$ which results in a simulator run-time $t_\textnormal{sim} = t/\omega^2 = \Theta(t \delta^{-1/2})$.
\end{repproposition}
\begin{proof}
Given a local observable, the error between the target expectation value $\expect{O}_\text{target}$ and the expectation value obtained on a noisy simulator, $\expect{O}_\text{sim}$ can be bounded by Eq.~\ref{eq:remainder_with_observable}, with the modified remainder $\mathcal{R}_{\omega, \delta}$ from Lemma~\ref{lemma:remainder_noisy} instead of $\mathcal{R}_{\omega}$:
\begin{widetext}
\begin{align}\label{eq:observable_remainder_with_noise}
&\abs{\expect{O}_\text{target} - \expect{O}_\text{sim}} = \frac{1}{\omega^2}\bigabs{\int_0^t \text{Tr}\bigg(O(t - s) \bigtr{\mathcal{A}}{\mathcal{R}_{\omega, \delta}\bigg(\frac{s}{\omega^2}\bigg)} \bigg)ds}, \nonumber \\
&\quad \leq  \sum_{\alpha} \bigabs{ \int_0^t \text{Tr}(O(t - s) q_\alpha) e^{-2s/\omega^2}ds} + \omega^2\sum_{j = 1}^4\sum_{\alpha, \alpha'} \bigabs{\int_0^t \int_0^{s/\omega^2}e^{-2(s/\omega^2 - s')} \text{Tr}(O(t - s) \mathcal{Q}^{(j)}_{\alpha, \alpha'}(s'))ds'ds} + \nonumber\\
&\quad \qquad \frac{\delta}{\omega^2} \sum_{\beta } \bigabs{\int_0^t \text{Tr}\bigg(O(t - s)\mathcal{K}_\beta^{(0)}\bigg(\frac{s}{\omega^2}\bigg)\bigg)ds} + \frac{\delta}{\omega^2} \sum_{j = 1}^2\sum_{\alpha, \beta} \omega^j \bigabs{\int_0^t \int_0^{s/\omega^2} e^{-2(s/\omega^2 - s')} \tr{}{O(t - s)\mathcal{K}^{(j)}_{\alpha, \beta}(s')}ds' ds}, \\
&\quad \numleq{1} \int_0^t \nu(t - s) e^{-2s/\omega^2}ds + 4\omega^2 \int_0^t \int_0^{s/\omega^2} \bigg(1 +\frac{2\delta \mathcal{Z}'}{\omega} + \frac{4\delta \mathcal{Z}'}{\omega^2}\bigg) e^{-2(s/\omega^2 - s')}\nu^2(t - s) ds' ds + \nonumber\\
&\quad \qquad \frac{\delta}{\omega^2} \int_0^t \nu'(t - s) ds + \frac{\delta}{\omega^2} \sum_{j = 1}^2 \omega^j \int_0^t \int_0^{s/\omega^2}e^{-2(s/\omega^2 - s')} {\nu'}^2(t - s) ds' ds, \nonumber \\
&\quad \numleq{2} \nu(t)\int_0^t e^{-2s/\omega^2}ds + \frac{\delta}{\omega^2}\nu'(t) \int_0^t  ds + \bigg(4(\omega^2 + 2\omega \delta  \mathcal{Z}' + 4\delta \mathcal{Z}')\nu^2(t) + \frac{\delta}{\omega}(1 + \omega){\nu'}^2(t)\bigg) \int_0^t \int_0^{s/\omega^2} e^{-2(s/\omega^2 - s')}ds' ds, \nonumber \\
&\quad \leq \frac{1}{2}\omega^2 \nu(t)  + \frac{t}{2}\bigg(4(\omega^2 + 2\omega \delta \mathcal{Z}' + 4\delta \mathcal{Z}')\nu^2(t) + \frac{\delta}{\omega}(1 + \omega){\nu'}^2(t) + \frac{2\delta}{\omega^2}\nu'(t)\bigg), \nonumber
\end{align}
\end{widetext}
where in (1) we have used Lemma~\ref{lemma:bounds_remainder_noisy_lr} and in (2), we have used the fact that the functions $\nu(t), \nu'(t)$ appearing in Lemma~\ref{lemma:bounds_remainder_noisy_lr} are non-decreasing in $t$. From this bound, which is uniform in the system size $n$, it can be seen that, for a fixed $t$, a choice of $\omega$ scaling as $\Theta(\delta^{1/4})$ minimizes the error $\abs{\expect{O}_\text{target} - \expect{O}_\text{sim}}$ as $\delta \to 0$. Furthermore, with this choice of $\omega$, we obtain an observable error of $O(\delta^{1/2} t^{2d + 1})$, as $\delta \to 0$ and $t\to \infty$, independent of the system size. 
\end{proof}

\emph{Long-time dynamics and fixed point}. Next, we consider rapidly mixing observables (Eq.~\ref{eq:rapid_mixing_observable_v2}) and analyze stability of these observables for long-time dynamics or fixed points. For this, we will establish the following modification of Lemma~\ref{lemma:error_term_rapid_mixing} to the account for noise in the quantum simulator.

\begin{lemma}\label{lemma:error_term_rapid_mixing_noisy}
    Suppose $O$ is a local observable with $\norm{O} \leq 1$ supported on $S_O$, and for $\tau>0$, let $O(\tau) = \exp(\mathcal{L}^\dagger \tau)(O)$ where $\mathcal{L}$ is a geometrically local Lindbladian of the form in Eq.~\ref{eq:geom_local_lind}. Furthermore, suppose $O$ is rapidly mixing with respect to $\mathcal{L}$ and satisfies Eq.~\ref{eq:rapid_mixing_observable_v2} with $k(\abs{S_O}, \gamma) \leq O(\exp(\gamma^{-\kappa}))$. Then for\ $q_\alpha$ as defined in Lemma~\ref{lemma:remainder},
        \[
        \sum_{\alpha}\bigabs{\int_0^t \tr{}{O(t - s) q_\alpha}e^{-2s/\omega^2}ds} \leq  \omega^2 \lambda^{(1)}(\gamma),
        \]
        where $\lambda^{(1)}(\gamma) \leq O(\gamma^{-d(\kappa + 1)})$ as $\gamma \to 0$ and for $j \in \{1,2,3,4\}$
        \begin{align*}
        &\sum_{\alpha, \alpha'}\bigabs{\int_0^t \int_0^{s/\omega^2} \tr{}{O(t - s)\mathcal{Q}_{\alpha, \alpha'}^{(j)}(s')} e^{-2(s/\omega^2 - s')}ds' ds} \nonumber\\
        &\qquad \qquad \qquad \qquad \qquad \leq  \bigg(1 + \frac{2\mathcal{Z}'\delta}{\omega} + \frac{4\mathcal{Z}'\delta}{\omega^2}\bigg)\lambda^{(2)}(\gamma),
        \end{align*}
where $\lambda^{(2)}(\gamma) \leq O(\gamma^{-(2d + 1)(\kappa + 1)})$ as $\gamma \to 0$ and for $j \in \{3, 4\}$, $\mathcal{Q}_{\alpha, \alpha'}^{(j)}$ is defined in Lemma~\ref{lemma:remainder} but with $\rho_\omega \to \rho_{\omega, \delta}$ and for $j\in \{1, 2\}$, we define $\mathcal{Q}_{\alpha, \alpha'}^{(j)} = \mathcal{Q}_{\alpha, h_{\alpha'}}^{(j)}$ where $\mathcal{Q}_{\alpha, h}^{(j)}$ is defined in Lemma~\ref{lemma:remainder} but with $\rho_\omega \to \rho_{\omega, \delta}$. Furthermore,
\begin{align*}
\sum_{\beta}\bigabs{\int_0^t \textnormal{Tr}\bigg(O(t - s)\mathcal{K}_\beta^{(0)}\bigg(\frac{s}{\omega^2}\bigg)\bigg)ds} \leq {\lambda'}^{(1)}(\gamma)
\end{align*}
where ${\lambda'}^{(1)}(\gamma) \leq O(\gamma^{-(d + 1)(\kappa + 1)})$ as $\gamma \to 0$ and for $j \in \{1,2\}$,
\begin{align*}
&\sum_{\alpha, \beta}\bigabs{\int_0^t \int_0^{s/\omega^2} \textnormal{Tr}\big(O(t - s) \mathcal{K}_{\alpha,\beta}^{(j)}(s') \big) e^{-2(s/\omega^2 - s')}ds' ds} \nonumber \\
&\qquad \qquad \qquad \qquad \qquad \leq {\lambda'}^{(2)}(\gamma),
\end{align*}
where ${\lambda'}^{(2)}(\gamma) \leq O(\gamma^{-(2d + 1)(\kappa + 1)})$ as $\gamma \to 0$.
\end{lemma}
\noindent The proof of this lemma closely follows the same strategy as the proof of Lemma~\ref{lemma:error_term_rapid_mixing} for the noiseless case and is detailed in Appendix~\ref{appendix:analysis_noisy}. With this lemma, we can then establish the stability of the quantum simulator while simulating the long-time dynamics or fixed point expectation values of rapidly mixing observables.
\begin{repproposition}{prop:fp_noisy}
Suppose $\mathcal{L}$ is a $d-$dimensional geometrically local Lindbladian and $O$ with $\norm{O}\leq 1$ is a local observable supported on $O(1)$ lattice sites satisfying rapid mixing (Eq.~\ref{eq:rapid_mixing_observable}). Then, in the presence of noise with noise rate $\delta$,  the expected local observable at any time $t$ can be obtained to a precision $\varepsilon = O(\delta^{1/2}\gamma^{-(\kappa + 1)(2d + 1)})$, independent of $n$ and $t$, with the analogue quantum simulator. Furthermore, to obtain this precision, we need to choose $\omega = \Theta(\delta^{1/4})$ which results in a simulator run-time $t_\textnormal{sim} = t/\omega^2 = O(t \delta^{-1/2})$.
\end{repproposition}
\begin{proof}
To bound the observable error, we start from Eq.~\ref{eq:observable_remainder_with_noise} and use Lemma~\ref{lemma:error_term_rapid_mixing_noisy} --- we then obtain
\begin{align*}
&\abs{\expect{O}_\text{target} - \expect{O}_\text{sim}} \leq \nonumber\\
&\qquad \omega^2 \lambda^{(1)}(\gamma) + 4(\omega^2 + 2\mathcal{Z}'\omega\delta + 4\mathcal{Z}'\delta) \lambda^{(2)}(\gamma) + \nonumber \\
&\qquad \frac{\delta}{\omega^2}{\lambda'}^{(1)}(\gamma) + \frac{\delta}{\omega}(1  + \omega) {\lambda'}^{(2)}(\gamma).
\end{align*}
From this bound, which is uniform in the system size $n$, as well as the scalings of $\lambda^{(j)}(\gamma), {\lambda'}^{(j)}(\gamma)$ from Lemma~\ref{lemma:error_term_rapid_mixing_noisy} it can be seen that, for a fixed $\gamma$, a choice of $\omega$ scaling as $\Theta(\delta^{1/4})$ minimizes the error $\abs{\expect{O}_\text{target} - \expect{O}_\text{sim}}$ as $\delta \to 0$. Furthermore, with this choice of $\omega$, we obtain an observable error of $O(\delta^{1/2} \gamma^{-(\kappa+1)(2d + 1)})$, as $\delta \to 0$ and $\gamma\to 0$, independent of the system size. 
\end{proof}

\subsection{Quantum advantage in the presence of noise}
As shown in the previous subsection, even though the dynamics and fixed point of geometrically local Lindbladians could be stable to noise on the analogue quantum simulator, a noisy quantum simulator cannot simulate observables perfectly but only to a noise-limited precision. We now adopt the perspective in Ref.~\cite{trivedi2024quantum}, and ask if there are family of stable Lindbladian problems where a reduction in the noise rate $\delta$ results in classical algorithms requiring a superpolynomially longer time to achieve this noise-limited precision. To show that these problems exist (subject to the complexity assumption of BQP $\neq$ BPP), we will leverage the quantum-circuit to 2D Lindbladian mapping presented in Lemma~\ref{lemma:encoded_lindbladian}.

We first consider the problem of rapidly mixing local observables in the Lindbladian fixed point in 2D. From Corollary~\ref{prop:fp_noisy_fixed_point}, it follows that, given $\delta > 0$, as long as $\gamma^{-1} \leq O(\delta^{\alpha - 1/(10(\kappa + 1))})$ for $\alpha > 0$, the noise-limited precision in the estimated observable $\varepsilon \leq O(\delta^{5\alpha (\kappa + 1)}) \to 0$ as $\delta \to 0$. Physically, this corresponds to the intuitively expected fact that a simulator with lower noise rate $\delta$ can be used to solve problems which have a smaller decay rate $\gamma$ (or equivalently those that take longer to reach their fixed point) without accumulating a large error. Now, using Lemma~\ref{lemma:encoded_lindbladian} we establish that there cannot exist a classical algorithm which for any $\alpha, \kappa > 0$ can compute the rapidly mixing local observable to a precision of $O\left(\delta^{5\alpha(\kappa+1)}\right)$ in time $\text{poly}(\delta^{-1})$ unless BQP=BPP.

\begin{repproposition}{prop:quantum_advantage_noisy}
A noisy analogue quantum simulator with noise rate $\delta$ can solve Problem~\ref{prob:fixed_points}, with parameters $\alpha, \kappa$, to a noise-limited precision $O(\delta^{5\alpha (\kappa + 1)})$ (which $\to 0$ as $\delta \to 0$). Furthermore, there cannot exist a randomized classical algorithm with run-time $\text{poly}(\delta^{-1})$ that can solve Problem~\ref{prob:fixed_points} to the same precisions for every given $\alpha, \kappa$ unless \textnormal{BQP = BPP}.
\end{repproposition}

\begin{proof}
The estimate of the noise-limited precision and the quantum simulator run-time follows immediately from Corollary~\ref{prop:fp_noisy_fixed_point} of Proposition~\ref{prop:fp_noisy}. We now show the classical hardness of the sequence of problems defined in Problem~\ref{prob:fixed_points}. Assume that such a classical algorithm did in fact exist --- then, we can use it to simulate the outcome of measuring an output qubit in an arbitrary poly-depth quantum circuit. To see this, suppose we had a quantum circuit $\mathcal{C}$ with architecture shown in Fig.~\ref{fig:circuit_format} on $N$ qubits with $R \leq O(N^m)$ rounds for some $m > 0$. We can then use Lemma~\ref{lemma:encoded_lindbladian} to produce a corresponding 2D geometrically local Lindbladian $\mathcal{L}$ and a rapidly mixing local observable $O$ which satisfies the rapid-mixing condition (Eq.~\ref{eq:rapid_mixing_observable}) with $\kappa = 1/3$ and $\gamma^{-1} \leq O(N^{3(m + 1)})$. Being able to simulate the local observable $O$ to a precision of $O((NR)^{-1}) = O(N^{-(m + 1)})$ would allow us to estimate, up to an $O(1)$ additive error, the probability of the output qubit (which is arbitrarily chosen to be the first qubit) to be in 1, thus successfully simulating the quantum circuit.

Now, if we indeed had a classical algorithm that did satisfy the conditions in the proposition, then for a small $\delta > 0$, we could use it simulate this Lindbladian for $N$ being chosen as a function, $N_\delta$, of $\delta$ to satisfy the constraint $\gamma^{-1} \leq O(\delta^{\alpha - 1/(10(\kappa + 1))})$ i.e.~
\[
N^{3(m + 1)}_\delta \leq O\left(\delta^{\alpha - 3/40}\right),
\]
and to satisfy the requirement that precision $\varepsilon = O(\delta^{\alpha(\kappa + 1)})$ be at-least $O((NR)^{-1})$, we impose that
\[
\delta^{4\alpha/3} \leq O\left(N_\delta^{-(m + 1)}\right).
\]
We note that both of these requirements can be satisfied by using $\alpha = 1/280$ and $N_\delta = \Theta(\delta^{-1/(42(m + 1))})$.

Thus, if there did exist a $\text{poly}(\delta^{-1})$ classical algorithm to simulate  the rapidly mixing Lindbladian problem, even with the constraint on $\gamma^{-1}$ as stated in the proposition to ensure a vanishing noise-limited error as $\delta \to 0$, for any $\alpha, \kappa > 0$, we could simulate the outcome of measurement of an output qubit in an arbitrary poly-depth quantum circuit, which would imply that BQP=BPP.
\end{proof}

Similarly, we can use Proposition~\ref{prop:quantum_advantage_noisy} to establish a similar result for dynamics of geometrically local Lindbladians.
\begin{repcorollary}{prop:quantum_advantage_dynamics_noisy}
  A noisy analogue quantum simulator with noise rate $\delta$ can solve Problem~\ref{prob:time_dynamics}, with parameter $\alpha$, to a noise-limited precision $O(\delta^{5\alpha})$ (which $\to 0$ as $\delta \to 0$) in simulator run-time $O(\textnormal{poly}(\delta^{-1}))$ and there cannot exist a $\textnormal{poly}(\delta^{-1})$ randomized classical algorithm to estimate this local observable to the same precision for every given $\alpha > 0$ unless \textnormal{BQP = BPP}.
\end{repcorollary}
\begin{proof}
    The estimate of the noise-limited precision and the quantum simulator run-time follows immediately from Proposition~\ref{prop:dynamics_noisy}. To show the classical hardness of the sequence of problems defined in Problem~\ref{prob:time_dynamics}, note that if such a classical algorithm did exist, then it could be used to simulate the sequence of fixed point problems in Proposition~\ref{prop:quantum_advantage_noisy} to the corresponding noise-limited precision corresponding to $\alpha, \kappa > 0$ by simply choosing $t_{(\delta)} = \Theta(\gamma^{-(\kappa + 1)}\log(\delta^{-1})) \leq O(\delta^{\alpha'- 1/10}\log(\delta^{-1}))$ where $\alpha' = (\kappa + 1)\alpha$ and thus implying that BQP = BPP.
\end{proof}

\section{A first experimental proposal: Non-thermality of Lindbladian fixed points}

Conventionally, a large many-body system interacting with an environment is expected to thermalize to the environment's temperature --- consequently, the fixed point of the corresponding Lindbladian is expected to be a Gibb's state at the environment's temperature \cite{davies1976quantum}. However, whether a precise microscopic open system model, described by a many-body Lindbladian, has a thermal fixed point remains much less clear largely due to the difficulty in classically numerically simulating Lindbladian fixed points \cite{cui2015variational, weimer2021simulation}. It has been argued, either using the Keldysch path-integral formalism \cite{diehl2008quantum, buchhold2015nonequilibrium, sieberer2016keldysh} or in exactly solvable models of gaussian fermions \cite{prosen_pizorn_2008}, that local Lindbladians could have non-thermal fixed points --- in particular, their correlation function structure as well as critical exponents cannot be explained by assuming them to be a thermal state. For instance, in 1D, there are Lindbladians which can support a long-range correlated fixed point (i.e.~fixed points with correlation function that decay polynomially with the separation between the sites being correlated) \cite{weimer2021simulation} while it is known that Gibb's states of 1D local Hamiltonians are necessarily short range correlated \cite{araki1969gibbs, bluhm2022exponential}. In 2 or higher dimensions, whether Lindbladian fixed points are thermal (i.e.~are either the Gibb's state of a geometrically local Hamiltonian or are in the same phase \cite{coser2019classification} as the Gibb's state of a geometrically local Hamiltonian) are also less well understood.

Outside of solvable models of gaussian fermions \cite{prosen_pizorn_2008, bravyi2011classical}, to the best of our knowledge, there are not  mathematically rigorous results or efficient numerical tools that precisely delineate under what conditions a many-body Lindbladian will have a thermal fixed point. This question could thus be potentially approached with an analog quantum simulator. However, we would then have to understand if noise in the simulator hampers our ability to answer this question at practically relevant system sizes. In the remainder of this section, we consider the 1D model with non-thermal fixed points proposed in Ref.~\cite{prosen_pizorn_2008} and provide numerical evidence of its noise robustness.

The model that we consider is shown in Fig.~\ref{fig:phase_change_with_noise}(a). We consider a 1D system of $n$ fermionic modes arranged with a  Hamiltonian $H$ given by
\begin{subequations}\label{eq:fermion_chain_model}
\begin{align}
    H = 2h\sum_{x= 1}^n  n_x + \sum_{x = 1}^{n - 1}  \big(a_x a^\dagger_{x + 1}  + \gamma a_x a_{x + 1}  + U n_x n_{x + 1}+ \text{h.c.}\big),
\end{align}
where $a_x$ is the annihilation operator for the fermion at site $x$ and $n_x = a_x^\dagger a_x$ is the corresponding number operator. Additionally, the chain of fermions is dissipated at the boundaries with jump operators $\sqrt{\Gamma_L} a_1, \quad \sqrt{\Gamma_R} a_n$, where $\Gamma_L$ and $\Gamma_R$ are the dissipation rates at the left and right respectively. The full Lindbladian describing the system is thus given by
\begin{align}
\mathcal{L}(X) = -i[H, X] + \Gamma_L \mathcal{D}_{a_1}(X) + \Gamma_R \mathcal{D}_{a_n}(X).
\end{align}
\end{subequations}
At $U = 0$, this model is exactly solvable using the formalism in Ref.~\cite{bravyi2011classical}, while for $U \neq 0$, the model is non-integrable. Furthermore, while we study this problem for fermions, it is equivalent to a 1D nearest neighbour spin model with boundary dissipation via the Jordan Wigner transformation.

At $U = 0$ and when $h \leq h_c = 1 -\gamma^2$, the steady state of this system is long range correlated and thus not thermal. This is shown in the first column in Figure \ref{fig:phase_change_with_noise}(b) --- the steady-state covariance between fermionic occupation numbers, $\text{cov}(n_x n_y) = \expect{n_x n_y} - \expect{n_x} \expect{n_y}$, shows a transition from the state being short-range correlated when $h > h_c$ to being long-range correlated when $h \leq h_c$. Next, we consider adding simulator noise to this model --- we do so by adding jump operators $\sqrt{\delta}a_x, \sqrt{\delta}a_x^\dagger$ at each site $x$, where $\delta$ is the noise rate, and simulate instead the fixed point of the Lindbladian $\mathcal{L}_\delta$:
\[
\mathcal{L}_\delta = \mathcal{L} + \delta \sum_{x = 1}^n \big(\mathcal{D}_{a_x} + \mathcal{D}_{a_x^\dagger}\big).
\]
In the absence of the boundary dissipation (i.e.~$\Gamma_L = \Gamma_R = 0$), this noise would have resulted the system being driven to the maximally mixed state in its steady state. However, in the presence of the boundary dissipation (i.e.~$\Gamma_L, \Gamma_R \neq 0$), we find that the covariance function measured on the simulator can still show signature of a long-range phase at $\delta \sim 0.5\% - 1\%$ at large system sizes ($n \sim 160$) as shown in Fig.~\ref{fig:phase_change_with_noise}(b). Specifically, we observe that the quantum simulator can show significant signatures of long-range correlations when the system is deep inside the long-range correlated phase ($h = 0.4 - 0.6$ in Fig.~\ref{fig:phase_change_with_noise}(b)), while near the phase transition point ($h = h_c$), the presence of noise in the simulator destroys the long-range correlations. Nevertheless, these simulations indicate that the noisy simulator would be able to detect non-thermality in the steady state for system sizes practically relevant for current experiments, but make a noise-rate limited error in detecting the point at which the phase transition occurs.

For $U \neq 0$, the Lindbladian does not remain integrable and the parameter regimes in which the Lindbladian fixed point is non-thermal have not been precisely delineated. Based on the noise analysis for $U = 0$ discussed above, we conjecture that a noisy simulator should be able to detect the existence of a non-thermal phase in the model and approximate, to a noise-rate limited precision, the model parameters ($h, U_0, \gamma$) at which this phase transition occurs. We remark that our theoretical result on the stability of fixed point (Proposition~\ref{prop:fp_noisy}) does not directly apply to this problem and rigorously ensure that a noisy simulator can detect non-thermality in Lindbladian fixed points --- we leave a thorough theoretical and numerical investigation of this problem for future work.

\usetikzlibrary{decorations.pathmorphing}
\begin{figure}
    \centering
    \begin{tikzpicture}
    \definecolor{sitecolor}{HTML}{74A0C8}
    \node[circle, fill=sitecolor, draw=black, line width=1pt, text=black] (c1) {1};    \node[circle, fill=sitecolor, draw=black, line width=1pt, text=black, right of=c1] (c2) {2};
    \node[circle, fill=sitecolor, draw=black, line width=1pt, text=black, right of=c2] (c3){3};
    \node[right of=c3] (dots) {\Large \ldots};
    \node[circle, fill=sitecolor, draw=black, line width=1pt, text=black, right of=dots] (cn) {n};
    \node[left of=c1] (leftbath) {};
    \node[right of=cn] (rightbath) {};
    
    \foreach \xa/\xb in {c1/c2, c2/c3, c3/dots, dots/cn}{
        \draw[<->] (\xa) edge[bend right, line width=1pt] node [right] {} (\xb);
    }
    \draw [->, line join=round, line width=0.5pt,
        decorate, decoration={
            zigzag,
            segment length=4,
            amplitude=1.5,post=lineto,
            post length=2pt,
        }]  (c1) -- (leftbath) node [label={[text=black,yshift=-2pt]:$\sqrt{\Gamma_L}a_1$}] {};
    \draw [->, line join=round, line width=0.5,
        decorate, decoration={
            zigzag,
            segment length=4,
            amplitude=1.5,post=lineto,
            post length=2pt
        }, draw=black]  (cn) -- (rightbath) node [label={[text=black,yshift=-2pt]:$\sqrt{\Gamma_R}a_n$}] {};
    \foreach \x/\xsubscript in {c1/1, c2/2, c3/3, cn/n}{
        \node[above of=\x] (a\x) {$\sqrt{\delta}a_\xsubscript$};
        \draw [->, densely dotted, line join=round, line width=0.5pt,
        decorate, decoration={
            zigzag,
            segment length=4,
            amplitude=1.5,post=lineto,
            post length=2pt
        }]  (\x) -- (a\x) {};
        \node[below of=\x] (b\x) {$\sqrt{\delta}a^\dagger_\xsubscript$};
        \draw [->, densely dotted, line join=round, line width=0.5pt,
        decorate, decoration={
            zigzag,
            segment length=4,
            amplitude=1.5,post=lineto,
            post length=2pt
        }]  (\x) -- (b\x) {};
    }
    \node (diagramlabel) at ([shift=({-1.7cm,1.2cm})]c1) {\bfseries(a)};
    \node (plotlabel) at ([shift=({-1.7cm,-1.5cm})]c1) {\bfseries(b)};
    \end{tikzpicture}
    \includegraphics[scale=0.44]{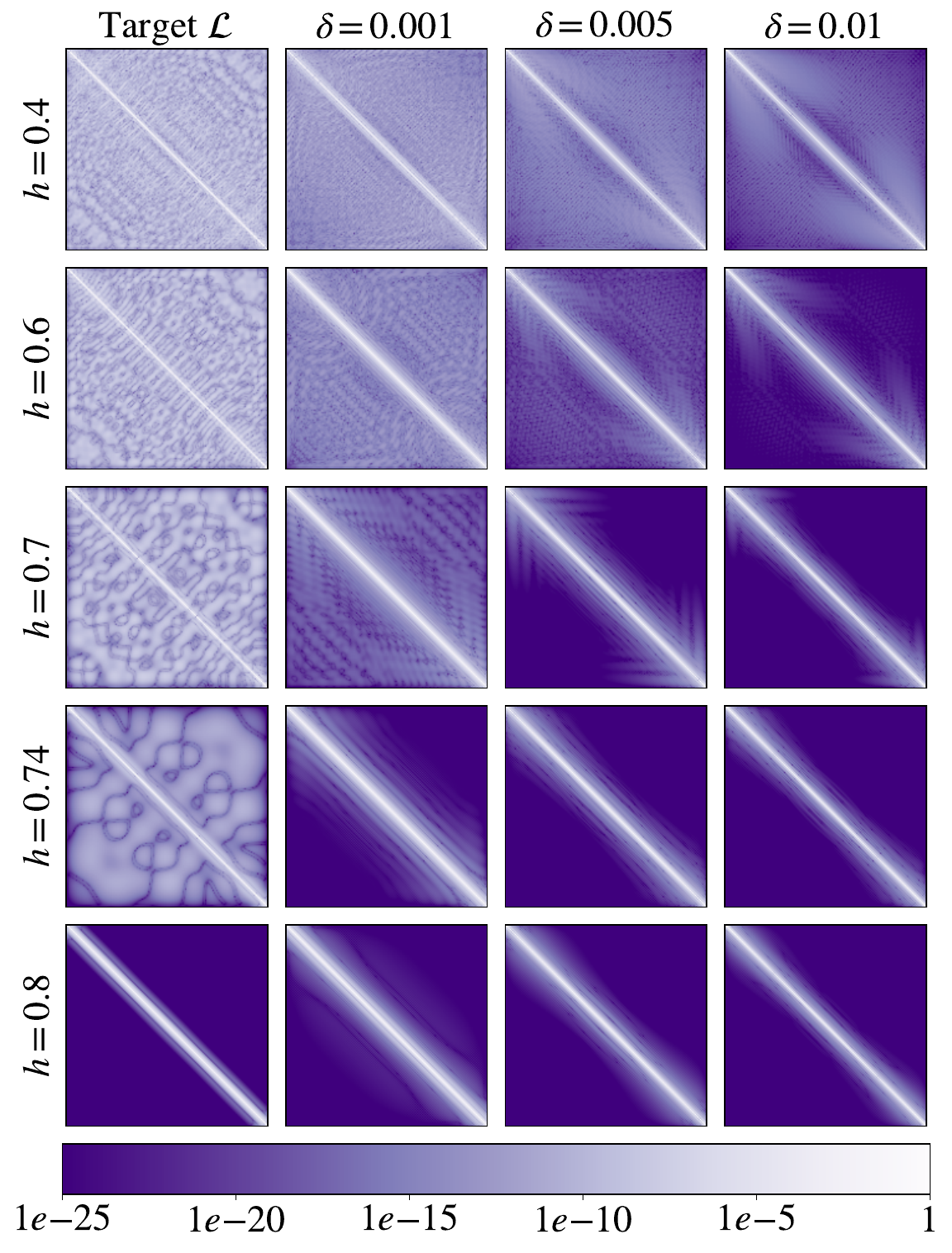}
    \caption{(a) Schematic of the open 1D fermion chain model described by Eq.~\ref{eq:fermion_chain_model}. Solid lines indicate hamiltonian (smooth) and dissipative (squiggly) terms in target model, and dotted lines indicate noise. (b) Numerically computed occupation covariance for pairs of fermions in the steady state of the noisy fermion chain model with $\gamma=0.5, \Gamma_R, \Gamma_L=0.5$. For each pair $x,y$ of fermions, we plot $\abs{\text{cov}(n_xn_y)}$. In the noiseless case ($\delta=0$), there is a phase change from short-range to long-range correlations as $h$ is brought lower than $h_c=1-\gamma^2=0.75$. When depolarizing noise is added ($\delta>0$), the phase transition point is effectively ``delayed" to lower values of $h$.}
    \label{fig:phase_change_with_noise}
\end{figure}

\section{Conclusion}
In conclusion, we have provided an analysis of analogue quantum simulation of physically motivated open quantum system simulation problems. Our analysis developed tools to study both a noiseless quantum simulator as well as a noisy quantum simulator and provided rigorous accuracy guarantees on the observables being measured on the quantum simulator. Furthermore, we also provided complexity theoretic evidence of classical hardness of the physically motivated problems that we considered. Our results provide theoretical evidence for the possibility of using near-term quantum devices for solving physically interesting and classically hard open system simulation problems, while remaining stable to noise. Our paper also introduces several new technical results that could be of independent interest --- the accuracy and stability guarantees that we develop in this paper were based on developing a mathematically rigorous adiabatic elimination analysis that could account for Lieb-Robinson bounds, and the classical hardness results were built on a quantum circuit to 2D Lindbladian encoding together with its convergence analysis.

Another promising direction investigating open system simulation problems on analog simulators can be inspired by Ref.~\cite{liu2024verifiable}, in which a protocol is proposed for verifying quantum advantage on analog hamiltonian simulators and is shown to have some robustness to noise in the simulation. This type of verification protocol is useful in the scenario that the quantum simulator is not implementing the claimed computation - either due to uncalibrated noise or if the operator of the simulator is adversarial to the party attempting to establish the presence of quantum advantage. In our work we make the physics-oriented assumption that the simulator implements the simulation as defined and that errors arise from noise that is well calibrated and known in advance. In future work, we hope to determine if the physically-motivated universal quantum simulation problem we consider can also be efficiently verified in the same sense as Ref.~\cite{liu2024verifiable}.

Our work leaves several important theoretical questions open --- most importantly, can we develop certifiable analogue simulation protocols for the analogue simulation for the fixed point of Lindbladians that are not rapidly mixing, and are these protocols stable to errors? A possible approach to this problem, which we leave for future work, is to identify a set of reasonable assumptions on the spectrum of Lindbladians, that may not be rapid mixing, but which still allow for analytical results related to their stability. An alternative could be studying Lindbladians showing dissipative phase transitions numerically near the phase transition, and understand if (and which) local observables are stable.

Furthermore, another direction would be extending the classical hardness results to 1D geometrically local Lindbladians by encoding a quantum circuit in the fixed point of such a Lindbladian, thereby making the case for quantum advantage in simpler (i.e.~1D) experimental setups. This should in principle be possible by adapting the techniques developed for the corresponding result that encodes a quantum circuit to the ground state of a 1D Hamiltonian \cite{aharonov2009power}. Furthermore, extending classical hardness results to translationally invariant Lindbladians, along the lines of similar results in Hamiltonian ground states \cite{gottesman_irani_2009}, also remain open.
\bibliographystyle{unsrt}
\bibliography{paper_bib}

\onecolumngrid
\appendix
\newpage
\section{Detailed proofs from section \ref{sec:analysis_noiseless}}
\label{appendix:analysis_noiseless_proofs}
\begin{replemma}{lemma:remainder}
For any $t > 0$, the remainder $\mathcal{R}_\omega(t)=\frac{d}{dt} \tr{\mathcal{A}}{\rho_\omega(t)}-\omega^2 \mathcal{L} \tr{\mathcal{A}}{\rho_\omega(t)}$ satisfies 
    \begin{subequations}
    \begin{align}
    \mathcal{R}_\omega(t) & 
        = \omega^2\sum_{\alpha}  e^{-2t} q_{\alpha} +\omega^4 \sum_{j \in \{1, 2\}}\sum_{\alpha} \int_0^t e^{-2(t - s)}\mathcal{Q}^{(j)}_{\alpha, H_\textnormal{sys}}(s)ds
        \nonumber \\
        & \quad 
        + \omega^4\sum_{j \in \{3, 4\}}\sum_{\substack{\alpha, \alpha'}}\int_0^t e^{-2(t - s)}Q_{\alpha, \alpha'}^{(j)}(s)ds.
    \end{align}
where
\begin{align*}
& q_\alpha = -\mathcal{D}_{L_\alpha}(\rho(0)), \\
&\mathcal{Q}_{\alpha, h}^{(1)}(t) = - \frac{1}{\omega}[L_\alpha^\dagger, [h, \tr{\mathcal{A}}{\sigma_\alpha \rho_\omega(t)}]] + \textnormal{h.c.}
\\
&\mathcal{Q}_{\alpha, h}^{(2)}(t) = -\frac{i}{2}[L_\alpha^\dagger, L_\alpha [h,\tr{\mathcal{A}}{\rho_{\omega}(t)}]] + \textnormal{h.c.}.\\
&\mathcal{Q}_{\alpha, \alpha'}^{(3)}(t) =
\frac{i}{\omega} \sum_{u \in \{+, -\}} \mathcal{D}_{L_\alpha}\big([L_{\alpha'}^{(u)}, \tr{\mathcal{A}}{\sigma_{\alpha'}^{(\bar{u})} \rho_\omega(t)}]\big),\\
&\text{If }\alpha = \alpha', \mathcal{Q}^{(4)}_{\alpha, \alpha'}(t) =
    \frac{2}{\omega^2}(\mathcal{D}_{L_\alpha^\dagger}-\mathcal{D}_{L_\alpha})(\tr{\mathcal{A}}{n_\alpha \rho_\omega(t)}), \\
&\text{If }\alpha \neq \alpha',  \mathcal{Q}_{\alpha, \alpha'}^{(4)}(t) = 
    -\frac{1}{\omega^2}\sum_{\substack{u, u' \\ \in \{+, -\}}} [L_\alpha^{(u)}, [L_{\alpha'}^{(u')}, \tr{\mathcal{A}}{\sigma_\alpha^{({\bar{u}})}\sigma_{\alpha'}^{(\bar{u}')}\rho_\omega(t)}]].
\end{align*}
\end{subequations}
\end{replemma}
\begin{proof}
In order to find $\mathcal{R}_\omega (t)$, we start from the equations of motion for $\tr{\mathcal{A}}{{\rho_\omega}(t)}$ and $\tr{\mathcal{A}}{\sigma_\alpha \rho_\omega(t)}$ as given in Eqs.~\ref{eq:dynamics}a,b. We first integrate Eq.~\ref{eq:dynamics}b to obtain
\begin{equation}\label{eq:tr_sigma_alpha_integrated_app}
    \tr{\mathcal{A}}{\sigma_\alpha \rho_\omega(t)} = \int_0^t e^{-2(t - s)} \bigg(
    -i \omega^2 [H_\text{sys}, \tr{\mathcal{A}}{\sigma_\alpha\rho_\omega(s)}] -i \omega L_\alpha \tr{\mathcal{A}}{\rho_\omega(s)} + \omega\sum_{\alpha'} E_{\alpha, \alpha'} (s) \bigg)ds,
\end{equation}
where we have used the fact that $\sigma_\alpha \rho_\omega(0) = 0$, since the ancillae are initialized in $\ket{0}$.
Next we note that, from integration by parts and the expression for $\frac{d}{dt}\tr{\mathcal{A}}{\rho_{\omega}(t)}$ in Eq.~\ref{eq:dynamics}a, it follows that
\begin{align}\label{eq:tr_sigma_alpha_int_part}
    &\int_0^t e^{-2(t - s)} L_\alpha \tr{\mathcal{A}}{\rho_\omega(s)}ds \nonumber\\
    &= \frac{1}{2}L_\alpha \tr{\mathcal{A}}{\rho_\omega(t)} -  \frac{e^{-2t}}{2}L_\alpha \tr{\mathcal{A}}{\rho_\omega(0)} - \frac{1}{2}\int_0^t e^{-2(t - s)}L_\alpha \frac{d}{ds}\tr{\mathcal{A}}{\rho_\omega(s)}ds, \nonumber\\
    & =\frac{1}{2}L_\alpha \tr{\mathcal{A}}{\rho_\omega(t)} -  \frac{e^{-2t}}{2}L_\alpha \tr{\mathcal{A}}{\rho_\omega(0)} - \frac{i\omega^2}{2}\int_0^t e^{-2(t - s)}L_\alpha [H_\text{sys}, \tr{\mathcal{A}}{\rho_\omega(s)}]ds + \frac{i\omega}{2}\sum_{\alpha'}\int_0^t e^{-2(t - s)}F_{\alpha, \alpha'}(s)ds,
\end{align}
where
\begin{equation}
    F_{\alpha, \alpha'}(s) = \sum_{u \in \{+,-\}} L_\alpha [L_{\alpha'}^{(u)}, \tr{\mathcal{A}}{\sigma_{\alpha'}^{(\bar{u})}\rho_\omega(s)}].
\end{equation}
From Eqs.~\ref{eq:tr_sigma_alpha_integrated_app} and \ref{eq:tr_sigma_alpha_int_part}, we now obtain that
\begin{align}
    \label{eq:tr_sigma_alpha_integrated_expanded}
    &\tr{\mathcal{A}}{\sigma_\alpha \rho_\omega(t)} = -\frac{i\omega}{2}L_\alpha \tr{\mathcal{A}}{\rho_\omega(t)} + \frac{i\omega}{2} e^{-2t} L_\alpha \rho(0)
    \nonumber \\
    &\qquad + \int_0^t e^{-2(t - s)}\Bigg( -i\omega^2 [H_\text{sys},\tr{\mathcal{A}}{\sigma_\alpha\rho_\omega(s))}] - \frac{\omega^3}{2} L_\alpha [H_\text{sys}, \text{Tr}_{\mathcal{A}}(\rho_{\omega}(s))] + \sum_{\alpha'} \bigg(\omega{E_{\alpha, \alpha'}(s)} + \frac{\omega^2}{2} {F_{\alpha, \alpha'}(s)}\bigg) \Bigg)ds.
\end{align}
 We now consider the remainder $\mathcal{R}_\omega(t)$ --- using the definition of the remainder (Eq.~\ref{eq:remainder_def}) together with Eq.~\ref{eq:dynamics}a, we obtain that
\begin{align}\label{eq:simplified_remainder}
    \mathcal{R}_\omega(t) & = \frac{d}{dt}\tr{\mathcal{A}}{\rho_\omega(t)} - \omega^2 \mathcal{L}\tr{\mathcal{A}}{\rho_\omega(t)} =  \sum_{\alpha}\bigg(\left(-i\omega[L_\alpha^\dagger, \tr{\mathcal{A}}{\sigma_\alpha {\rho}_\omega(t)}] + \text{h.c}\right)- \omega^2 \mathcal{D}_{L_\alpha} \tr{A}{\rho_\omega(t)}\bigg).
\end{align}
It follows from Eq.~\ref{eq:tr_sigma_alpha_integrated_expanded} that
\begin{align}
    \label{eq:expressions_commutator}
     &-i\omega[L_\alpha^\dagger, \tr{\mathcal{A}}{\sigma_\alpha \rho_\omega(t)}] = -\frac{\omega^2}{2}[L_\alpha^\dagger, L_\alpha \tr{\mathcal{A}}{\rho_\omega(t)}] +\frac{\omega^2}{2}e^{-2t}[L_\alpha^\dagger, L_\alpha \rho(0)] + \int_0^t e^{2(t - s)}\bigg(-\omega^3 [L_\alpha^\dagger, [H_\text{sys}, \tr{\mathcal{A}}{\sigma_\alpha\rho_\omega(t)} ] \nonumber\\
     &\qquad \qquad -i\frac{\omega^4}{2}[L_\alpha^\dagger, L_\alpha[H_\text{sys}, \tr{\mathcal{A}}{\rho_\omega(s)}]] - i \sum_{\alpha'}\bigg(\omega^2[L_\alpha^\dagger, E_{\alpha, \alpha'}(s)] + \frac{\omega^3}{2}[L_\alpha^\dagger, F_{\alpha, \alpha'}(s)]\bigg)\bigg) ds, 
\end{align}
Using Eqs.~\ref{eq:simplified_remainder} and \ref{eq:expressions_commutator}a, b, we note that $-[L_\alpha^\dagger,L\tr{\mathcal{A}}{\rho_{\omega}(s)}]/2 + \text{h.c} = \mathcal{D}_{L_\alpha}$ and so arrive at
\begin{align}
     \mathcal{R}_\omega(t) = \sum_{\alpha} \left(\omega^2 e^{-2t} q_{\alpha} +  \omega^4 \int_0^t e^{-2(t-s)} \left( {\mathcal{Q}^{(1)}_{\alpha,H_\textnormal{sys}}(s)} + \omega^4 \mathcal{Q}^{(2)}_{\alpha,H_\textnormal{sys}}(s) \right) ds \right) + \omega^4\sum_{\alpha, \alpha'} \int_0^t e^{-2(t - s)}\big({\mathcal{Q}^{(3)}_{\alpha, \alpha'}(s)} +  \mathcal{Q}^{(4)}_{\alpha, \alpha'}(s)\big) ds,
\end{align}
where
\begin{align}
    q_\alpha & := \frac{1}{2}[L_\alpha^\dagger, L_\alpha \rho(0) ] + \text{h.c.},
    \nonumber \\
    \mathcal{Q}^{(1)}_{\alpha,H_\textnormal{sys}}(t) & := -\frac{1}{\omega}\left[ L_\alpha^\dagger, [H_\text{sys}, \tr{\mathcal{A}}{\sigma_\alpha\rho_\omega(t)}]\right] + \text{h.c.},
    \nonumber \\
    \mathcal{Q}^{(2)}_{\alpha,H_\textnormal{sys}}(t) & := -\frac{i}{2}[L_\alpha^\dagger, L_\alpha [H_\text{sys},\tr{\mathcal{A}}{\rho_{\omega}(t)}]] + \text{h.c.},
    \nonumber \\
    \mathcal{Q}^{(3)}_{\alpha, \alpha'}(t) & := \frac{-i}{2\omega}[L_\alpha^\dagger, F_{\alpha, \alpha'}(t)] + \text{h.c},
    \nonumber \\
    \mathcal{Q}^{(4)}_{\alpha, \alpha'}(t) & := -\frac{i}{\omega^2} [L_\alpha^\dagger, E_{\alpha, \alpha'}(t)] + \text{h.c.}.
\end{align}
The expressions for $q_\alpha$, $\mathcal{Q}_{\alpha,H_{sys}}^{(1)}(t)$, $\mathcal{Q}_{\alpha,H_{sys}}^{(2)}(t)$, $\mathcal{Q}_{\alpha,\alpha'}^{(3)}(t)$, and $\mathcal{Q}_{\alpha,\alpha'}^{(4)}(t)$ can be seen to be the same as those provided in the lemma statement.
\end{proof}

\begin{replemma}{lemma:tr_sigma_bounds}
    Suppose $\rho_\omega(t)$ is the joint state of the system and the ancilla qubits with the ancilla qubits initially being in state $\ket{0}$, then for all $\alpha, \alpha'$
    \begin{align*}
        &\norm{\sigma_{\alpha}\rho_\omega(t)}_1 \leq \frac{\omega}{2} \text{ and }\norm{\tr{\mathcal{A}}{\sigma_{\alpha}^\dagger \sigma_{\alpha'}\rho_\omega(t)}}_1,
        \norm{\tr{\mathcal{A}}{\sigma_{\alpha} \sigma_{\alpha'}\rho_\omega(t)}}_1\leq \frac{\omega^2}{4}.
    \end{align*}
\end{replemma}
\begin{proof}
    The proof of the bound on $\norm{\sigma_\alpha\rho_{\omega}(t)}_1$ is provided after the statement of Lemma~\ref{lemma:tr_sigma_bounds} in Subsection~\ref{subsec:rigorous_adiabatic_elimination}. Here we bound $\norm{\tr{\mathcal{A}}{\sigma_\alpha^\dagger \sigma_{\alpha'}\rho_{\omega}(t)}}$ and $\norm{\tr{\mathcal{A}}{\sigma_\alpha \sigma_{\alpha'}\rho_{\omega}(t)}}$ for any $\alpha,\alpha'$. It is convenient to define the ``Heisenberg-like" picture of any operator or superoperator $O$ to be $O(t):=\mathcal{E}_{\omega}^{-1}(t,0)O\mathcal{E}_{\omega}(t,0)$. Accordingly we may write in vectorized form that, for $u \in \{-,+\}$,
    \[
    \frac{d}{dt}\left(\sigma_{\alpha,u}^{(u)}(t)\sigma_{\alpha',l}(t)\right) = [\sigma_{\alpha,u}^{(u)}(t) \sigma_{\alpha',l}(t),\mathcal{L_{\omega}}(t)] = -4 \sigma_{\alpha,u}^{(u)}(t) - i\omega L_{\alpha',l}(t)\sigma_{\alpha',l}^z(t) \sigma_{\alpha,u}^{(u)}(t) + u i \omega L_{\alpha,u}^{(u)}(t) \sigma_{\alpha,u}^z(t) \sigma_{\alpha',l}(t),
    \]
    where we interpret the subscripts as $-=l,+=r$ and $\sigma_{\alpha}^z:=[\sigma_\alpha,\sigma_\alpha^\dagger]$ is the pauli-Z operator acting on ancilla qubit $\alpha$. Integrating the above equation, we obtain a relation between $\sigma_{\alpha,u}^{(u)}(t)\sigma_{\alpha,l}(t)$ and $\sigma_{\alpha,u}^{(u)}\sigma_{\alpha,l}$,
    \[
        \sigma_{\alpha,u}^{(u)}(t)\sigma_{\alpha,l}(t) = e^{-4t} \sigma_{\alpha,u}^{(u)} \sigma_{\alpha',l} + i \omega \int_0^t e^{-4(t-t')} \left(-L_{\alpha',l}(t')\sigma_{\alpha',l}^z(t') \sigma_{\alpha,u}^{(u)}(t') + uL_{\alpha,u}^{(u)}(t') \sigma_{\alpha,u}^z(t')\sigma_{\alpha',l}(t')\right) dt'.
    \]
    We express the bound on $\norm{\tr{\mathcal{A}}{\sigma_{\alpha}^{(u)} \sigma_{\alpha'}\rho_{\omega}(t)}}_1$ in vectorized form using the above equation as
    \begin{align}
    & \norm{\vecbra{\text{Tr}_{\mathcal{A}}}\sigma_{\alpha,l}^{(u)} \sigma_{\alpha',l} \vecket{\rho_{\omega}(t)}}_1 = \norm{\vecbra{\text{Tr}_{\mathcal{A}}}\mathcal{E}_{\omega}(t,0) \sigma_{\alpha,u}^{(u)}(t)\sigma_{\alpha,l}(t) \vecket{\rho(0)}}_1
    \nonumber \\
    & \quad \leq \omega \int_0^t e^{-4(t-t')} \left(\norm{\vecbra{\text{Tr}_{\mathcal{A}}}\mathcal{E}_{\omega}(t,0) L_{\alpha',l}(t')\sigma_{\alpha',l}^z(t') \sigma_{\alpha,u}^{(u)}(t') \vecket{\rho(0)}}_1 + \norm{\vecbra{\text{Tr}_{\mathcal{A}}}\mathcal{E}_{\omega}(t,0) L_{\alpha,u}^{(u)}(t') \sigma_{\alpha,u}^z(t')\sigma_{\alpha',l}(t')\vecket{\rho(0)}}_1\right) dt',
    \end{align}
    where we have used the assumption that $\rho(0)$ is initialized with all ancilla in vacuum to apply $\sigma_{\alpha',l}\vecket{\rho(0)}=0$. Using the bound $\norm{\sigma_\alpha \rho_{\omega}(t')}_1\leq\omega/2$ for any $\alpha$ from the first part of Lemma~\ref{lemma:tr_sigma_bounds} in conjunction with the contracting effect of channels on the 1-norm, we have $\norm{\sigma_{\alpha,l}(t)\vecket{\rho(0)}}_1,\norm{\sigma_{\alpha,r}^\dagger(t)\vecket{\rho(0)}}_1\leq\omega/2$ for any $\alpha$ and $t\geq0$. Inserting these bounds into the above equation and applying the assumption $\norm{L_\alpha} \leq 1$, we find
    \begin{align}
        \norm{\tr{\mathcal{A}}{\sigma_\alpha^{(u)}\sigma_\alpha\rho(t)}}_1 \leq \omega \int_0^t e^{-4(t-t')} \left( \frac{\omega}{2} + \frac{\omega}{2}\right)dt' \leq \frac{\omega^2}{4}.
    \end{align}
    The lemma statement is obtained by either setting $u=-$ or $u=+$.
\end{proof}
\section{Detailed proofs of lemmas from section \ref{sec:geometric_locality}}
\label{appendix:geometric_locality_lemma_proofs}
\noindent Throughout this appendix, we will focus on geometrically local Lindbladians introduced in section~\ref{sec:geometric_locality}. These will be Lindbladians of the form 
\[
\mathcal{L} = \sum_\alpha \mathcal{L}_\alpha,
\]
where we will assume that $\mathcal{L}_\alpha(X) = -i[h_\alpha, X] + \mathcal{D}_{L_\alpha}(X)$ for operators $h_\alpha, L_\alpha$ supported on $S_\alpha$. Furthermore, we will assume that diameter $\text{diam}(S_\alpha) \leq a$ for all $\alpha$, and the degree $\partial(S_\alpha) = \{\alpha': S_\alpha \cap S_{\alpha'} = \emptyset\} \leq \mathcal{Z}$ for all $\alpha$. As was mentioned in section~\ref{sec:geometric_locality}, a key ingredient in our analysis is the Lieb-Robinson bound from Ref.~\cite{barthel_kliesch_2012}.
\begin{replemma}{lemma:lieb_robinson}
Suppose $\mathcal{K}$ is a superoperator supported on region $S_\mathcal{K}$ and satisfies $\mathcal{K}(I) = 0$, and $O$ is a local observable, with $\norm{O} \leq 1$, supported on region $S_O$, then
\begin{align*}
    \norm{\mathcal{K}  e^{\mathcal{L}^\dagger t}(O)} \leq \eta_{S_O} \norm{\mathcal{K}}_{cb,\infty\to\infty}\exp\bigg(4e\mathcal{Z}t - \frac{d({S_\mathcal{K}, S_O})}{a}\bigg).
\end{align*}
\end{replemma}
\noindent An immediate consequence of this lemma is the following modified ``2-point" Lieb-Robinson bound.
\begin{replemma}{lemma:2superopLR}
$\mathcal{K}_X$ and $\mathcal{K}_Y$ are superoperators supported on regions $X$ and $Y$ such that $\mathcal{K}_X(I) = \mathcal{K}_Y(I) = 0$ and $\textnormal{diam}(X),\textnormal{diam}(Y)\leq a$, then for any local observable $O$, with $\norm{O} \leq 1$, supported on region $S_O$
\begin{align*}
    &\norm{\mathcal{K}_X \mathcal{K}_Y e^{\mathcal{L}^\dagger t}(O)} \leq
    e \eta_{S_O}\norm{\mathcal{K}_X}_{cb,\infty\to\infty} \norm{\mathcal{K}_Y}_{cb,\infty\to\infty} \exp\bigg(4e\mathcal{Z}t - \frac{1}{2a}(d({X, S_O}) + d({Y, S_O}))\bigg).
\end{align*}
\end{replemma}
\begin{proof} Note that
\begin{align}\label{eq:X_first_then_Y}
    \norm{\mathcal{K}_X \mathcal{K}_Y e^{\mathcal{L}^\dagger t}(O)}  &\leq \norm{\mathcal{K}_X}_{cb,\infty\to\infty} \norm{\mathcal{K}_Y e^{\mathcal{L}^\dagger t}(O)}
 \leq \frac{\partial(S_O)}{\mathcal{Z}} \norm{\mathcal{K}_X}_{cb,\infty\to\infty}\norm{\mathcal{K}_Y}_{cb,\infty\to\infty}\exp\bigg(4e\mathcal{Z}t - \frac{1}{a} d(Y, S_O)\bigg).
\end{align}
Next, consider the case $X \cap Y = \emptyset$, then $\mathcal{K}_X \mathcal{K}_Y = \mathcal{K}_Y \mathcal{K}_X$. In this case, the above equation holds with $X$ and $Y$ swapped i.e.
\begin{align}\label{eq:Y_first_then_X}
    \norm{\mathcal{K}_X \mathcal{K}_Y e^{\mathcal{L}^\dagger t}(O)}    \leq \frac{\partial(S_O)}{\mathcal{Z}}\norm{\mathcal{K}_X}_{cb,\infty\to\infty}\norm{\mathcal{K}_Y}_{cb,\infty\to\infty}\exp\bigg(4e\mathcal{Z}t - \frac{1}{a} d(X, S_O)\bigg).
\end{align}
Multiplying Eqs.~\ref{eq:X_first_then_Y} with \ref{eq:Y_first_then_X}, we obtain that
\begin{align}\label{eq:final_result_non_intersect}
 \norm{\mathcal{K}_X \mathcal{K}_Y e^{\mathcal{L}^\dagger t}(O)} \leq \frac{\partial(S_O)}{\mathcal{Z}}\norm{\mathcal{K}_X}_{cb,\infty\to\infty}\norm{\mathcal{K}_Y}_{cb,\infty\to\infty}\exp\bigg(4e\mathcal{Z}t - \frac{1}{2a} \big(d(X, S_O) + d(Y, S_O)\big)\bigg)
\end{align}
Consider now the case when $X\cap Y \neq \emptyset$. In this case, we can treat $\mathcal{K}_X \mathcal{K}_Y$ as a superoperator supported on $X \cup Y$. From the Lieb-Robinson bounds, it then follows that
\begin{align}\label{eq:lieb_robinsion_union_region}
    \norm{\mathcal{K}_X \mathcal{K}_Y e^{\mathcal{L}^\dagger t}(O)} & \leq \frac{\partial(S_O)}{\mathcal{Z}}\norm{\mathcal{K}_X \mathcal{K}_Y}_{cb,\infty\to\infty} \exp\bigg(4e\mathcal{Z}t- \frac{1}{a}d(X \cup Y, S_O)\bigg).
\end{align}
We note that $\norm{\mathcal{K}_X \mathcal{K}_Y}_{cb,\infty\to\infty} \leq \norm{\mathcal{K}_X }_{cb,\infty\to\infty}\norm{\mathcal{K}_Y}_{cb,\infty\to\infty} $. Next, we note that
\[
    d(X \cup Y, S_O) = \min(d(X,S_O), d(Y,S_O) = \min(d(x_0, S_O), d(y_0,S_O)),
\]
where $x_0 = \textnormal{argmin}_{x\in X}d(x,S_O)$ and $y_0 = \textnormal{argmin}_{y\in Y}d(y,S_O)$. Since $\textnormal{diam}(X),\textnormal{diam}(Y)\leq a$, we have $d(x_0,y_0) \leq a$ and so $\min(d(x_0, S_O), d(y_0,S_O)) \geq d(x_0, S_O) - d(x_0, y_0) = d(X,S_0) - a$. By the same logic we also have $\min(d(x_0, S_O), d(y_0,S_O)) \geq d(Y,S_0) - a$ and consequently we have that
\[
    d(X \cup Y, S_O) \geq \frac{1}{2}(d(X(S_O)) + d(Y(S_O)))-a.
\]
Therefore, from Eq.~\ref{eq:lieb_robinsion_union_region}, we obtain that when $X \cap Y \neq \emptyset$,
\begin{align}\label{eq:final_result_intersect}
\norm{\mathcal{K}_X\mathcal{K}_Y e^{\mathcal{L}^\dagger t}(O)} \leq \frac{e  \partial(S_O)}{\mathcal{Z}} \norm{\mathcal{K}_X}_{cb,\infty\to\infty}\norm{\mathcal{K}_Y}_{cb,\infty\to\infty}\exp\bigg(4e\mathcal{Z}t - \frac{1}{2a}(d(X, S_O) + d(Y, S_O))\bigg).
\end{align}
From Eqs.~\ref{eq:final_result_non_intersect} and \ref{eq:final_result_intersect}, the lemma statement follows.
\end{proof}

An upper bound on the sum of an exponential function over a $d-$dimensional lattice is a useful technical result.
\begin{lemma}\label{lemma:lr_summation}
    Suppose $\tilde{S}_\alpha$, for $\alpha \in \mathbb{N}$, are subsets of $\mathbb{Z}^d$ which satisfy that for some $\tilde{\mathcal{Z}}> 0$, $\partial_{\tilde{S}_\alpha} = \abs{\{\alpha' : \tilde{S}_\alpha \cap \tilde{S}_{\alpha'} \neq \emptyset\}} \leq \tilde{\mathcal{Z}}$ for all $\alpha \in \mathbb{N}$. Suppose $S \subseteq \mathbb{Z}^d$ with $\abs{S} < \infty$ is a fixed subset of $\mathbb{Z}^d$. For any $k \in \{1, 2 \dots \}$, $m \in \{0, 1\}$ and $T, \lambda, x \geq 0$, define
    \[
    \xi^{(m, k)}_{\lambda, x}(T) = \sum_{\alpha_1, \alpha_2 \dots \alpha_k} d_{\alpha_1, \alpha_2 \dots \alpha_k}^m \min\bigg(x\exp\bigg(T - \frac{1}{\lambda} d_{\alpha_1, \alpha_2 \dots \alpha_k}\bigg), 1\bigg),
    \]
    where $d_{\alpha_1, \alpha_2 \dots \alpha_k} = \sum_{j  =1}^k d(S, \tilde{S}_{\alpha_j})$. Then, 
    \[
    \xi^{(m, k)}_{\lambda, x}(T)\leq \max(x, 1)^k k^m \big(\nu^{(0)}(\lambda, T)\big)^{k - m}\big(\nu^{(1)}(\lambda, T)\big)^{m},
    \]
    where,  $\nu^{(m)}(\lambda, T)$ are piecewise continuous non-decreasing functions of $\lambda$ (for a fixed $T$) and $T$ (for a fixed $\lambda$) which only depend on $\tilde{\mathcal{Z}}$ and $d$. Furthermore, 
    \[
    \nu^{(m)}(\lambda, T) \leq O((\lambda T)^{m + d}) \text{ as }\lambda, T \to \infty \text{ and }\nu^{(m)}(\lambda, 0) \leq O(\lambda^{m + d}) \text{ as }\lambda \to \infty.
    \]
\end{lemma}
\begin{proof}
    We can reduce the multiple summation to a single summation --- and important fact that we will use here is that for $0 < x_1, x_2 \dots x_k \leq 1$ and $A > 1$,
    \begin{align}\label{eq:min_inequality}
    \min(Ax_1 x_2 \dots x_k, 1) \leq \min(Ax_1, 1) \min(Ax_2, 1) \dots \prod_{i = 1}^k \min(Ax_i, 1). 
    \end{align}
    To see that Eq.~\ref{eq:min_inequality} is true, consider two cases --- if $Ax_1 x_2 \dots x_k \geq 1$, then for any $i \in [1:k]$,
    \[
    Ax_i \geq \frac{1}{x_1 \dots x_{i - 1} x_{i + 1} \dots x_k}\geq 1 \implies \min(Ax_i, 1) = 1
    \]
    and hence satisfies Eq.~\ref{eq:min_inequality}. On the other hand if $Ax_1 x_2 \dots x_k \leq 1$, then we again have that
    \[
    \min(Ax_1 x_2 \dots x_k, 1) = Ax_1 x_2 \dots x_k \leq \prod_{i \in S} Ax_i \text{ for any } S \subseteq [1:k] \implies \min(Ax_1 x_2 \dots x_k, 1) \leq \prod_{i = 1}^k\min(Ax_i, 1). 
    \]
    where we have used that $0 < x_i \leq 1$ and $A > 1$.
    Now, noting that for any $x, y > 0$, $\min(y, x) \leq \max(x, 1) \min(y, 1)$. Using this together with Eq.~\ref{eq:min_inequality}, we obtain that,
    \begin{subequations}\label{eq:single_sum_to_multiple}
    \begin{align}
    \xi^{(0, k)}_{\lambda, x}(T) &= \sum_{\alpha_1, \alpha_2 \dots \alpha_k}\min\bigg(x\exp\bigg(T - \frac{1}{\lambda}d_{\alpha_1, \alpha_2 \dots \alpha_k}\bigg), 1\bigg) \nonumber \\
    &\leq \max(x, 1) \sum_{\alpha_1, \alpha_2 \dots \alpha_k}\prod_{i = 1}^k \min\bigg(\exp\bigg(T -\frac{1}{\lambda} d(\tilde{S}_{\alpha_i}, S)\bigg), 1\bigg), \nonumber\\
    & \leq \max(x, 1)^k \bigg(\sum_\alpha \min\bigg(\exp\bigg(T -\frac{1}{\lambda} d(\tilde{S}_{\alpha}, S)\bigg), 1\bigg)\bigg)^k.
    \end{align}
    and
    \begin{align}
    \xi^{(1, k)}_{\lambda, x}(T) &= \sum_{\alpha_1, \alpha_2 \dots \alpha_k}d_{\alpha_1, \alpha_2 \dots \alpha_k}\min\bigg(\exp\bigg(T - \frac{1}{\lambda}d_{\alpha_1, \alpha_2 \dots \alpha_k}\bigg), x\bigg) \nonumber \\
    &\leq \max(x, 1) \sum_{\alpha_1, \alpha_2 \dots \alpha_k}\bigg(\sum_{i = 1}^k d(\tilde{S}_{\alpha_i}, S)\bigg)\prod_{i = 1}^k \min\bigg(\exp\bigg(T -\frac{1}{\lambda} d(\tilde{S}_{\alpha_i}, S)\bigg), 1\bigg), \nonumber\\
    &= k\max(x, 1) \sum_{\alpha_1, \alpha_2 \dots \alpha_k}d(\tilde{S}_{\alpha_1}, S)\prod_{i = 1}^k \min\bigg(\exp\bigg(T -\frac{1}{\lambda} d(\tilde{S}_{\alpha_i}, S)\bigg), 1\bigg), \nonumber\\
    & \leq k\max(x, 1)^k \bigg(\sum_\alpha \min\bigg(\exp\bigg(T -\frac{1}{\lambda} d(\tilde{S}_{\alpha}, S)\bigg), 1\bigg)\bigg)^{k - 1}\sum_\alpha d(\tilde{S}_{\alpha}, S)\min\bigg(\exp\bigg(T -\frac{1}{\lambda} d(\tilde{S}_{\alpha}, S)\bigg), 1\bigg).
    \end{align}
    \end{subequations}
    Having reduced the multiple summation to a single summation, we re-express the single summation with respect to $\alpha$ with a summation with respect to the grid distances. In particular, we note that
    \begin{align}\label{eq:summation_to_distance}
        \sum_{\alpha} d^m(\tilde{S}_\alpha, S)\min\bigg(\exp\bigg(T - \frac{1}{\lambda} d(\tilde{S}_{\alpha}, S)\bigg), 1 \bigg) = \sum_{n = 0}^\infty n^m\min\bigg(\exp\bigg(T - \frac{1}{\lambda}n\bigg), 1 \bigg) \bigabs{\{\alpha : d(\tilde{S}_\alpha, S) = n\}}.
    \end{align}
    We next upper bound $\abs{\{\alpha: d(\tilde{S}_\alpha, S) = n\}}$ --- note that
    \begin{align}\label{eq:upper_bound_d_Salpha_S}
        \abs{\{\alpha : d(\tilde{S}_\alpha, S) = n\}} \numleq{1} \tilde{\mathcal{Z}}\abs{\{x \in \mathbb{Z}^d : d(x, S) = n\}} \numleq{2} \tilde{\mathcal{Z}}\abs{S}\abs{\{x \in \mathbb{Z}^d : d(x, 0) = n\}},
    \end{align}
    where in (1) we have used the fact that any one site $x \in \mathbb{Z}^d$ can be in at most $\tilde{\mathcal{Z}}$ sets $\tilde{S}_\alpha$ (since otherwise there would be more than $\tilde{\mathcal{Z}}$ sets $\tilde{S}_\alpha$ that would intersect each other), and in (2) we have simply upper bounded the number of sites at a distance of $n$ from $S$ by the number of sites at a distance of $n$ from any site in $S$.

    A bound on $\abs{\{x \in \mathbb{Z}^d : d(x, 0) = n\}}$ can be obtained by noting that, since $d(x, y) = \sum_{i = 1}^d \abs{x_i - y_i}$ is the Manhattan distance on the lattice, 
    \begin{align}\label{eq:upper_bound_circum_zd}
    \abs{\{x \in \mathbb{Z}^d : d(x, 0) = n\}} &= \bigabs{\bigg\{x \in \mathbb{Z}^d : \sum_{i = 1}^d \abs{x_i} = n\bigg\}}\leq 2^d \bigabs{\bigg\{x\in \mathbb{N}_0^d : \sum_{i = 1}^d x_i = n \bigg\}}, \nonumber \\
    &= 2^d {n +d - 1\choose d - 1} \leq \frac{2^d}{(d - 1)!}(n + d - 1)^{d - 1}.
    \end{align}
    Combining Eqs.~\ref{eq:upper_bound_d_Salpha_S} and \ref{eq:upper_bound_circum_zd}, we obtain that
    \[
    \abs{\{\alpha:d(\tilde{S}_\alpha, S) = n\}} \leq \frac{2^d\tilde{\mathcal{Z}}\abs{S}}{(d - 1)!}(n + d - 1)^{d - 1}.
    \]
    Returning to Eq.~\ref{eq:summation_to_distance}, we obtain that
    \begin{align}\label{eq:final_sum_alpha_to_n}
        \sum_{\alpha}d^m(\tilde{S}_\alpha, S)\min\bigg(\exp\bigg(T - \frac{1}{\lambda} d(\tilde{S}_{\alpha}, S)\bigg), 1 \bigg) &\leq \frac{2^d \tilde{\mathcal{Z}}\abs{S}}{(d - 1)!} \sum_{n = 0}^\infty \min\bigg(\exp\bigg(T - \frac{1}{\lambda}n\bigg), 1 \bigg)  n^m (n + d - 1)^{d - 1},
    \end{align}
    Now, we note that
    \begin{align}\label{eq:summation_n_to_explicit_nu}
        &\sum_{n = 0}^\infty \min\bigg(\exp\bigg(T -\frac{n}{\lambda}\bigg), 1\bigg)n^m(n + d - 1)^k \nonumber\\
        &\qquad\numleq{1} \sum_{n = 0}^{\lceil \lambda T \rceil } n^m (n + d - 1)^{d - 1} + \sum_{n = \lfloor \lambda T \rfloor}^{\infty} n^{m} (n + d - 1)^{d - 1} \exp\bigg(T - \frac{n}{\lambda}\bigg), \nonumber\\
        &\qquad\leq \sum_{n = 0}^{\lceil \lambda T \rceil }n^m (n + d - 1)^{d - 1} + \sum_{n = 0 }^{\infty} (n + \lfloor \lambda T \rfloor)^m(n + \lfloor \lambda T \rfloor + d - 1)^{d - 1}\exp\bigg(T - \frac{\lfloor \lambda T \rfloor}{\lambda} - \frac{n}{\lambda}\bigg), \nonumber\\
        &\qquad\numleq{2} \sum_{n = 0}^{\lceil \lambda T \rceil}n^m(n + d - 1)^{d - 1} + \sum_{n = 0}^\infty (n +\lfloor \lambda T \rfloor)^m (n + \lfloor \lambda T \rfloor + d - 1)^{d - 1} \exp(-n/\lambda), \nonumber \\
        &\qquad\numleq{3} Y_m(\lceil \lambda T \rceil) +2^{d-2} \sum_{\sigma, \sigma' \in \{0, 1\}}  \lfloor \lambda T\rfloor^{(1 - \sigma)m}(d - 1 + \lfloor \lambda T \rfloor)^{(1 - \sigma')(d - 1)} X_{\sigma m + \sigma'(d - 1)}(\lambda),
    \end{align}
    where for $p \in \{0, 1, 2 \dots \}$ and $m \in \{0, 1\}$
    \begin{align*}
        X_p(\lambda) = \sum_{n = 0}^\infty n^p \exp(-n / \lambda) \text{ and }Y_m(z) = \sum_{n = 0}^{\lceil z  \rceil}n^m(n + d - 1)^{d - 1}.
    \end{align*}
    In (1), we have split the summation over $n$ into two summation --- one summation involving only $n \leq \lceil \lambda T \rceil$, where we have used $\min(\exp(T - n/\lambda), 1) \leq 1$, and the other summation involving $n \geq \lfloor \lambda T \rfloor$, where we have used $\min(\exp(T - n/\lambda), 1) \leq \exp(T - n/\lambda)$. In (2), we have used the fact that $\lambda T \leq \lfloor \lambda T \rfloor$, and thus $\exp(T - \lfloor \lambda T \rfloor / \lambda) \leq 1$. In (3), we have used the fact that for $a, b > 0$, $(a + b)^d \leq 2^{d - 1}(a^d + b^d)$. We will now pick the functions $\nu^{(m)}(\lambda, T)$ from the lemma statement to be,
    \[
    \nu^{(m)}(\lambda, T) =\frac{2^d \tilde{\mathcal{Z}}\abs{S}}{(d - 1)!}  \bigg(Y_m(\lceil \lambda T \rceil) +2^{d-2} \sum_{\sigma, \sigma' \in \{0, 1\}}  \lfloor \lambda T\rfloor^{(1 - \sigma)m}(d - 1 + \lfloor \lambda T \rfloor)^{(1 - \sigma')(d - 1)} X_{\sigma m + \sigma'(d - 1)}(\lambda)\bigg).
    \]
    From Eqs.~\ref{eq:single_sum_to_multiple}, \ref{eq:final_sum_alpha_to_n} and \ref{eq:summation_n_to_explicit_nu}, it follows that the lemma statement holds with this choice of $\nu^{(m)}(\lambda, T)$. Furthermore, since the dependence of $\nu^{(m)}(\lambda, T)$ on $T$ is entirely through $\lfloor \lambda T\rfloor$ and $\lceil \lambda T \rceil$, it follows that it is a piecewise continuous function of $T$ for a fixed $\lambda$. It is also clear that $\nu^{(m)}(\lambda, T)$ is non-decreasing in $T$ for a fixed $\lambda$. Finally, to understand the asymptotic behavior of $\nu^{(m)}(\lambda, T)$ as $\lambda, T \to \infty$, we note that $X_p(\lambda) \leq O(\lambda^{p + 1})$ and $Y_m(\lceil \lambda T \rceil) \leq O((\lambda T)^{m  +d})$. Therefore, we have that
    \begin{align*}
    \nu^{(m)}(\lambda, T) &\leq O((\lambda T)^{m  +d}) + \sum_{\sigma, \sigma' \in \{0, 1\}}O((\lambda T)^{(1 - \sigma)m + (1 - \sigma')(d -1)}) O(\lambda^{(\sigma m + \sigma' (d - 1) + 1)}), \nonumber \\
    &\leq O((\lambda T)^{m  +d}) + \sum_{\sigma, \sigma' \in \{0, 1\}}O((\lambda T)^{(1 - \sigma)m + (1 - \sigma')(d -1)}) O((\lambda T)^{(\sigma m + \sigma' (d - 1) + 1)}), \nonumber \\
    &\leq O((\lambda T)^{m + d}).
    \end{align*}
    Furthermore, for $T = 0$, 
    \[
    \nu^{(m)}(\lambda, 0) =\frac{2^d \tilde{\mathcal{Z}}\abs{S}}{(d - 1)!}\bigg(Y_m(0) + 2^{d - 2} (d - 1)^{d - 1} X_{m}(\lambda) + 2^{d- 2} X_{m + d - 1}(\lambda)\bigg) \leq O(\lambda^{m + d}) \text{ as }\lambda \to \infty. 
    \]
This concludes the proof of the lemma.
\end{proof}
\subsection{Proof of Lemma~\ref{lemma:bounds_remainder_lr}}
\noindent In this appendix, we will establish Lemma~\ref{lemma:bounds_remainder_lr}. We first establish the following
\begin{lemma}\label{lemma:bounds_remainder_lr_intermediate}
    Suppose $O$ is a local observable with $\norm{O}\leq 1$ supported on $S_O$, and for $\tau>0$, let $O(\tau) = \exp(\mathcal{L}^\textnormal{\dagger} \tau)(O)$ where $\mathcal{L}$ is a geometrically local Lindbladian of the form in Eq.~\ref{eq:geom_local_lind}. Then for $q_\alpha(s), \mathcal{Q}_{\alpha,\alpha'}^{(j)}(s)$, with $\mathcal{Q}_{\alpha, \alpha'}^{(1)} =\mathcal{Q}_{\alpha, h_{\alpha'}}^{(1)}$ and $\mathcal{Q}_{\alpha, \alpha'}^{(2)} =\mathcal{Q}_{\alpha, h_{\alpha'}}^{(2)}$, as defined in Lemma~\ref{lemma:remainder},
        \[
        \bigabs{\tr{}{O(\tau) q_\alpha(s)}} \leq 2  \min\bigg(\eta_{S_O} \exp\bigg(4e\mathcal{Z}\tau - \frac{d(S_\alpha, S_O)}{a}\bigg), 1\bigg) \text{ for any }\alpha,
        \]
        and
        \[
        \abs{\tr{}{O(\tau)\mathcal{Q}^{(j)}_{\alpha, \alpha'}(s)}} \leq 4\min\bigg(e \eta_{S_O}\exp\bigg(4e\mathcal{Z}\tau - \frac{1}{2a}\big( d(S_\alpha, S_O) + d(S_{\alpha'}, S_O)\big)\bigg), 1 \bigg) \text{ for any }\alpha, \alpha', j.
        \]
\end{lemma}
\begin{proof}
    Using the Lieb-Robinson bounds (Lemma~\ref{lemma:lieb_robinson}) together with the fact that $\norm{\mathcal{D}_{L_\alpha}}_{\infty \to \infty} \leq 2$, we obtain that
    \begin{align*}
    \abs{\tr{}{O(\tau) q_\alpha}} &= \abs{\tr{}{O(\tau) \mathcal{D}_{L_\alpha}(\rho(0))}} \leq \norm{\mathcal{D}_{L_\alpha}^\dagger (O(\tau))}\norm{\rho(0)}_1, \nonumber\\
    &\leq 2 \min\bigg(\eta_{S_O}\exp\bigg(4e\mathcal{Z}\tau - \frac{d(S_\alpha, S_O)}{a} \bigg), 1\bigg).
    \end{align*}
    Next, consider $\text{Tr}(O(\tau) \mathcal{Q}_{\alpha, \alpha'}^{(1)}(s))$. For any $\alpha, \alpha'$ we obtain that
    \begin{align*}
    \abs{\tr{}{O(\tau) \mathcal{Q}^{(1)}_{\alpha, \alpha'}(s)}} &\leq \frac{2}{\omega} \bigabs{\tr{}{O(\tau) [L_\alpha^\dagger, [h_{\alpha'}, \tr{\mathcal{A}}{(\sigma_\alpha \rho_\omega(s)}]]}}, \nonumber\\
    &\leq \frac{2}{\omega} \bignorm{[h_{\alpha'}, [L_\alpha^\dagger, O(\tau)]]} \bignorm{\tr{\mathcal{A}}{(\sigma_\alpha\rho_\omega(s)}}_1, \\
    &\leq \bignorm{[h_{\alpha'}, [L_\alpha^\dagger, O(\tau)]]},
    \end{align*}
    where, in the last step, we have used Lemma~\ref{lemma:tr_sigma_bounds}. Next, we can use Lemma~\ref{lemma:2superopLR} together with the fact that $\norm{[h_{\alpha'}, \cdot]}_{cb, \infty \to \infty},\norm{[L_{\alpha}^{(\pm)}, \cdot]}_{cb, \infty \to \infty}  \leq 2$ to obtain
    \begin{align*}
    &\abs{\tr{}{O(\tau)\mathcal{Q}^{(1)}_{\alpha, h_{\alpha'}}(s)}}  \leq 4 \norm{O}\min\bigg(e \eta_{S_O}\exp\bigg(4e\mathcal{Z}\tau - \frac{1}{2a}\big( d(S_\alpha, S_O) + d(S_{\alpha'}, S_O)\big)\bigg), 1 \bigg).
    \end{align*}

    Next, for any $\alpha, \alpha'$, we obtain that,
    \begin{align*}
        \bigabs{\tr{}{O(\tau)\mathcal{Q}^{(2)}_{\alpha, {\alpha'}}(s)}}
        & \leq \bigabs{\tr{}{O(\tau) [L_\alpha^\dagger, L_\alpha [h_{\alpha'}, \rho_\omega(s)]]}},
        \nonumber \\
        & \leq \norm{[h_{\alpha'},[L_\alpha^\dagger, O(\tau)]L_\alpha]} \norm{\tr{A}{\rho_\omega(s)}}_1,
        \nonumber \\
        & \leq 4 \norm{O} \min\left(e\eta_{S_O} \exp\left(4e\mathcal{Z}\tau - \frac{1}{2a}( d(S_\alpha,S_O) + d(S_{\alpha'},S_O) )\right), 1 \right).
    \end{align*}

    Next, we bound $\abs{\tr{\mathcal{A}}{\mathcal{Q}^{(3)}_{\alpha, \alpha'}}}$. For any $\alpha, \alpha'$, we obtain that
    \begin{align*}
    \bigabs{\tr{}{O(\tau)\mathcal{Q}^{(3)}_{\alpha, \alpha'}(s)}} &\leq \frac{1}{\omega} \sum_{u \in \{+, -\}} \abs{\tr{}{O(\tau)\mathcal{D}_{L_\alpha}\left([L_{\alpha'}^{(u)}, \tr{\mathcal{A}}{\sigma_{\alpha'}^{(\bar u)}\rho_\omega(s)}]\right)}}, \nonumber\\
    &\leq \frac{1}{\omega} \sum_{u \in \{+, -\}} \bignorm{[L_{\alpha'}^{(u)}, \mathcal{D}_{L_\alpha}^\dagger(O(\tau))]} \bignorm{\tr{\mathcal{A}}{\sigma_{\alpha'}^{(\bar u)} \rho_\omega(s)}}_1, \nonumber \\
    &\leq \frac{1}{2} \sum_{u \in \{+, -\}}\bignorm{[L_{\alpha'}^{(u)}, \mathcal{D}_{L_\alpha}^\dagger(O(\tau))]},
    \end{align*}
    where, in the last step, we have used Lemma~\ref{lemma:tr_sigma_bounds}. Next, we use Lemma~\ref{lemma:2superopLR} together with $\norm{[L_{\alpha}^{(u)}, \cdot]}_{cb,\infty\to\infty} \leq 2$ and $\norm{\mathcal{D}_{L_\alpha}^\dagger}_{cb,\infty\to\infty} \leq 2$ to obtain
    \[
    \bigabs{\tr{}{O(\tau)\mathcal{Q}^{(3)}_{\alpha, \alpha'}(s)}} \leq 4 \norm{O}\min\bigg(e\eta_{S_O} \exp\bigg(4e\mathcal{Z}\tau - \frac{1}{2a}\big(d(S_\alpha, S_O) + d(S_{\alpha'}, S_O)\big)\bigg), 1\bigg). 
    \]
    
    Next we bound $\abs{\tr{}{O(\tau)\mathcal{Q}^{(4)}_{\alpha,\alpha'}(s)}}_1$. For $\alpha \neq \alpha'$
    \begin{align}
    \abs{\tr{}{O(\tau) \mathcal{Q}^{(4)}_{\alpha, \alpha'}(s)}} &\leq \frac{1}{\omega^2}\abs{ \tr{}{O(\tau) [L_\alpha^{(u)}, [L_{\alpha'}^{(u')}, \tr{\mathcal{A}}{\sigma_\alpha^{(\bar{u})} \sigma_{\alpha'}^{(\bar{u}')} \rho_\omega(s)}]]}}, \nonumber\\
    &\leq \frac{1}{\omega^2} \sum_{u, u' \in \{-, +\}} \bignorm{[L_{\alpha'}^{({u}')}, [L_\alpha^{(u)}, O(\tau)]]}\bignorm{\tr{\mathcal{A}}{\sigma_\alpha^{(\bar{u})} \sigma_{\alpha'}^{(\bar{u}')} \rho_\omega(s)}}_1, \nonumber\\
    &\leq \frac{1}{4}\sum_{u, u' \in \{-, +\}} \bignorm{[L_{\alpha'}^{({u}')}, [L_\alpha^{(u)}, O(\tau)]]}\nonumber,
    \end{align}
    where in the last step we have used Lemma~\ref{lemma:tr_sigma_bounds}. Furthermore, from Lemma~\ref{lemma:2superopLR} and the fact that $\norm{[L_\alpha^{(u)}, \cdot]}_{cb,\infty\to\infty} \leq 2$ it follows that
    \begin{align}\label{eq:remainder_lr_bound_alphas_unequal}
    \bigabs{\tr{}{O(\tau) \mathcal{Q}^{(4)}_{\alpha, \alpha'}(s)}} \leq 4\norm{O} \min\bigg(e \eta_{S_O}\exp\bigg(4e\mathcal{Z}\tau - \frac{1}{2a}\big( d(S_\alpha, S_O) + d(S_{\alpha'}, S_O)\big)\bigg), 1 \bigg).
    \end{align}
    Similarly, for $\alpha = \alpha'$, 
    \begin{align*}
        \bigabs{\tr{}{O(\tau) \mathcal{Q}^{(4)}_{\alpha, \alpha}(s)}} &\leq \frac{2}{\omega^2} \bigabs{\tr{}{O(\tau) (\mathcal{D}_{L_\alpha} - \mathcal{D}_{L_\alpha^\dagger})(\tr{\mathcal{A}}{n_\alpha \rho_\omega(s)})}}, \nonumber \\
        &\leq \frac{2}{\omega^2} \norm{(\mathcal{D}_{L_\alpha} - \mathcal{D}_{L_\alpha^\dagger})^\dagger \big(O(\tau)\big)} \norm{\tr{\mathcal{A}}{n_\alpha \rho_\omega(s)}}_1, \nonumber\\
        &\leq \frac{1}{2}\norm{(\mathcal{D}_{L_\alpha} - \mathcal{D}_{L_\alpha^\dagger})^\dagger \big(O(\tau)\big)},
    \end{align*}
    where, again, in the last step, we have used Lemma~\ref{lemma:tr_sigma_bounds}. Next, we can use Lemma~\ref{lemma:lieb_robinson} together with the fact that $\norm{(\mathcal{D}_{L_\alpha} - \mathcal{D}_{L_\alpha^\dagger})^\dagger}_{cb,\infty\to\infty} \leq 4$, we obtain that
    \begin{align}\label{eq:remainder_lr_bound_alphas_equal}
    \abs{\tr{}{O(\tau)\mathcal{Q}^{(4)}_{\alpha, \alpha}(s)}} &\leq 2 \norm{O} \min\bigg(\eta_{S_O}\exp\bigg(4e\mathcal{Z}\tau - \frac{1}{a}d(S_\alpha, S_O)\bigg), 1\bigg).
    \end{align}
    Eqs.~\ref{eq:remainder_lr_bound_alphas_equal} and \ref{eq:remainder_lr_bound_alphas_unequal} together establish establish a bound on $\abs{\tr{}{O(\tau)\mathcal{Q}^{(4)}_{\alpha, \alpha'}(s)}}$ for any $\alpha,\alpha'$ and is consistent with the lemma statement.
\end{proof}
\begin{replemma}{lemma:bounds_remainder_lr}
Suppose $O$ is a local observable with $\norm{O} \leq 1$ supported on $S_O$, and for $\tau>0$, let $O(\tau) = \exp(\mathcal{L}^\textnormal{\dagger} \tau)(O)$ where $\mathcal{L}$ is a geometrically local Lindbladian of the form in Eq.~\ref{eq:geom_local_lind}. Then for\ $q_\alpha$ as defined in Lemma~\ref{lemma:remainder}, then there is a non-decreasing piecewise continuous function $\nu$ such that $\nu(t) \leq O(t^d)$ as $t \to \infty$ and
        \[
        \sum_{\alpha}\bigabs{\tr{}{O(\tau) q_\alpha}} \leq  \nu(\tau),
        \]
and for $j \in \{1, 2, 3, 4\}$
        \begin{align*}
        &\sum_{\alpha, \alpha'}\bigabs{\tr{}{O(\tau)\mathcal{Q}_{\alpha, \alpha'}^{(j)}(s)}} \leq \nu^2(\tau),
        \end{align*}
where, for $j \in \{3, 4\}$, $\mathcal{Q}_{\alpha, \alpha'}^{(j)}$ is defined in Lemma~\ref{lemma:remainder} and for $j\in \{1, 2\}$, we define $\mathcal{Q}_{\alpha, \alpha'}^{(j)} = \mathcal{Q}_{\alpha, h_{\alpha'}}^{(j)}$ where $\mathcal{Q}_{\alpha, h}^{(j)}$ is defined in Lemma~\ref{lemma:remainder}.
\end{replemma}
\begin{proof}
    We will prove this lemma using Lemmas \ref{lemma:lr_summation} and \ref{lemma:bounds_remainder_lr_intermediate}. Consider first
    \begin{align}\label{eq:function_bound_single_sum}
        \sum_{\alpha}\abs{\text{Tr}(O(\tau) q_\alpha)}
        & \leq 2 \sum_{\alpha} \min\bigg(\eta_{S_O} \exp\bigg(4e\mathcal{Z}\tau - \frac{d(S_\alpha, S_O)}{a}\bigg), 1\bigg), \nonumber\\
        & = 2\xi^{(0, 1)}_{a, \eta_{S_O}}(4e\mathcal{Z}\tau) \leq 2 \max(\eta_{S_O}, 1)\nu^{(0)}(a, 4e\mathcal{Z}\tau) \leq 2e\max(\eta_{S_O}, 1) \nu^{(0)}(2a, 4e\mathcal{Z}\tau),
    \end{align}
    where in the last step, we have used the fact that as per Lemma~\ref{lemma:lr_summation} $\nu^{(0)}(\lambda, T)$ is a non-decreasing function of $\lambda$ for a fixed $T$.
    Similarly, for $j \in \{1, 2, 3, 4\}$,
    \begin{align}\label{eq:function_bound_double_sum}
        &\sum_{\alpha_1, \alpha_2}\abs{\text{Tr}(O(\tau) \mathcal{Q}_{\alpha, \alpha'}^{(j)}(s)}, \nonumber\\
        &\qquad \leq 4\sum_{\alpha, \alpha'} \min\bigg(e\eta_{S_O} \exp\bigg(4e\mathcal{Z}\tau - \frac{d(S_\alpha, S_O) + d(S_{\alpha'}, S_O)}{2a}\bigg), 1\bigg),\nonumber\\
        &\qquad = 4 \xi_{2a, e\eta_{S_O}}^{(0, 2)}(2a, 4e\mathcal{Z}\tau) \leq 4 \big(\max(e\eta_{S_O}, 1)\big)^2 \left(\nu^{(0)}(2a, 4e\mathcal{Z}\tau)\right)^2 \leq 4e^2 \big(\max(\eta_{S_O}, 1)\big)^2 \big(\nu^{(0)}(2a, 4e\mathcal{Z}\tau)\big)^2.
     \end{align}
     From Eqs.~\ref{eq:function_bound_single_sum} and \ref{eq:function_bound_double_sum}, it follows that choosing $\nu(\tau) = 2e \max(\eta_{S_O}, 1) \nu^{(0)}(2a, 4e\mathcal{Z}\tau)$ satisfies the lemma statement. It can also be noted that, from the asymptotics of $\nu^{(0)}(\lambda, T) $ in Lemma~\ref{lemma:lr_summation}, $\nu(\tau) \leq O(\tau^d)$ and that since $\nu^{(0)}(\lambda, T)$, for a fixed $\lambda$, is a non-decreasing and piecewise continuous function of $T$, $\nu(\tau)$ is also a piecewise continuous non-decreasing function of $\tau$.
\end{proof}

\subsection{Proof of Lemma~\ref{lemma:error_term_rapid_mixing}}
\noindent In this appendix, we will outline a proof of Lemma~\ref{lemma:error_term_rapid_mixing}. A convenient technical lemma to use for this proof is below.
\begin{lemma}
\label{lemma:local_rapid_mixing_superop_bounds}
 Suppose $O$ is a spatially local observable and let $O(\tau) = \exp(\mathcal{L}^\textnormal{\dagger}\tau)(O)$ be its Heisenberg-picture evolution with respect to a geometrically local Lindbladian $\mathcal{L}$ of the form of Eq.~\ref{eq:geom_local_lind} and with fixed point $\sigma$. Suppose that $O$ obeys the local rapid mixing condition (Eq.~\ref{eq:rapid_mixing_observable_v2}). Furthermore, suppose $\sigma_\omega(s)$ is a time-dependent operator which satisfies $\norm{\sigma_\omega(s)}_1 \leq f(\omega)$ for some $f(\omega) > 0$. Furthermore, let $\mathcal{K}_\alpha, \mathcal{J}_{\alpha'}$ be superoperators supported on $\tilde{S}_\alpha$ and $\tilde{S}_{\alpha'}$ respectively which satisfy $\mathcal{K}_\alpha^\textnormal{\dagger}(I) = \mathcal{J}_{\alpha'}^\textnormal{\dagger}(I) = 0$ and $\norm{\mathcal{K}_\alpha}_{\diamond}, \norm{\mathcal{J}_\alpha}_{\diamond} \leq 2$ and  $\forall \alpha: \abs{\{\alpha' : \tilde{S}_\alpha \cap \tilde{S}_{\alpha'} \neq \emptyset \}}\leq \tilde{\mathcal{Z}}$ for some $\mathcal{Z}' > 0$. Then,
 \begin{enumerate}
     \item [(a)] For any $t > 0$,
     \begin{align}\label{eq:general_rapid_mixing_bound_single_with_decay}
    \sum_{\alpha} \bigabs{\int_0^t e^{-2s/\omega^2} \tr{}{O(t - s) \mathcal{K}_\alpha(\sigma_\omega(s))} ds}
    \leq \omega^2 f(\omega) \zeta^{(1)}(\gamma),
    \end{align}
    where if $k(\abs{S_O}, \gamma)$, defined in Eq.~\ref{eq:rapid_mixing_observable_v2}, is $O(\exp(\gamma^{-\kappa}))$ as $\gamma \to 0$, then $\zeta^{(1)}(\gamma) \leq O(\gamma^{-d (\kappa + 1)})$ as $\gamma \to 0$.

    \item [(b)] For any $t > 0$,
    \begin{align}\label{eq:general_rapid_mixing_bound_single_without_decay}
    \sum_{\alpha} \bigabs{\int_0^t  \tr{}{O(t - s) \mathcal{K}_\alpha(\sigma_\omega(s))} ds}
    \leq f(\omega) \zeta^{(2)}(\gamma),
    \end{align}
    where if $k(\abs{S_O}, \gamma)$, defined in Eq.~\ref{eq:rapid_mixing_observable_v2}, is $O(\exp(\gamma^{-\kappa}))$ as $\gamma \to 0$, then $\zeta^{(2)}(\gamma) \leq O(\gamma^{-(d + 1) (\kappa + 1)})$ as $\gamma \to 0$.
    \item [(c)] For any $t > 0$,
    \begin{align}\label{eq:general_rapid_mixing_bound_double}
        \sum_{\alpha, \alpha'} \bigabs{\int_0^t \int_0^s e^{-2(s - s')/\omega^2}\tr{}{O(t - s) \mathcal{K}_\alpha \mathcal{J}_{\alpha'}(\sigma_\omega(s'))} ds' ds } \leq \omega^2 f(\omega)\zeta^{(3)}(\gamma),
    \end{align}
    where if $k(\abs{S_O}, \gamma)$, defined in Eq.~\ref{eq:rapid_mixing_observable_v2}, is $O(\exp(\gamma^{-\kappa}))$ as $\gamma \to 0$, then $\zeta^{(3)}(\gamma) \leq O(\gamma^{-(\kappa + 1)(2d + 1)})$.
 \end{enumerate}
 In Eqs.~\ref{eq:general_rapid_mixing_bound_single_with_decay}-\ref{eq:general_rapid_mixing_bound_double}, the functions $\zeta^{(1)}(\gamma)$, $\zeta^{(2)}(\gamma)$ and $\zeta^{(3)}(\gamma)$ depend on $\mathcal{Z}$, $a$, and $\tilde{\mathcal{Z}}$.
\end{lemma}
\begin{proof}

(a) The idea behind this proof is similar to the example detailed in the main text i.e.~to separate out the integral in the summation with into a short-time and long-time integral. The short-time integral is upper bounded using the Lieb-Robinson bounds, and the long-time integral is bounded using the local rapid-mixing property (Eq.~\ref{eq:rapid_mixing_observable_v2}). We begin by noting that
\begin{align}
    \bigabs{\int_0^t e^{-2s/\omega^2} \tr{}{O(t - s) \mathcal{K}_\alpha(\sigma_\omega(s))} ds}  = \bigabs{\int_0^{t}e^{-2(t - s)/\omega^2} \tr{}{O(s) \mathcal{K}_\alpha(\sigma_\omega(t - s))} ds}\leq e_{\alpha}^{\leq}(t, t_\alpha) + e_{\alpha}^{\geq}(t, t_\alpha),
\end{align}
where
\begin{align*}
    &e_{\alpha}^{\leq}(t, t_\alpha) = \bigabs{\int_0^{\min(t, t_\alpha)} e^{-2(t - s)/\omega^2}\tr{}{O(s) \mathcal{K}_\alpha(\sigma_\omega(t - s))}ds},\\
    &e_{\alpha}^{\leq}(t, t_\alpha) = \bigabs{\int_{\min(t, t_\alpha)}^t e^{-2(t - s)/\omega^2}\tr{}{O(s) \mathcal{K}_\alpha(\sigma_\omega(t - s))}ds},
\end{align*}
and $t_\alpha>0$ is to be chosen later. We bound $e_{\alpha}^{\leq}(t,t_\alpha)$ using Lieb-Robinson bounds. First we note that by Lemma~\ref{lemma:lieb_robinson},
\begin{align}\label{eq:lr_bound_short_time_integrand}
\bigabs{\tr{}{O(s)\mathcal{K}_\alpha(\sigma_\omega(t-s))}} \leq \norm{\mathcal{K}^\dagger_\alpha O(s)} \norm{\sigma_\omega(t-s)}_1 \leq 2f(\omega) \min\left( \eta_{S_O} \exp\left(4e\mathcal{Z}s-\frac{d(S_O,\tilde{S}_\alpha)}{a}\right),1\right).
\end{align}
This allows us to bound $e^\leq_\alpha(t,t_\alpha)$:
\begin{align}
    e^\leq_\alpha(t,t_\alpha) & \leq 2f(\omega)\eta_{S_O} \int_0^{\min(t,t_\alpha)} e^{-2(t-s)/\omega^2} \min\bigg(\eta_{S_O} \exp \left( 4e \mathcal{Z}s - \frac{d(S_O,\tilde{S}_\alpha)}{a}\right), 1\bigg) ds,
    \nonumber \\
    & \leq 2f(\omega) \min\bigg(\eta_{S_O} \exp \left( 4e \mathcal{Z}\min(t, t_\alpha) - \frac{d(S_O,\tilde{S}_\alpha)}{a}\right), 1\bigg) \int_0^{\min(t, t_\alpha)}e^{-2(t - s)/\omega^2}ds, \nonumber \\
    &\leq \omega^2 f(\omega) \min\bigg(\eta_{S_O}\exp \left( 4e \mathcal{Z}t_\alpha - \frac{d(S_O,\tilde{S}_\alpha)}{a}\right), 1\bigg)
\end{align}
We bound $e^\geq_\alpha(t,t_\alpha)$ using the local rapid mixing condition. Note that
\begin{align}\label{eq:bound_single_sum_integrand}
    & \bigabs{\tr{}{O(s)\mathcal{K}_\alpha(\sigma_\omega(t-s))}} \nonumber
    \\ 
    & \qquad \leq \bigabs{\tr{}{O\sigma}\tr{}{\mathcal{K}_\alpha(\sigma_\omega(t-s))}} + \bigabs{\tr{}{(O(s)-\tr{}{O\sigma}I)\mathcal{K}_\alpha(\sigma_\omega(t-s))}},\nonumber
    \\
    & \qquad \numeq{1} \norm{O(s)-\tr{}{O\sigma}I} \norm{\mathcal{K}_\alpha}_\diamond \norm{\sigma_\omega(t-s)}_1, \nonumber 
    \\
    &  \qquad \numeq{2} 2f(\omega) k(\abs{S_O}, \gamma)e^{-\gamma s},
\end{align}
where in (1) we have the fact that for any operator $X$, $\text{Tr}(\mathcal{K_\alpha}(X)) = \text{Tr}(\mathcal{K}_\alpha^\dagger(I)X) = 0$ and in (2), we have used Eq.~\ref{eq:rapid_mixing_observable_v2}. For $t \leq t_\alpha$, since $\min(t, t_\alpha) = t$, it follows immediately that
\begin{align}\label{eq:long_time_sindex_eval1}
e_\alpha^{\geq}(t, t_\alpha) = 0.
\end{align}
For $t \geq t_\alpha$, we can obtain two upper bounds on $e_\alpha^{\geq}(t, t_\alpha)$. The first upper bound is using Eq.~\ref{eq:bound_single_sum_integrand} which yields,
\begin{align}\label{eq:long_time_sindex_eval2}
e^{\geq}_\alpha({t, t_\alpha}) \leq 2 f(\omega) k(\abs{S_O}, \gamma) \int_{t_\alpha}^t e^{-\gamma s} e^{-2(t - s)/\omega^2}ds \leq \omega^2 f(\omega) k(\abs{S_O}, \gamma) e^{-\gamma t_\alpha}.
\end{align}
Alternatively, we can also bound
\begin{align}\label{eq:long_time_sindex_eval3}
e^{\geq}_\alpha(t, t_\alpha) \leq \int_{t_\alpha}^t e^{-2(t - s)/\omega^2}\abs{\text{Tr}(O(s) \mathcal{K}_\alpha(\sigma_\omega(t - s)))}ds \leq \int_{t_\alpha}^t e^{-2(t - s)/\omega^2} \norm{O}\norm{\mathcal{K}_\alpha}_\diamond \norm{\sigma_\omega(t - s)}_1 ds \leq \omega^2 f(\omega).
\end{align}
From the bounds in Eqs.~\ref{eq:long_time_sindex_eval1}, \ref{eq:long_time_sindex_eval2} and \ref{eq:long_time_sindex_eval3}, we obtain that,
\[
e^{\geq}_\alpha(t, t_\alpha) \leq \omega^2 f(\omega) \min\big(k(\abs{S_O}, \gamma)e^{-\gamma t_\alpha}, 1\big) \text{ for any }t, t_\alpha.
\]
Combining the bounds for the short-time and long-time integrals, we obtain that
\begin{align}
    &\sum_{\alpha}\bigabs{\int_0^t e^{-2s/\omega^2}\text{Tr}(O(t - s) \mathcal{K}_\alpha(\sigma_\omega(s)))ds}  \nonumber\\
    &\qquad \leq \omega^2 f(\omega) \sum_{\alpha} \bigg(\min\bigg(\eta_{S_O} \exp\bigg(4e\mathcal{Z}t_\alpha - \frac{d(S_O, \tilde{S}_\alpha)}{a}\bigg), 1\bigg) + \min\big(k(\abs{S_O}, \gamma)e^{-\gamma t_\alpha}, 1)\bigg), \nonumber \\
    &\qquad \leq \omega^2 f(\omega) \sum_{\alpha} \bigg(\min\bigg(\eta_{S_O} \exp\bigg(4e\mathcal{Z}t_\alpha - \frac{d(S_O, \tilde{S}_\alpha)}{a}\bigg), 1\bigg) + \min\big(\exp\big(T_{\gamma} -\gamma t_\alpha\big), 1)\bigg),
\end{align}
where $T_\gamma = \log(\max(k(\abs{S_O}, \gamma), 1))$ --- note also that $T_\gamma > 0$. We now pick $t_\alpha = d(S_O,\tilde{S}_\alpha)/(8e\mathcal{Z}a)$ and apply Lemma~\ref{lemma:lr_summation} to obtain
\begin{align}
    \sum_{\alpha}\bigabs{\int_0^t e^{-2s/\omega^2}\text{Tr}(O(t - s) \mathcal{K}_\alpha(\sigma_\omega(s)))ds}&\leq  \omega^2 f(\omega) \big(\xi_{2a, \eta_{S_O}}^{(0, 1)}(0) + \xi^{(0, 1)}_{{8e\mathcal{Z}a \gamma^{-1}}, 1}(T_\gamma)\big), \nonumber\\
    &\leq \omega^2 f(\omega) \underbrace{\big(\max(\eta_{S_O}, 1) \nu^{(0)}(2a, 0) + \nu^{(0)}(8e\mathcal{Z}a \gamma^{-1}, T_\gamma)\big)}_{\zeta^{(0)}(\gamma)}.
\end{align}
Now, if $k(\abs{S_O}, \gamma) \leq O(\exp(\gamma^{-\kappa}))$ as $\gamma \to 0$, then $T_\gamma \leq O(\gamma^{-\kappa})$. Therefore, as $\gamma \to 0$
\[
\zeta^{(0)}(\gamma) = \max(\eta_{S_O}, 1) \nu^{(0)}(2a, 0) + \nu^{(0)}(8e\mathcal{Z}a \gamma^{-1}, T_\gamma) \leq O((\gamma^{-1}T_\gamma)^d)\leq O(\gamma^{-d(\kappa + 1)}),
\]
which proves part (a) of the lemma. Note that $\zeta^{(0)}(\gamma)$ also depends on $\tilde{\mathcal{Z}}$ through $\nu^{(0)}(\lambda, T)$.\\
\noindent (b) We proceed similarly to part (a), and note that
\[
\bigabs{\int_0^t \text{Tr}(O(t - s)\mathcal{K}_\alpha(\sigma_\omega(s))ds} = \bigabs{\int_0^t \text{Tr}(O(s)\mathcal{K}_\alpha(\sigma_\omega(t - s))ds} \leq e_{\alpha}^{\leq}(t, t_\alpha) + e_{\alpha}^{\geq}(t, t_\alpha),
\]
where 
\begin{align*}
    &e_{\alpha}^{\leq}(t, t_\alpha) = \bigabs{\int_0^{\min(t, t_\alpha)} \tr{}{O(s) \mathcal{K}_\alpha(\sigma_\omega(t - s))}ds},\\
    &e_{\alpha}^{\geq}(t, t_\alpha) = \bigabs{\int_{\min(t, t_\alpha)}^t \tr{}{O(s) \mathcal{K}_\alpha(\sigma_\omega(t - s))}ds},
\end{align*}
where $t_\alpha > 0$ is to be chosen later. We proceed similarly to part (a). To bound $e_{\alpha}^{\leq}(t, t_\alpha)$, we use Eq.~\ref{eq:lr_bound_short_time_integrand} to obtain 
\begin{align}
    e_\alpha^{\leq}(t, t_\alpha) & \leq 2 f(\omega) \int_0^{\min(t,t_\alpha)} \min\left( \eta_{S_O} \exp \left( 4e\mathcal{Z}s - \frac{d(S_O,\tilde{S}_\alpha)}{a}\right),1\right) ds,
    \nonumber \\
    & \leq 2 f(\omega) \min\left( \eta_{S_O} \exp \left( 4e\mathcal{Z}\min(t,t_\alpha) - \frac{d(S_O,\tilde{S}_\alpha)}{a}\right),1\right) \int_0^{\min(t,t_\alpha)}ds,
    \nonumber \\
    & \leq 2 f(\omega) t_\alpha \min\left( \eta_{S_O} \exp \left( 4e\mathcal{Z}t_\alpha - \frac{d(S_O,\tilde{S}_\alpha)}{a}\right),1\right).
\end{align}
Next we bound $e_{\alpha}^{\geq}(t, t_\alpha)$. If $t \leq t_\alpha$, then $\min(t,t_\alpha)=t$ and we have
\begin{align}
    e^\geq_\alpha(t,t_\alpha)=0.
    \label{eq:longtime_plainint_eval1}
\end{align}
If instead $t \geq t_\alpha$, we can use Eq.~\ref{eq:bound_single_sum_integrand} to write
\begin{align}
    e^{\geq}_\alpha(t,t_\alpha) \leq 2f(\omega)k(\abs{S_O},\gamma) \int_{t_\alpha}^t e^{-\gamma s} ds \leq 2f(\omega)k(\abs{S_O},\gamma)\gamma^{-1}e^{-\gamma t_\alpha}.
    \label{eq:longtime_plainint_eval2}.
\end{align}
Combining Eqs.~\ref{eq:longtime_plainint_eval1} and \ref{eq:longtime_plainint_eval2}, we obtain that
\begin{align}
    e^{\geq}_\alpha(t,t_\alpha) \leq 2f(\omega) \gamma^{-1}k(\abs{S_O},\gamma) e^{-\gamma t_\alpha} \text{ for any } t, t_\alpha. 
\end{align}
We now pick
\[
    t_{\alpha} = T'_\gamma + \frac{d(S_O,\tilde{S}_\alpha)}{8e\mathcal{Z}a},
\]
where $T'_\gamma=\gamma^{-1}\max(\log(\gamma^{-1}k(\abs{S_O},\gamma)))$. We accordingly obtain the following expressions for the short-time and long-time integrals:
\begin{align}
    \sum_\alpha e^{\leq}_{\alpha}(t,t_\alpha) & \leq f(\omega) \left( 2T'_\gamma \xi^{(0,1)}_{2a, \eta_{S_O}}(4e\mathcal{Z}T'_\gamma) + \frac{1}{4e\mathcal{Z}a} \xi^{(1,1)}_{2a, \eta_{S_O}}(4e\mathcal{Z}T'_\gamma)\right),
    \nonumber \\
    \sum_\alpha e^{\geq}_\alpha(t,t_\alpha) & \leq 2f(\omega) \xi^{(0,1)}_{8e\mathcal{Z}a\gamma^{-1},1}(0),
\end{align}
where we have used $\xi^{(m,k)}_{\lambda,x}(t)$ from Lemma~\ref{lemma:lr_summation}. Using Lemma~\ref{lemma:lr_summation} we obtain that,
\[
    \sum_\alpha \bigabs{\int_0^t \tr{}{O(t-s)\mathcal{K}_\alpha(\sigma_\omega(s))}ds} \leq \sum_\alpha e^{\leq}_{\alpha}(t,t_\alpha) + e^\geq_\alpha(t,t_\alpha) \leq f(\omega) \zeta^{(2)}(\gamma),
\]
where
\[
    \zeta^{(2)}(\gamma) = \max(\eta_{S_O},1)\left( 2T'_\gamma\nu^{(0)}(2a,4e\mathcal{Z}T'_\gamma) + \frac{1}{4e\mathcal{Z}}\nu^{(1)}(2a, 4e\mathcal{Z}T'_\gamma) + 2\nu^{(0)}(8e\mathcal{Z}a\gamma^{-1}, 0)\right).
\]
Note that $\zeta^{(2)}(\gamma)$ also depends on $\tilde{\mathcal{Z}}$ through its dependence on $\nu^{(m)}(\lambda, T)$. If $k(\abs{S_O}, \gamma) \leq O(\exp(\gamma^{-\kappa}))$ as $\gamma \to 0$, then $T'_\gamma \leq O(\gamma^{-(\kappa + 1)})$ and so, using the asymptotics of $\nu^{(m)}(\lambda, T)$ from Lemma~\ref{lemma:lr_summation}, we obtain that
\[
\zeta^{(2)}(\gamma) \leq O(\gamma^{-(\kappa+1)}) \times O(\gamma^{-d(\kappa+1)}) + O(\gamma^{-(d+1)(\kappa+1)}) + O(\gamma^{-d}) \leq O(\gamma^{-(d+1)(\kappa+1)}),
\]
which establishes the lemma.

\noindent (c) We begin by rewriting
\begin{align}\label{eq:split_integral_term}
 &\bigabs{\int_0^t \int_0^s e^{-2(s - s')/\omega^2}\tr{}{O(t - s) \mathcal{K}_\alpha \mathcal{J}_{\alpha'}(\sigma_\omega(s'))} ds' ds }  =  \bigabs{\int_0^t \int_s^t e^{-2(s' - s)/\omega^2}\tr{}{O(s) \mathcal{K}_\alpha \mathcal{J}_{\alpha'}(\sigma_\omega(t - s'))} ds' ds },
\end{align}
where we have made the change of variables $s \to t - s$ and $s' \to t - s'$. Next, we define the short-time and long-time integrals as follows:
\begin{subequations}
\begin{align*}
    \bigabs{\int_0^t \int_s^t e^{-2(s' - s)/\omega^2}\tr{}{O(s) \mathcal{K}_\alpha \mathcal{J}_{\alpha'}(\sigma_\omega(t - s'))} ds' ds}  \leq e^{\leq}_{\alpha, \alpha'}(t, t_{\alpha,\alpha'}) + e^{\geq}_{\alpha, \alpha'}(t, t_{\alpha,\alpha'}),
\end{align*}
where
    \begin{align*}
    &e^{\leq}_{\alpha, \alpha'}(t, t_{\alpha, \alpha'}) = \bigabs{\int_0^{\min(t, t_{\alpha, \alpha'})} \int_s^t e^{-2(s' - s)/\omega^2}\tr{}{O(s) \mathcal{K}_\alpha \mathcal{J}_{\alpha'}(\sigma_\omega(t - s'))} ds' ds}, \nonumber \\
    &e^{\geq}_{\alpha, \alpha'}(t, t_{\alpha, \alpha'}) = \bigabs{\int_{\min(t, t_{\alpha, \alpha'})}^t \int_s^t e^{-2(s' - s)/\omega^2}\tr{}{O(s) \mathcal{K}_\alpha \mathcal{J}_{\alpha'}(\sigma_\omega(t - s'))} ds' ds},
    \nonumber
\end{align*}
\end{subequations}
and $t_{\alpha, \alpha'} > 0$ will be chosen later. Consider first $e_{\alpha, \alpha'}^{\leq}(t, t_{\alpha, \alpha'})$ --- this can be upper bounded using the Lieb-Robinson bounds. Applying Lemma~\ref{lemma:2superopLR} and using $\norm{\mathcal{K}_\alpha^\dagger}_{cb,\infty\to\infty}, \norm{\mathcal{J}_\alpha^\dagger}_{cb,\infty\to\infty} \leq 2$, we see that
\begin{align}
    \bigabs{\tr{}{O(s)\mathcal{K}_\alpha \mathcal{J}_{\alpha'}(\sigma_\omega(t-s'))}} &\leq \bignorm{\mathcal{J}_{\alpha'}^\dagger \mathcal{K}_\alpha^\dagger(O(s))} \bignorm{\sigma_\omega(t - s')}_1,
    \nonumber \\
    & \leq 4 f(\omega) \min\left(e\eta_{S_O} \exp\left(4e\mathcal{Z}s - \frac{1}{2a}(d(S_O, \tilde{S}_\alpha) + d(S_O, \tilde{S}_{\alpha'})\right), 1\right),
    \nonumber
\end{align}
leading to the bound on $e^{\leq}_{\alpha, \alpha'}(t, t_{\alpha, \alpha'})$:
\begin{align}
    e^{\leq}_{\alpha, \alpha'}(t, t_{\alpha, \alpha'}) & \leq 4f(\omega)\int_0^{\min(t, t_{\alpha, \alpha'})} \int_s^t e^{-2(s' - s)/\omega^2}\min\left(e\eta_{S_O} \exp\left(4e\mathcal{Z}s - \frac{1}{2a}(d(S_O, \tilde{S}_\alpha) + d(S_O, \tilde{S}_{\alpha'}))\right), 1 \right) ds' ds,
    \nonumber \\
    & \leq 4f(\omega)\min\bigg(e\eta_{S_O} \exp\left(4e\mathcal{Z}t_{\alpha,\alpha'} - \frac{1}{2a}(d(S_O, \tilde{S}_\alpha) + d(S_O, \tilde{S}_{\alpha'}))\right), 1\bigg) \int_0^{\min(t, t_{\alpha, \alpha'})} \int_s^t e^{-2(s' - s)/\omega^2}ds'ds,
    \nonumber \\
    & \leq 2\omega^2f(\omega) t_{\alpha,\alpha'} \min\bigg(e\eta_{S_O}\exp\left(4e\mathcal{Z}t_{\alpha,\alpha'} - \frac{1}{2a}(d(S_O, \tilde{S}_\alpha) + d(S_O, \tilde{S}_{\alpha'}))\right), 1\bigg).
\end{align}
Next we consider $e_{\alpha, \alpha'}^{\geq}(t, t_{\alpha, \alpha'})$ --- again, to upper bound this term, we use the local rapid mixing property (Eq.~\ref{eq:rapid_mixing_observable_v2}). For any $0 \leq s, s' \leq t$, we have
\begin{align}\label{eq:long_time_integrad_dind_bound}
    \bigabs{\tr{}{O(s)\mathcal{K}_\alpha\mathcal{J}_{\alpha'}(\sigma_\omega(t-s'))}} & \leq \bigabs{\tr{}{O\sigma}\tr{}{\mathcal{K}_\alpha\mathcal{J}_{\alpha'}(\sigma_\omega(t-s')}} + \bigabs{\tr{}{(O(s)-\tr{}{O\sigma}I)\mathcal{K}_\alpha \mathcal{J}_{\alpha'}(\sigma_\omega(t-s'))}},
    \nonumber \\
    & \numleq{1} \norm{O(s)-\tr{}{O\sigma}I} \norm{\mathcal{K}_\alpha}_\diamond \norm{\mathcal{J}_{\alpha'}}_\diamond \norm{\sigma_\omega(t-s')}_1,
    \nonumber \\
    & \numleq 4f(\omega)k(\abs{S_O}, \gamma)e^{-\gamma s}.
\end{align}
where in (1) we have used the fact that $\mathcal{J}_{\alpha'}^\dagger(I)=0$ and therefore for any $X$, $\text{Tr}(\mathcal{J}_{\alpha'}(X)) = \text{Tr}(\mathcal{J}_{\alpha'}^\dagger(I)X) = 0$. In (2), we have used Eq.~\ref{eq:rapid_mixing_observable_v2}. Now, for $t < t_{\alpha, \alpha'}$, $\min(t, t_{\alpha, \alpha'}) = t$ and therefore
\begin{align}\label{eq:long_time_dindex_eval1}
e^{\geq}_{\alpha, \alpha'}(t, t_{\alpha, \alpha'}) = 0.
\end{align}
Next, for $t \geq t_{\alpha, \alpha'}$, using Eq.~\ref{eq:long_time_integrad_dind_bound}, we obtain that
\begin{align}\label{eq:long_time_dindex_eval2}
    e^{\geq}_{\alpha, \alpha'}(t, t_{\alpha, \alpha'}) &\leq  4f(\omega)k(\abs{S_O}, \gamma) \int_{ t_{\alpha, \alpha'}}^t \int_s^t e^{-2(s' - s)/\omega^2} e^{-\gamma s} ds' ds,
    \nonumber \\
    &\leq 2\omega^2 f(\omega)k(\abs{S_O}, \gamma) \int_{t_{\alpha, \alpha'}}^t e^{-\gamma s} ds \nonumber \\
    & \leq 2\gamma^{-1}\omega^2f(\omega)k(\abs{S_O}, \gamma) e^{-\gamma t_{\alpha,\alpha'}}.   
\end{align}
From the bounds in Eq.~\ref{eq:long_time_dindex_eval1} and \ref{eq:long_time_dindex_eval2}, we obtain that
\[
e^{\geq}_{\alpha, \alpha'}(t, t_{\alpha, \alpha'}) \leq 2\gamma^{-1}\omega^2 f(\omega) k(\abs{S_O}, \gamma)e^{-\gamma t_{\alpha, \alpha'}}  \text{ for any }t, t_{\alpha, \alpha'}.
\]
We now pick
\[
t_{\alpha, \alpha'} =  {T}'_\gamma + \frac{1}{16e\mathcal{Z}a}\bigg(d(S_O, \tilde{S}_\alpha) + d(S_O, \tilde{S}_{\alpha'})\bigg)
\]
where $T_\gamma' = \gamma^{-1}\max\big(\log(\gamma^{-1}k(\abs{S_O}, \gamma), 1)\big)$. With this choice, we obtain the following expressions for the long-time and short-time integrals:
\begin{align*}
    \sum_{\alpha,\alpha'} e^\leq_{\alpha,\alpha'}(t,t_{\alpha,\alpha'})
    & \leq \omega^2 f(\omega)\bigg(2T_{\gamma}'\xi_{4a, e\eta_{S_O}}^{(0, 2)}\big(4e\mathcal{Z}{T_\gamma'}\big) + \frac{1}{8e\mathcal{Z}a}\xi_{4a, e\eta_{S_O}}^{(1, 2)}\big(4e\mathcal{Z}T_\gamma'\big)\bigg),
    \nonumber \\
    \sum_{\alpha,\alpha'} e^\geq_{\alpha,\alpha'}(t,t_{\alpha,\alpha'})
    & \leq {2\omega^2f(\omega)} \xi^{(0, 2)}_{ 16e\gamma^{-1}\mathcal{Z}a, 1}(0),
\end{align*}
where we have used $\xi^{m, k}_{\lambda, x}(t)$ defined in Lemma~\ref{lemma:lr_summation}. Using Lemma~\ref{lemma:lr_summation}, we obtain that
\begin{align*}
\sum_{\alpha, \alpha'} \bigabs{\int_0^t \int_0^s e^{-2(s - s')/\omega^2}\tr{}{O(t - s) \mathcal{K}_\alpha \mathcal{J}_{\alpha'}(\sigma_\omega(s'))} ds' ds }  \leq \sum_{\alpha, \alpha'}\big(e^{\leq}_{\alpha, \alpha'}(t, t_{\alpha, \alpha'}) + e^{\geq}_{\alpha, \alpha'}(t, t_{\alpha, \alpha'})\big) &\leq \omega^2 f(\omega)\zeta^{(3)}(\gamma),
\end{align*}
where
\begin{align*}
    \zeta^{(3)}(\gamma) =  \max(e\eta_{S_O},1)^2\left(2T'_\gamma\left(\nu^{(0)}(4a, 4e\mathcal{Z}T'_\gamma)\right)^2 + \frac{1}{4e\mathcal{Z}a}\nu^{(0)}(4a, 4e\mathcal{Z}T'_\gamma)\nu^{(1)}(4a, 4e\mathcal{Z}T'_\gamma)\right)+2\left(\nu^{(0)}(16e\gamma^{-1}\mathcal{Z})\right)^2.
\end{align*}
Note that $\zeta^{(3)}(\gamma)$ also depends on $\tilde{\mathcal{Z}}$ through its dependence on $\nu^{(m)}(\gamma)$. If $k(\abs{S_O}, \gamma) \leq O(\exp(\gamma^{-\kappa}))$ as $\gamma \to 0$, $T_\gamma \leq O(\gamma^{-(\kappa + 1)})$ and using the asymptotics of $\nu^{(m)}(\lambda, T)$ from Lemma~\ref{lemma:lr_summation}, we obtain that
\[
\zeta^{(3)}(\gamma) \leq O\left(\gamma^{-(\kappa+1)}\right)\times  O\left(\gamma^{-2d(\kappa+1)}\right) +  O\left(\gamma^{-d(\kappa+1)}\right)\times O\left(\gamma^{-(d+1)(\kappa+1)}\right) +  O\left(\gamma^{-2d}\right) \leq  O\left(\gamma^{-(2d + 1)(\kappa + 1)}\right),
\]
which establishes the lemma.
\end{proof}
\begin{replemma}{lemma:error_term_rapid_mixing}
    Suppose $O$ is a local observable with $\norm{O} \leq 1$ supported on $S_O$, and for $\tau>0$, let $O(\tau) = \exp(\mathcal{L}^\textnormal{\dagger} \tau)(O)$ where $\mathcal{L}$ is a geometrically local Lindbladian of the form in Eq.~\ref{eq:geom_local_lind}. Furthermore, suppose $O$ is rapidly mixing with respect to $\mathcal{L}$ and satisfies Eq.~\ref{eq:rapid_mixing_observable_v2} with $k(\abs{S_O}, \gamma) \leq O(\exp(\gamma^{-\kappa}))$ as $\gamma\to0$. Then for\ $q_\alpha$ as defined in Lemma~\ref{lemma:remainder},
        \[
        \sum_{\alpha}\bigabs{\int_0^t \tr{}{O(t - s) q_\alpha}e^{-2s/\omega^2}ds} \leq  \omega^2 \lambda^{(1)}(\gamma),
        \]
        where $\lambda^{(1)}(\gamma) \leq O(\gamma^{-d(\kappa + 1)})$ as $\gamma \to 0$; and for $j \in \{1, 2, 3, 4\}$
        \begin{align*}
        &\sum_{\alpha, \alpha'}\bigabs{\int_0^t \int_0^{s/\omega^2} \tr{}{O(t - s)\mathcal{Q}_{\alpha, \alpha'}^{(j)}(s')} e^{-2(s/\omega^2 - s')}ds' ds} \leq \lambda^{(2)}(\gamma),
        \end{align*}
where $\lambda^{(2)}(\gamma) \leq O(\gamma^{-(2d + 1)(\kappa + 1)})$ as $\gamma\to0$ and for $j \in \{3, 4\}$, $\mathcal{Q}_{\alpha, \alpha'}^{(j)}$ is defined in Lemma~\ref{lemma:remainder}, and for $j\in \{1, 2\}$, we define $\mathcal{Q}_{\alpha, \alpha'}^{(j)} = \mathcal{Q}_{\alpha, h_{\alpha'}}^{(j)}$ where $\mathcal{Q}_{\alpha, h}^{(j)}$ is defined in Lemma~\ref{lemma:remainder}.
\end{replemma}
\begin{proof}
    Noting from Lemma~\ref{lemma:remainder} that $q_\alpha = \mathcal{D}_{L_\alpha}(\rho(0))$, and that $\norm{\rho(0)}_1 = 1$, from the application of Lemma~\ref{lemma:local_rapid_mixing_superop_bounds}(a), we obtain that
    \[
    \sum_{\alpha}\bigabs{\int_0^t \text{Tr}(O(t - s) q_\alpha) e^{-2s/\omega^2}ds} \leq \omega^2 \zeta^{(1)}(\gamma),
    \]
    and thus we may choose $\lambda^{(1)}(\gamma)=\zeta^{(1)}(\gamma)$. By Lemma~\ref{lemma:local_rapid_mixing_superop_bounds}, so long as $k(\abs{S_O},\gamma) \leq O\left(\exp(\gamma^{-\kappa})\right)$ as $\gamma \to 0$, then $\zeta^{(1)}(\gamma)\leq O\left(\gamma^{-d(\kappa+1)}\right)$ as $\gamma \to 0$.
    Next, we note that
    \begin{align}\label{eq:change_of_variable}
    \sum_{\alpha, \alpha'}\bigabs{\int_0^t \int_0^{s/\omega^2} \tr{}{O(t - s)\mathcal{Q}_{\alpha, \alpha'}^{(j)}(s')} e^{-2(s/\omega^2 - s')}ds' ds} = \frac{1}{\omega^2}\sum_{\alpha, \alpha'}\bigabs{\int_0^t \int_0^{s} \text{Tr}\bigg({O(t - s)\mathcal{Q}_{\alpha, \alpha'}^{(j)}\bigg(\frac{s'}{\omega^2}\bigg)\bigg)} e^{-2(s - s')/\omega^2}ds' ds}
    \end{align}
    We can now consider $j \in \{1, 2, 3, 4\}$ separately --- for $j = 1$, we have from Lemma~\ref{lemma:remainder} that $\mathcal{Q}_{\alpha, \alpha'}^{(1)}(s'/\omega^2) = \mathcal{K}_{\alpha}\mathcal{J}_{\alpha'}(\sigma_\omega(s')) + \textnormal{h.c.}$ with $\mathcal{K}_\alpha(X) = [L_\alpha^\dagger, X]$, $\mathcal{J}_{\alpha'}(X) = [h_{\alpha'}, X]$ and $\sigma_\omega(s') = -\omega^{-1}\tr{\mathcal{A}}{\sigma_\alpha \rho_\omega(s'/\omega^2)}$. We note that $\norm{\mathcal{K}_\alpha}_{\diamond}, \norm{\mathcal{J}_{\alpha'}}_\diamond \leq 2$ and that, by Lemma~\ref{lemma:tr_sigma_bounds}, $\norm{\sigma_\omega(s')}_1 \leq 1/2$. Thus, from Eq.~\ref{eq:change_of_variable} and Lemma~\ref{lemma:local_rapid_mixing_superop_bounds}(c), we obtain that
    \begin{equation}
    \sum_{\alpha, \alpha'}\bigabs{\int_0^t \int_0^{s/\omega^2} \tr{}{O(t - s)\mathcal{Q}_{\alpha, \alpha'}^{(1)}(s')} e^{-2(s/\omega^2 - s')}ds' ds} \leq \zeta^{(3)}(\gamma).
    \label{eq:Q1_lrm_bound}
    \end{equation}
    Similarly, for $j = 2$, we have from Lemma~\ref{lemma:remainder} that $\mathcal{Q}_{\alpha, \alpha'}^{(2)}(s'/\omega^2) = \mathcal{K}_\alpha \mathcal{J}_{\alpha'}(\sigma_\omega(s')) + \textnormal{h.c.}$ where $\mathcal{K}_\alpha(X) = [L_\alpha^\dagger, L_\alpha X]$, $\mathcal{J}_{\alpha'}(X) = [h_{\alpha'}, X]$ and $\sigma_\omega(s') = -i\text{Tr}_\mathcal{A}(\rho_\omega(s'/\omega^2))/2$. We note that $\norm{\mathcal{K}_\alpha}_{\diamond}, \norm{\mathcal{J}_{\alpha'}}_\diamond \leq 2$ and that $\norm{\sigma_\omega(s')}_1 \leq 1/2$. Thus, from Eq.~\ref{eq:change_of_variable} and Lemma~\ref{lemma:local_rapid_mixing_superop_bounds}, we obtain that
    \begin{equation}
        \sum_{\alpha, \alpha'}\bigabs{\int_0^t \int_0^{s/\omega^2} \tr{}{O(t - s)\mathcal{Q}_{\alpha, \alpha'}^{(2)}(s')} e^{-2(s/\omega^2 - s')}ds' ds} \leq \zeta^{(3)}(\gamma).
        \label{eq:Q2_lrm_bound}
    \end{equation}
    For $j = 3$, we have from Lemma~\ref{lemma:remainder} that $\mathcal{Q}_{\alpha, \alpha'}^{(3)}(s'/\omega^2) = \sum_{u \in \{-, +\}}\mathcal{K}_\alpha^{(u)} \mathcal{J}_{\alpha'}^{(u)}(\sigma_\omega^{(u)}(s'))$ where $\mathcal{K}^{(u)}_\alpha = \mathcal{D}_{L_\alpha}$, $\mathcal{J}^{(u)}_{\alpha'}(X) = [L_{\alpha'}^{(u)}, X]$ and $\sigma_\omega^{(u)}(s') = i\omega^{-1}\text{Tr}_\mathcal{A}(\sigma_{\alpha'}^{(\bar{u})}\rho_\omega(s'/\omega^2))$. We note that $\norm{\mathcal{K}^{(u)}_\alpha}_{\diamond}, \norm{\mathcal{J}^{(u)}_{\alpha'}}_\diamond \leq 2$ and, using Lemma~\ref{lemma:tr_sigma_bounds}, $\norm{\sigma^{(\bar{u})}_\omega(s')}_1 \leq 1/2$. Thus, from Eq.~\ref{eq:change_of_variable} and Lemma~\ref{lemma:local_rapid_mixing_superop_bounds}, we obtain that
    \begin{align}
    \sum_{\alpha, \alpha'}\bigabs{\int_0^t \int_0^{s/\omega^2} \tr{}{O(t - s)\mathcal{Q}_{\alpha, \alpha'}^{(3)}(s')} e^{-2(s/\omega^2 - s')}ds' ds} \leq \zeta^{(3)}(\gamma).
    \label{eq:Q3_lrm_bound}
    \end{align}
    Finally, for $j = 4$, we have from Lemma~\ref{lemma:remainder} that $\mathcal{Q}_{\alpha, \alpha'}^{(4)}(s'/\omega^2) = \sum_{u, u' \in \{-, +\}}\mathcal{K}_\alpha \mathcal{J}_{\alpha'}^{(u,u')}(\sigma_{\omega}^{(u,u')}(s')) + \textnormal{h.c.}$ according to the following definitions: We define $\mathcal{K}_\alpha(X)=[L^\dagger_\alpha,X]$.  $J_{\alpha'}^{(u,u')}(X)=L^{(u)}X$ if $u'=-$  and $J_{\alpha'}^{(u,u')}(X)=XL^{(u)}$ if $u'=+$. Finally, $\sigma_{\omega}^{(u,u')}(s')=\omega^{-2}\tr{\mathcal{A}}{\sigma_{\alpha}^{(\bar{u})}\sigma_{\alpha}\rho_{\omega,\delta}(s'/\omega^2)}$ if $\alpha=\alpha'$ and $\sigma_{\omega}^{(u,u')}(s')=u'\omega^{-2}\tr{\mathcal{A}}{\sigma_{\alpha'}^{(\bar{u})}\sigma_{\alpha}\rho_{\omega,\delta}(s'/\omega^2)}$ if $\alpha\neq\alpha'$. We note that for any $u,u'$, $\norm{\mathcal{K}_\alpha}_\diamond, \norm{\mathcal{J_{\alpha'}^{(u,u')}}}_\diamond\leq 2$ and, using Lemma~\ref{lemma:tr_sigma_bounds}, $\norm{\sigma_\omega^{(u,u')}(s')}_1 \leq 1/4$. Thus, applying Eq.~\ref{eq:change_of_variable} and Lemma~\ref{lemma:local_rapid_mixing_superop_bounds}(c), we obtain that
    \begin{equation}
        \sum_{\alpha, \alpha'}\bigabs{\int_0^t \int_0^{s/\omega^2} \tr{}{O(t - s)\mathcal{Q}_{\alpha, \alpha'}^{(4)}(s')} e^{-2(s/\omega^2 - s')}ds' ds} \leq 2\zeta^{(3)}(\gamma).
        \label{eq:Q4_lrm_bound}
    \end{equation}
    By Lemma~\ref{lemma:local_rapid_mixing_superop_bounds}, so long as $k(\abs{S_O},\gamma) \leq O\left(\exp(\gamma^{-\kappa})\right)$ as $\gamma \to 0$, then $\zeta^{(3)}(\gamma)\leq O\left(\gamma^{-(2d+1)(\kappa+1)}\right)$ as $\gamma \to 0$. Considering Eqs.~\ref{eq:Q1_lrm_bound}-\ref{eq:Q4_lrm_bound}, we may choose $\lambda^{(2)}(\gamma)=2\zeta^{(3)}(\gamma)$ to satisfy the lemma statement.
\end{proof}
\section{Circuit-to-geometrically local Lindbladian encoding}
\label{appendix:circuit_to_geom_local_lindbladian_encoding}
In this appendix, we construct an encoding of a quantum circuit into a 2D geometrically local Lindbladian with a unique fixed point. We prove a useful result about its convergence to fixed point. Our construction follows from combining the dissipative quantum computing gadgets in Ref.~\cite{verstraete2009quantum} with the quantum circuit to Hamiltonian encoding in Ref.~\cite{aharonov2008adiabatic}. We remark that Ref.~\cite{verstraete2009quantum} already suggested, in its supplementary, that it should be possible to perform dissipative quantum computation with just a 2D nearest-neighbour Lindbladian but did not provide an explicit construction. We fill that gap in this appendix.

\subsection{Review of the $O(\log N)$-local}\label{sec:log_N_local}
As a review of the basic construction and proof technique involved in encoding a quantum circuit into the fixed point of an Lindbladian, we first analyze the 5-local construction in Ref.~\cite{verstraete2009quantum}. While Ref.~\cite{verstraete2009quantum} laid out the circuit to Lindbladian encoding and analyzed the eigenvalue spectrum of the resulting Lindbladian, they did not explicitly bound the error between the state obtained by evolution under the Lindbladian and the fixed point. This subsection also serves to fill this gap.
\begin{figure}
\includegraphics[scale=0.5]{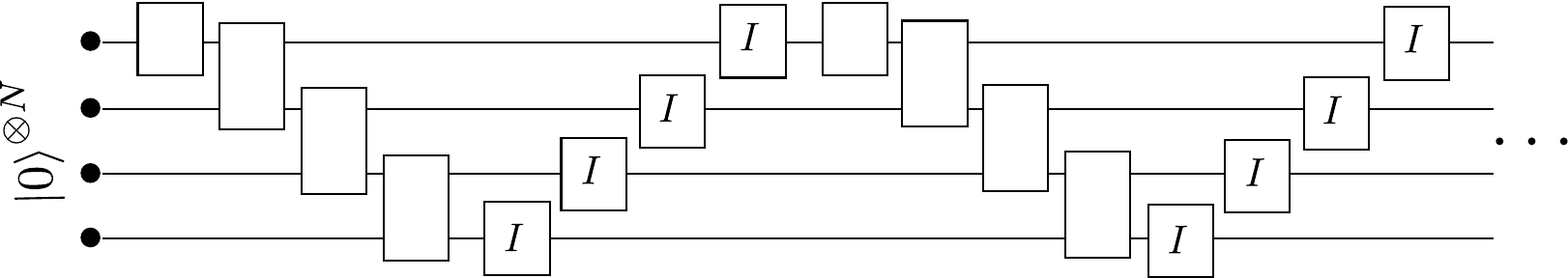}
\caption{Architecture of the circuit being encoded into the Lindbladian assumed in section \ref{sec:log_N_local}. The identity gates are treated as separate time-steps. Any poly$(N)$ depth quantum circuit on $N$ qubits can be brought into this architecture with at-most a $\text{poly}(N)$ overhead in the circuit depth.}
\label{fig:log_N_local}
\end{figure}

We will begin by first reviewing the circuit to Lindbladian encoding that is a slight modification of the one shown in Ref.~\cite{verstraete2009quantum} --- consider a quantum circuit on $N$ qubits and of depth $T = \text{poly}(N)$. In preparation for the analysis of the geometrically local Lindbladian, we will assume the  unitaries $U_1, U_2 \dots U_T$ applied in each time-step of the circuit are in the architecture as shown in Fig.~\ref{fig:log_N_local}. We will construct a master equation on the Hilbert space $(\mathbb{C}^{2})^{\otimes N} \otimes \mathbb{C}^{T + 1}$ of $N$ data qubits and a $(T + 1)-$level system, which will be the clock qudit. We consider a Lindbladian $\mathcal{L}_\text{ref}$ with two sets of jump operators, \emph{first}, $\{S_\alpha\}_{\alpha \in \{1, 2 \dots N\}}$ encoding the initialization of the data qubits to $\ket{0}$ and \emph{second}, $\{L_t\}_{t\in\{1, 2 \dots T\}}$ to encode the quantum gates:
\begin{subequations}\label{eq:logN_lind}
\begin{align}
&\mathcal{L}_\text{ref} = \sum_{\alpha = 1}^N \mathcal{D}_{S_\alpha}  + \sum_{s = 1}^T \mathcal{D}_{L_s}, \\
        &S_\alpha = \ket{0_\alpha}\!\bra{1_\alpha} \otimes \sum_{t = 0}^{\alpha - 1}\ket{t}\!\bra{t} \text{ for }\alpha \in [1:N], \\
        &L_s = U_s\otimes\ket{s}\!\bra{s - 1 } + \text{h.c.} \text{ if } s \in [1:T].
\end{align}
\end{subequations}
As we will demonstrate below, the unique fixed point of this master equation is given by the state
\begin{align}\label{eq:logN_lind_fp}
\sigma = \frac{1}{T + 1}\sum_{s = 0}^T \ket{\phi_s}\!\bra{\phi_s} \otimes \ket{s}\!\bra{s},
\end{align}
where $\ket{\phi_s} = U_s U_{s - 1} \dots U_1 \ket{0}^{\otimes N}$, with $\ket{\phi_0} = \ket{0}^{\otimes N}$ Furthermore, we will also establish an upper bound on $\norm{e^{\mathcal{L}_\text{ref}t} - \text{Tr}(\cdot) \sigma}_{1\to 1}$ to assess the time taken to converge to the fixed point. We remark that our analysis closely follows Ref.~\cite{verstraete2009quantum}, where they calculated the eigenvalues $\mathcal{L}_\text{ref}$. However, since the eigenvector matrix of a Lindbladian superoperator is not necessarily unitary and can be badly conditioned, just the eigenvalues are not enough bound the error $\norm{e^{\mathcal{L}_\text{ref}t} - \text{Tr}(\cdot) \sigma}_{1\to 1}$. Below, we build on the analysis in Ref.~\cite{verstraete2009quantum} and provide a concrete bound on $\norm{e^{\mathcal{L}_\text{ref}t} - \text{Tr}(\cdot) \sigma}_{1\to 1}$.

To proceed with our analysis, we need the following technical lemma.
\begin{lemma}\label{lemma:tridiagonal_matrices_analysis}
    Suppose $M \in \mathbb{R}^{k \times k}$ is the tridiagonal matrix given by
    \[
    M \cong \sum_{i = 0}^{k - 2}\big(\ket{i}\!\bra{i + 1} + \ket{i + 1}\!\bra{i} - \ket{i}\!\bra{i} - \ket{i + 1}\!\bra{i + 1}\big),
    \]
    then
    \[
    \norm{\exp(Mt) - \ket{\xi}\!\bra{\xi}}_\textnormal{op} \leq \exp\bigg(-4t\sin^2\bigg(\frac{\pi}{2k}\bigg)\bigg) = \exp\bigg(-\Omega\bigg(\frac{t}{k^2}\bigg)\bigg) \text{ as } k \to \infty,
    \]
    where $\ket{\xi} =k^{-1/2}\sum_{n = 0}^{k - 1}\ket{n}$. Furthermore, if $d \in \mathbb{R}^k$ is a vector with only non-negative elements and with smallest non-zero element $d_0$ then 
    \[
    \norm{\exp((M - \textnormal{diag}(d))t)}_\textnormal{op} \leq \exp\bigg(-\frac{t}{k}\min\bigg(d_0, 4\sin^2\bigg(\frac{\pi}{2k}\bigg)\bigg)\bigg) = \exp\bigg(-\Omega\bigg(\frac{t}{k^3}\bigg)\bigg) \text{ as } k \to \infty.
    \]
\end{lemma}
\begin{proof}
    We can analytically diagonalize $M$ --- consider the vector $\ket{v_n} \in \mathbb{R}^k$, for $n \in \{0, 1, 2\dots k - 1\}$ given by
    \[
    \ket{v_n} = \bigg(\frac{2}{k}\bigg)^{1/2} \sum_{m = 0}^{k - 1} \cos\bigg(\frac{\pi n}{k}\bigg(m + \frac{1}{2}\bigg)\bigg)\ket{m},
    \]
    $\ket{v_n}$ is an eigenvector of $M$ with eigenvalue $\lambda_n = -4\sin^2(n\pi / 2k)$. Note that $\ket{v_0} = \ket{\xi}$ with $\lambda_0 = 0$. It then follows immediately from the diagonalization of $M$ that
    \[
    \norm{\exp(Mt) - \ket{\xi}\!\bra{\xi}}_\text{op} =\bignorm{\sum_{m = 1}^{k - 1}e^{\lambda_m t}\ket{v_m}\! \bra{v_m}}_\text{op} = \exp({\lambda_1 t}) = \exp\bigg(-4t\sin^2\bigg(\frac{\pi}{2k}\bigg)\bigg) 
    \]
    To bound $\norm{\exp((M - \textnormal{diag}(d))t)}$, we use Ref.~\cite{kitaev2002classical} Lemma 14.4 which is as follows: Suppose $A_1, A_2 \succeq 0$ are non-negative operators with the minimum positive eigenvalue being at least $c_0$. Suppose $\mathcal{N}_1, \mathcal{N}_2$ are the null spaces of $A_1, A_2$ respectively and $\theta = \cos^{-1}\big(\max_{\ket{n_1} \in \mathcal{N}_1} \max_{\ket{n_2} \in \mathcal{N}_2} \abs{\bra{n_1}n_2\rangle} / \norm{n_1} \norm{n_2} \big)$ be the angle between $\mathcal{N}_1, \mathcal{N}_2$, then $\lambda_\text{min}(A_1 + A_2) \geq 2c_0 \sin^2(\theta / 2)$. We now apply this with $A_1 =\textnormal{diag}(d)$ and $A_2 = -M$. Note that this implies that $c_0 = \min(d_0, 4 \sin^2(\pi / 2k))$ and $\cos \theta \leq (k - 1) / k$. Thus, we obtain that 
    \[
    \lambda_\text{max}(M - \textnormal{diag}(d)) = -\lambda_\text{min}(-M + \textnormal{diag}(d)) \leq -\frac{1}{k}\min\bigg(d_0,  4\sin^2\bigg(\frac{\pi}{2k}\bigg)\bigg).
    \]
    Now, using the fact that $\norm{\exp((M - \textnormal{diag}(d))t)} \leq \exp(\lambda_\text{max}(M - \textnormal{diag}(d)) t)$, the lemma statement follows.
\end{proof}
\begin{lemma}[Expanding on Ref.~\cite{verstraete2009quantum}] \label{prop:logN_conv_bound} The state $\sigma$ described in Eq.~\ref{eq:logN_lind_fp} is the unique fixed point of the lindbladian in Eq.~\ref{eq:logN_lind}, and for sufficiently large $T, N$,
\[
\norm{e^{\mathcal{L}t} - \textnormal{Tr}(\cdot)\sigma}_{1 \to 1} \leq c_0(T, N) \exp(-a_0(T) t),
\]
where
\begin{align*}
    c_0(T, N) = 64(T + 1)^{1/2}N^4 2^{9N/2} \text{ and }a_0(T) = \frac{4}{T + 1}\sin^2\bigg(\frac{\pi}{2(T + 1)}\bigg)
\end{align*}
\end{lemma}
\begin{proof}
In this proof, we will work with the vectorized notation introduced in section \ref{sec:notation}. In both vectorized and unvectorized states or operators, we will consistently write the operator acting on the space of the data qubits to the left of the tensor product $\otimes$, and the operator acting on the clock qudit on the right of the tensor product $\otimes$. Following Ref.~\cite{verstraete2009quantum}, as a first step, we introduce the unitary
\[
V = \sum_{s = 0}^{T} U_s U_{s - 1} \dots U_0 \otimes \ket{s}\! \bra{s},
\]
with $U_0 = I$. Note that $\forall \alpha \in [1:N]:$
\begin{align}
V^\dagger S_\alpha V &= \sum_{t = 0}^{\alpha - 1} V^\dagger \big(\ket{0_\alpha}\!\bra{1_\alpha} \otimes \ket{t}\!\bra{t}\big) V, \nonumber\\
&=\sum_{t = 0}^{\alpha - 1} U_0^\dagger U_1^\dagger \dots U_t^\dagger (\ket{0_\alpha}\!\bra{1_\alpha})U_t U_{t - 1}\dots U_0 \otimes \ket{t}\!\bra{t}, \nonumber\\
&\numeq{1}\sum_{t = 0}^{\alpha - 1}\ket{0_\alpha}\!\bra{1_\alpha}\otimes \ket{t}\!\bra{t},
\end{align}
where in (1), we have used the fact that the unitaries $ U_1 \dots U_{\alpha - 1}$ act only on qubits $1, 2\dots \alpha - 1$ due to the assumption that the circuit is of the form shown in Fig. . Furthermore, $\forall s \in [1:T]:$
\begin{align}
\tilde{L}_s &= V^\dagger L_s V = V^\dagger \big(U_s \otimes \ket{s}\!\bra{s - 1}\big)V + \text{h.c.} = \ket{s}\!\bra{s - 1} + \text{h.c.}.
\end{align}
Therefore, we obtain that
\begin{align}\label{eq:unitary_trans_lind}
\tilde{\mathcal{L}}_\text{ref} := V_l^\dagger V_r^T \mathcal{L}_\text{ref} V_l V_r^* = \sum_{\alpha = 1}^N \mathcal{D}_{S_\alpha} + \sum_{s = 1}^T \mathcal{D}_{\tilde{L}_s}.
\end{align}
Next, we diagonalize the Lindbladian in the space of the data qubits. For this, we consider the following (non-unitary) transformation 
\[
X = \sum_{j = 1}^{N + 1 } X_j X_{j + 1} \dots X_N \otimes \Pi_j,
\]
where for $j = N + 1$, $X_j X_{j + 1} \dots X_N := I$ and
\begin{align*}
& X_i = P_{i, l}^{0} P_{i, r}^{0} + \sigma_{i, l}\sigma_{i, r} - P_{i, l}^{1} P_{i, r}^{1} + P_{i, l}^0 P_{i, r}^1 + P_{i, l}^1 P_{i, r}^0 , \text{ and }\\
&\Pi_1 = \sum_{s = 0}^T \ket{s_l s_r}\!\bra{s_l s_r}, \Pi_j = \sum_{s = 0}^{j - 1}\big(\ket{j_l s_r}\!\bra{j_l s_r} + \ket{ s_l j_r}\!\bra{s_l j_r}\big) \text{ and }\Pi_{N + 1} = I - \sum_{j = 1}^N \Pi_j,
\end{align*}
with $P^\alpha_{i, \lambda} = \ket{\alpha_{i, \lambda}}\!\bra{\alpha_{i, \lambda}}$, for $\alpha \in \{0, 1\}, \lambda \in \{l, r\}$. It can be checked that $X_i$ satisfies $X_i \mathcal{D}_{\sigma_i} X_i = -\big(P_{i, r}^1  + P_{i, l}^1\big)/2$ and that $X_i = X_i^{-1}$ (and therefore $X = X^{-1}$). Next, we compute $X\mathcal{D}_{S_\alpha}X$ and $X\mathcal{D}_{\tilde{L}_s}X$ --- note that $\mathcal{D}_{S_\alpha}$ is explicitly given by
    \[
    \mathcal{D}_{S_\alpha} = \mathcal{D}_{\sigma_\alpha} \otimes \sum_{s, s'=0}^{\alpha - 1}\ket{s_l s'_r}\!\bra{s_l s'_r} - \frac{1}{2}\bigg(P_{\alpha, l}^1 \otimes \sum_{s = 0}^{\alpha - 1}\sum_{s' = \alpha}^T\ket{s_l s'_r}\!\bra{s_l s'_r} + P_{\alpha, r}^1 \otimes \sum_{s = 0}^{\alpha - 1}\sum_{s' = \alpha}^T\ket{s_l' s_r}\!\bra{s_l' s_r}\bigg)
    \]
    To compute $X \mathcal{D}_{S_\alpha} X$, we observe that for $\alpha \in [1:N]$
    \begin{subequations}\label{eq:projector_property_wrt_diagonal}
        \begin{align}
    &\Pi_j \bigg(\sum_{s, s' = 0}^{\alpha - 1}\ket{s_l s_r'}\!\bra{s_l s_r'}\bigg)\Pi_{j'} = 
        0 \text{ if } j\neq j' \text{ or } j = j' > \alpha, 
\\
    &\Pi_j \bigg(\sum_{s = 0}^{\alpha - 1} \sum_{s'=\alpha}^T \ket{s_l s_r'}\!\bra{s_l s_r'}\bigg) \Pi_{j'} = \Pi_j \bigg(\sum_{s = 0}^{\alpha - 1} \sum_{s'=\alpha}^T \ket{s'_l s_r}\!\bra{s'_l s_r}\bigg) \Pi_{j'} = 0 \text{ if }j \neq j' \text{ or } j = j' \leq \alpha
    \end{align}
    \end{subequations}
    Therefore, we obtain that, for all $\alpha \in [1:N]$
    \begin{align*}
    X \mathcal{D}_{S_\alpha}X^{-1} &= X  \bigg(\mathcal{D}_{\sigma_\alpha} \otimes \sum_{s, s'=0}^{\alpha - 1}\ket{s_l s'_r}\!\bra{s_l s'_r}\bigg) X  - \frac{1}{2}X\bigg(P_{\alpha, l}^1 \otimes \sum_{s = 0}^{\alpha - 1}\sum_{s' = \alpha}^T\ket{s_l s'_r}\!\bra{s_l s'_r}\bigg)X - \frac{1}{2} X\bigg(P_{\alpha, r}^1 \otimes \sum_{s = 0}^{\alpha - 1}\sum_{s' = \alpha}^T\ket{s_l' s_r}\!\bra{s_l' s_r}\bigg)X, \nonumber \\
    &=X_\alpha \mathcal{D}_{\sigma_\alpha}X_\alpha \otimes \sum_{s, s' = 0}^{\alpha - 1}\ket{s_l s_r'}\!\bra{s_l s_r'} - \frac{1}{2}\bigg(P_{\alpha, l}^{1} \otimes \sum_{s = 0}^{\alpha - 1}\sum_{s' = \alpha}^T \ket{s_l s_r'}\!\bra{s_l s_r'} + P_{\alpha, r}^{1}\otimes \sum_{s = 0}^{\alpha - 1}\sum_{s' = \alpha}^{T} \ket{s_l' s_r}\!\bra{s_l' s_r}\bigg), \\
    &=-\frac{1}{2}\big(P_{\alpha, r}^{(1)} + P_{\alpha, l}^{(1)}\big) \otimes \sum_{s, s' = 0}^{\alpha - 1}\ket{s_l s'_r}\!\bra{s_l s'_r} - \frac{1}{2}\bigg(P_{\alpha, l}^{1} \otimes \sum_{s = 0}^{\alpha - 1}\sum_{s' = \alpha}^T \ket{s_l s_r'}\!\bra{s_l s_r'} + P_{\alpha, r}^{1}\otimes \sum_{s = 0}^{\alpha - 1}\sum_{s' = \alpha}^{T} \ket{s_l' s_r}\!\bra{s_l' s_r}\bigg), \nonumber \\
    &=-\frac{1}{2}\bigg(P_{\alpha, l}^{(1)}\otimes \sum_{s = 0}^{\alpha - 1} \ket{s_l}\!\bra{s_l} + P_{\alpha, r}^{(1)} \otimes\sum_{s = 0}^{\alpha - 1} \ket{s_r}\!\bra{s_r}\bigg). 
    \end{align*}
    Furthermore, it can also be easily seen that for any $s \in [1:T]$ and $j \in [1:N + 1]$, $[\Pi_j, \mathcal{D}_{\tilde{L}_s]}] = 0$, and $\mathcal{D}_{\tilde{L}_s}$ acts as identity on the space of the data qubits. Therefore, we obtain that $X \mathcal{D}_{\tilde{L}_s} X^{-1} = \mathcal{D}_{\tilde{L}_s}$.

    Thus, by performing the transformation given by $X$, the vectorized Lindbladian is now diagonal on the space of the data qubits. In particular, we can express
    \[
    X\tilde{\mathcal{L}}_\text{ref}X^{-1} = \sum_{x, y \in \{0, 1\}^N} \ket{x_l y_r}\!\bra{x_l y_r} \otimes \mathcal{A}_{x, y},
    \]
    where for $x, y \in \{0, 1\}^N$, $\mathcal{A}_{x, y} \in \text{L}(\mathbb{C}^{T + 1} \otimes \mathbb{C}^{T + 1})$ given by
    \[
    \mathcal{A}_{x, y} = -\frac{1}{2}\sum_{s = 0}^{T}f_s(x) \ket{s_l}\!\bra{s_l} - \frac{1}{2}\sum_{s = 0}^{T}f_s(y) \ket{s_r}\!\bra{s_r} + \sum_{s = 1}^T \mathcal{D}_{\ket{s - 1}\!\bra{s} + \text{h.c.}},
    \]
    where $f_s(x) =  x_{s + 1} + x_{s + 2} \dots x_N$ if $0\leq s \leq N - 1$ and $f_s(x) = 0$ if $s \geq N $. Furthermore, $\mathcal{A}_{x, y}$ is block diagonal --- to see this, we perform the following decomposition
    \[
    \mathbb{C}^{T + 1}\otimes \mathbb{C}^{T + 1} = \mathcal{S}^{(0)}\cup \bigg(\bigcup_{t \in [1:T]} \mathcal{S}^{(1)}_t\bigg) \cup \bigg(\bigcup_{\substack{t, s \in [0:T] \\ \abs{t - s}\geq 2}} S^{(2)}_{t, s}\bigg),
    \]
    where
    \begin{align*}
    &\mathcal{S}^{(0)} = \text{span}(\{\ket{t, t} \text{ for } t\in [0:T]\}) \cong \mathbb{C}^{T + 1}, \\
    &\mathcal{S}^{(1)}_t = \text{span}\big(\{\ket{t, t - 1}, \ket{t - 1, t} \}\big) \cong \mathbb{C}^2 \text{ for }t \in [1:T],\\
    &\mathcal{S}^{(2)}_{t, s} = \text{span}\big(\{\ket{t, s}\}\big) \cong \mathbb{C}^1 \text{ for }t, s \in [0:T] \text{ with }\abs{t - s} \geq 2.
    \end{align*}
    Then, we can note that $\forall x, y \in \{0, 1\}^N$, $\mathcal{A}_{x, y}$ is block diagonal on the subspaces $\mathcal{S}^{(0)}, \mathcal{S}_t^{(1)}, \mathcal{S}_{t, s}^{(2)}$. Furthermore,
    \begin{subequations}
    \begin{align}
    &\mathcal{A}_{x, y}\big|_{\mathcal{S}^{(0)}} \cong -\frac{1}{2}\sum_{s = 0}^{N - 1}\big( f_s(x) + f_s(y)\big) \ket{s}\!\bra{s} + \sum_{s= 1}^T \bigg(\ket{s - 1}\!\bra{s} + \ket{s - 1}\!\bra{s} - \ket{s}\!\bra{s} - \ket{s - 1}\!\bra{s - 1}\bigg), \\
    &\mathcal{A}_{x, y}\big|_{\mathcal{S}^{(1)}_t} \cong -\bigg(2 - \frac{1}{2}(\delta_{t, 1} + \delta_{t, T})\bigg)I + \ket{1}\!\bra{0} + \ket{0}\!\bra{1} - \frac{1}{2}\bigg(\big(f_t(x) + f_{t-1}(x)\big)\ket{0}\!\bra{0} +\big(f_t(y) + f_{t-1}(y)\big) \ket{1}\!\bra{1}\bigg), \\
    &\mathcal{A}_{x, y}\big|_{\mathcal{S}^{(2)}_{t, s}} \cong -\frac{1}{2}\big(4 - \delta_{t, 0} - \delta_{t, T} - \delta_{s, 0} - \delta_{s, T}\big) - \frac{1}{2}\big(f_t(x) + f_s(y)\big).
     \end{align}
    \end{subequations}
We now immediately obtain the eigenvalues of $\mathcal{L}_\text{ref}$, which are just the eigenvalues of the block diagonal operator $\mathcal{A}_{x, y}$. It can be noted that $\mathcal{A}_{x, y}|_{\mathcal{S}_t^{(1)}} \preceq -I/2$ and $\mathcal{A}_{x, y}|_{\mathcal{S}_{t, s}^{(2)}}\leq -1 $ for any $x, y, s, t$. Finally, it also follows from Lemma~\ref{lemma:tridiagonal_matrices_analysis} that only $\mathcal{A}_{x = 0, y = 0}|_{\mathcal{S}^{(0)}}$ has a zero eigenvalue, with the eigenvector $\vecket{\xi} =(T + 1)^{-1/2} (\ket{0_l0_r} + \ket{1_l 1_r} +\dots \ket{T_l T_r})$. This allows us to obtain the fixed point of $\mathcal{L}_\text{ref}$ (where an additional factor of $(T + 1)^{-1/2}$ is introduced to make the fixed point a normalized density matrix),
\begin{align*}
\vecket{\sigma} &= \frac{1}{\sqrt{T + 1}}VX\ket{0_l 0_r}^{\otimes N}\otimes \vecket{\xi} =\frac{1}{\sqrt{T + 1}} V_l V_r^*\ket{0_l 0_r}^{\otimes N} \otimes \vecket{\xi},
\end{align*}
from which it follows that
\begin{align*}
 \sigma &= \frac{1}{T + 1}V \bigg((\ket{0}\!\bra{0})^{\otimes N} \otimes  \sum_{t = 0}^T\ket{t}\!\bra{t}\bigg)V^\dagger =\frac{1}{T + 1}\sum_{t = 0}^T \ket{\phi_t}\!\bra{\phi_t} \otimes \ket{t}\!\bra{t}.
\end{align*}
Next, we bound $\norm{e^{\mathcal{L}_\text{ref}t} - \textnormal{Tr}(\cdot) \sigma}_{1 \to 1}$. Out of the two transformations used to diagonalize $\mathcal{L}_\text{ref}$ in the Hilbert space of the data qubits, the unitary $V$ does not effect this norm i.e.~
\begin{align}\label{eq:error_does_not_change_under_unitary}
    \norm{e^{\mathcal{L}_\text{ref}t} - \textnormal{Tr}(\cdot) \sigma}_{1 \to 1} = \norm{e^{\tilde{\mathcal{L}}t} - \textnormal{Tr}(\cdot) \tilde{\sigma}}_{1 \to 1},
\end{align}
where $\tilde{L}$ is defined in Eq.~\ref{eq:unitary_trans_lind} and $\tilde{\sigma} = (T + 1)^{-1}(\ket{0}\!\bra{0})^{\otimes N} \otimes \sum_{t= 0}^T\ket{t}\!\bra{t} $. Furthermore, we can note that, in the vectorized picture, $\textnormal{Tr}(\cdot)\tilde{\sigma} = \vecket{\tilde{\sigma}}\vecbra{I}$, where $\vecket{\tilde{\sigma}} = (T + 1)^{-1}\ket{0_l 0_r}^{\otimes N} \otimes \sum_{t = 0}^T \ket{t_l t_r} = (T + 1)^{-1/2}\ket{0_l 0_r}^{\otimes N} \otimes \vecket{\xi}$. Also note that $X\vecket{\tilde{\sigma}} = \vecket{\tilde{\sigma}}$ and $\vecbra{I}X = \bra{0_l 0_r}^{\otimes N}\otimes \sum_{t = 0}^T \bra{t_l t_r} = (T + 1)^{1/2} \bra{0_l 0_r}^{\otimes N} \otimes \vecbra{\xi}$.

\begin{align}\label{eq:error_in_terms_of_X}
    \norm{e^{\tilde{\mathcal{L}}_\text{ref}t} - \text{Tr}(\cdot)\tilde{\sigma}}_{1\to 1} &= \bignorm{X\bigg[\sum_{x, y \in \{0, 1\}^N} \ket{x_l y_r}\!\bra{x_l y_r} \otimes e^{\mathcal{A}_{x, y}t} \bigg]X - \vecket{\tilde{\sigma}}\vecbra{I}}_{1\to 1}, \nonumber\\
    &=\bignorm{X\bigg[\sum_{x, y \in \{0, 1\}^N} \ket{x_l y_r}\!\bra{x_l y_r}\otimes e^{\mathcal{A}_{x, y}t}  - \big(\ket{0_l 0_r}\!\bra{0_l 0_r}\big)^{\otimes N} \otimes \vecket{\xi}\!\vecbra{\xi}\bigg]X}_{1 \to 1}, \nonumber \\
    &\leq \norm{X}_{1\to 1}^2 \bignorm{\sum_{x, y \in \{0, 1\}^N} \ket{x_l y_r}\!\bra{x_l y_r}\otimes e^{\mathcal{A}_{x, y}t}  - \big(\ket{0_l 0_r}\!\bra{0_l 0_r}\big)^{\otimes N} \otimes \vecket{\xi}\!\vecbra{\xi}}_{1\to 1}, \nonumber \\
    &\leq (2^N (T + 1))^{1/2}\norm{X}_{1\to 1}^2 \bignorm{\sum_{x, y \in \{0, 1\}^N} \ket{x_l y_r}\!\bra{x_l y_r}\otimes e^{\mathcal{A}_{x, y}t}  - \big(\ket{0_l 0_r}\!\bra{0_l 0_r}\big)^{\otimes N} \otimes \vecket{\xi}\!\vecbra{\xi}}_{2\to 2}, \nonumber \\
    &\leq (2^N (T + 1))^{1/2} \norm{X}_{1 \to 1}^2 \sum_{x, y \in \{0, 1\}^N}\bignorm{\ket{x_l y_r}\!\bra{x_l y_r}\otimes \mathcal{Q}_{x, y}(t)}_{2\to 2}, \nonumber\\
    & = (2^N (T + 1))^{1/2} \norm{X}_{1 \to 1}^2 \sum_{x, y \in \{0, 1\}^N}\bignorm{ \mathcal{Q}_{x, y}(t)}_{2\to 2}
\end{align}
where 
\begin{align*}
    \mathcal{Q}_{x, y}(t) = \begin{cases}e^{\mathcal{A}_{0^N, 0^N}t} - \vecket{\xi}\vecbra{\xi} & \text{ if }x = 0^N \text{ and }y = 0^N, \\
    e^{\mathcal{A}_{x, y}t} & \text{ otherwise}.
    \end{cases}
\end{align*}
Note that, for any superoperator $\mathcal{E}$, $\norm{\mathcal{E}}_{2\to 2}$ coincides with its operator norm when interpreted as an operator in the vectorized picture. This fact together with Lemma~\ref{lemma:tridiagonal_matrices_analysis} allows us, for sufficiently large $T$, to provide the upper bound
\[
\norm{\mathcal{Q}_{x, y}}_{2 \to 2} \leq \exp\bigg(-\frac{4t}{T + 1}\sin^2\bigg(\frac{\pi}{2(T + 1)}\bigg)\bigg) \ \forall x, y \in \{0, 1\}^N.
\]
Therefore, we obtain that
\[
\norm{e^{\tilde{\mathcal{L}}_\text{ref}t} - \textnormal{Tr}(\cdot)\tilde{\sigma}}_{1\to 1} \leq (T  +1)^{1/2}2^{5N/2}\norm{X}_{1\to 1}^2 \exp\bigg(-\frac{4t}{T + 1}\sin^2\bigg(\frac{\pi}{2(T + 1)}\bigg)\bigg)
\]
It remains to bound $\norm{X}_{1\to 1}$ (recall that $X$ is interpreted as a superoperator) --- we start with the simple upper bound
\begin{align}\label{eq:upper_bound_X}
\norm{X}_{1\to 1} \leq \norm{I \otimes \Pi_{N + 1}}_\diamond + \sum_{j = 1}^N \bigg(\prod_{i = j}^{N} \norm{X_i}_\diamond\bigg) \norm{I \otimes \Pi_j}_{1 \to 1} \leq 1 + \sum_{j = 1}^N \bigg(1 + \prod_{i = j}^{N} \norm{X_i}_\diamond\bigg) \norm{I \otimes \Pi_j}_{1 \to 1}
\end{align}
The diamond norm $\norm{X_i}_{\diamond}$ can be computed explicitly. For this, we note that, as a single qubit superoperator, $X_i(\Omega) = \Omega + \bra{0_i}\Omega\ket{0_i} \otimes (\ket{0_i}\!\bra{0_i} - 2\ket{1_i}\!\bra{1_i})$. Therefore, for $\Omega$ on a possibly larger space, $\norm{X_i(\Omega)}_1 \leq \norm{\Omega}_1 + \norm{\bra{0_i}\Omega\ket{0_i} \otimes (\ket{0_i}\!\bra{0_i} - 2\ket{1_i}\!\bra{1_i})}_1 \leq \norm{\Omega}_1 + 3 \norm{\bra{0_i}\Omega\ket{0_i}}_1 \leq 4 \norm{\Omega}_1$, which yields $\norm{X}_\diamond \leq 4$. Consequently, we obtain from Eq.~\ref{eq:upper_bound_X} that
\begin{align}\label{eq:upper_bound_X_new}
\norm{X}_{1\to 1}  \leq 1 + \sum_{j = 1}^N \big(1 + 4^{N - j + 1}\big) \norm{ \Pi_j}_{\diamond}.
\end{align}
Next, we provide upper bounds on $\norm{\text{I}\otimes \Pi_j}$. Note first that $\Pi_1$, as a superoperator on $\mathbb{C}^{T + 1}$, is a channel since it admits the Kraus representation
\[
\Pi_1(X) = \sum_{s = 0}^{T} \ket{s}\!\bra{s} X \ket{s}\!\bra{s}.
\]
Therefore, we obtain that $\norm{\Pi_1}_{\diamond}\leq 1$. Next, we note that $\Pi_j$ for $j \in [2:N]$ can be expressed as
\[
\Pi_j(X) = \sum_{s = 1}^{j - 1}\bigg(\ket{j}\!\bra{j} X \ket{s}\!\bra{s} + \ket{s}\!\bra{s} X \ket{j}\!\bra{j}\bigg),
\]
and consequently we obtain that
\[
\norm{(\Pi_j \otimes I)(X)}_1 \leq \sum_{s = 1}^{j - 1}\big( \norm{\ket{j}\!\bra{j} X \ket{s}\!\bra{s}}_1 + \norm{\ket{s}\!\bra{s} X \ket{j}\!\bra{j}}_1\big) \leq 2(j - 1) \implies \norm{\Pi_j}_{\diamond} \leq 2(j - 1).
\]
Using these bounds for $\norm{\Pi_j}_\diamond$ together with Eq.~\ref{eq:upper_bound_X_new}, we then obtain that
\begin{align}\label{eq:bound_on_X_final}
\norm{X}_{1\to 1} \leq 1 + \sum_{j = 1}^N \big(1 + 4^{N - j + 1}\big) 2(j - 1) \leq 1 + 2N^2 ( 1 + 4^N) \leq 8N^2 4^N.
\end{align}
Combining Eqs.~\ref{eq:error_does_not_change_under_unitary}, \ref{eq:error_in_terms_of_X} and \ref{eq:bound_on_X_final}, we obtain the lemma statement.
\end{proof}
\noindent From Lemma~\ref{prop:logN_conv_bound}, it follows that starting from any initial state $\rho(0)$, $\rho(t) = e^{\mathcal{L}_\text{ref}t}\rho(0)$ is an $\varepsilon$ approximation to the fixed point $\sigma$ given in Eq.~\ref{eq:logN_lind_fp} after time $t = \Theta(T^3\log(\varepsilon^{-1})) + \Theta(T^3 \log T) + \Theta(T^3 \log N) + \Theta(T^3 N)$. Then, the fixed point can be used to effectively perform the quantum computation by first measuring the clock qudit on the computational basis and post-selecting on the measurement outcome being $\ket{T}$ --- if successful, the data qubits will then be in the state $\ket{\phi_T}$ that would have been obtained at the output of the encoded quantum circuit. Since the probability of the clock qudit being in $\ket{T}$ is $1/(T + 1)$, this can be repeated $O(T) \leq O(\text{poly}(N))$ to have successfully prepare $\ket{\phi_T}$ on the data qubits.

\subsection{2D local Lindbladian}
\emph{Construction}. Consider a quantum circuit on $N$ qubits, with depth $\textnormal{poly}(N)$. As described in Ref.~\cite{aharonov2008adiabatic}, we use the fact that any quantum circuit (Fig.~\ref{fig:encoding_circuit_format}a) can be brought into a specific topology shown in Fig.~\ref{fig:encoding_circuit_format}b, where a quantum gate first acts on qubits 1, 2, then on 2, 3, then on 3, 4 and so on,  using SWAP gates and with at most a polynomial overhead in the depth. We assume that the quantum circuit, in expressed in this format, has depth $NR$, where $R$ is $\text{poly}(N)$, where each round is composed of 1 single qubit and $N - 1$ two qubit gates. We will encode this circuit into a 2D dissipative system with only spatially local jump operators --- we consider $NR$ qudits with $d = 6$ arranged in a 2D grid with $N$ rows and $R$ columns.
\begin{figure}
\centering
    \includegraphics[scale=0.475]{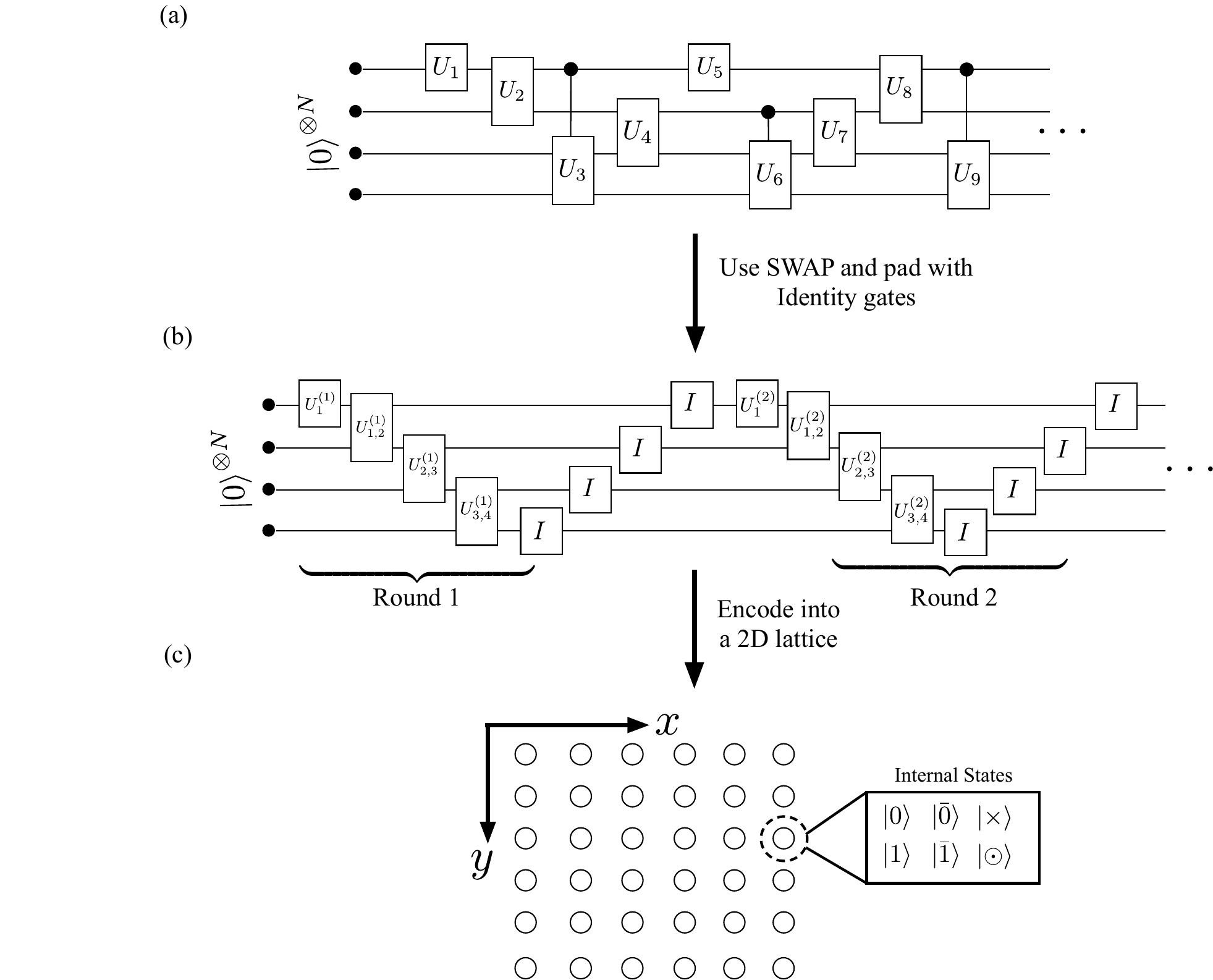}
    \caption{Schematic depiction of the circuit to geometrically local Lindbladian encoding. Given an arbitrary quantum circuit on $N$ qubits and depth $\text{poly}(N)$ (a), it can always be mapped to a quantum circuit on $N$ qubits with architecture shown in (b) where the quantum gates are applied on the qubits sequentially and in several rounds. The total number of rounds needed for a circuit that was originally of depth $\text{poly}(N)$ is at most $\text{poly}(N)$. Finally, this circuit is then encoded into a geometrically local Lindbladian on 2D grid of 6-level qudits as shown in (c). }
    \label{fig:encoding_circuit_format}
\end{figure}
The qudit internal states will be denoted by $\{\ket{0}, \ket{1}, \ket{\bar{0}}, \ket{\bar{1}}, \ket{\times}, \ket{\odot}\}$. We make the following remarks about notation used in this subsection:
\begin{enumerate}
\item[(a)] For $z \in \{0, 1, \bar{0}, \bar{1}, \times, \odot\}$, $\ket{z_{x, y}}$ will mean that the qudit at the $x^\text{th}$ column and $y^\text{th}$ row is in the state $z$.
\item[(b)]We will also define the operators $\pi = \ket{0}\!\bra{0} + \ket{1}\!\bra{1}$ and $\bar{\pi} = \ket{\bar{0}}\!\bra{\bar{0}} + \ket{\bar{1}}\!\bra{\bar{1}}$, and $\pi_{x, y}$ or $\bar{\pi}_{x, y}$ denote the operators $\pi$ or $\bar{\pi}$ respectively when applied to the qudit in the $x^\text{th}$ row and $y^\text{th}$ column. 
\item[(c)] We will refer to the qudit being in $\ket{\times}, \ket{\odot}$, barred $0/1$ and unbarred $0/1$ as being of distinct types. Thus, an operation that changes the `type' of a qudit could be an operation that converts a qudit from $\times$ to $\odot$ or from barred $0/1$ to $\times$ etc. Likewise, an operation that simply rotates a qudit within the space of barred (or unbarred) 0/1 states will not change its type. 
\item[(d)]We will often use operators defined on qudits that are horizontally adjacent, vertically adjacent or arranged on a plaquette. It will be convenient and visually clearer to express tensor products of operators as a matrix of operators associated with the coordinate of the lower leftmost vertex. For example,
\begin{align*}
& \begin{bmatrix}
 D & C \\
 A & B
 \end{bmatrix}_{x, y} := A_{x, y - 1} \otimes B_{x, y} \otimes C_{x - 1, y} \otimes D_{x - 1, y - 1} \text{  } 
  \begin{tikzpicture}[baseline={([yshift=-.5ex]current bounding box.center)},vertex/.style={anchor=base,
    circle,fill=black!25,minimum size=18pt,inner sep=2pt}]
    \node[vertex] (G1) at (0,0)   {$A$};
    \node[vertex] (G2) at (1,0)   {$B$};
    \node[vertex] (G3) at (1,1)   {$C$};
    \node[vertex] (G4) at (0,1)   {$D$};
    \node[] at (1.75, -0.25){$(x, y)$};
    \node[] at (-1.0, -0.25){$(x - 1, y )$};
    \node[] at (-1.25, 1.5){$(x - 1, y - 1)$};
    \node[] at (2.0, 1.5){$(x , y - 1)$};
    \draw (G1) -- (G2) -- (G3) -- (G4) -- (G1) -- cycle;
  \end{tikzpicture}, \nonumber \\
  & \begin{bmatrix}
 A & B
 \end{bmatrix}_{x, y} := A_{x, y - 1} \otimes B_{x, y} \text{  } \qquad \qquad \qquad   \qquad  \qquad \ 
  \begin{tikzpicture}[baseline={([yshift=-.5ex]current bounding box.center)},vertex/.style={anchor=base,
    circle,fill=black!25,minimum size=18pt,inner sep=2pt}]
    \node[vertex] (G1) at (0,0)   {$A$};
    \node[vertex] (G2) at (1,0)   {$B$};
    \node[] at (1.75, -0.25){$(x, y)$};
    \node[] at (-1.0, -0.25){$(x - 1, y)$};
    \draw (G1) -- (G2);
  \end{tikzpicture}, \nonumber \\
    & \begin{bmatrix}
 A \\
 B
 \end{bmatrix}_{x, y} := A_{x - 1, y} \otimes B_{x, y} \text{  } \qquad \qquad \qquad   \qquad  \qquad \quad \  \
  \begin{tikzpicture}[baseline={([yshift=-.5ex]current bounding box.center)},vertex/.style={anchor=base,
    circle,fill=black!25,minimum size=18pt,inner sep=2pt}]
    \node[vertex] (G1) at (0, 1)   {$A$};
    \node[vertex] (G2) at (0, 0)   {$B$};
    \node[] at (-1, 1.5){$(x, y - 1)$};
    \node[] at (-0.75, -0.25){$(x, y )$};
    \draw (G1) -- (G2);
  \end{tikzpicture}, \nonumber \\
      & \begin{bmatrix}
 A 
 \end{bmatrix}_{x, y} := A_{x, y} \text{  } \qquad \qquad \qquad   \qquad  \qquad \qquad \qquad \qquad \  \ \
  \begin{tikzpicture}[baseline={([yshift=-.5ex]current bounding box.center)},vertex/.style={anchor=base,
    circle,fill=black!25,minimum size=18pt,inner sep=2pt}]
    \node[vertex] (G2) at (0, 0)   {$A$};
    \node[] at (-0.75, -0.25){$(x, y )$};
  \end{tikzpicture}.
\end{align*}
Often, for convenience, we will omit elements in the matrix --- the omitted elements should be interpreted as identity operators on the corresponding qudits. For example, 
\begin{align*}
\begin{bmatrix}
 & A & \\
 B & C & D \\
 & E & 
\end{bmatrix}_{x, y} = A_{x - 1, y - 2} \otimes B_{x - 2, y - 1} \otimes C_{x - 1, y - 1} \otimes D_{x, y-1} \otimes E_{x-1, y} \text{ or }\begin{tikzpicture}[baseline={([yshift=-.5ex]current bounding box.center)},vertex/.style={anchor=base,
    circle,fill=black!25,minimum size=18pt,inner sep=2pt}]
    \draw[dashed] (0, -1) -- (0, 0) -- (0, 1);
    \draw[dashed] (0, 1.1) -- (1, 1.1) -- (2, 1.1);
    \draw[dashed] (0, -1.0) -- (1, -1.0) -- (2, -1.0);
    \draw[dashed] (0, 1.1) -- (1, 1.1) -- (2, 1.1);
    \draw[dashed] (2, 1.1) -- (2, -1.0);
    \node[vertex] (G1) at (0,0)   {$B$};
    \node[vertex] (G2) at (1,0)   {$C$};
    \node[vertex] (G3) at (1,1)   {$A$};
    \node[vertex] (G4) at (2,0)   {$D$};
    \node[vertex] (G5) at (1,-1)   {$E$};
    \node[] at (2.1, -1.25){$(x, y)$};
    \node[] at (-0.1, -1.25){$(x - 2, y)$};
    \node[] at (-0.3, 1.4){$(x - 2, y - 2)$};
    \node[] at (2.1, 1.4){$(x, y-2)$};
    \draw (G3) -- (G2) -- (G5);
    \draw (G1) -- (G2) -- (G4);
  \end{tikzpicture},
\end{align*}
\end{enumerate}

First, we map the quantum circuit on $N$ qubits into a larger quantum circuit on the $NR$ qudits with the property that we only ever apply a gate between two (horizontally or vertically) neighbouring qudits and every neighbouring pair of qudits experience a gate only once (note that any one qudit might experience multiple two-qudit gates, but they will always be with a different second qudit). Similar to the situation with Hamiltonians, this will avoid the need of an explicit clock qudit. Also similar to the situation with Hamiltonian encodings, we point out that removing the clock qudit to allow for a geometrically local Lindbladian would come at the cost of a much larger Hilbert space than that of the encoded qubits --- consequently, only a small subspace of this large Hilbert space will be used for the encoded computation. We will provide a precise definition of this subspace later --- in the following qualitative description of the how to construct the encoding Lindbladian, we will simply refer to as the ``space of valid states".
\begin{figure}
\centering
    \includegraphics[scale=0.42]{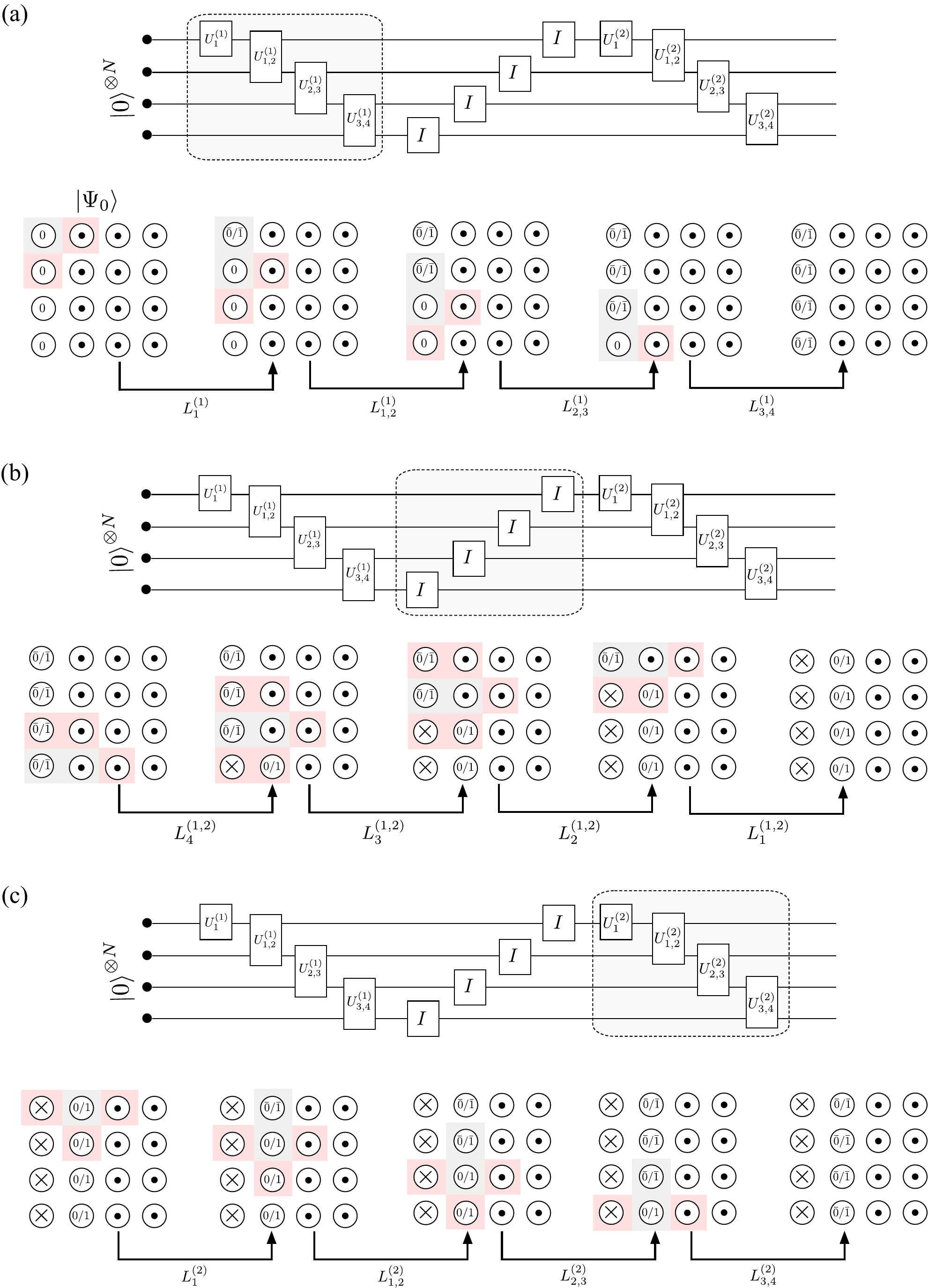}
    \caption{Schematic depiction of the jump operators used to implement the circuit. (a) In the first round, the qudits are initially in the state shown in $\ket{\Psi_0}$. The gates are then applied through the jump operators $L_1^{(1)}, L_{1, 2}^{(1)}, L_{2, 3}^{(1)} \dots $ which change the state of the qudits in the first column and also convert them from unbarred to barred states. (b) In the second round, the qudits are swapped from the first column to the second column by the jump operators $L_1^{(1, 2)}, L_2^{(1, 2)}\dots$ --- simultaneously, the qudits in the first column are converted to $\times$ and the qudits in the second column are converted to unbarred states. (c) Next, we apply the jump $L_1^{(2)}, L_{1, 2}^{(2)}, L_{2, 3}^{(2)}\dots $ which execute the second round of gates. The gray-shaded qudits are the qudits whose state can be changed by the applied operation, and the pink-shaded qudits are qudits whose state does not change on the application of the jump operator but are checked to be in the valid state by each jump operator. In particular, a gray shaded qudit whose type is being changed is always surrounded by a gray or pink-shaded qudit in our construction.}
    \label{fig:encoding_circuit_format}
\end{figure}

Initially, all the qudits other than those in the first column will be in $\ket{\odot}$, all the qudits in the first column will be in $\ket{0}$:
\begin{align}\label{eq:init_state_2D_grid}
    \ket{\Psi} = \bigotimes_{x = 1}^{N} \bigg(\ket{0_{x, 1}}\otimes \bigotimes_{y = 2}^{R}\ket{\odot_{x, y}}\bigg).
\end{align}
We then apply the first round of unitaries, $ U_1^{(1)}, U_{12}^{(1)}, U_{23}^{(1)} \dots U_{N - 1, N}^{(1)}$ and simultaneously transform the qudit state from $\ket{a} \to \ket{\bar{a}}$ for $a \in \{0, 1\}$. This is accomplished by applying the operators:
\begin{align}\label{eq:unitary_round_1}
 &L^{(1)}_{1} = \sum_{a, b = 0}^1 \big( U^{(1)}_{1}\big)_{a, b}\begin{bmatrix}\ket{\bar{a}}\!\bra{b}  & \ket{\odot}\!\bra{\odot} \\
 \pi & \end{bmatrix}_{2, 2}  + \text{h.c.}, \\
 &L^{(1)}_{y, y + 1} = \begin{dcases}\sum_{a, a', b, b' = 0}^1 \big(U^{(1)}_{y, y + 1}\big)_{a, a'; b, b'}\begin{bmatrix}
 \ket{\bar{a}}\!\bra{\bar{b}} & \\
 \ket{\bar{a}'}\!\bra{b} & \ket{\odot}\!\bra{\odot} \\
 \pi & 
 \end{bmatrix}_{2, y + 2} + \text{h.c.} & \text{ if } y\in [1: N - 2], \\
 \sum_{a, a', b, b' = 0}^1  \big(U^{(1)}_{y, y + 1}\big)_{a, a'; b, b'}\begin{bmatrix}
 \ket{\bar{a}}\!\bra{\bar{b}} & \\
 \ket{\bar{a}'}\!\bra{b} & \ket{\odot}\!\bra{\odot}
 \end{bmatrix}_{2, y + 1}+ \text{h.c.} & \text{ if } y = N - 1
 \end{dcases}
\end{align}
Note that these operators perform two tasks --- (1) they apply the unitaries $U^{(1)}_1, U^{(1)}_{1, 2} \dots U^{(1)}_{N - 1, N}$ and (2) they also check if the qudit whose type is being changed is horizontally and vertically adjacent to a qudit in the correct ``type". After application of the operators $L_1^{(1)}, L_{1, 2}^{(1)} \dots L_{N - 1, N}^{(1)}$, all the qudits in the first row will be in the states $\ket{\bar{0}}, \ket{\bar{1}}$. We also note that the order in which the operators $L_1^{(1)}, L_{1, 2}^{(1)} \dots L_{n - 1, n}^{(1)}$ are applied does not matter if we apply sufficiently large number of them --- this is because the operator $L_{x, x + 1}^{(1)}$ only applies $U^{(1)}_{x, x + 1}$ if the $i^\text{th}$ qudit in the first column is in an barred state, and the $(x + 1)^\text{th}$ qudit is in an unbarred state, else it annihilates the state (i.e.~maps the state to 0). Furthermore, only if this unitary is successfully applied would the $(x + 1)^\text{th}$ qudit will be mapped to a barred state, and consequently the only next operator that will act without annihilating the state is $V^{(1)}_{x + 1, x+2}$. Consequently, the choice of operators in Eq.~\ref{eq:unitary_round_1} naturally enforces the ordering of the unitaries in circuit being encoded. Another important point about the operators constructed in Eq.~\ref{eq:unitary_round_1} is that the operators also check if the qudit whose type is being changed is horizontally and vertically adjacent to a valid type --- in particular, this qudit should have a barred state vertically above it, an unbarred state vertically below it and a $\ket{\odot}$ state horizontally to its right. While this check might seem vacuous if the initial state is valid, it becomes important if the initial state is invalid and needs to be brought into the valid subspace via additional jump operators (which we describe below).

Next, we apply operators that swap the barred states from the first column to unbarred states in the second column, and simultaneously map the states in the first column to $\ket{\times}$. This is accomplished by applying the operators
\begin{align}\label{eq:operator_round_1_swap}
L_y^{(1, 2)} = \begin{cases}\sum_{a \in \{0, 1\}}\begin{bmatrix}
\bar{\pi} & \ket{\odot}\!\bra{\odot} & \\
\ket{\times}\!\bra{\bar{a}} & \ket{a}\!\bra{\odot} & \ket{\odot}\!\bra{\odot} \\
\ket{\times}\!\bra{\times} & \pi & 
\end{bmatrix}_{3, y + 1} & \text{ if }y\in [2 : N - 1], \\
\sum_{a \in \{0, 1\}}\begin{bmatrix}
\ket{\times}\!\bra{\bar{a}} & \ket{a}\!\bra{\odot} & \ket{\odot}\!\bra{\odot} \\
\ket{\times}\!\bra{\times} & \pi & 
\end{bmatrix}_{3, y + 2} & \text{ if }y = 1, \\
\sum_{a \in \{0, 1\}}\begin{bmatrix}
\bar{\pi} & \ket{\odot}\!\bra{\odot} & \\
\ket{\times}\!\bra{\bar{a}} & \ket{a}\!\bra{\odot} & \ket{\odot}\!\bra{\odot} 
\end{bmatrix}_{3, y} & \text{ if }y = N, 
\end{cases}
\end{align}
In the next step, we again apply a round of operators similar to Eq.~\ref{eq:unitary_round_1} on the second column, followed by operators similar to Eq.~\ref{eq:operator_round_1_swap} on the second and third columns. More specifically, we apply the jump operators,
\begin{subequations}\label{eq:jump_op_computation}
    \begin{align}
    &L^{(x)}_1 = \sum_{a, b = 0}^1 \big(U^{(x)}_1\big)_{a, b}\begin{bmatrix} \ket{\times}\!\bra{\times}& \ket{\bar{a}}\!\bra{b} & \ket{\odot}\!\bra{\odot} \\ &
    \pi &
    \end{bmatrix}_{x + 1, 2}, L^{(x)}_{y, y + 1} = \sum_{a, a', b, b' =0}^1 \big(U^{(x)}_{y, y + 1})_{a, a'; b, b'}
    \begin{bmatrix}
    & \ket{\bar{a}}\!\bra{\bar{b}} & \\
    \ket{\times}\!\bra{\times} & \ket{\bar{a}'}\!\bra{b'} & \ket{\odot}\!\bra{\odot} \\
    & \pi &
    \end{bmatrix}_{x + 1, y + 2}, \\
    &L^{(x, x + 1)}_{y} = \sum_{a \in \{0, 1\}}\begin{bmatrix}
& \bar{\pi} & \ket{\odot}\!\bra{\odot} & \\
\ket{\times}\!\bra{\times} & \ket{\times}\!\bra{\bar{a}} & \ket{a}\!\bra{\odot} & \ket{\odot}\!\bra{\odot} \\
& \ket{\times}\!\bra{\times} & \pi & 
\end{bmatrix}_{x + 2, y + 1}
\end{align}
\end{subequations}
where any single qudit operator in these expressions which is defined outside the lattice should be treated as identity. For simplicity, we will often denote the set containing all of these jump operators by $\{L_t : t\in \mathcal{T}\}$, with the set $\mathcal{T}$ indexing different jump operators.

We define $\mathcal{S}$ as the set of states which are of the form shown in Fig.~\ref{fig:valid_states} with the barred or unbarred 0/1 qubits being in a possibly entangled state. $\mathcal{S}$ contains the states that are obtained from the application of the jump operators $L^{(x)}_{y,  y+1}$ and $L^{(x, x + 1)}_y$ to the initial state $\ket{\Psi_0}$ and will be referred to the space of ``valid" configurations. Furthermore, $\mathcal{S}\cong (\mathbb{C}^2)^{\otimes{N}}\otimes \mathbb{C}^{\otimes {2NR + 1}}$ i.e.~every state in $\mathcal{S}$ is a linear combination of states of the form $\ket{\psi}\otimes \ket{\gamma_t}$, where $\ket{\gamma_t}$ for $t \in \{0, 1, 2 \dots 2NR\}$ indicates the ``shape" of the state (i.e.~which qudits are barred 0/1, unbarred 0/1, $\odot$ or $\times$) and $\ket{\psi}$ is a $N-$qubit entangled state. We can also think of the states $\{\ket{\gamma_t}\}_{t\in [0:(2nR +1)]}$ as the state of a clock, even though we do not have any explicit qubits designated to be the clock. Similar to the 5-local case, we will need to penalize invalid configurations. As is shown Ref.~\cite{aharonov2008adiabatic}, the invalid configurations in the 2D encoding can be checked 2-locally. In particular, the show the following lemma.
\begin{lemma}[Claim 4.2 of Ref.~\cite{aharonov2008adiabatic}]
A computational basis state $\ket{x}\in \{\times, \odot, 0, 1, \bar{0}, \bar{1}\}$ is a valid state (i.e.~$\ket{x}\in \mathcal{S}$) if none of the following two-qudit horizontally or vertically adjacent configurations are present in $\ket{x}$ ($a$ or $b$ denotes unbarred $0$ or $1$ state and $\bar{a}$ or $\bar{b}$ denotes barred $0$ or $1$ state):
\begin{enumerate}
    \item $(\odot, a)$, $(\odot, \bar{a}) $, since in all valid states, a qudit in the state $\odot$ is to the right of all the other qudits.
    \item $(a, \times)$, $(\bar{a}, \times)$, since in all valid states, a qudit in the state $\times$ is to the left of all the other qudits.
    \item $(\odot, \times)$, $(\times, \odot)$ since in all valid states, qubits in $\times$ and qubits in $\odot$ are separated by a barred or unbarred qudit.
    \item $(a, b)$, $(\bar{a}, b)$, $(a, \bar{b})$ or $(\bar{a}, \bar{b})$ since there is only one qudit in a barred or unbarred 0/1 state per row.
    \item $\big( {\odot \atop \bar{a}}\big), \big({\times \atop \bar{a}}\big), \big({b \atop \bar{a}}\big)$ since only a barred 0/1 qudit can be vertically above a barred 0/1 qudit.
    \item $\big( {a \atop \odot }\big)$, $\big( {a \atop \times }\big)$ since only an unbarred 0/1 qudit can be vertically below an unbarred 0/1 qudit.
    \item $\big({\odot \atop \times}\big)$, $\big( {\times \atop \odot}\big)$ since qudits in $\odot$ state and $\times$ state cannot be vertically adjacent.
    \item $\big({\bar{a} \atop \odot} \big)$, $\big({\times \atop a }\big)$ since there is no $\odot$ qudit below a barred 0/1 qudit and no $\times$ qudit above an unbarred 0/1 qudit.
    \item $\odot$ in the first column of qudits and $\times$ in the last column of qudits.
\end{enumerate}
\end{lemma}
For each of the invalid configurations from 1 to 9, which we now construct a jump operator which maps the invalid qudits' state to valid qudit states by just changing the state of one of the qudits. More importantly, for a set of two neighbouring vertical qudits that are in an invalid configuration, we will make the choice to change the bottom qudit to make the configuration valid and for a set of two neighbouring horizontal qudits that are in an invalid configuration, we will make the choice to change the left qudit to make the configuration valid. This can be accomplished by adding the jump operators shown in the table below corresponding to each of the two-qudit invalid configurations.
\begin{table}[htpb]
    \centering
    \begin{tabular}{| c | c || c | c |}
    \hline
        \text{Configuration}  & \text{Jump operator(s)} & \text{Configuration} & \text{Jump operator(s)} \\
        \hline &&&\\[-0.75em]
          $(\odot, a)$&  $\ket{\odot, \odot}\!\bra{\odot, a}$ & $\big( {\odot \atop \bar{a}}\big)$ & $\frac{1}{\sqrt{3}}\bigket{\odot \atop z}\!\bigbra{\odot \atop \bar{a}}, \ z \in \{\odot, 0, 1\}$   \\
          \hline  &&&\\[-0.75em]
          $(\odot, \bar{a})$ & $\ket{\odot, \odot}\!\bra{\odot, \bar{a}}$ & $\big( {\times \atop \bar{a}}\big)$ & $\bigket{\times \atop \times}\!\bigbra{\times \atop \bar{a}}$ \\
          \hline &&&\\[-0.75em]
          $(a, \times)$ & $\ket{a, \odot}\!\bra{a, \times}$ & $\big( {b \atop \bar{a}}\big)$ & $\frac{1}{\sqrt{2}}\bigket{b \atop z}\!\bigbra{b \atop \bar{a}}, \ z\in \{0,1\}$  \\
          \hline &&&\\[-0.75em]
          $(\bar{a}, \times)$ & $\ket{\bar{a}, \odot}\!\bra{\bar{a}, \times}$ & $\big( {a \atop \odot}\big)$ & $\frac{1}{\sqrt{2}}\bigket{{a \atop z }}\!\bra{{a \atop \odot}}, \ z\in\{0, 1\}$  \\
          \hline &&&\\[-0.75em]
          $(\odot, \times)$ & $\ket{\odot, \odot}\!\bra{\odot, \times}$ & $\big( {a \atop \times}\big)$ & $\frac{1}{\sqrt{2}}\bigket{{a \atop z}}\!\bigbra{{a \atop \times}}, \ z\in\{0, 1\}$  \\
          \hline  &&&\\[-0.75em]
          $(\times, \odot)$ & $\frac{1}{\sqrt{5}}\ket{\times, z}\!\bra{\times, \odot}, \ z \in\{\times, 0, 1, \bar{0}, \bar{1}\}$ & $\big( {\odot \atop \times}\big)$ & $
          \frac{1}{\sqrt{3}}\bigket{\odot \atop z}\!\bigbra{\odot \atop \times}, \ z \in \{\odot, 0, 1\}$  \\
          \hline &&&\\[-0.75em]
          $(a, b)$ & $\ket{a, \odot}\!\bra{a, b}$ & $\big( {\times \atop \odot}\big)$ & $\bigket{{\times \atop \times}}\!\bigbra{\times \atop \odot}$  \\
          \hline &&&\\[-0.75em]
          $(\bar{a}, b)$ & $\ket{\bar{a}, \odot}\!\bra{\bar{a}, b}$ & $\big( {\bar{a} \atop \odot}\big)$ & $\frac{1}{\sqrt{5}}\bigket{\bar{a}\atop z}\!\bigbra{\bar{a} \atop \odot}, z \in \{\times, 0, 1, \bar{0}, \bar{1}\}$  \\
          \hline &&&\\[-0.75em]
          $(a, \bar{b})$ & $\ket{a, \odot}\bra{a, \bar{b}}$ & $\big( {\times \atop a}\big)$ & $\bigket{{\times \atop \times}}\!\bigbra{{\times \atop {a}}}$  \\
          \hline &&&\\[-0.75em]
          $(\bar{a}, \bar{b})$ & $\ket{\bar{a}, \odot}\!\bra{\bar{a}, \bar{b}}$ & \ & \ \\
          \hline &&&\\[-0.75em]
          $(\odot_{1, 1})$ & $\frac{1}{2}\ket{z_{1, 1}}\bra{\odot_{1, 1}} \ z \in \{0, 1, \bar{0}, \bar{1}\}$ & $(\times_{R, 1})$ & $\frac{1}{2}\ket{z_{x, R}}\!\bra{\times_{1, R}}  \text{ for } z\in \{0, 1, \bar{0}, \bar{1}\}$  \\
          \hline
          
    \end{tabular}
    \caption{Jump operators corresponding to different invalid configurations. Each of the horizontal two-qudit jump operator corrects the invalid configuration by changing the right-most qudit, and the vertical two-qudit jump operator corrects the invalid configuration by changing the bottom-most qudit.}
    \label{tab:jump_operator}
\end{table}
We will denote the set of jump operators penalizing these incorrect configurations by $\{{P}_j : j \in \mathcal{J}\}$ --- this set will include the horizontal and vertical two-qudit jump operators at every point on the lattice, as well as the single qudit jump operators on the first and last column of the lattice. 

\begin{figure}[b]
    \centering
    \includegraphics[scale=0.5]{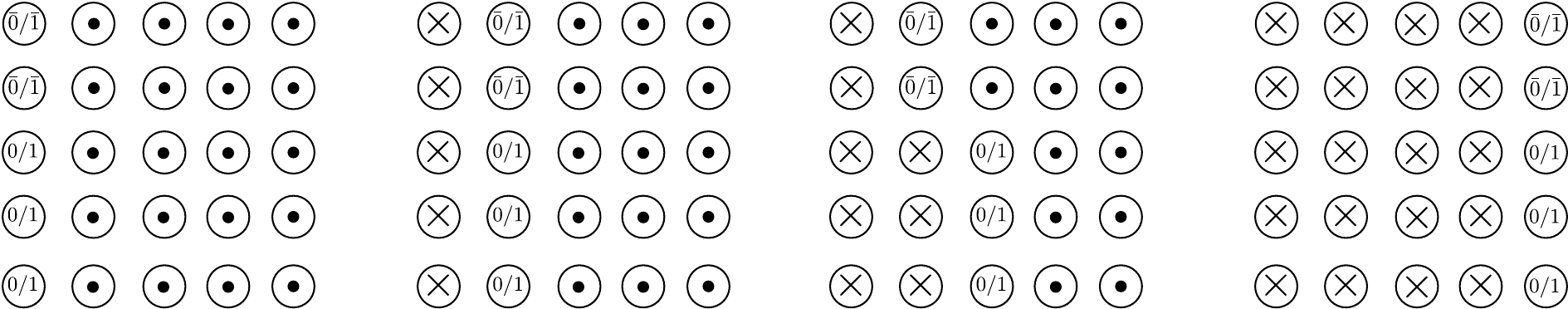}
    \caption{Schematic depiction of typical valid states i.e.~states that lie in the subspace $\mathcal{S}$ and are reached on applying the jump operators $\{L_t: t\in \mathcal{T}\}$ on the initial state $\ket{\Psi_0}$.}
    \label{fig:valid_states}
\end{figure}
Finally, to initialize the qubits to $\ket{\Psi_0}$, we add jump operators $\{S_\alpha\}_{\alpha \in [1:N]}$ which act only on the first column of qudits and map them to unbarred 0 from unbarred 1 i.e $S_\alpha = \ket{0_{1, \alpha}}\!\bra{1_{1, \alpha}}$.

\emph{Analysis}. We begin by establishing that the operators $\{P_j\}_{j \in \mathcal{J}}$ enforce that, at long times, the state of the qudits is in the valid subspace. For this, we consider the operator $F$, which counts the number of invalid configurations in the state of the qudits:
\[
F = \sum_{j \in \mathcal{J}}P_j^\dagger P_j,
\]
and analyze its dynamics in the Heisenberg picture. We will use the following technical lemma.
\begin{lemma}\label{lemma:bound_sequence_inequalities}
Suppose $x_1(t), x_2(t), x_3(t) \dots x_n(t) \geq 0$, with $x_i(0) \leq 1 \ \forall i \in [1:n]$, satisfy the inequalities
\[
\frac{d}{dt}x_1(t) \leq -\alpha_1 x_1(t) \text{ and }\frac{d}{dt}x_i(t) \leq -\alpha_i x_i(t) + \sum_{j = 1}^{i - 1} \beta_{i, j}x_j(t),
\] 
where $\alpha_i \geq \alpha>0$, $\beta_{i, j} \geq 0, \sum_{j} \beta_{i, j} \leq \beta$ for some $\alpha, \beta > 0$ and $\forall i$, then
\[
x_i(t) \leq (\max(1, 2\beta/\alpha) )^n e^{-\alpha t/2}\text{ for all } t\geq 0, i \in [1:n].
\]
\end{lemma}
\begin{proof}
We immediately obtain that $x_1(t) \leq x_1(0) e^{-\alpha_1 t} \leq e^{-\alpha t}$. We can rewrite the inequalities for $x_2(t), x_3(t) \dots x_n(t)$ as integral inequalities:
\[
x_i(t) \leq e^{-\alpha_i t} + \sum_{j = 1}^{i - 1}\beta_{i, j} \int_0^t x_j(s) e^{-\alpha_i(t - s)} ds \leq e^{-\alpha t} + \sum_{j = 1}^{i - 1}\beta_{i, j} \int_0^t x_j(s) e^{-\alpha (t - s)} ds.
\]
Define $\tilde{x}_i(t) = \max_{j \in [1:i]} x_j(t)$ --- we then obtain that $\tilde{x}_1(t) \leq e^{-\alpha t}$ and for $i \geq 2$, 
\[
\tilde{x}_i(t) \leq e^{-\alpha t} + \sum_{j = 1}^{i - 1}\beta_{i, j} \int_0^t \tilde{x}_{i - 1}(s) e^{-\alpha (t - s)} ds \leq e^{-\alpha t} + \beta \int_0^t \tilde{x}_{i - 1}(s) e^{-\alpha (t - s)} ds.
\]
We can integrate these equations recursively to obtain that $\forall i \in [1:n]$,
\[
x_i(t) \leq \tilde{x}_i(t) \leq \bigg(1 + \beta t + \frac{\beta^2 t^2}{2!} + \frac{\beta^3 t^3}{3!} + \dots \frac{\beta^{i - 1}t^{i - 1}}{(i - 1)!} \bigg)e^{-\alpha t},
\]
Noting that for $t \geq 0$ and $i \in [1:n]$, $(\beta t)^i = (\alpha t / 2)^i (2\beta / \alpha)^i \leq (\max(1, 2\beta/\alpha) )^n (\alpha t/2)^i$, we obtain that
\[
x_i(t) \leq (\max(1, 2\beta/\alpha) )^n \bigg(1 + \frac{\alpha t}{2} + \frac{\alpha^2 t^2}{2! 2^2} + \frac{\alpha^3 t^3}{3! 2^3} \dots \frac{\alpha^{i - 1}t^{i - 1}}{(i - 1)! 2^{i - 1}}\bigg)e^{-\alpha t} \leq  (\max(1, 2\beta/\alpha) )^n e^{-\alpha t/2},
\]
which establishes the lemma statement.
\end{proof}

\begin{lemma}\label{lemma:encoding_num_error_bound}
    For any initial state $\rho(0)$ of the 2D grid of qudits, $\expect{F(t)} = \textnormal{Tr}(F e^{\mathcal{L}t} (\rho(0))) \leq (NR + 1) 20^{NR + 1}e^{-t/5}$.
\end{lemma}
\begin{proof}
It will  be convenient to define the sets $\mathcal{J}_{x, y}^{(H)} = \{j \in \mathcal{J} : P_j \text{ is a 2-qudit operator supported on } (x - 1, y), (x, y)\} $ and $\mathcal{J}_{x, y}^{(V)} = \{j \in \mathcal{J} : P_j \text{ is a 2-qudit operator supported on }(x, y - 1), (x, y)\}$.
We first re-express $F$ as 
\[
F =  \pi^{\odot}_{1, 1} + \pi^{\times}_{R, 1} + \sum_{\substack{x\in[1:R] \\ y \in [1:N]}} F_{x, y},
\]
where $\pi^{\odot}_{x, y} = \ket{\odot_{x, y}}\!\bra{\odot_{x, y}}$, $\pi^\times_{x, y} = \ket{\times_{x, y}}\!\bra{\times_{x, y}}$ and
\begin{align}
&F_{x, y} = \delta_{x \neq 1} F_{x, y}^{(H)} +  \delta_{y \neq 1} F_{x, y}^{(V)} \text{ where }\nonumber \\ 
&\text{for }x\neq 1:  F^{(H)}_{x, y} = \sum_{j \in \mathcal{J}_{x, y}^{(H)}} P_j^\dagger P_j =  \begin{bmatrix}  \pi + \bar{\pi} + \ket{\odot}\!\bra{\odot} & \pi + \bar{\pi} + \ket{\times}\!\bra{\times}\end{bmatrix}_{x, y} + \begin{bmatrix} \ket{\times}\!\bra{\times} & \ket{\odot}\! \bra{\odot}\end{bmatrix}_{x, y} \text{ and } \nonumber \\
&\text{for } y\neq 1: F^{(V)}_{x, y } = \sum_{j \in \mathcal{J}_{x, y}^{(V)}} P_j^\dagger P_j = \begin{bmatrix} \ket{\odot}\!\bra{\odot} + \ket{\times}\!\bra{\times} + {\pi} \\ \bar{\pi} \end{bmatrix}_{x, y} + \begin{bmatrix} \ket{\times}\!\bra{\times} + \pi + \bar{\pi} \\ \ket{\odot}\!\bra{\odot} \end{bmatrix}_{x, y} + \begin{bmatrix} \ket{\times}\!\bra{\times} \\ \pi \end{bmatrix}_{x, y} + \begin{bmatrix} \pi + \ket{\odot}\!\bra{\odot} \\ \ket{\times}\!\bra{\times}\end{bmatrix}_{x, y}.
\end{align}
Here, we define $\delta_{a\neq b} = 1 - \delta_{a, b}$. We also define
\begin{align}
    \mathcal{L}_{x, y}^{(H)} = \sum_{j \in \mathcal{J}_{x, y}^{(H)}}\mathcal{D}_{P_j}  \text{ and }\mathcal{L}_{x, y}^{(V)} = \sum_{j \in \mathcal{J}_{x, y}^{(V)}}\mathcal{D}_{P_j}.
\end{align}
As laid out in Table~\ref{tab:jump_operator}, since for any horizontally (vertically) invalid two-qudit configuration, the horizontal (vertical) two-qudit jump operators $P_j$ correct it by changing the right (bottom) qudit, it follows that 
\begin{subequations}\label{eq:penality_operator_penalty}
\begin{align}
    &\text{If } (x', y') \notin \{(x - 1, y), (x, y)\}, \text{ then } \forall \ j \in \mathcal{J}_{x', y'}^{(H)}: [P_j, F^{(H)}_{x, y}] = 0 \implies \mathcal{L}^{(H)\dagger}_{x', y'}(F^{(H)}_{x, y}) = 0, \\
    &\text{If } (x', y') \notin \{(x, y - 1), (x, y) \} \text{ then } \forall \ j \in \mathcal{J}_{x', y'}^{(H)} : [P_j, F^{(V)}_{x, y}] = 0 \implies \mathcal{L}^{(H)\dagger}_{x', y'}(F^{(V)}_{x, y}) = 0, \\
    &\text{If } (x', y') \notin\{(x, y - 1), (x, y)\} \text{ then } \forall \ j \in \mathcal{J}_{x', y'}^{(V)} : [P_j, F^{(V)}_{x, y}] = 0 \implies \mathcal{L}^{(V)\dagger}_{x', y'}(F^{(V)}_{x, y}) = 0, \\
    &\text{If } (x', y') \notin \{(x - 1, y), (x, y)\}, \text{ then } \forall \ j \in \mathcal{J}_{x', y'}^{(V)}: [P_j, F^{(H)}_{x, y}] = 0 \implies \mathcal{L}^{(V)\dagger}_{x', y'}(F^{(H)}_{x, y}) = 0.
\end{align}
\end{subequations}
\begin{figure}
\includegraphics[scale=0.475]{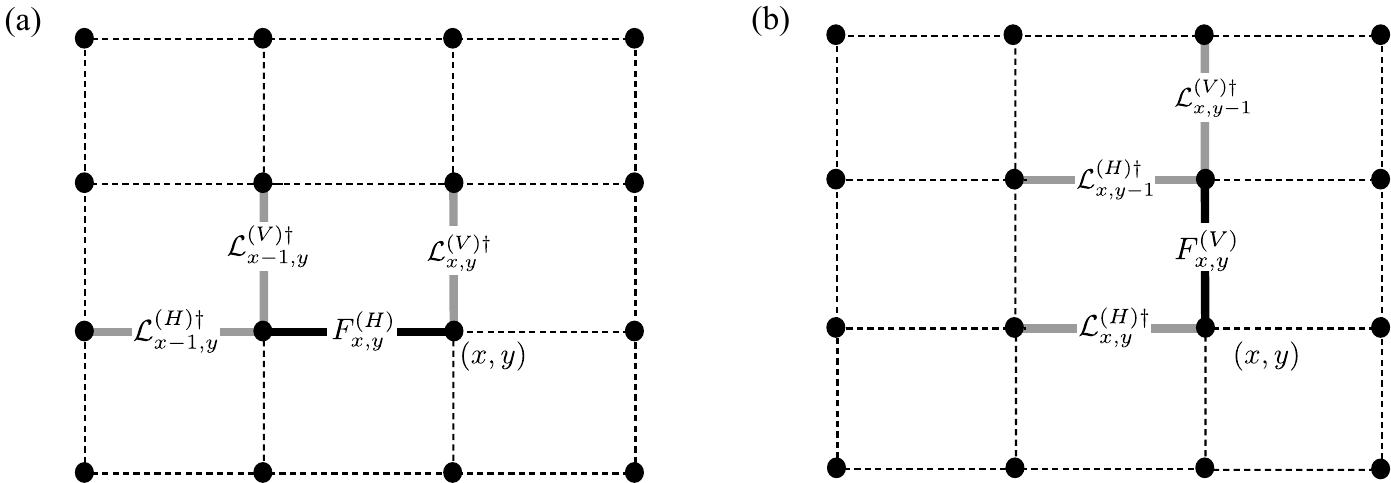}
\caption{Schematic depiction of which jump operators $P_j$ effect error number operators $F_{x, y}^{(H)}, F_{x, y}^{(V)}$ corresponding to a given edge. The edges whose associated $P_j$ can change $F_{x, y}^{(H)}, F_{x, y}^{(V)}$ corresponding to a dark edge are shown in gray, while the $P_j$ operators associated with the dotted edges do not impact the error number operator corresponding to the dark edge.}
\label{fig:plaq_comm}
\end{figure}
\noindent These relationships are depicted in Fig.~\ref{fig:plaq_comm} which shows, for a given horizontal or vertical edge associated with a $F_{x, y}^{(H/V)}$, which edges are associated with terms $\mathcal{L}_{x, y}^{(H/V)\dagger}$ that result in a non-zero value when applied on the $F_{x, y}^{(H/V)}$. Furthermore, it can be verified by explicit computation that
\begin{align}\label{eq:top_right_commutation}
\sum_{z\in \{0, 1, \bar{0}, \bar{1}\}} \mathcal{D}_{\ket{z_{R, 1}}\!\bra{\times_{R, 1}}}^\dagger \left(F^{(H)}_{R, 1}\right) = 0.
\end{align}
It can also be noted that, as a consequence of the specific choice of $L_t$ (Eqs.~\ref{eq:jump_op_computation}) which checks the validity of the qudits horizontally and vertically adjacent to the qudit whose type is being changed, 
\begin{align}\label{eq:penality_operator_L}
[F_{x, y}^{(H)}, L_t] = 0 = [F_{x, y}^{(V)}, L_t] \ \forall \ t \in \mathcal{T}.
\end{align}
Furthermore, since $F^{(H)}_{x, y}$ and $F^{(V)}_{x, y}$ are expressible entirely in terms of the projector $\pi = \ket{0}\!\bra{0} + \ket{1}\!\bra{1}$, as opposed to the individual projectors $\ket{0}\!\bra{0}$ and $\ket{1}\!\bra{1}$, it also follow that
\begin{align}\label{eq:penalty_operator_S}
[F_{x, y}^{(H)}, S_\alpha] = 0 = [F_{x, y}^{(V)}, S_\alpha] \ \forall \alpha \in \mathcal{A}.
\end{align}
We now obtain upper bounds on $\langle F_{x, y}^{(H)}(t)\rangle = \textnormal{Tr}(F_{x, y}^{(H)}\rho(t))$ and $\langle F_{x, y}^{(V)}(t)\rangle = \textnormal{Tr}(F^{(V)}_{x, y} \rho(t))$. It follows from Eqs.~\ref{eq:penality_operator_penalty}, \ref{eq:top_right_commutation}, \ref{eq:penality_operator_L} and \ref{eq:penalty_operator_S} that $\forall x \in [2: R], y\in [1:N]$, 
\begin{subequations}\label{eq:dynamics_F}
\begin{align}
\frac{d}{dt} \langle F_{x, y}^{(H)}(t)\rangle  + \langle F_{x, y}^{(H)}(t)\rangle  &= \delta_{x \neq 2} \tr{}{\mathcal{L}^{(H)\dagger}_{x - 1, y} (F_{x, y}^{(H)}) \rho(t)} + \delta_{y \neq 1} \tr{}{\mathcal{L}^{(V)\dagger}_{x, y}(F_{x, y}^{(H)}) \rho(t)} + \delta_{y \neq 1} \tr{}{\mathcal{L}^{(V)\dagger}_{x - 1, y}(F_{x, y}^{(H)}) \rho(t)} + \nonumber\\
&\qquad \qquad \frac{1}{4}\delta_{x, 2} \delta_{y, 1}\sum_{z \in \{0, 1, \bar{0}, \bar{1}\}}\text{Tr}\big(\mathcal{D}_{\ket{z_{1, 1}}\!\bra{\odot_{1, 1}}}^\dagger(F_{x, y}^{(H)})\rho(t)\big),
\end{align}
and $\forall x \in [1:R], y \in [2:N]$, 
\begin{align}
&\frac{d}{dt} \langle F_{x, y}^{(V)}(t)\rangle  + \langle F_{x, y}^{(V)}(t)\rangle  = \delta_{y \neq 2} \tr{}{\mathcal{L}^{(V)\dagger}_{x, y - 1} (F_{x, y}^{(V)}) \rho(t)} + \delta_{x \neq 1} \tr{}{\mathcal{L}^{(H)\dagger}_{x, y}(F_{x, y}^{(V)}) \rho(t)} + \delta_{x \neq 1} \tr{}{\mathcal{L}^{(H)\dagger}_{x, y - 1}(F_{x, y}^{(V)}) \rho(t)} + \nonumber\\
&\qquad \qquad \frac{1}{4}\delta_{x, 1}\delta_{y, 2}\sum_{z \in \{0, 1, \bar{0}, \bar{1}\}}\text{Tr}\big({\mathcal{D}^\dagger_{\ket{z_{1, 1}}\!\bra{\odot_{1, 1}}}}(F_{x, y}^{(V)})\rho(t)\big) + \frac{1}{4}\delta_{x, R}\delta_{y, 2}\sum_{z \in \{0, 1, \bar{0}, \bar{1}\}}\text{Tr}\big({\mathcal{D}^\dagger_{\ket{z_{R, 1}}\!\bra{\times_{R, 1}}}}(F_{x, y}^{(V)})\rho(t)\big)
\end{align}
\end{subequations}
where we have used the fact that $\mathcal{L}_{x, y}^{(H)\dagger}(F_{x, y}^{(H)}) = -F_{x, y}^{(H)}$ and $\mathcal{L}_{x, y}^{(V)\dagger}(F_{x, y}^{(V)}) = -F_{x, y}^{(V)}$. We can upper bound some of the terms on the right hand side of Eqs.~\ref{eq:dynamics_F}a and b with a simple argument --- note that
\begin{align}\label{eq:upper_bound_simple_argument}
 \tr{}{\mathcal{L}^{(H)\dagger}_{x - 1, y} (F_{x, y}^{(H)}) \rho(t)} &= \sum_{j \in \mathcal{J}^{(H)}_{x - 1, y}} \text{Tr}(P_j^\dagger F_{x, y}^{(H)} P_j \rho(t)) - \tr{}{F_{x - 1, y}^{(H)} F_{x, y}^{(H)} \rho(t)}, \nonumber \\
& \numleq{1} \sum_{j \in \mathcal{J}^{(H)}_{x - 1, y}} \text{Tr}( F_{x, y}^{(H)} P_j \rho(t)P_j^\dagger)\numleq{2} \sum_{j \in \mathcal{J}^{(H)}_{x - 1, y}} \text{Tr}(P_j \rho(t) P_j^\dagger) = \langle F_{x - 1, y}^{(H)}(t) \rangle,
\end{align}
where in (1) we have used the fact that $F^{(H)}_{x - 1, y} F^{(H)}_{x, y} \succeq 0$, and in (2) we have used the fact that $\norm{F}_{x, y}^{(H)} \leq 1$ and consequently for any $\sigma \succeq 0, \text{Tr}(F^{(H)}_{x, y} \sigma) \leq \norm{F_{x, y}^{(H)}} \norm{\sigma}_1 \leq \text{Tr}(\sigma)$. Proceeding similarly, we can obtain the upper bounds
\begin{align}
\tr{}{\mathcal{L}^{(V)\dagger}_{x - 1, y}(F_{x, y}^{(H)})\rho(t)}, \tr{}{\mathcal{L}^{(V)\dagger}_{x, y - 1}(F^{(V)}_{x, y}) \rho(t)} \leq \langle F_{x, y - 1}^{(V)}(t) \rangle,  \tr{}{\mathcal{L}^{(H)\dagger}_{x - 1, y}(F_{x, y}^{(V)})\rho(t)} \leq \langle F_{x - 1, y}^{(H)}(t) \rangle.
\end{align}
Furthermore,
\begin{align}
\sum_{z \in \{0, 1, \bar{0}, \bar{1}\}} \text{Tr}\big({\mathcal{D}^\dagger_{\ket{z_{1, 1}}\!\bra{\odot_{1, 1}}}}(F_{1, 2}^{(V)})\rho(t)\big), \sum_{z \in \{0, 1, \bar{0}, \bar{1}\}} \text{Tr}\big({\mathcal{D}^\dagger_{\ket{z_{1, 1}}\!\bra{\odot_{1, 1}}}}(F_{2, 1}^{(H)})\rho(t)\big) \leq 4\langle\pi^{\odot}_{1, 1}(t)\rangle,
\end{align}
and
\begin{align}
\sum_{z \in \{0, 1, \bar{0}, \bar{1}\}} \text{Tr}\big({\mathcal{D}^\dagger_{\ket{z_{R, 1}}\!\bra{\times_{R, 1}}}}(F_{R, 2}^{(V)})\rho(t)\big) \leq 4\langle\pi^{\times}_{R, 1}(t)\rangle
\end{align}
Using these bounds and Eq.~\ref{eq:dynamics_F}, we obtain that $\forall \ x\in [2:R], y \in [1:N]$,
\begin{subequations}\label{eq:dynamics_F_ineq}
\begin{align}
\frac{d}{dt} \langle F_{x, y}^{(H)}(t)\rangle  + \langle F_{x, y}^{(H)}(t)\rangle  &\leq \delta_{x \neq 2} \langle F_{x - 1, y}^{(H)}(t)\rangle + \delta_{y \neq 1}\langle F_{x - 1, y}^{(V)}(t) \rangle + \delta_{y \neq 1} \tr{}{\mathcal{L}^{(V)\dagger}_{x, y}(F_{x, y}^{(H)}) \rho(t)} + \delta_{x, 2}\delta_{y, 1}\langle \pi^{\odot}_{1, 1}(t)\rangle,
\end{align}
and $\forall x \in [1:R], y \in [2:N]$, 
\begin{align}
\frac{d}{dt} \langle F_{x, y}^{(V)}(t)\rangle  + \langle F_{x, y}^{(V)}(t)\rangle  \leq \delta_{y \neq 2} \langle F_{x, y - 1}^{(V)}(t)\rangle + \delta_{x \neq 1}\langle F_{x, y - 1}^{(H)} (t)\rangle+ \delta_{x \neq 1} \tr{}{\mathcal{L}^{(H)\dagger}_{x, y}(F_{x, y}^{(V)}) \rho(t)}  + \delta_{x, 1}\delta_{y, 2} \langle\pi_{1, 1}^{\odot}(t)\rangle + \nonumber\\
 \delta_{x, R} \delta_{y, 2} \langle \pi^{\times}_{R, 1}(t) \rangle.
\end{align}
\end{subequations}
We will now arrange these inequalities into a sequence which is of the form analyzed in Lemma~\ref{lemma:bound_sequence_inequalities}, and then use this lemma to obtain an upper bound on $\langle F(t) \rangle$.  We will first consider the first column of qudits and upper bound $\langle F^{(V)}_{1, y}(t)\rangle$, as well as $\langle \pi^{\odot}_{1, 1}(t) \rangle$. For $\langle \pi^{\odot}_{1, 1}(t) \rangle$, we obtain that
\begin{align}\label{eq:sequence_1}
\frac{d}{dt}\langle \pi^{\odot}_{1, 1}(t) \rangle + \langle \pi^{\odot}_{1, 1}(t) \rangle = 0.
\end{align}
Specializing Eq.~\ref{eq:dynamics_F_ineq}b to the first column of qudits ($x = 1$), we obtain
\begin{align}\label{eq:sequence_2}
\frac{d}{dt}\langle F_{1, y}(t) \rangle + \langle F_{1, y}(t) \rangle \leq \delta_{y \neq 2} \langle F_{1, y - 1}(t) \rangle + \delta_{y, 2} \langle \pi^{\odot}_{1, 1}(t) \rangle,
\end{align}
where we have used the fact that $F_{1, y} = F_{1, y}^{(V)}$.

We next consider inequalities for all the remaining columns. For $\langle F_{x, 1}^{(H)}(t) \rangle = \langle F_{x, 1}(t) \rangle$, we obtain from Eq.~\ref{eq:dynamics_F_ineq}a that
\begin{align}\label{eq:sequence_3}
\frac{d}{dt} \langle F_{x, 1}(t) \rangle  + \langle F_{x, 1}(t) \rangle \leq \delta_{x\neq 2} \langle F_{x - 1, 1}(t) \rangle + \delta_{x, 2}\langle \pi^{\odot}_{1, 1}(t) \rangle .
\end{align}
To use Eqs.~\ref{eq:dynamics_F_ineq}a, b when $y \in [2:N ]$, we need to provide upper bounds on $\tr{}{\mathcal{L}^{(H)\dagger}_{x, y}(F_{x, y}^{(V)})\rho(t)}$ and $\tr{}{\mathcal{L}^{(V)\dagger}_{x, y}(F_{x, y}^{(H)})\rho(t)}$. While these terms could be bounded by following the argument in Eq.~\ref{eq:upper_bound_simple_argument}, the bound thus obtained is too loose to be useful.  Instead, a more careful analysis is needed --- in particular, one important fact that we will use is that if $F_{x - 1, y}^{(V)}$ and $F_{x, y - 1}^{(H)}$ have small expected value (i.e.~the invalid configurations on the edges $((x - 1, y - 1), (x - 1, y))$ and $((x - 1, y - 1), (x, y - 1))$) are unlikely), then the probable configurations on the diagonally opposite sites $(x - 1, y), (x, y - 1)$ are also constrained. This can be made formal by noting the simple inequality
\begin{align}\label{eq:penalty_diagonal}
&\bigtr{}{\begin{bmatrix}
\ & \ket{\times}\!\bra{\times} + \pi + \bar{\pi} \\
\ket{\odot}\!\bra{\odot} + \pi + \bar{\pi} & \
\end{bmatrix}_{x, y}\rho(t)} \nonumber\\
&\qquad = \bigtr{}{\begin{bmatrix}
\ket{\odot}\!\bra{\odot} + \pi + \bar{\pi} & \ket{\times}\!\bra{\times} + \pi + \bar{\pi} \\
\ket{\odot}\!\bra{\odot} + \pi + \bar{\pi} & \
\end{bmatrix}_{x, y}\rho(t)}  + \bigtr{}{\begin{bmatrix}
\ket{\times}\!\bra{\times} & \ket{\times}\!\bra{\times} + \pi + \bar{\pi} \\
\ket{\odot}\!\bra{\odot} + \pi + \bar{\pi} & \
\end{bmatrix}_{x, y}\rho(t)}, \nonumber \\
&\qquad \leq \bigtr{}{\begin{bmatrix}
\ket{\odot}\!\bra{\odot} + \pi + \bar{\pi} & \ket{\times}\!\bra{\times} + \pi + \bar{\pi} 
\end{bmatrix}_{x, y - 1}\rho(t)}  + \bigtr{}{\begin{bmatrix}
\ket{\times}\!\bra{\times} \\
\ket{\odot}\!\bra{\odot} + \pi + \bar{\pi} 
\end{bmatrix}_{x - 1, y}\rho(t)}, \nonumber\\
&\qquad \leq \langle F^{(H)}_{x, y - 1}(t) \rangle + \langle F^{(V)}_{x - 1, y}(t) \rangle.
\end{align}
Equation \ref{eq:penalty_diagonal} can be interpreted as indicating that if the configurations on the edges  $((x - 1, y - 1), (x - 1, y))$ and $((x - 1, y - 1), (x, y - 1))$ are valid (i.e. $\langle F^{(H)}(x, y)\rangle = \langle F^{(V)}_{x, y}\rangle = 0$), then the diagonally opposite vertices, $(x - 1, y)$ and $(x, y - 1)$, cannot have the configurations $(a, b)$ where $a \in \{\odot, \bar{0}, \bar{1}, 0, 1\}$ and $b \in \{\times, \bar{0}, \bar{1}, 0, 1\}$. We now consider upper bounding $\tr{}{\mathcal{L}_{x, y}^{(V)\dagger}(F_{x, y}^{(H)})\rho(t)}$:
\begin{subequations}\label{eq:upper_bound_cross_term}
\begin{align}
&\tr{}{\mathcal{L}_{x, y}^{(V)\dagger}(F_{x, y}^{(H)})\rho(t)}= \sum_{j \in \mathcal{J}_{x, y}^{(V)}} \text{Tr}(P_j^\dagger F_{x, y}^{(H)} P_j \rho(t)) - \text{Tr}(F_{x, y}^{(H)} F_{x, y}^{(V)} \rho(t)), \nonumber \\
&=\frac{1}{3} \bigtr{}{\begin{bmatrix}
 \ & \ket{\odot}\!\bra{\odot} \\
 \ket{\times}\!\bra{\times} & \bar{\pi} + \ket{\times}\!\bra{\times}
\end{bmatrix}_{x, y}\rho(t)} + \frac{2}{3}\bigtr{}{\begin{bmatrix}
 \ & \ket{\odot}\!\bra{\odot} \\
 \pi + \bar{\pi} + \ket{\odot}\!\bra{\odot} & \bar{\pi}
\end{bmatrix}_{x, y}\rho(t)}  +  \bigtr{}{\begin{bmatrix} 
 \ & \pi + \ket{\times}\!\bra{\times} \\
 \pi + \bar{\pi} + \ket{\odot}\!\bra{\odot} &  \bar{\pi}
\end{bmatrix}_{x, y}\rho(t)} \nonumber \\
& \qquad+ \bigtr{}{\begin{bmatrix} 
 \ & \ket{\times}\!\bra{\times} \\
 \pi + \bar{\pi} + \ket{\odot}\!\bra{\odot} & {\pi} + \ket{\odot}\!\bra{\odot}
\end{bmatrix}_{x, y}\rho(t)}  + \bigtr{}{\begin{bmatrix} 
 \ & \bar{\pi} \\
 \pi + \bar{\pi} + \ket{\odot}\!\bra{\odot} & \ket{\times}\!\bra{\times}
\end{bmatrix}_{x, y}\rho(t)} - \tr{}{F_{x, y}^{(H)} F_{x, y}^{(V)}\rho(t)}, \nonumber \\
& \numleq{1}  \frac{1}{3} \bigtr{}{\begin{bmatrix}
 \ & \ket{\odot}\!\bra{\odot} \\
 \ket{\times}\!\bra{\times} & \bar{\pi} + \ket{\times}\!\bra{\times}
\end{bmatrix}_{x, y}\rho(t)} + \bigtr{}{\begin{bmatrix} 
 \ & \bar{\pi} \\
 \pi + \bar{\pi} + \ket{\odot}\!\bra{\odot} & \ket{\times}\!\bra{\times}
\end{bmatrix}_{x, y}\rho(t)} + \bigtr{}{\begin{bmatrix} 
 \ & \ket{\times}\!\bra{\times} \\
 \pi + \bar{\pi} + \ket{\odot}\!\bra{\odot} &  \ket{\odot}\!\bra{\odot}
\end{bmatrix}_{x, y}\rho(t)} , \nonumber\\
& \leq \frac{1}{3} \bigtr{}{\begin{bmatrix}
\ket{\odot}\!\bra{\odot} \\
\bar{\pi} + \ket{\times}\!\bra{\times}
\end{bmatrix}_{x, y}\rho(t)} + \bigtr{}{\begin{bmatrix} 
 \ & \bar{\pi} + \ket{\times}\!\bra{\times} \\
 \pi + \bar{\pi} + \ket{\odot}\!\bra{\odot} & \
\end{bmatrix}_{x, y}\rho(t)}, \nonumber\\
& \numleq{2} \frac{1}{3}\langle F_{x, y}^{(V)}(t) \rangle +  \langle F_{x - 1, y}^{(H)}(t) \rangle + \langle F_{x, y - 1}^{(V)}(t) \rangle,
\end{align}
where to obtain (1), we need to use the explicit expression for $F_{x, y}^{(H)} F_{x, y}^{(V)}$ to cancel terms from $\sum_{j \in \mathcal{J}^{(V)}_{x, y}} P_j ^\dagger F_{x, y}^{(H)} P_j$, and in (2) we have used the bound from Eq.~\ref{eq:penalty_diagonal}. Proceeding similarly, we can upper bound $\text{Tr}(\mathcal{L}_{x, y}^{(H)\dagger}(F_{x, y}^{(V)})\rho(t))$:
\begin{align}
&\text{Tr}(\mathcal{L}_{x, y}^{(H)\dagger}(F_{x, y}^{(V)}) \rho(t)) = \sum_{j \in \mathcal{J}_{x, y}^{(H)}} \text{Tr}(P_j^\dagger F_{x, y}^{(V)} P_j \rho(t)) - \tr{}{F_{x, y}^{(V)} F_{x, y}^{(H)} \rho(t)}, \nonumber\\
&=\bigtr{}{\begin{bmatrix}
&\ket{\times}\!\bra{\times} + \pi + \bar{\pi} \\
\ket{\odot}\!\bra{\odot} + \pi + \bar{\pi} &  \ket{\times}\!\bra{\times} + \pi + \bar{\pi}
\end{bmatrix}\rho(t)}_{x, y} + \frac{3}{5}\bigtr{}{\begin{bmatrix}
& \ket{\odot}\!\bra{\odot} + \pi + \frac{4}{3}\ket{\times}\!\bra{\times} \\
\ket{\times}\!\bra{\times} &  \ket{\odot}\!\bra{\odot} 
\end{bmatrix}\rho(t)}_{x, y}  - \textnormal{Tr}(F_{x, y}^{(V)} F_{x, y}^{(H)} \rho(t)), \nonumber\\
&\leq\bigtr{}{\begin{bmatrix}
&\ket{\times}\!\bra{\times} + \pi + \bar{\pi} \\
\ket{\odot}\!\bra{\odot} + \pi + \bar{\pi} &  \ket{\times}\!\bra{\times} + \pi + \bar{\pi}
\end{bmatrix}\rho(t)}_{x, y} + \frac{3}{5}\bigtr{}{\begin{bmatrix}
& \ket{\odot}\!\bra{\odot}  \\
\ket{\times}\!\bra{\times} &  \ket{\odot}\!\bra{\odot} 
\end{bmatrix}\rho(t)}_{x, y}, \nonumber\\
&\leq\bigtr{}{\begin{bmatrix}
&\ket{\times}\!\bra{\times} + \pi + \bar{\pi} \\
\ket{\odot}\!\bra{\odot} + \pi + \bar{\pi} & 
\end{bmatrix}\rho(t)}_{x, y} + \frac{3}{5}\bigtr{}{\begin{bmatrix}
\ket{\times}\!\bra{\times} &  \ket{\odot}\!\bra{\odot} 
\end{bmatrix}\rho(t)}_{x, y}, \nonumber\\
&\leq \frac{3}{5}\langle F_{x, y}^{(H)}(t)\rangle +\langle F_{x - 1, y}^{(H)}\rangle + \langle F_{x, y - 1}^{(V)}(t) \rangle.
\end{align}
\end{subequations}
From Eqs.~\ref{eq:upper_bound_cross_term}a, b together with Eqs.~\ref{eq:dynamics_F_ineq}a, b, we obtain that for $x \in [2:R - 1], y \in [2: N]$
\begin{align}\label{eq:sequence_4}
\frac{d}{dt}\expect{F_{x, y}(t)} + \frac{2}{5} \expect{F_{x, y}(t)} \leq 2\big( \langle F_{x - 1, y}(t) \rangle +  \langle F_{x, y - 1}(t) \rangle\big) + \delta_{x, R}\delta_{y, 2} \langle \pi^{\times}_{R, 1}(t) \rangle.
\end{align}
Finally, we can similarly obtain an inequality for $\langle \pi_{1, R}^\times(t) \rangle$ from
\begin{align}\label{eq:sequence_5}
\frac{d}{dt} \langle \pi_{R, 1}^\times(t) \rangle + \langle \pi_{1, R}^\times(t) \rangle \leq \text{Tr}\big(\mathcal{L}_{R, 1}^{(H)\dagger}(\pi_{R, 1}^\times)\rho(t)\big) \leq \langle F_{R, 1}(t) \rangle.
\end{align}
Now, we can see that the inequalities in Eqs.~\ref{eq:sequence_1}, \ref{eq:sequence_2}, \ref{eq:sequence_3}, \ref{eq:sequence_4} and \ref{eq:sequence_5} satisfy the conditions of Lemma~\ref{lemma:bound_sequence_inequalities} with $x_1(t), x_2(t), x_3(t) \dots x_n(t)$ being chosen as $\expect{\pi^{\odot}_{1, 1}(t)}, \expect{F_{1, 2}(t)}, \expect{F_{1, 3}(t)} \dots \expect{F_{1, N}(t)} , \expect{F_{2, 1}(t)}, \expect{F_{2, 2}(t)} \dots \expect{F_{2, N}(t)} \dots   \expect{F_{R - 1, 1}(t)}, \expect{F_{R - 1, 2}(t)} \dots \expect{F_{R - 1, N}(t)}$, $\expect{F_{R, 1}(t)}, \expect{\pi^{\times}_{R, 1}(t)},  \expect{F_{R, 2}(t)} \dots \expect{F_{R, N}(t)}$ and with $\alpha = 2/5, \beta = 4$. We thus obtain from \ref{lemma:bound_sequence_inequalities} that
\[
\expect{F_{x, y}(t)}, \expect{\pi^{\odot}_{1, 1}(t)}, \expect{\pi^{\times}_{R, 1}(t)}, \leq 20^{N R + 1} e^{-t/5} \implies \expect{F(t)} \leq (NR + 1) 20^{N  R + 1} e^{-t/5},
\]
which proves the lemma.
\end{proof}
We can now finally prove Lemma~\ref{lemma:encoded_lindbladian} from the main text.

\begin{replemma}{lemma:encoded_lindbladian}
    Suppose we are given a quantum circuit $\mathcal{C}$ on $N$ qubits with architecture shown in Fig.  depth and with $R$ rounds of gates. Then, there exists a two-dimensional spatially local Lindbladian $\mathcal{L}$ on $NR$ 6-level qudits with a unqiue fixed point $\sigma$, as well as a local observable $O$ such that 
    \[
    \textnormal{Tr}(O \sigma) = \frac{1}{2NR }z_{\mathcal{C}},
    \]
    where $z_{\mathcal{C}}$ is the expected value of a pauli-\textnormal{Z} operator on the last qubit at the output of $\mathcal{C}$. Furthermore, for any initial state $\rho(0)$ of the two-dimensional grid of qudits,
    \[
    \norm{e^{\mathcal{L}t}(\rho(0)) - \sigma_{\mathcal{C}}}_1 \leq c_0(N, R) \exp(-\gamma_0(N, R) t),
    \]
    where $c_0(N, R) = O(N^6 R^2 \exp(O(NR)))$ and $\gamma_0(N, R) = \Theta(N^{-3}R^{-3})$.
\end{replemma}
\begin{proof}
We first bound the 1-norm distance between $\rho(t) = e^{\mathcal{L}t}(\rho_0)$ and $\Pi_\mathcal{S}\rho(t) \Pi_\mathcal{S}$ to show that as $t \to \infty$, due to the action of the penalizing jump operators, $\rho(t)$ evolves to be almost supported in the valid subspace $\mathcal{S}$. For this, we will use Lemma~\ref{lemma:encoding_num_error_bound} --- note that by definition of $\mathcal{S}$, a state $\ket{\psi} \in \mathcal{S}^\perp$ is orthogonal to the kernel of at least one of the projectors $P_j^\dagger P_j$ and thus $\bra{\psi}F\ket{\psi}\geq \bra{\psi}P_j^\dagger P_j \ket{\psi} = \bra{\psi}\psi\rangle$. Also, we observe that $F - \Pi_{S^\perp} \succeq 0$ since for any $\ket{\psi} \neq 0$,
\begin{align*}
    \bra{\psi}\big(F - \Pi_{S^\perp} \big) \ket{\psi} &= \bra{\psi}F\ket{\psi} - \bra{\psi} \Pi_{S^\perp} \ket{\psi}, \nonumber \\
    &\numeq{1}\bra{\psi}\Pi_{S^\perp} F \Pi_{S^\perp}\ket{\psi} - \bra{\psi} \Pi_{S^\perp} \ket{\psi}, \nonumber \\
    &\geq \bra{\psi} \Pi_{S^\perp} \Pi_{S^\perp} \ket{\psi} -  \bra{\psi} \Pi_{S^\perp} \ket{\psi} = 0,
\end{align*}
where, in (1) we have used the fact that, since $P_j \Pi_{\mathcal{S}} = 0 \ \forall \ j \in \mathcal{J}$, $F\Pi_\mathcal{S} = \Pi_{\mathcal{S}} F = 0 \implies \Pi_{\mathcal{S}^\perp} F\Pi_{\mathcal{S}^\perp} = F $. Using $F \succeq \Pi_{\mathcal{S}^\perp}$, we obtain that
\[
\textnormal{Tr}(\Pi_{\mathcal{S}^\perp}\rho(t)) \leq \textnormal{Tr}(F\rho(t)) \leq (NR + 1)20^{NR + 1}e^{-t/5}.
\]
Suppose for $\rho \in \text{D}_1((\mathbb{C}^6)^{\otimes NR})$, $\rho(t) = e^{\mathcal{L}t}(\rho)$, then
\begin{align}\label{eq:particle_num_bound_encoding}
    \norm{\rho(t) - \Pi_\mathcal{S}\rho(t) \Pi_{\mathcal{S}}}_1 &\leq \norm{\Pi_{\mathcal{S}^\perp}\rho(t) }_1 + \norm{\Pi_{\mathcal{S}}\rho(t)\Pi_{\mathcal{S}^\perp} }_1, \nonumber\\
    &\leq \sqrt{\textnormal{Tr}(\Pi_{\mathcal{S}^\perp} \rho(t))\textnormal{Tr}(\rho(t))} +\sqrt{\textnormal{Tr}(\Pi_{\mathcal{S}^\perp}\rho(t))\textnormal{Tr}(\Pi_{\mathcal{S}}\rho(t))}, \nonumber\\
    &\leq 2\sqrt{\textnormal{Tr}(\Pi_{\mathcal{S}^\perp}\rho(t))} \leq 2(NR + 1)^{1/2}20^{(NR + 1)/2}e^{-t/10},
\end{align}
Next, it is convenient to introduce the superoperators $\mathcal{P}_{\mathcal{S}}$ and $\mathcal{Q}_{\mathcal{S}}$ defined by
\[
\mathcal{P}_{\mathcal{S}}(X) = \Pi_{\mathcal{S}} X \Pi_{\mathcal{S}} \text{ and }\mathcal{Q}_{\mathcal{S}} = \text{id} - \mathcal{P}_{\mathcal{S}}.
\]
$\mathcal{P}_{\mathcal{S}}$ can be interpreted to project an operator onto another operator that is entirely supported on $\mathcal{S}$. Now, the master equation can be expressed as the following two equations:
\begin{subequations}\label{eq:coupled_ad_elim_eq}
\begin{align}
&\frac{d}{dt}\mathcal{P}_\mathcal{S}(\rho) = \mathcal{P}_\mathcal{S}\mathcal{L} \mathcal{P}_\mathcal{S}(\rho) + \mathcal{P}_\mathcal{S}\mathcal{L}\mathcal{Q}_\mathcal{S}(\rho), \\
&\frac{d}{dt}\mathcal{Q}_\mathcal{S}(\rho) = \mathcal{Q}_\mathcal{S} \mathcal{L} \mathcal{P}_\mathcal{S}(\rho) + \mathcal{Q}_\mathcal{S} \mathcal{L} \mathcal{Q}_\mathcal{S}(\rho),
\end{align}
\end{subequations}
Furthermore, since $\forall \ j \in \mathcal{J}: P_j \Pi_\mathcal{S} = 0$, it immediately follows that
\begin{align}\label{eq:penalty_operator}
\mathcal{D}_{P_j}  \mathcal{P}_\mathcal{S}(X) = P_j \Pi_\mathcal{S}X\Pi_\mathcal{S}P_j^\dagger - \frac{1}{2}\{P_j^\dagger P_j, \Pi_\mathcal{S}X\Pi_\mathcal{S}\} = 0.
\end{align}
Since $\forall \ \alpha \in \mathcal{A}, t \in \mathcal{T}: [L_t, \Pi_\mathcal{S}] = [S_\alpha, \Pi_\mathcal{S}] = 0$, $\mathcal{D}_{S_\alpha} \mathcal{P}_\mathcal{S} = \mathcal{P}_{\mathcal{S}}  \mathcal{D}_{S_\alpha}$, $\mathcal{D}_{L_t} \mathcal{P}_\mathcal{S} = \mathcal{P}_\mathcal{S} \mathcal{D}_{L_t}$. Thus, we obtain that
\begin{align*}
\mathcal{Q}_\mathcal{S}  \mathcal{L} \mathcal{P}_\mathcal{S} &= \sum_{\alpha \in \mathcal{A}} \mathcal{Q}_{\mathcal{S}} \mathcal{D}_{S_\alpha} \mathcal{P}_{\mathcal{S}} + \sum_{t \in \mathcal{T}} \mathcal{Q}_{\mathcal{S}} \mathcal{D}_{L_t} \mathcal{P}_{\mathcal{S}} +  \sum_{j\in \mathcal{J}}\mathcal{Q}_{\mathcal{S}} \mathcal{D}_{P_j} \mathcal{P}_{\mathcal{S}}, \nonumber \\
&=\sum_{\alpha \in \mathcal{A}} \mathcal{Q}_{\mathcal{S}} \mathcal{P}_{\mathcal{S}} \mathcal{D}_{S_\alpha} + \sum_{t \in \mathcal{T}}\mathcal{Q}_{\mathcal{S}} \mathcal{P}_{\mathcal{S}} \mathcal{D}_{L_t} = 0,
\end{align*}
where we have used that $\mathcal{Q}_{\mathcal{S}} \mathcal{P}_{\mathcal{S}} = 0$. Similarly, 
\begin{align*}
\norm{\mathcal{P}_\mathcal{S}  \mathcal{L} \mathcal{Q}_\mathcal{S}}_{1\to 1} &= \bignorm{\sum_{\alpha \in \mathcal{A}} \mathcal{P}_{\mathcal{S}} \mathcal{D}_{S_\alpha} \mathcal{Q}_{\mathcal{S}} + \sum_{t = 1}^T \mathcal{P}_{\mathcal{S}} \mathcal{D}_{L_t} \mathcal{Q}_{\mathcal{S}} +  \sum_{j \in \mathcal{J}}\mathcal{P}_{\mathcal{S}} \mathcal{D}_{P_j} \mathcal{Q}_{\mathcal{S}}}_{1\to 1}, \nonumber \\
&=\bignorm{\sum_{\alpha \in \mathcal{A}} \mathcal{P}_{\mathcal{S}} \mathcal{Q}_{\mathcal{S}} \mathcal{D}_{S_\alpha} + \sum_{t \in \mathcal{T}} \mathcal{P}_{\mathcal{S}} \mathcal{Q}_{\mathcal{S}} \mathcal{D}_{L_t} + \sum_{j \in \mathcal{J}}\mathcal{P}_\mathcal{S} \mathcal{D}_{P_j} \mathcal{Q}_{\mathcal{S}}}_{1\to 1}, \nonumber \\
&=\bignorm{\sum_{j \in \mathcal{J}}\big( \mathcal{P}_\mathcal{S} \mathcal{D}_{P_j} -  \mathcal{P}_\mathcal{S} \mathcal{D}_{P_j}  \mathcal{P}_\mathcal{S} \big)}_{1\to 1}, \nonumber \\
&\leq \sum_{j \in \mathcal{J}} \bigg( \norm{\mathcal{P}_\mathcal{S} \mathcal{D}_{P_j}}_{1\to 1} + \bignorm{\mathcal{P}_\mathcal{S} \mathcal{D}_{P_j}  \mathcal{P}_\mathcal{S} }_{1\to 1}\bigg), \\
&\numleq{1} 2\sum_{j \in \mathcal{J}} \norm{\mathcal{D}_{P_j}}_{1\to 1} \numleq{2} 4\sum_{j \in \mathcal{J}} \norm{P_j}^2 \leq 4\abs{\mathcal{J}},
\end{align*}
where in (1) we use that $\norm{\mathcal{P}_\mathcal{S}}_{1\to 1} \leq 1$ and in (2) we use that $\norm{\mathcal{D}_{P_j}}_{1 \to 1} \leq 2\norm{P_j}^2 = 2$. Finally, it also follows that $\mathcal{P}_\mathcal{S}\mathcal{D}_{L_t}\mathcal{P}_\mathcal{S} = \mathcal{D}_{\Pi_\mathcal{S}L_t \Pi_\mathcal{S}}$ and $\mathcal{P}_\mathcal{S}\mathcal{D}_{S_\alpha}\mathcal{P}_\mathcal{S} = \mathcal{D}_{\Pi_\mathcal{S}S_\alpha \Pi_\mathcal{S}}$ \ $\forall\ t\in \mathcal{T}, \alpha \in \mathcal{A}$. To see this, note that since any $A \in \{S_\alpha\}_{\alpha \in \mathcal{A}} \cup \{L_t\}_{t\in \mathcal{T}}$ commutes with $\Pi_{\mathcal{S}}$, 
\begin{align}\label{eq:commuting_with_PS}
\mathcal{P}_S\mathcal{D}_A\mathcal{P}_\mathcal{S}(X) &= \Pi_\mathcal{S}A \Pi_\mathcal{S}X\Pi_{\mathcal{S}} A^\dagger\Pi_{\mathcal{S}} - \frac{1}{2}\big(\Pi_{\mathcal{S}}A^\dagger A \Pi_{\mathcal{S}} X\Pi_{\mathcal{S}} + \Pi_{\mathcal{S}}A^\dagger A \Pi_{\mathcal{S}} X\Pi_{\mathcal{S}} \big), \nonumber \\
&=\Pi_\mathcal{S}A \Pi_\mathcal{S}X\Pi_{\mathcal{S}} A^\dagger \Pi_\mathcal{S}- \frac{1}{2}\big(\Pi_{\mathcal{S}}A^\dagger\Pi_{\mathcal{S}} A \Pi_{\mathcal{S}} X\Pi_{\mathcal{S}} + \Pi_{\mathcal{S}}A^\dagger \Pi_{\mathcal{S}} A \Pi_{\mathcal{S}} X\Pi_{\mathcal{S}} \big), \nonumber\\
&=\mathcal{D}_{\Pi_{\mathcal{S}} A \Pi_{\mathcal{S}}}(X).
\end{align}
From Eqs.~\ref{eq:penalty_operator} and \ref{eq:commuting_with_PS}, it then follows that
\[
\mathcal{L}_\mathcal{S}:= \mathcal{P}_{\mathcal{S}} \mathcal{L} \mathcal{P}_{\mathcal{S}} = \sum_{\alpha \in \mathcal{A}} \mathcal{P}_{\mathcal{S}} \mathcal{D}_{S_\alpha}  \mathcal{P}_{\mathcal{S}} + \sum_{t \in \mathcal{T}} \mathcal{P}_{\mathcal{S}} \mathcal{D}_{L_t}  \mathcal{P}_{\mathcal{S}} = \sum_{\alpha \in \mathcal{A}} \mathcal{D}_{\Pi_\mathcal{S} S_\alpha \Pi_{\mathcal{S}}} + \sum_{t\in \mathcal{T}} \mathcal{D}_{\Pi_\mathcal{S} L_t \Pi_{\mathcal{S}}}.
\]
Furthermore, $\mathcal{L}_\mathcal{S}$, which is a Lindbladian superoperator corresponding to the restriction of the full Linbladian $\mathcal{L}$ onto the valid subspace of states, is by construction unitarily equivalent to $\mathcal{L}_\text{ref}$ with each gate (including the identity gates) in Fig.~\ref{fig:encoding_circuit_format}b treated as an independent time-step.

Using Eq.~\ref{eq:coupled_ad_elim_eq}, we can now bound $\norm{e^{\mathcal{L}t}(\rho_0) - \sigma}_1$. From Eq.~\ref{eq:coupled_ad_elim_eq}a, we obtain that $\mathcal{Q}_\mathcal{S}(\rho(t)) = e^{\mathcal{Q}_\mathcal{S} \mathcal{L}\mathcal{Q}_\mathcal{S}}(\rho_0)$ and from Eq.~\ref{eq:coupled_ad_elim_eq}b, we obtain that
\begin{align}\label{eq:p_s_integrated}
\mathcal{P}_\mathcal{S}(\rho(t)) = e^{\mathcal{L}_\mathcal{S}t}\mathcal{P}_\mathcal{S}(\rho_0) + \int_0^t e^{\mathcal{P}_\mathcal{S}\mathcal{L}\mathcal{P}_\mathcal{S}(t - t')}\mathcal{P}_\mathcal{S} \mathcal{L}\mathcal{Q}_\mathcal{S}(\rho(t'))dt'.
\end{align}
    Now, since $\mathcal{L}_\mathcal{S} $ is unitarily equivalent to $\mathcal{L}_\text{ref}$, we can use Lemma~\ref{prop:logN_conv_bound} to obtain that for any $A \in \text{L}((\mathbb{C}^{6})^{\otimes NR}; \mathcal{S})$
    \[
    e^{\mathcal{L}_\mathcal{S}t}(A) = \text{Tr}(A) \sigma + \chi_t(A), \text{ where }\norm{\chi_t}_{1\to 1}\leq c_0(T, N) e^{-a_0(T) t},
    \]
    where $T = 2NR$ and $c_0(T, N), a_0(T)$ are defined in Lemma~\ref{prop:logN_conv_bound}. With this decomposition, we obtain from Eq.~\ref{eq:p_s_integrated} that
    \begin{align*}
    &\mathcal{P}_{\mathcal{S}}(\rho(t)) \nonumber \\
    &= \bigg[\text{Tr}(\mathcal{P}_S(\rho_0))  +  \int_0^t \textnormal{Tr}\big(\mathcal{P}_\mathcal{S} \mathcal{L} \mathcal{Q}_{\mathcal{S}}(\rho(t'))\big) dt'\bigg]\sigma + \chi_t(\mathcal{P}_\mathcal{S}(\rho_0)) + \int_0^t \chi_{t - t'}(\mathcal{P}_\mathcal{S} \mathcal{L} \mathcal{Q}_{\mathcal{S}}(\rho(t'))dt', \nonumber\\
    &\numeq{1} \bigg[\text{Tr}(\mathcal{P}_S(\rho_0))  -  \int_0^t \textnormal{Tr}\big(\mathcal{Q}_\mathcal{S} \mathcal{L} \mathcal{Q}_{\mathcal{S}}(\rho(t'))\big)  dt'\bigg]\sigma + \chi_t(\mathcal{P}_\mathcal{S}(\rho_0)) + \int_0^t \chi_{t - t'}(\mathcal{P}_\mathcal{S} \mathcal{L} \mathcal{Q}_{\mathcal{S}}(\rho(t'))dt', \nonumber\\
    &\numeq{2}\bigg[\text{Tr}(\mathcal{P}_S(\rho_0)) -  \int_0^t \textnormal{Tr}\bigg(\frac{d}{dt}\mathcal{Q}_\mathcal{S}(\rho(t'))\bigg)  dt'\bigg]\sigma + \chi_t(\mathcal{P}_\mathcal{S}(\rho_0)) + \int_0^t \chi_{t - t'}(\mathcal{P}_\mathcal{S} \mathcal{L} \mathcal{Q}_{\mathcal{S}}(\rho(t'))dt', \nonumber \\
    &= \text{Tr}(\rho_0)\sigma - \text{Tr}(\mathcal{Q}_\mathcal{S}(\rho(t))) + \chi_t(\mathcal{P}_{\mathcal{S}}(\rho_0)) + \int_0^t \chi_{t - t'}(\mathcal{P}_{\mathcal{S}} \mathcal{L} \mathcal{Q}_\mathcal{S}(\rho(t'))dt',
    \end{align*}
    where in (1) we have used the fact that $\text{Tr}(\mathcal{P}_\mathcal{S} \mathcal{L}(\cdot)) = \text{Tr}(\mathcal{L}(\cdot)) - \text{Tr}(\mathcal{Q}_\mathcal{S} \mathcal{L}(\cdot)) = - \text{Tr}(\mathcal{Q}_\mathcal{S} \mathcal{L}(\cdot))$, since every operator in the image of a Lindbladian is traceless and in (2), we have used Eq.~\ref{eq:coupled_ad_elim_eq}b. Noting that $\norm{\mathcal{P}_\mathcal{S}(\rho(0))}_1 = \norm{\Pi_\mathcal{S}\rho(0)\Pi_\mathcal{S}}_1 \leq 1$, from Eq~\ref{eq:particle_num_bound_encoding}, $\norm{\mathcal{Q}_\mathcal{S}(\rho(t))}_1 \leq 2(NR + 1)^{1/2}10^{NR + 1}e^{-t/5}$, and we obtain
    \begin{align*}
        &\norm{\mathcal{P}_\mathcal{S}(\rho(t)) - \text{Tr}(\rho_0)\sigma}_{1} \nonumber\\
        &\leq \norm{\mathcal{Q}_{\mathcal{S}}(\rho(t))}_1 + \norm{\chi_t(\mathcal{P}_{\mathcal{S}}(\rho(0)))}_{1} + \bignorm{\int_0^t \chi_{t - t'}(\mathcal{P}_{\mathcal{S}} \mathcal{L} \mathcal{Q}_\mathcal{S}(\rho(t'))dt'}_{1}, \nonumber \\
        &\leq 2({NR + 1})^{1/2}10^{NR + 1}e^{-t/5}+ \norm{\chi_t}_{1\to 1} \norm{\mathcal{P}_\mathcal{S}(\rho(0))}_1 + \int_0^t \norm{\chi_{t - t'}}_{1\to 1} \norm{\mathcal{P}_\mathcal{S}  \mathcal{L} \mathcal{Q}_\mathcal{S}}_{1\to 1} \norm{\mathcal{Q}_{\mathcal{S}}(\rho(t'))}_1 dt', \nonumber \\
        &\leq 2({NR + 1})^{1/2}10^{NR + 1}e^{-t/5} + \norm{\chi_t}_{1\to 1} \norm{\mathcal{P}_\mathcal{S}(\rho(0))}_1 + \int_0^t \norm{\chi_{t - t'}}_{1\to 1} \norm{\mathcal{P}_\mathcal{S}  \mathcal{L} \mathcal{Q}_\mathcal{S}}_{1\to 1} \norm{\mathcal{Q}_{\mathcal{S}}\rho(t')}_1 dt', \nonumber \\
        &\leq 2({NR + 1})^{1/2}10^{NR + 1}e^{-t/5} + c_0(T, N) e^{-a_0(T)t}\norm{\mathcal{P}_\mathcal{S}(\rho(0))}_1 + 8\abs{\mathcal{J}}({NR + 1})^{1/2}10^{NR + 1} c_0(T, N) \int_0^t e^{-a_0(T) (t - t')}e^{-t'/5}dt', \nonumber \\
        &\leq 2({NR + 1})^{1/2}10^{NR + 1}e^{-t/5} + c_0(T, N) e^{-a_0(T)t} + \frac{40\abs{\mathcal{J}} (NR + 1)^{1/2} 10^{NR + 1} c_0(T, N)}{1 - 5a_0(T)}  \big(e^{-t/5} - e^{-a_0(T)t}\big), \nonumber
    \end{align*}
Finally, we obtain from $\norm{e^{\mathcal{L}t}(\rho_0) - \text{Tr}(\rho_0)\sigma}_1 \leq \norm{\mathcal{Q}_\mathcal{S}(\rho(t))}_1 + \norm{\mathcal{P}_\mathcal{S}(\rho(t))}_1$ that
\[
\norm{e^{\mathcal{L}t}(\rho_0) - \text{Tr}(\rho_0)\sigma}_{1} \leq 4({NR + 1})^{1/2}10^{NR + 1}e^{-t/5} + c_0(T, N) e^{-a_0(T)t} + \frac{40\abs{\mathcal{J}} (NR + 1)^{1/2} 10^{NR + 1} c_0(T, N)}{1 - 5a_0(T)}  \big(e^{-t/5} - e^{-a_0(T)t}\big).
\]
From this bound, we obtain that as $N, R \to \infty$,
\[
\norm{e^{\mathcal{L}t}(\rho_0) - \text{Tr}(\rho_0)\sigma}_{1} \leq O(N^6 R^2 \exp(O(NR))) \exp\big(-\Omega(N^{-3}R^{-3}) t\big).
\]
Finally, we construct a local observable whose expected value is the expected value of the pauli $Z$-operator on the last qubit at the output of the quantum circuit --- this local observable will be $O = \ket{\bar{1}_{R, N}}\!\bra{\bar{1}_{R, N}} -  \ket{\bar{0}_{R, N}}\!\bra{\bar{0}_{R, N}}$ and it acts on the lower right-most qudit on the 2D grid. To see that this observable indeed measures the pauli $Z$-operator on the last qubit of the encoded circuit, recall from Lemma~\ref{prop:logN_conv_bound} that the fixed point of the Lindbladian contains a mixture of the (valid) state of the qudits corresponding to different time-steps of the encoded circuit. Note that in all of these time-steps except for the last one, the lower rightmost qudit is in the state $\ket{\times}$ or in an unbarred $0/$ state and thus yields a $0$ expected value for the observable $O$. Consequently, the expected value of $O$ in the fixed point only has a contribution from the state at the last time-step and is equal to the expected value of the pauli-$Z$ on the last qubit of the encoded circuit multiplied by a normalization factor of $1 / 2NR$, which arises from the fact that the fixed point contains a convex combination of the state of the qudits at all time-steps.
\end{proof}

\section{Stability to noise and errors}
\label{appendix:analysis_noisy}
\noindent As detailed in the main text, we will asssume that the quantum simulator due to noise, instead of implementing $\mathcal{L}_{\omega}$, will implement the Lindbladian $\mathcal{L}_{\omega, \delta}$ 
\[
\mathcal{L}_{\omega, \delta} = \mathcal{L}_\omega + \delta \sum_{\beta}\mathcal{N}_\beta,
\]
where $\mathcal{N}_\beta$, which itself is a Lindbladian, acts on the system qudits and ancilla qubits and is supported on ${S}'_\beta$ and $\delta$ is the noise rate. We will assume that there is a $\mathcal{Z}' \geq 0$ such that
\[
\abs{\{\beta' : S'_{\beta'}\cap S'_\beta \neq \emptyset \}} + \abs{\{\alpha : S'_{\beta'}\cap S_\alpha \neq \emptyset \}} \leq \mathcal{Z}' \ \text{for all }\beta.
\]
i.e.~${S}_\beta'$ intersects with at most $\mathcal{Z}'$ other subsets $S'_\beta$ or $\mathcal{S}_\alpha$ (on which the jump operators $L_\alpha$ and Hamiltonian terms $h_\alpha$ corresponding to the target Lindbladian are supported). In particular, this implies that any one ancilla qubit is in the support of at most $\mathcal{Z}'$ Lindbladians $\mathcal{N}_\beta$. We will denote by $\mathcal{E}_{\omega, \delta}(t, s) = \exp(\mathcal{L}_{\omega, \delta}(t - s))$ and $\rho_{\omega, \delta}(t) = \mathcal{E}_{\omega, \delta}(t, 0)\rho(0)$, where we will again assume the ancillae to initially be in the $\ket{0}$ state. We first modify the remainder expression given in Lemma~\ref{lemma:remainder} for the noiseless case.
\begin{replemma}{lemma:remainder_noisy}
If the ancillae are initially in $\ket{0}$ then 
\[
\frac{d}{dt}\textnormal{Tr}_\mathcal{A}(\rho_{\omega, \delta}(t)) - \omega^2 \mathcal{L} \textnormal{Tr}_\mathcal{A}(\rho_{\omega, \delta}(t)) = \mathcal{R}_{\omega}(t) + \delta \sum_{\beta} \mathcal{K}^{(0)}_{\beta}(t) + \delta \sum_{\alpha, \beta} \int_0^t e^{-2(t-s)} \left(\omega \mathcal{K}^{(1)}_{\alpha, \beta}(s) + \omega^2 \mathcal{K}^{(2)}_{\alpha, \beta}(s)\right) ds,
\]
where $\mathcal{R}_{\omega}(t)$ is as defined in Lemma~\ref{lemma:remainder} but with $\rho_{\omega} \to \rho_{\omega, \delta}$ and
\begin{align*}
    \mathcal{K}_{\beta}^{(0)}(t) &= \tr{\mathcal{A}}{\mathcal{N}_\beta(\rho_{\omega,\delta}(t)},
    \nonumber \\
    \mathcal{K}_{\alpha,\beta}^{(1)}(t) &= -i[L_\alpha^\dagger, \tr{\mathcal{A}}{\sigma_\alpha\mathcal{N}_\beta(\rho_{\omega,\delta}(t))}] + \textnormal{h.c.},
    \nonumber \\
    \mathcal{K}_{\alpha,\beta}^{(2)}(t) &= \frac{1}{2}[L_\alpha^\dagger,L_\alpha\tr{\mathcal{A}}{\mathcal{N}_\beta(\rho_{\omega,\delta}(t)}] + \textnormal{h.c.}.
\end{align*}
\end{replemma}
\begin{proof} We follow the method as in the proof of Lemma~\ref{lemma:remainder}. We begin by writing the equations of motion for $\tr{\mathcal{A}}{\rho_{\omega,\delta}(t)}$ and $\tr{\mathcal{A}}{\sigma_\alpha \rho_{\omega,\delta}(t)}$,
\begin{subequations}
\begin{align}
    \label{eq:dynamics_trA_rho_witherrors}
    & \frac{d}{dt}\tr{\mathcal{A}}{\rho_{\omega,\delta}(t)}
    = -\omega \sum_{\alpha}(i[L_\alpha^\dagger, \tr{\mathcal{A}}{\sigma_\alpha {\rho}_{\omega,\delta}(t)}]+\textnormal{h.c.})  -i\omega^2 [H_\textnormal{sys}, \tr{\mathcal{A}}{{\rho}_{\omega,\delta}(t)}] + \delta \sum_\beta \tr{\mathcal{A}}{\mathcal{N}_\beta(\rho_{\omega,\delta}(t))},
    \\
    \label{eq:dynamics_trA_sigma_rho_witherrors}
    & \frac{d}{dt}\tr{\mathcal{A}}{\sigma_\alpha \rho_{\omega,\delta}(t)}
    \nonumber \\
    & \quad = 
    - 2 \tr{\mathcal{A}}{\sigma_\alpha \rho_{\omega,\delta} (t)} - i\omega L_\alpha \tr{\mathcal{A}}{\rho_{\omega,\delta}(t)}
  - i\omega^2[H_\text{sys}, \tr{\mathcal{A}}{\sigma_\alpha\rho_{\omega,\delta}(t)}]+ \omega\sum_{\alpha'} {E_{\alpha,\alpha'}(t)} + \delta \sum_\beta \tr{\mathcal{A}}{\sigma_{\alpha}\mathcal{N}_\beta(\rho_{\omega,\delta}(t))}. 
\end{align}
\end{subequations}
We now consider the remainder $\mathcal{R}_{\omega,\delta}(t)$ --- using the definition of the remainder (Eq.~\ref{eq:remainder_def}) together with Eq.~\ref{eq:dynamics}a, we obtain that
\begin{align}\label{eq:simplified_remainder_witherrors}
    \mathcal{R}_{\omega,\delta}(t) & = \frac{d}{dt}\tr{\mathcal{A}}{\rho_{\omega,\delta}(t)} - \omega^2 \mathcal{L}\tr{\mathcal{A}}{\rho_{\omega,\delta}(t)}
    \nonumber \\
    & = \sum_{\alpha}\left(\left(-i\omega[L_\alpha^\dagger, \tr{\mathcal{A}}{\sigma_\alpha {\rho}_{\omega,\delta}(t)}] + \text{h.c}\right)- \omega^2 \mathcal{D}_{L_\alpha} \tr{A}{\rho_{\omega,\delta}(t)}\right) + \delta \sum_\beta \tr{\mathcal{A}}{\mathcal{N}_\beta(\rho_{\omega,\delta}(t))}.
\end{align}
Integrating Eq.~\ref{eq:dynamics_trA_sigma_rho_witherrors}, we find
\begin{align}
    & \tr{\mathcal{A}}{\sigma_\alpha \rho_{\omega,\delta}(t)}
    \nonumber \\
    & \quad = \int_0^t e^{-2(t - s)} \bigg(
    -i \omega^2 [H_\text{sys}, \tr{\mathcal{A}}{\sigma_\alpha\rho_\omega(s)}] -i \omega L_\alpha \tr{\mathcal{A}}{\rho_\omega(s)} + \omega\sum_{\alpha'} E_{\alpha, \alpha'}(s) + \delta \sum_\beta \tr{\mathcal{A}}{\sigma_\alpha\mathcal{N}_\beta(\rho_{\omega,\delta}(s))}\bigg)ds.
    \label{eq:trA_sigma_rho_integral_form_witherrors}
\end{align}
Using integration by parts and the expression for $\frac{d}{dt}\tr{\mathcal{A}}{\rho_{\omega,\delta}(t)}$ from Eq.~\ref{eq:dynamics_trA_rho_witherrors}, we obtain
\begin{align}
    & -i\omega \int_0^t e^{-2(t-s)} L_\alpha \tr{\mathcal{A}}{\rho_{\omega,\delta}(s)} ds = -i\frac{\omega}{2} L_\alpha \tr{\mathcal{A}}{\rho_{\omega,\delta}(t)} + i \frac{\omega}{2}e^{-2t} \rho(0) + i \frac{\omega}{2} \int_0^t e^{-2(t-s)}L_\alpha\frac{d}{ds}\tr{\mathcal{A}}{\rho_{\omega,\delta}(s)} ds,
    \nonumber \\
    & \quad = -i\frac{\omega}{2} L_\alpha \tr{\mathcal{A}}{\rho_{\omega,\delta}(t)} + i \frac{\omega}{2}e^{-2t} \tr{\mathcal{A}}{\rho(0)} - \frac{\omega^3}{2} \int_0^t e^{-2(t-s)} L_\alpha [H_\textnormal{sys}, \tr{\mathcal{A}}{\rho_{\omega,\delta}(s)}] ds + \frac{\omega^2}{2} \sum_{\alpha'} \int_0^t e^{-2(t-s)} F_{\alpha,\alpha'}(s) ds
    \nonumber \\
    & \quad \quad + i \frac{\omega \delta}{2} \sum_{\beta} \int_0^t e^{-2(t-s)} L_\alpha \tr{\mathcal{A}}{\mathcal{N}_\beta(\rho_{\omega,\delta}(s))} ds.
    \label{eq:trA_rho_byparts_witherrors}
\end{align}
Inserting Eq.~\ref{eq:trA_rho_byparts_witherrors} into Eq.~\ref{eq:trA_sigma_rho_integral_form_witherrors}, we obtain
\begin{align}
    \label{eq:Ldag_tr_sigma_rho_commutator_witherrors}
     -i\omega[L_\alpha^\dagger, \tr{\mathcal{A}}{\sigma_\alpha \rho_{\omega,\delta}(t)}] & = -\frac{\omega^2}{2}[L_\alpha^\dagger, L_\alpha \tr{\mathcal{A}}{\rho_{\omega,\delta}(t)}] +\frac{\omega^2}{2}e^{-2t}[L_\alpha^\dagger, L_\alpha \rho(0)]
     \nonumber \\
     & \quad + \int_0^t e^{2(t - s)}\bigg(-\omega^3 [L_\alpha^\dagger, [H_\textnormal{sys}, \tr{\mathcal{A}}{\sigma_\alpha\rho_{\omega,\delta}(t)} ] -i\frac{\omega^4}{2}[L_\alpha^\dagger, L_\alpha[H_\textnormal{sys}, \tr{\mathcal{A}}{\rho_{\omega,\delta}(s)}]]\bigg) ds
     \nonumber \\
     & \quad + \int_0^t e^{-2(t-s)} \left(- i \sum_{\alpha'}\bigg(\omega^2[L_\alpha^\dagger, E_{\alpha, \alpha'}(s)] + \frac{\omega^3}{2}[L_\alpha^\dagger, F_{\alpha, \alpha'}(s)]\right) ds
     \nonumber \\
     & \quad + \sum_\beta \int_0^t e^{-2(t-s)} \left( - i \omega \delta [L_\alpha^\dagger, \tr{\mathcal{A}}{\sigma_\alpha \mathcal{N}_\beta(\rho_{\omega,\delta}(s))}] +\frac{\omega^2\delta}{2} L_{\alpha} \tr{\mathcal{A}}{\mathcal{N}_\beta(\rho_{\omega,\delta}(s))} \right) ds,
\end{align}
Using Eq.~\ref{eq:Ldag_tr_sigma_rho_commutator_witherrors} and Eq.~\ref{eq:simplified_remainder_witherrors}, we note that $-[L_\alpha^\dagger,L_\alpha\tr{\mathcal{A}}{\rho_{\omega,\delta}(t)}]/2 + \text{h.c} = \mathcal{D}_{L_\alpha}$ and so arrive at
\begin{align}
     \mathcal{R}_{\omega,\delta}(t) & = \sum_{\alpha} \left(\omega^2 e^{-2t} q_{\alpha} +  \omega^4 \int_0^t e^{-2(t-s)} \left( {\mathcal{Q}^{(1)}_{\alpha,H_\textnormal{sys}}(s)} + \mathcal{Q}^{(2)}_{\alpha,H_\textnormal{sys}}(s) \right) ds \right) + \omega^4\sum_{\alpha, \alpha'} \int_0^t e^{-2(t - s)}\big({\mathcal{Q}^{(3)}_{\alpha, \alpha'}(s)} +  \mathcal{Q}^{(4)}_{\alpha, \alpha'}(s)\big) ds
     \nonumber \\
     & \quad +\omega^2 \delta \sum_\beta \mathcal{K}_\beta^{(0)}(t) + \delta \sum_{\alpha,\beta} \int_0^t e^{-2(t-s)} \left( \omega \mathcal{K}_{\alpha,\beta}^{(1)}(s) + \omega^2 \mathcal{K}_{\alpha,\beta}^{(2)}(s) \right),
\end{align}
where $q_\alpha$, $\mathcal{Q}_{\alpha,H_{sys}}^{(1)}(t)$, $\mathcal{Q}_{\alpha,H_{sys}}^{(2)}(t)$, $\mathcal{Q}_{\alpha,\alpha'}^{(3)}(t)$, and $\mathcal{Q}_{\alpha,\alpha'}^{(4)}(t)$ are defined in Lemma~\ref{lemma:remainder} (with the substitution $\rho_{\omega}(t)\to\rho_{\omega,\delta}(t)$) and $\mathcal{K}^{(0)}_\beta(t)$, $\mathcal{K}^{(1)}_{\alpha,\beta}(t)$, $\mathcal{K}^{(2)}_{\alpha,\beta}(t)$ are defined in the statement of this lemma.
\end{proof}
Next we provide bounds on the excitation of the ancillae in the presence of noise. The following lemma is the counterpart of Lemma~\ref{lemma:tr_sigma_bounds} in the noiseless case.
\begin{replemma}{lemma:tr_sigma_bounds_noisy}
Suppose $\rho_{\omega, \delta}(t)$ is the joint state of the system and the ancilla qubits with the ancilla qubits initially being in the state $\ket{0}$, then for all $\alpha, \alpha'$,
\begin{align*}
    &\norm{\sigma_\alpha \rho_{\omega, \delta}(t)}_1 \leq \frac{\omega}{2} + \mathcal{Z}'\delta  \text{ and }\norm{\tr{\mathcal{A}}{\sigma_\alpha \sigma_{\alpha'} \rho_{\omega, \delta}(t)}}_1, \norm{\tr{\mathcal{A}}{\sigma_\alpha^\dagger \sigma_{\alpha'} \rho_{\omega, \delta}(t)}}_1 \leq \frac{\omega^2}{4}  + \frac{\omega\mathcal{Z}'\delta}{2} + \mathcal{Z}'\delta.
\end{align*}
\end{replemma}
\begin{proof}
    The proof follows the same steps as the proof of Lemma~\ref{lemma:tr_sigma_bounds}. We will denote by $\vecket{0_{\mathcal{A}}}$ the vectorization of $\ket{0}\!\bra{0}$ on all the ancilla qubits. It will be convenient to note the following for $\sigma_{\alpha, l}(t) = \mathcal{E}_{\omega, \delta}^{-1}(t, 0)\sigma_{\alpha, l} \mathcal{E}_{\omega, \delta}(t, 0)$. We can now obtain
    \begin{align*}
    \frac{d}{dt}\sigma_{\alpha, l}(t) &= \mathcal{E}_{\omega, \delta}^{-1}(t,0) [\sigma_{\alpha, l}, \mathcal{L}_{\omega, \delta}]\mathcal{E}_{\omega, \delta}(t, 0), \\
    &= -2\sigma_{\alpha, l}(t) -i\omega \mathcal{E}_{\omega, \delta}^{-1}(t, 0)[\sigma_{\alpha, l}, \sigma_{\alpha, l}^\dagger]L_{\alpha, l}\mathcal{E}_{\omega, \delta}(t, 0) + \delta \mathcal{E}_{\omega, \delta}^{-1}(t, 0) [\sigma_{\alpha, l}, \mathcal{N}]\mathcal{E}_{\omega, \delta}(t, 0).
    \end{align*}
    This can be integrated to obtain
    \begin{align}
        \label{eq:vectorized_input_output_commutator_sigma}
        \sigma_{\alpha, l}\mathcal{E}_{\omega, \delta}(t, 0) = \mathcal{E}_{\omega, \delta}(t, 0)\sigma_{\alpha, l}e^{-2t} + \int_0^t e^{-2(t - s)}\mathcal{E}_{\omega, \delta}(t,s)\big(-i\omega[\sigma_{\alpha, l}, \sigma_{\alpha, l}^\dagger]L_{\alpha, l} + \delta [\sigma_{\alpha, l}, \mathcal{N}]\big)\mathcal{E}_{\omega, \delta}(s, 0) ds.
    \end{align}
Consider now upper bounding $\norm{\tr{\mathcal{A}}{\sigma_\alpha \rho_{\omega, \delta}(t)}}_1$. We can first upper bound
\[
\norm{\tr{\mathcal{A}}{\sigma_\alpha \rho_{\omega, \delta}(t)}}_1 \leq \norm{\vecbra{\text{Tr}_\mathcal{A}}\sigma_{\alpha, l}\mathcal{E}_{\omega, \delta}(t, 0)\vecket{0_{\mathcal{A}}}}_\diamond, 
\]
where $\vecbra{\text{Tr}_\mathcal{A}}\sigma_{\alpha, l}\mathcal{E}_{\omega, \delta}(t, 0)\vecket{0_{\mathcal{A}}}$ is interpreted as a superoperator on the system qudits. Using Eq.~\ref{eq:vectorized_input_output_commutator_sigma} and the fact that $\sigma_{\alpha, l}\vecket{0_\mathcal{A}} = 0$, we obtain that,
\begin{align*}
    &\bignorm{\vecbra{\text{Tr}_\mathcal{A}}\sigma_{\alpha, l}\mathcal{E}_{\omega, \delta}(t, 0)\vecket{0_{\mathcal{A}}}}_\diamond \nonumber\\
    &\leq \int_0^t e^{-2(t - s)}\bigg(\omega \norm{\vecbra{\text{Tr}_\mathcal{A}} \mathcal{E}_{\omega, \delta}(t, s)[\sigma_{\alpha, l}, \sigma_{\alpha, l}^\dagger]L_{\alpha, l}\mathcal{E}_{\omega, \delta}(s, 0) \vecket{0_\mathcal{A}}}_\diamond + \delta \bignorm{\vecbra{\text{Tr}_\mathcal{A}} \mathcal{E}_{\omega, \delta}(t, s)[\sigma_{\alpha, l},\mathcal{N} ]\mathcal{E}_{\omega, \delta}(s, 0)]\vecket{0_\mathcal{A}}}_\diamond \bigg)ds,\nonumber\\
    &\leq \int_0^t e^{-2(t - s)}\bigg(\omega \norm{[\sigma_{\alpha, l}, \sigma_{\alpha, l}^\dagger]}_\diamond \norm{L_{\alpha, l}}_\diamond + \delta \norm{[\sigma_{\alpha, l}, \mathcal{N}]}_\diamond\bigg) ds,\nonumber\\
    &\numleq{1} \int_0^t e^{-2(t - s)} \big(\omega + 2  \mathcal{Z}'\delta\big)ds = \frac{\omega}{2} +  \mathcal{Z}'\delta,
\end{align*}
where in (1) we have used the fact that at most only $\mathcal{Z}'$ terms $\mathcal{N}_\beta$ do not commute with $\sigma_{\alpha, l}$ i.e. $\abs{\{\beta : [\mathcal{N}_\beta, \sigma_{\alpha, l}] \neq 0\}} \leq  \mathcal{Z}'$.

We follow a similar method to derive a bound on $\tr{\mathcal{A}}{\sigma_\alpha^\dagger \sigma_{\alpha'} \rho_{\omega,\delta}(t)}$ and $\tr{\mathcal{A}}{\sigma_{\alpha}^\dagger \sigma_{\alpha'}\rho_{\omega,\delta}(t)}$. For $u\in\{-,+\}$, we write in vectorized notation,
\begin{equation}
    \label{eq:trA_sigmadag_sigma_rho_vectorized_witherrors}
    \tr{\mathcal{A}}{\sigma_\alpha^{(u)} \sigma_{\alpha'} \rho_{\omega,\delta}(t)} = \vecbra{\textnormal{Tr}_{\mathcal{A}}} \sigma_{\alpha,u}^{(u)} \sigma_{\alpha',l}\vecket{\rho_{\omega,\delta}(t)} = \vecbra{\textnormal{Tr}_{\mathcal{A}}} \mathcal{E}_{\omega,\delta}(t,0) \sigma_{\alpha,u}^{(u)}(t) \sigma_{\alpha',l}(t) \vecket{\rho(0)},
\end{equation}
where we interpret the subcripts as $-=l, +=r$. We have that
\begin{align}
    &\frac{d}{dt} (\sigma_{\alpha,u}^{(u)}(t)\sigma_{\alpha,l}(t)) 
     = [\sigma_{\alpha,u}^{(u)}(t)\sigma_{\alpha,l}(t), \mathcal{L}_{\omega,\delta}(t)],
    \nonumber \\
    & \quad = -4\sigma_{\alpha,u}^{(u)}(t)\sigma_{\alpha,l}(t) -i\omega L_{\alpha',l}(t)\sigma_{\alpha',l}^z\sigma_{\alpha,u}^{(u)}(t) + u i\omega L_{\alpha,u}^{(u)}(t)\sigma_{\alpha,u}^z(t)\sigma_{\alpha',l}(t) + \delta [\sigma_{\alpha,u}^{(u)}(t)\sigma_{\alpha',l}(t),\mathcal{N}(t)],
\end{align}
which can be integrated to obtain
\begin{align}
    & \sigma_{\alpha,u}^{(u)}(t)\sigma_{\alpha,l}(t) = e^{-4t}\sigma_{\alpha,u}^{(u)}\sigma_{\alpha',l} \nonumber\\
    & \quad + \int_0^t e^{-4(t-s)} \left( -i\omega L_{\alpha',l}(s)\sigma_{\alpha',l}^z(s)\sigma_{\alpha,u}^{(u)}(s) + u i\omega L_{\alpha,u}^{(u)}(s)\sigma_{\alpha,u}^z(s)\sigma_{\alpha',l}(s) + \delta [\sigma_{\alpha,u}^{(u)}(s)\sigma_{\alpha',l}(s),\mathcal{N}(s)] \right) ds.
    \label{eq:sigmadag_sigma_integral_form_witherrors}
\end{align}
Inserting Eq.~\ref{eq:sigmadag_sigma_integral_form_witherrors} into Eq.~\ref{eq:trA_sigmadag_sigma_rho_vectorized_witherrors} and applying the fact that $\sigma_{\alpha,l}\vecket{\rho(0)}=0$, we obtain
\begin{align}
    \norm{\tr{\mathcal{A}}{\sigma_{\alpha}^{(u)}\sigma_{\alpha'}\rho_{\omega,\delta}(t)}}_1 & \leq \int_0^t e^{-4(t-s)} \omega\norm{\vecbra{\textnormal{Tr}_{\mathcal{A}}}\mathcal{E}_{\omega,\delta}(t,0) L_{\alpha',l}(s) \sigma_{\alpha',l}^z(s) \sigma_{\alpha,u}^{(u)}(s)\vecket{\rho(0)}}_1 ds
    \nonumber \\
    & \quad + \int_0^t e^{-4(t-s)} \omega\norm{\vecbra{\textnormal{Tr}_{\mathcal{A}}}\mathcal{E}_{\omega,\delta}(t,0) L_{\alpha,u}^{(u)}(s) \sigma_{\alpha,u}^z(s) \sigma_{\alpha',l}(s)\vecket{\rho(0)}}_1 ds
    \nonumber \\
    & \quad + \int_0^t e^{-4(t-s)} \delta \norm{\vecbra{\textnormal{Tr}_{\mathcal{A}}}\mathcal{E}_{\omega,\delta}(t,0) [\sigma_{\alpha,u}^{(u)}(s)\sigma_{\alpha',l}(s),\mathcal{N}(s)] \vecket{\rho(0)}}_1 ds.
    \nonumber \\
    & \leq \int_0^t e^{-4(t-s)} \left( \omega \norm{\sigma_{\alpha,u}^{(u)}\vecket{\rho_{\omega,\delta}(s)}}_1 + \omega \norm{\sigma_{\alpha,l}\vecket{\rho_{\omega,\delta}(s)}}_1 + \delta \norm{[\sigma_{\alpha,u}^{(u)}\sigma_{\alpha',l},\mathcal{N}]}_\diamond \norm{\vecket{\rho(0)}}_1 \right) ds
    \nonumber \\
    & \numleq{1} \frac{\omega^2}{4} + \frac{\omega\mathcal{Z'}\delta}{2} + \mathcal{Z}'\delta.
\end{align}
In (1) we have used the bound $\norm{\sigma_\alpha \rho(t)}_1<\omega/2+\mathcal{Z}'\delta$ derived in the first part of this lemma to insert the bounds $\norm{\sigma_{\alpha,r}^\dagger(s)\vecket{\rho(0)}}_1, \norm{\sigma_{\alpha,l}(s)\vecket{\rho(0)}}_1 \leq \omega/2+\mathcal{Z}'\delta$. We have also used the fact that at most $2\mathcal{Z}'$ terms $\mathcal{N}_{\beta}$ do not commute with $\sigma_{\alpha}^{(u)}\sigma_{\alpha'}$. The lemma statement is confirmed choosing $u\in\{-,+\}$.
\end{proof}

\emph{Dynamics}. We now establish the counterpart of Lemma~\ref{lemma:bounds_remainder_lr} in the presence of errors. This requires additionally providing bounds on $\sum_{\beta}\mathcal{K}_\beta^{(0)}(t)$, $\sum_{\alpha, \beta} \mathcal{K}_{\alpha, \beta}^{(1)}(t)$ and $\sum_{\alpha, \beta} \mathcal{K}_{\alpha, \beta}^{(2)}(t)$. We begin by providing a counterpart to Lemma~\ref{lemma:bounds_remainder_lr_intermediate}.
\begin{lemma}\label{lemma:bounds_remainder_lr_noisy_intermediate}
    Suppose $O$ is a local observable with $\norm{O}\leq 1$ supported on $S_O$, and for $\tau>0$, let $O(\tau) = \exp(\mathcal{L}^\textnormal{\dagger} \tau)(O)$ where $\mathcal{L}$ is a geometrically local Lindbladian of the form in Eq.~\ref{eq:geom_local_lind}. Then for $q_{\alpha}, \mathcal{Q}_{\alpha,\alpha'}^{(j)}(s)$, with $\mathcal{Q}_{\alpha, \alpha'}^{(1)} =\mathcal{Q}_{\alpha, h_{\alpha'}}^{(1)}$ and $\mathcal{Q}_{\alpha, \alpha'}^{(2)} =\mathcal{Q}_{\alpha, h_{\alpha'}}^{(2)}$, as defined in Lemma~\ref{lemma:remainder_noisy},
        \[
        \bigabs{\tr{}{O(\tau) q_{\alpha,\delta}}} \leq 2  \min\bigg(\eta_{S_O} \exp\bigg(4e\mathcal{Z}\tau - \frac{d(S_\alpha, S_O)}{a}\bigg), 1\bigg) \text{ for any }\alpha,
        \]
        and,
        \[
        \abs{\tr{}{O(\tau)\mathcal{Q}^{(j)}_{\alpha, \alpha'}(s)}} \leq \left(4+\frac{8\mathcal{Z}'\delta}{\omega}+\frac{16\mathcal{Z}'\delta}{\omega^2}\right) \min\bigg(e \eta_{S_O}\exp\bigg(4e\mathcal{Z}\tau - \frac{1}{2a}\big( d(S_\alpha, S_O) + d(S_{\alpha'}, S_O)\big)\bigg), 1 \bigg) \text{ for any }\alpha, \alpha', j.
        \]
        Furthermore, 
        \[
        \bigabs{\textnormal{Tr}\big(O(\tau)\mathcal{K}_\beta^{(0)}(s)\big)} \leq \min\bigg(\eta_{S_O}\exp\bigg(4e\mathcal{Z}\tau - \frac{d(\tilde{S}'_\beta, S_O)}{a}\bigg), 1\bigg) \text{ for any }\beta,
        \]
        and
        \[
        \abs{\tr{}{O(\tau)\mathcal{K}^{(j)}_{\alpha,\beta}(s)}} \leq 2 \min\bigg(e\eta_{S_O}\exp\bigg(4e\mathcal{Z}\tau - \frac{1}{2a}\big(d(\tilde{S}'_\beta, S_O) + d({S}_\alpha, S_O) \big)\bigg), 1\bigg). \text{ for any }\alpha,\beta, j
        \]
    where $\tilde{S}'_\beta$ is the set of system qudits contained in $S'_\beta$.
\end{lemma}
\begin{proof}
    The bound on $\abs{\tr{}{O(\tau)q_{\alpha}}}$ is a restatement of the bound given and proven in Lemma~\ref{lemma:bounds_remainder_lr_intermediate}. The bounds on $\abs{\tr{}{O(\tau)\mathcal{Q}^{(j)}_{\alpha, \alpha'}(s)}}$ are proven in the same manner as the bounds on $\abs{\tr{}{O(\tau)\mathcal{Q}^{(j)}_{\alpha, \alpha'}(s)}}$ in Lemma~\ref{lemma:bounds_remainder_lr_intermediate}.
    
    First we consider $\text{Tr}(O(\tau) \mathcal{Q}_{ \alpha, \alpha'}^{(1)}(s))$. For any $\alpha, \alpha'$ we obtain that
    \begin{align*}
    \abs{\tr{}{O(\tau) \mathcal{Q}^{(1)}_{\alpha, \alpha'}(s)}} &\leq \frac{2}{\omega} \bigabs{\tr{}{O(\tau) [L_\alpha^\dagger, [h_{\alpha'}, \tr{\mathcal{A}}{(\sigma_\alpha \rho_{\omega,\delta}(s)}]]}}, \nonumber\\
    &\leq \frac{2}{\omega} \bignorm{[h_{\alpha'}, [L_\alpha^\dagger, O(\tau)]]} \bignorm{\tr{\mathcal{A}}{(\sigma_\alpha\rho_{\omega,\delta}(s)}}_1, \\
    &\leq \left( 1 + \frac{2\mathcal{Z}'\delta}{\omega}\right) \bignorm{[h_{\alpha'}, [L_\alpha^\dagger, O(\tau)]]},
    \end{align*}
    where, in the last step, we have used Lemma~\ref{lemma:tr_sigma_bounds_noisy}. Next, we can use Lemma~\ref{lemma:2superopLR} together with the fact that $\norm{[h_{\alpha'}, \cdot]}_{cb, \infty \to \infty},\norm{[L_{\alpha}^{(\pm)}, \cdot]}_{cb, \infty \to \infty}  \leq 2$ to obtain
    \begin{align*}
    &\abs{\tr{}{O(\tau)\mathcal{Q}^{(1)}_{\alpha, h_{\alpha'}}(s)}}  \leq \left(4 + \frac{8\mathcal{Z}'\delta}{\omega} \right) \norm{O}\min\bigg(e \eta_{S_O}\exp\bigg(4e\mathcal{Z}\tau - \frac{1}{2a}\big( d(S_\alpha, S_O) + d(S_{\alpha'}, S_O)\big)\bigg), 1 \bigg).
    \end{align*}
    Next, for any $\alpha, \alpha'$, we obtain that,
    \begin{align*}
        \bigabs{\tr{}{O(\tau)\mathcal{Q}^{(2)}_{ \alpha, {\alpha'}}(s)}}
        & \leq \bigabs{\tr{}{O(\tau) [L_\alpha^\dagger, L_\alpha [h_{\alpha'}, \rho_{\omega,\delta}(s)]]}},
        \nonumber \\
        & \leq \norm{[h_{\alpha'},[L_\alpha^\dagger, O(\tau)]L_\alpha]} \norm{\tr{A}{\rho_{\omega,\delta}(s)}}_1,
        \nonumber \\
        & \leq 4 \norm{O} \min\left(e\eta_{S_O} \exp\left(4e\mathcal{Z}\tau - \frac{1}{2a}( d(S_\alpha,S_O) + d(S_{\alpha'},S_O) )\right), 1 \right).
    \end{align*}
    Next, we bound $\abs{\tr{\mathcal{A}}{\mathcal{Q}^{(3)}_{ \alpha, \alpha'}}}$. For any $\alpha, \alpha'$, we obtain that
    \begin{align*}
    \bigabs{\tr{}{O(\tau)\mathcal{Q}^{(3)}_{ \alpha, \alpha'}(s)}} &\leq \frac{1}{\omega} \sum_{u \in \{+, -\}} \abs{\tr{}{O(\tau)\mathcal{D}_{L_\alpha}\left([L_{\alpha'}^{(u)}, \tr{\mathcal{A}}{\sigma_{\alpha'}^{(\bar u)}\rho_{\delta,\omega}(s)}]\right)}}, \nonumber\\
    &\leq \frac{1}{\omega} \sum_{u \in \{+, -\}} \bignorm{[L_{\alpha'}^{(u)}, \mathcal{D}_{L_\alpha}^\dagger(O(\tau))]} \bignorm{\tr{\mathcal{A}}{\sigma_{\alpha'}^{(\bar u)} \rho_{\delta,\omega}(s)}}_1, \nonumber \\
    &\leq \left(\frac{1}{2}+\frac{\mathcal{Z}'\delta}{\omega}\right) \sum_{u \in \{+, -\}}\bignorm{[L_{\alpha'}^{(u)}, \mathcal{D}_{L_\alpha}^\dagger(O(\tau))]},
    \end{align*}
    where, in the last step, we have used Lemma~\ref{lemma:tr_sigma_bounds_noisy}. Next, we use Lemma~\ref{lemma:2superopLR} together with $\norm{[L_{\alpha}^{(u)}, \cdot]}_{cb,\infty\to\infty} \leq 2$ and $\norm{\mathcal{D}_{L_\alpha}^\dagger}_{cb,\infty\to\infty} \leq 2$ to obtain
    \[
    \bigabs{\tr{}{O(\tau)\mathcal{Q}^{(3)}_{\alpha, \alpha'}(s)}} \leq \left(4+\frac{8\mathcal{Z}'\delta}{\omega}\right) \norm{O}\min\bigg(e\eta_{S_O} \exp\bigg(4e\mathcal{Z}\tau - \frac{1}{2a}\big(d(S_\alpha, S_O) + d(S_{\alpha'}, S_O)\big)\bigg), 1\bigg). 
    \]
    Next we bound $\abs{\tr{}{O(\tau)\mathcal{Q}^{(4)}_{ \alpha,\alpha'}(s)}}_1$. For $\alpha \neq \alpha'$
    \begin{align}
    \abs{\tr{}{O(\tau) \mathcal{Q}^{(4)}_{ \alpha, \alpha'}(s)}} &\leq \frac{1}{\omega^2}\sum_{u, u' \in \{-, +\}}\abs{ \tr{}{O(\tau) [L_\alpha^{(u)}, [L_{\alpha'}^{(u')}, \tr{\mathcal{A}}{\sigma_\alpha^{(\bar{u})} \sigma_{\alpha'}^{(\bar{u}')} \rho_{\omega,\delta}(s)}]]}}, \nonumber\\
    &\leq \frac{1}{\omega^2} \sum_{u, u' \in \{-, +\}} \bignorm{[L_{\alpha'}^{({u}')}, [L_\alpha^{(u)}, O(\tau)]]}\bignorm{\tr{\mathcal{A}}{\sigma_\alpha^{(\bar{u})} \sigma_{\alpha'}^{(\bar{u}')} \rho_{\omega,\delta}(s)}}_1, \nonumber\\
    &\leq \left(\frac{1}{4} + \frac{\mathcal{Z}'\delta}{2\omega} +\frac{\mathcal{Z}'\delta}{\omega^2}\right)\sum_{u, u' \in \{-, +\}} \bignorm{[L_{\alpha'}^{({u}')}, [L_\alpha^{(u)}, O(\tau)]]}\nonumber,
    \end{align}
    where in the last step we have used Lemma~\ref{lemma:tr_sigma_bounds_noisy}. Furthermore, from Lemma~\ref{lemma:2superopLR} and the fact that $\norm{[L_\alpha^{(u)}, \cdot]}_{cb,\infty\to\infty} \leq 2$ it follows that
    \begin{align}\label{eq:remainder_lr_bound_alphas_unequal_noisy}
    \bigabs{\tr{}{O(\tau) \mathcal{Q}^{(4)}_{\alpha, \alpha'}(s)}} \leq \left(4 + \frac{8\mathcal{Z}'\delta}{\omega} + \frac{16\mathcal{Z}'\delta}{\omega^2}\right) \norm{O} \min\bigg(e \eta_{S_O}\exp\bigg(4e\mathcal{Z}\tau - \frac{1}{2a}\big( d(S_\alpha, S_O) + d(S_{\alpha'}, S_O)\big)\bigg), 1 \bigg).
    \end{align}
    Similarly, for $\alpha = \alpha'$, 
    \begin{align*}
        \bigabs{\tr{}{O(\tau) \mathcal{Q}^{(4)}_{\alpha, \alpha}(s)}} &\leq \frac{2}{\omega^2} \bigabs{\tr{}{O(\tau) (\mathcal{D}_{L_\alpha} - \mathcal{D}_{L_\alpha^\dagger})(\tr{\mathcal{A}}{n_\alpha \rho_{\omega,\delta}(s)})}}, \nonumber \\
        &\leq \frac{2}{\omega^2} \norm{(\mathcal{D}_{L_\alpha} - \mathcal{D}_{L_\alpha^\dagger})^\dagger \big(O(\tau)\big)} \norm{\tr{\mathcal{A}}{n_\alpha \rho_{\omega,\delta}(s)}}_1, \nonumber\\
        &\leq \left(\frac{1}{2} + \frac{\mathcal{Z}'\delta}{\omega} + \frac{2\mathcal{Z}'\delta}{\omega^2}\right)\norm{(\mathcal{D}_{L_\alpha} - \mathcal{D}_{L_\alpha^\dagger})^\dagger \big(O(\tau)\big)},
    \end{align*}
    where, again, in the last step, we have used Lemma~\ref{lemma:tr_sigma_bounds_noisy}. Next, we can use Lemma~\ref{lemma:lieb_robinson} together with the fact that $\norm{(\mathcal{D}_{L_\alpha} - \mathcal{D}_{L_\alpha^\dagger})^\dagger}_{cb,\infty\to\infty} \leq 4$, we obtain that
    \begin{align}\label{eq:remainder_lr_bound_alphas_equal_noisy}
    \abs{\tr{}{O(\tau)\mathcal{Q}^{(4)}_{\alpha, \alpha}(s)}} &\leq \left(2 + \frac{4\mathcal{Z}'\delta}{\omega} + \frac{8\mathcal{Z}'\delta}{\omega^2}\right) \norm{O} \min\bigg(\eta_{S_O}\exp\bigg(4e\mathcal{Z}\tau - \frac{1}{a}d(S_\alpha, S_O)\bigg), 1\bigg).
    \end{align}
    Eqs.~\ref{eq:remainder_lr_bound_alphas_equal_noisy} and \ref{eq:remainder_lr_bound_alphas_unequal_noisy} together establish establish a bound on $\abs{\tr{}{O(\tau)\mathcal{Q}^{(4)}_{\alpha, \alpha'}(s)}}$ for any $\alpha,\alpha'$ and is consistent with the lemma statement.

    The bound on $\abs{\text{Tr}(O(\tau) \mathcal{K}^{(0)}_\beta(s))}$ is proven in the main text. Next we bound $\abs{\tr{}{O(\tau)\mathcal{K}^{(1)}_{\alpha,\beta}}}$. For any $\alpha, \beta$,
    \begin{align}
        \abs{\tr{}{O(\tau)\mathcal{K}^{(1)}_{\alpha,\beta}(s)}}
        & \leq 2\abs{\tr{}{O(\tau)[L_{\alpha}^\dagger, \sigma_\alpha \mathcal{N}_\beta( \rho_{\omega,\delta}(s))]}}
        \nonumber \\
        & \leq 2\norm{\tilde{\mathcal{N}}_\beta^\dagger([L_\alpha^\dagger,O(\tau)]\sigma_\alpha)} \norm{\rho_{\omega,\delta}(s)}_1
        \nonumber \\
        & \leq 2 \norm{O} \min\bigg(e\eta_{S_O}\exp\bigg(4e\mathcal{Z}\tau - \frac{1}{2a}\big(d(\tilde{S}'_\beta, S_O) + d({S}_\alpha, S_O) \big)\bigg), 1\bigg),
    \end{align}
    where $\tilde{\mathcal{N}}_\beta(X)=\mathcal{N}_\beta(X\otimes I_\mathcal{A})$ and $\tilde{\mathcal{N}}_\beta$ is supported only on the system qudits in $\tilde{S}'_\beta$.
    Finally, we bound $\abs{\tr{}{O(\tau)\mathcal{K}^{(2)}_{\alpha,\beta}}}$. For any $\alpha, \beta$,
    \begin{align}
        \abs{\tr{}{O(\tau)\mathcal{K}^{(2)}_{\alpha,\beta}}}
        & \leq \abs{\tr{}{O(\tau)[L_\alpha^\dagger, L_\alpha \mathcal{N}_\beta(\rho_{\omega,\delta}(s)]}}
        \nonumber \\
        & \leq \norm{\tilde{\mathcal{N}}_\beta^\dagger\left([L^\dagger_\alpha,O(\tau)]L_\alpha\right)} \norm{\rho_{\omega,\delta}(s)}_1
        \nonumber \\
        & \leq 2 \norm{O} \min\bigg(e\eta_{S_O}\exp\bigg(4e\mathcal{Z}\tau - \frac{1}{2a}\big(d(\tilde{S}'_\beta, S_O) + d({S}_\alpha, S_O) \big)\bigg), 1\bigg),
    \end{align}
    producing a bound identical to that for $\abs{\tr{}{O(\tau)\mathcal{K}^{(1)}_{\alpha,\beta}}}$.
\end{proof}
\begin{replemma}{lemma:bounds_remainder_noisy_lr}
Suppose $O$ is a local observable on the system qudits with $\norm{O} \leq 1$ supported on $S_O$, and for $\tau>0$, let $O(\tau) = \exp(\mathcal{L}^\textnormal{\dagger} \tau)(O)$ where $\mathcal{L}$ is a geometrically local target Lindbladian of the form in Eq.~\ref{eq:geom_local_lind}. Then for\ $q_\alpha$ as defined in Lemma~\ref{lemma:remainder}, then there are non-decreasing piecewise continuous function $\nu, \nu'$ such that $\nu(t), \nu'(t) \leq O(t^d)$ as $t \to \infty$ and for $\omega \leq 2$
        \begin{align*}
         \sum_{\alpha}\bigabs{\tr{}{O(\tau) q_\alpha}} \leq  \nu(\tau) \text{ and }\sum_{\alpha, \alpha'}\bigabs{\tr{}{O(\tau)\mathcal{Q}_{\alpha, \alpha'}^{(j)}(s)}} \leq  \bigg(1 + \frac{2\mathcal{Z}' \delta}{\omega} + \frac{4\delta \mathcal{Z}'}{\omega^2}\bigg)\nu^2(\tau).
        \end{align*}
where, for $j \in \{3, 4\}$, $\mathcal{Q}_{\alpha, \alpha'}^{(j)}$ is defined in Lemma~\ref{lemma:remainder} and for $j\in \{1, 2\}$, we define $\mathcal{Q}_{\alpha, \alpha'}^{(j)} = \mathcal{Q}_{\alpha, h_{\alpha'}}^{(j)}$ where $\mathcal{Q}_{\alpha, h}^{(j)}$ is defined in Lemma~\ref{lemma:remainder}. Furthermore,
\begin{align*}
    \sum_{\beta}\bigabs{\textnormal{Tr}(O(\tau)\mathcal{K}_{\beta}^{(0)}(s)} \leq \nu'(\tau) & \text{ and for }j\in \{1, 2\}, \ \sum_{\alpha, \beta}\bigabs{\textnormal{Tr}(O(\tau)\mathcal{K}_{\beta}^{(j)}(s)} \leq (\nu'(\tau))^2,
\end{align*}
where $\mathcal{K}^{(0)}_\beta$ and $\mathcal{K}^{(j)}_{\alpha, \beta}$ for $j \in \{1,2\}$ are defined in Lemma~\ref{lemma:remainder_noisy}.
\end{replemma}
\begin{proof}
    We prove this lemma in the same manner as Lemma~\ref{lemma:bounds_remainder_lr}, using Lemmas~\ref{lemma:lr_summation} and \ref{lemma:bounds_remainder_lr_noisy_intermediate}. Consider first,
    \begin{align}\label{eq:function_bound_single_sum_noisy}
        \sum_{\alpha}\abs{\text{Tr}(O(\tau) q_\alpha)}
        & \leq 2 \sum_{\alpha} \min\bigg(\eta_{S_O} \exp\bigg(4e\mathcal{Z}\tau - \frac{d(S_\alpha, S_O)}{a}\bigg), 1\bigg), \nonumber\\
        & = 2\xi^{(0, 1)}_{a, \eta_{S_O}}(4e\mathcal{Z}\tau) \leq 2 \max(\eta_{S_O}, 1)\nu^{(0)}(a, 4e\mathcal{Z}\tau) \leq 2e\max(\eta_{S_O}, 1) \nu^{(0)}(2a, 4e\mathcal{Z}\tau),
    \end{align}
    where in the last step, we have used the fact that, as per Lemma~\ref{lemma:lr_summation}, $\nu^{(0)}(\lambda, T)$ is a non-decreasing function of $\lambda$ for a fixed $T$.

    Similarly, for $j \in \{1, 2, 3, 4\}$,
    \begin{align}\label{eq:function_bound_double_sum_noisy}
        &\sum_{\alpha_1, \alpha_2}\abs{\text{Tr}(O(\tau) \mathcal{Q}_{\alpha, \alpha'}^{(j)}(s)}, \nonumber\\
        &\qquad \leq \left( 4 + 8\omega\mathcal{Z}'\delta + 16\mathcal{Z}'\delta \right) \sum_{\alpha, \alpha'} \min\bigg(e\eta_{S_O} \exp\bigg(4e\mathcal{Z}\tau - \frac{d(S_\alpha, S_O) + d(S_{\alpha'}, S_O)}{2a}\bigg), 1\bigg),\nonumber\\
        &\qquad = \left( 4 + 8\omega\mathcal{Z}'\delta + 16\mathcal{Z}'\delta \right) \xi_{2a, e\eta_{S_O}}^{(0, 2)}(2a, 4e\mathcal{Z}\tau) \leq \left( 4 + 8\omega\mathcal{Z}'\delta + 16\mathcal{Z}'\delta \right) \big(\max(e\eta_{S_O}, 1)\big)^2 \left(\nu^{(0)}(2a, 4e\mathcal{Z}\tau)\right)^2
        \nonumber \\
        & \qquad \leq \left( 4 + 8\omega\mathcal{Z}'\delta + 16\mathcal{Z}'\delta \right) e^2 \big(\max(\eta_{S_O}, 1)\big)^2 \big(\nu^{(0)}(2a, 4e\mathcal{Z}\tau)\big)^2.
     \end{align}
     From Eqs.~\ref{eq:function_bound_single_sum_noisy} and \ref{eq:function_bound_double_sum_noisy}, it follows that choosing $\nu(\tau) = 2e \max(\eta_{S_O}, 1) \nu^{(0)}(2a, 4e\mathcal{Z}\tau)$ satisfies the lemma statement.

     Next we apply Lemmas~\ref{lemma:bounds_remainder_lr_noisy_intermediate} and \ref{lemma:lr_summation} to obtain
     \begin{align}
         \sum_\beta\bigabs{\tr{}{O(\tau)\mathcal{K}_\beta^{(0)}(s)}} & \leq \min\left(\eta_{S_O}\exp\left(4e\mathcal{Z}\tau - \frac{d(S_O,\tilde{S}'_\beta)}{a}\right),1\right)
         \nonumber \\
         & = \xi^{(0,1)}_{a, \eta_{S_O}}(4e\mathcal{Z}\tau) \leq \max(\eta_{S_O},1)\nu^{(0)}(a,4e\mathcal{Z}\tau)
         \nonumber \\
         & \leq \max(\eta_{S_O},1)\nu^{(0)}(2a,4e\mathcal{Z}\tau),
         \label{eq:single_sum_K_bound}
     \end{align}
     where in the last step we have used the fact that, as per Lemma~\ref{lemma:lr_summation}, $\nu^{(0)}(\lambda,T)$ is a non-decreasing function of $\lambda$ for fixed $T$.

     For $j\in\{1,2\}$,
     \begin{align}
         \sum_{\alpha,\beta} \bigabs{\tr{}{O(\tau)\mathcal{K}_{\beta}^{(j)}(s)}} & \leq 2\sum_{\alpha,\beta}\min\left(\eta_{S_O}\exp\left(4e\mathcal{Z}\tau - \frac{d(\tilde{S}'_\beta,S_O)+d(S_\alpha,S_O)}{2a}\right),1\right)
         \nonumber \\
         &  = \xi_{2a,e\eta_{S_O}}(4e\mathcal{Z}) \leq 2(\max(e\eta_{S_O},1))^2\left(\nu^{(0)}(2a,4e\mathcal{Z}\tau)\right)^2
         \nonumber \\
         & \leq 2 e^2 (\max(\eta_{S_O},1))^2\left(\nu^{(0)}(2a,4e\mathcal{Z}\tau)\right)^2.
         \label{eq:double_sum_K_bound}
     \end{align}
    Considering Eqs.~\ref{eq:single_sum_K_bound} and \ref{eq:double_sum_K_bound}, we may choose $\nu'(\tau)=\sqrt{2}e\max(\eta_{S_O},1)\nu^{(0)}(2a,4e\mathcal{Z}\tau)$ to satisfy the lemma statement.
    It can also be noted from the asymptotics of $\nu^{(0)}(\lambda, T) $ in Lemma~\ref{lemma:lr_summation} that both $\nu(\tau),\nu'(\tau) \leq O(\tau^d)$. Additionally, since $\nu^{(0)}(\lambda, T)$, for a fixed $\lambda$, is a non-decreasing and piecewise continuous function of $T$, $\nu(\tau)$ and $\nu'(\tau)$ are also piecewise continuous non-decreasing functions of $\tau$.
\end{proof}

\emph{Long-time dynamics or fixed points}. Next, we consider the problem of long-time dynamics or fixed points.
\begin{replemma}{lemma:error_term_rapid_mixing_noisy}
    Suppose $O$ is a local observable with $\norm{O} \leq 1$ supported on $S_O$, and for $\tau>0$, let $O(\tau) = \exp(\mathcal{L}^\textnormal{\dagger} \tau)(O)$ where $\mathcal{L}$ is a geometrically local Lindbladian of the form in Eq.~\ref{eq:geom_local_lind}. Furthermore, suppose $O$ is rapidly mixing with respect to $\mathcal{L}$ and satisfies Eq.~\ref{eq:rapid_mixing_observable_v2} with $k(\abs{S_O}, \gamma) \leq O(\exp(\gamma^{-\kappa}))$. Then for\ $q_\alpha$ as defined in Lemma~\ref{lemma:remainder},
        \[
        \sum_{\alpha}\bigabs{\int_0^t \tr{}{O(t - s) q_\alpha}e^{-2s/\omega^2}ds} \leq  \omega^2 \lambda^{(1)}(\gamma),
        \]
        where $\lambda^{(1)}(\gamma) \leq O(\gamma^{-d(\kappa + 1)})$ as $\gamma \to 0$ and for $j \in \{1,2,3,4\}$
        \begin{align*}
        &\sum_{\alpha, \alpha'}\bigabs{\int_0^t \int_0^{s/\omega^2} \tr{}{O(t - s)\mathcal{Q}_{\alpha, \alpha'}^{(j)}(s')} e^{-2(s/\omega^2 - s')}ds' ds} \leq  \bigg(1 + \frac{2\mathcal{Z}'\delta}{\omega} + \frac{4\mathcal{Z}'\delta}{\omega^2}\bigg)\lambda^{(2)}(\gamma),
        \end{align*}
where $\lambda^{(2)}(\gamma) \leq O(\gamma^{-(2d + 1)(\kappa + 1)})$ as $\gamma \to 0$ and for $j \in \{3, 4\}$, $\mathcal{Q}_{\alpha, \alpha'}^{(j)}$ is defined in Lemma~\ref{lemma:remainder} but with $\rho_\omega \to \rho_{\omega, \delta}$ and for $j\in \{1, 2\}$, we define $\mathcal{Q}_{\alpha, \alpha'}^{(j)} = \mathcal{Q}_{\alpha, h_{\alpha'}}^{(j)}$ where $\mathcal{Q}_{\alpha, h}^{(j)}$ is defined in Lemma~\ref{lemma:remainder} but with $\rho_\omega \to \rho_{\omega, \delta}$. Furthermore,
\begin{align*}
\sum_{\beta}\bigabs{\int_0^t \textnormal{Tr}\bigg(O(t - s)\mathcal{K}_\beta^{(0)}\bigg(\frac{s}{\omega^2}\bigg)\bigg)ds} \leq {\lambda'}^{(1)}(\gamma)
\end{align*}
where ${\lambda'}^{(1)}(\gamma) \leq O(\gamma^{-(d + 1)(\kappa + 1)})$ as $\gamma \to 0$ and for $j \in \{1,2\}$,
\begin{align*}
&\sum_{\alpha, \beta}\bigabs{\int_0^t \int_0^{s/\omega^2} \textnormal{Tr}\big(O(t - s) \mathcal{K}_{\alpha,\beta}^{(j)}(s') \big) e^{-2(s/\omega^2 - s')}ds' ds} \leq {\lambda'}^{(2)}(\gamma),
\end{align*}
where ${\lambda'}^{(2)}(\gamma) \leq O(\gamma^{-(2d + 1)(\kappa + 1)})$ as $\gamma \to 0$.

\end{replemma}
\begin{proof}
    We follow the same method as the proof of Lemma~\ref{lemma:error_term_rapid_mixing}. Since $q_\alpha$ is defined identically in the noisy and noiseless models, the first statement of the lemma is simply a restatement of the bound from Lemma~\ref{lemma:error_term_rapid_mixing}.
    Next we note that, for a function $\mathcal{F}_{\alpha,\alpha'}(t)$,
    \begin{align}\label{eq:change_of_variable_restate}
    \sum_{\alpha, \alpha'}\bigabs{\int_0^t \int_0^{s/\omega^2} \tr{}{O(t - s)\mathcal{F}_{\alpha, \alpha'}(s')} e^{-2(s/\omega^2 - s')}ds' ds} = \frac{1}{\omega^2}\sum_{\alpha, \alpha'}\bigabs{\int_0^t \int_0^{s} \text{Tr}\bigg({O(t - s)\mathcal{F}_{\alpha, \alpha'}\bigg(\frac{s'}{\omega^2}\bigg)\bigg)} e^{-2(s - s')/\omega^2}ds' ds}.
    \end{align}
    We now consider the terms involving $Q^{(j)}$ for $j \in \{1, 2, 3, 4\}$ separately --- for $j = 1$, we have from Lemma~\ref{lemma:remainder_noisy} that $\mathcal{Q}_{\alpha, \alpha'}^{(1)}(s'/\omega^2) = \mathcal{K}_{\alpha}\mathcal{J}_{\alpha'}(\sigma_\omega(s')) +\textnormal{h.c.}$ with $\mathcal{K}_\alpha(X) = [L_\alpha^\dagger, X]$, $\mathcal{J}_{\alpha'}(X) = [h_{\alpha'}, X]$ and $\sigma_\omega(s') = -\omega^{-1}\tr{\mathcal{A}}{\sigma_\alpha \rho_{\omega,\delta}(s'/\omega^2)}$. We note that $\norm{\mathcal{K}_\alpha}_{\diamond}, \norm{\mathcal{J}_{\alpha'}}_\diamond \leq 2$ and, using Lemma~\ref{lemma:tr_sigma_bounds_noisy}, $\norm{\sigma_\omega(s')}_1 \leq 1/2+\mathcal{Z}'\delta/\omega$. Thus, applying Eq.~\ref{eq:change_of_variable_restate} and Lemma~\ref{lemma:local_rapid_mixing_superop_bounds}, we obtain that
    \begin{align}
    \sum_{\alpha, \alpha'}\bigabs{\int_0^t \int_0^{s/\omega^2} \tr{}{O(t - s)\mathcal{Q}_{\alpha, \alpha'}^{(1)}(s')} e^{-2(s/\omega^2 - s')}ds' ds} \leq \left(1+\frac{2\mathcal{Z'}\delta}{\omega}\right)\zeta^{(3)}(\gamma).
    \label{eq:Q1_lrm_bound_noisy}
    \end{align}
    Similarly, for $j = 2$, we have from Lemma~\ref{lemma:remainder_noisy} that $\mathcal{Q}_{\alpha, \alpha'}^{(2)}(s'/\omega^2) = \mathcal{K}_\alpha \mathcal{J}_{\alpha'}(\sigma_\omega(s')) + \textnormal{h.c}$ where $\mathcal{K}_\alpha(X) = [L_\alpha^\dagger, L_\alpha X]$, $\mathcal{J}_{\alpha'}(X) = [h_{\alpha'}, X]$ and $\sigma_\omega(s') = -i\text{Tr}_\mathcal{A}(\rho_{\omega,\delta}(s'/\omega^2))/2$. We note that $\norm{\mathcal{K}_\alpha}_{\diamond}, \norm{\mathcal{J}_{\alpha'}}_\diamond \leq 2$ and that $\norm{\sigma_\omega(s')}_1 \leq 1/2$. Thus, from Eq.~\ref{eq:change_of_variable_restate} and Lemma~\ref{lemma:local_rapid_mixing_superop_bounds}, we obtain that
    \begin{equation}
    \sum_{\alpha, \alpha'}\bigabs{\int_0^t \int_0^{s/\omega^2} \tr{}{O(t - s)\mathcal{Q}_{\alpha, \alpha'}^{(2)}(s')} e^{-2(s/\omega^2 - s')}ds' ds} \leq \zeta^{(3)}(\gamma).
    \label{eq:Q2_lrm_bound_noisy}
    \end{equation}
    For $j = 3$, we have from Lemma~\ref{lemma:remainder_noisy} that $\mathcal{Q}_{\alpha, \alpha'}^{(3)}(s'/\omega^2) = \sum_{u \in \{-, +\}}\mathcal{K}_\alpha^{(u)} \mathcal{J}_{\alpha'}^{(u)}(\sigma_\omega^{(u)}(s'))$ where $\mathcal{K}^{(u)}_\alpha = \mathcal{D}_{L_\alpha}$, $\mathcal{J}^{(u)}_{\alpha'}(X) = [L_{\alpha'}^{(u)}, X]$ and $\sigma_\omega^{(u)}(s') = -i\omega^{-1}\text{Tr}_\mathcal{A}(\sigma_{\alpha'}^{(\bar{u})}\rho_{\omega,\delta}(s'/\omega^2))$. We note that $\norm{\mathcal{K}^{(u)}_\alpha}_{\diamond}, \norm{\mathcal{J}^{(u)}_{\alpha'}}_\diamond \leq 2$ and, using Lemma~\ref{lemma:tr_sigma_bounds_noisy}, $\norm{\sigma^{(\bar{u})}_\omega(s')}_1 \leq 1/2+\mathcal{Z'}\delta/\omega$. Thus, from Lemma~\ref{lemma:local_rapid_mixing_superop_bounds} and Eq.~\ref{eq:change_of_variable_restate}, we obtain that
    \begin{equation}
    \sum_{\alpha, \alpha'}\bigabs{\int_0^t \int_0^{s/\omega^2} \tr{}{O(t - s)\mathcal{Q}_{\alpha, \alpha'}^{(3)}(s')} e^{-2(s/\omega^2 - s')}ds' ds} \leq \left(1+2\frac{\mathcal{Z}'\delta}{\omega}\right)\zeta^{(3)}(\gamma).
    \label{eq:Q3_lrm_bound_noisy}
    \end{equation}
    For $j = 4$, we have from Lemma~\ref{lemma:remainder_noisy} that $\mathcal{Q}_{\alpha, \alpha'}^{(4)}(s'/\omega^2) = \sum_{u, u' \in \{-, +\}}\mathcal{K}_\alpha \mathcal{J}_{\alpha'}^{(u,u')}(\sigma_{\omega}^{(u,u')}(s')) + \textnormal{h.c.}$ according to the following definitions: We define $\mathcal{K}_\alpha(X)=[L^\dagger_\alpha,X]$.  $J_{\alpha'}^{(u,u')}(X)=L^{(u)}X$ if $u'=-$  and $J_{\alpha'}^{(u,u')}(X)=XL^{(u)}$ if $u'=+$. Finally, $\sigma_{\omega}^{(u,u')}(s')=\omega^{-2}\tr{\mathcal{A}}{\sigma_{\alpha}^{(\bar{u})}\sigma_{\alpha}\rho_{\omega,\delta}(s'/\omega^2)}$ if $\alpha=\alpha'$ and $\sigma_{\omega}^{(u,u')}(s')=u'\omega^{-2}\tr{\mathcal{A}}{\sigma_{\alpha'}^{(\bar{u})}\sigma_{\alpha}\rho_{\omega,\delta}(s'/\omega^2)}$ if $\alpha\neq\alpha'$. We note that for any $u,u'$, $\norm{\mathcal{K}_\alpha}_\diamond, \norm{\mathcal{J_{\alpha'}^{(u,u')}}}_\diamond\leq 2$ and, using Lemma~\ref{lemma:tr_sigma_bounds_noisy}, $\norm{\sigma_\omega^{(u,u')}(s')}_1 \leq 1/4+\omega^{-1}\mathcal{Z}'\delta/2+\omega^{-2}\mathcal{Z}'\delta$. Thus, applying Eq.~\ref{eq:change_of_variable_restate} and Lemma~\ref{lemma:local_rapid_mixing_superop_bounds}(c), we obtain that
    \begin{equation}
        \sum_{\alpha, \alpha'}\bigabs{\int_0^t \int_0^{s/\omega^2} \tr{}{O(t - s)\mathcal{Q}_{\alpha, \alpha'}^{(4)}(s')} e^{-2(s/\omega^2 - s')}ds' ds} \leq \left(2+4\frac{\mathcal{Z}'\delta}{\omega}+8\frac{\mathcal{Z}'\delta}{\omega^2}\right) \zeta^{(3)}(\gamma).
        \label{eq:Q4_lrm_bound_noisy}
    \end{equation}
    Considering Eqs.~\ref{eq:Q1_lrm_bound_noisy}-\ref{eq:Q4_lrm_bound_noisy}, we may choose $\lambda^{(2)}(\gamma)=2\zeta^{(3)}(\gamma)$ to satisfy the lemma statement. Note that $\zeta^{(3)}(\gamma)$ depends on $a$ and $\mathcal{Z}$ and that by Lemma~\ref{lemma:local_rapid_mixing_superop_bounds}, so long as $k(\abs{S_O},\gamma) \leq O\left(\exp(\gamma^{-\kappa})\right)$ as $\gamma \to 0$, then $\zeta^{(3)}(\gamma)\leq O\left(\gamma^{-(2d+1)(\kappa+1)}\right)$ as $\gamma \to 0$.
    
    Next, using the definition of $\mathcal{K}^{(0)}_\beta(t)$ from Lemma~\ref{lemma:remainder_noisy} and the fact that $O(t-s)$ is not supported on the ancillae, we have that $\tr{}{O(t-s)\mathcal{K}_\beta^{(0)}\left(s/\omega^2\right)}=\tr{}{O(t-s)\mathcal{N}_\beta(\sigma_{\omega}(s))}$ where $\sigma_{\omega}(s) = \rho_{\omega,\delta}(s/\omega^2)$. Noting that $\norm{\mathcal{N}_\beta}_\diamond\leq2$ and that $\norm{\sigma_\omega(s)}_1\leq 1$, we apply Lemma~\ref{lemma:local_rapid_mixing_superop_bounds}(b) to obtain
    \begin{align}
        \sum_\beta \bigabs{\int_0^t \tr{}{O(t-s)\mathcal{K}_\beta^{(0)}\left(s/\omega^2\right)}ds} \leq \zeta^{(2)}(\gamma),
    \end{align}
    where in the application of Lemma~\ref{lemma:local_rapid_mixing_superop_bounds} we consider $d(S_O, S'_\beta)=d(S_O,\tilde{S}'_\beta)$. Picking $\lambda'^{(1)}(\gamma)=\zeta^{(2)}(\gamma)$ satisfies the lemma statement. Note that $\zeta^{(2)}(\gamma)$ depends on $\mathcal{Z}$, $\mathcal{Z}'$, and $a$.

    In order to prove the last statement in this lemma, we must define a way to sum over all regions $S_\alpha$ and $S'_\beta$ in a single summation. For $c\in\{1,2,...,M+M'\}$, we define $S''_c=S_c$ if $c\leq M$ and $S''_c=S'_{c-M}$ if $c>M$. Remember that $M$ is the number of target Lindbladian terms $\mathcal{L}_\alpha$ and that $M'$ is the number of error Lindbladian terms $\mathcal{N}_\beta$. Given any $c$, we define the number other regions $S''_{c'}$ that intersect $S''_c$ to be bounded by some $\mathcal{Z}''\geq0$, i.e. $\abs{\{c' | S''_{c'} \cap S''_c \neq \emptyset\}} \leq \mathcal{Z}''$. If $c>M$, then $S''_c=S'_{c-M}$ can intersect a maximum of $\mathcal{Z}'$ other regions $S''_{c'}$. If $c \leq M$, then $S''_c=S_c$ can intersect a maximum of $\mathcal{Z}$ other regions $S_{\alpha}$ and each of the $\abs{S_c}$ points in $S_c$ can be in the support of a maximum of $\mathcal{Z}'$ regions $S''_\beta$. We can bound $\max_{\alpha}\abs{S_\alpha}$ by some $V_S$ and note that $V_S$ is a function of only $a$ and $d$. Hence we can choose $\mathcal{Z}'' = \max(\mathcal{Z'}, \mathcal{Z}+V_S\mathcal{Z}')$.
    
    Using these definitions, from Lemma~\ref{lemma:remainder_noisy}, we have that $\sum_{\alpha,\beta}\tr{}{O(t-s)\mathcal{K}^{(1)}_{\alpha,\beta}(s'/\omega^2)} = \sum_{c,c'} \tr{}{O(t-s)\mathcal{K}_c\mathcal{J}_{c'}(\sigma_\omega(s'))} + \textnormal{h.c.}$ where $\mathcal{K}_\alpha(X) = -i[L_c^\dagger, \sigma_c X]$ if $c\leq M$ and $0$ otherwise, $\mathcal{J}_{c'}=\mathcal{N}_{c'}$ if $c' > M$ and $0$ otherwise, and $\sigma_\omega(s')=\rho_{\omega,\delta}(s'/\omega^2)$. We note that $\norm{\mathcal{K}_c}_\diamond, \norm{\mathcal{J}_{c'}}_\diamond \leq 2$ and $\norm{\sigma_\omega(s')}_1 \leq 1$. Thus, applying Eq.~\ref{eq:change_of_variable_restate} and Lemma~\ref{lemma:local_rapid_mixing_superop_bounds}(c), we obtain that
    \begin{align}
        \sum_{\alpha, \beta}\bigabs{\int_0^t \int_0^{s/\omega^2} e^{-2(s/\omega^2 - s')} \textnormal{Tr}\big(O(t - s) \mathcal{K}_{\alpha,\beta}^{(1)}(s') \big) ds' ds} \leq 2{\zeta'}^{(3)}(\gamma).
    \end{align}
    
    Using the definition of $\mathcal{K}^{(2)}_{\alpha,\beta}$ in Lemma~\ref{lemma:remainder_noisy}, we have that $\sum_{\alpha,\beta}\tr{}{O(t-s)K^{(2)}_{\alpha,\beta}(s'/\omega^2)} = \sum_{c,c'}\tr{}{O(t-s)\mathcal{K}_c\mathcal{J}_{c'}(\sigma_\omega(s'))} + \textnormal{h.c.}$
    where $\mathcal{K}_c(X) = [L_c^\dagger, L_c X]$ if $c\leq M$ and $0$ otherwise, $\mathcal{J}_{c'}=\mathcal{N}_{c'}$ if $c' > M$ and $0$ otherwise, and $\sigma_\omega(s')=\rho_{\omega,\delta}(s'/\omega^2)/2$. We note that $\norm{\mathcal{K}_c}_\diamond, \norm{\mathcal{J}_{c'}}_\diamond \leq 2$ and $\norm{\sigma_\omega(s')}_1 \leq 1/2$. Thus, applying Eq.~\ref{eq:change_of_variable_restate} and Lemma~\ref{lemma:local_rapid_mixing_superop_bounds}(c), we obtain that
    \begin{align}
        \sum_{\alpha, \beta}\bigabs{\int_0^t \int_0^{s/\omega^2} e^{-2(s/\omega^2 - s')} \textnormal{Tr}\big(O(t - s) \mathcal{K}_{\alpha,\beta}^{(2)}(s') \big) ds' ds} \leq {\zeta'}^{(3)}(\gamma).
    \end{align}
    By Lemma~\ref{lemma:local_rapid_mixing_superop_bounds}, $\zeta^{(3)}(\gamma)\leq O\left(\gamma^{-(2d+1)(\kappa+1)}\right)$ as $\gamma \to 0$ and $\zeta^{(3)}(\gamma)$ depends on $\mathcal{Z}$, $a$, and  $\mathcal{Z}''$. Picking $\lambda'^{(2)}(\gamma)=2\zeta^{(3)}(\gamma)$ satisfies the lemma statement.
\end{proof}
\end{document}